\documentclass[11pt]{article}
\usepackage{jheppub}
\usepackage{amsmath,amssymb,amsfonts,amsthm,graphicx,physics}
\usepackage{setspace}
\usepackage{subcaption}
\usepackage{dsdshorthand}
\usepackage{leftindex}
\usepackage{bbm}

\usepackage[left=1in,right=1in,top=1in,bottom=1.2in]{geometry}
\usepackage{xcolor}
\usepackage{simpler-wick}

\usepackage{tikz}
    \usepackage{amssymb,amsfonts,amsmath}
    \usepackage{tkz-euclide}
        \usetikzlibrary{arrows,calc,patterns}
\usepackage{pgfplots}

\definecolor{darkblue}{rgb}{0.1,0.1,.7}
\definecolor{purple}{rgb}{0.6,0,0.6}
\definecolor{orange}{rgb}{0.9,0.6,0}
\definecolor{llgray}{rgb}{0.9,0.9,1}
\definecolor{dgreen}{rgb}{0,0.5,0}
\definecolor{outsideyellow}{rgb}{0.96,0.86,0.41}
\definecolor{insideyellow}{rgb}{0.99,0.96,0.82}
\definecolor{observerred}{rgb}{0.93,0.27,0.13}
\definecolor{mattergreen}{rgb}{0,0.63,0.29}
\usepackage[]{latexsym}
\usepackage[utf8]{inputenc}
\usepackage{geometry}
\usepackage{amscd}
\usepackage[all,cmtip]{xy}
\usepackage{mathrsfs}

\usepackage[margin=10pt,font=small,labelfont=bf]{caption}
\usepackage{dsdshorthand}
\usepackage{changepage}
\usepackage{setspace}
\usepackage{tikz}
\usepackage{subcaption}
\usepackage{ marvosym }
\usepackage{tensor}
\usepackage{CJKutf8}
\usepackage{empheq}

\title{\begin{center}
    The gravitational path integral\\ from  an observer's point of view
\end{center}} 
\author[1]{Ahmed I. Abdalla,}
\author[1]{Stefano Antonini,}
\author[1]{Luca V.~Iliesiu,}
\author[2]{Adam Levine}

\def\Hp{\cH_{\text{pert}}}
\def\Hnonp{\cH_{\text{non-pert}}}
\def\Hrelp{\cH_{\text{pert}}^{\text{rel}}}
\def\Hr{\cH_{\text{non-pert}}^{\text{rel}}}
\def\HO{\cH^\text{obs}}
\def\Oo{\mathcal{O}_O}
\def\Hrl{\cH_{\text{non-pert,l}}^{\text{rel}}}
\def\Hrr{\cH_{\text{non-pert,r}}^{\text{rel}}}
\def\HrL{\cH_{\text{non-pert,L}}^{\text{rel}}}
\def\HrR{\cH_{\text{non-pert,R}}^{\text{rel}}}
\def\OR{\cO_O}

\DeclareMathOperator*{\sumint}{%
\mathchoice%
  {\ooalign{$\displaystyle\sum$\cr\hidewidth$\displaystyle\int$\hidewidth\cr}}
  {\ooalign{\raisebox{.14\height}{\scalebox{.7}{$\textstyle\sum$}}\cr\hidewidth$\textstyle\int$\hidewidth\cr}}
  {\ooalign{\raisebox{.2\height}{\scalebox{.6}{$\scriptstyle\sum$}}\cr$\scriptstyle\int$\cr}}
  {\ooalign{\raisebox{.2\height}{\scalebox{.6}{$\scriptstyle\sum$}}\cr$\scriptstyle\int$\cr}}
}
\newcommand{\inlinefig}[2][10]{
    \raisebox{-0.5\totalheight}{\includegraphics[height=#1\fontcharht\font`\B]{#2}}
}

\affiliation[1]{Center for Theoretical Physics and Department of Physics, University of California, Berkeley, California 94720, U.S.A.}

\affiliation[2]{Center for Theoretical Physics, Massachusetts Institute of Technology, \\Cambridge, MA 02139, USA}

\emailAdd{aiabdalla@berkeley.edu}
\emailAdd{santonini@berkeley.edu}
\emailAdd{liliesiu@berkeley.edu}
\emailAdd{arlevine@mit.edu}

\abstract{

One of the fundamental problems in quantum gravity is to describe the experience of a gravitating observer in generic spacetimes. In this paper, we develop a framework 
for describing non-perturbative physics relative to an observer using the gravitational path integral.
We apply our proposal to an observer that lives in a closed universe and one that falls behind a black hole horizon. We find that the Hilbert space that describes the experience of the observer is much larger than the Hilbert space in the absence of an observer. In the case of closed universes, the Hilbert space is not one-dimensional, as calculations in the absence of the observer suggest. Rather, its dimension scales exponentially with $G_N^{-1}$. Similarly, from an observer’s perspective, the dimension of the Hilbert space in a two-sided black hole is increased. We compute various observables probing the experience of a gravitating observer in this Hilbert space. We find that an observer experiences non-trivial physics in the closed universe in contrast to what it would see in a one-dimensional Hilbert space. In the two-sided black hole setting, our proposal implies that non-perturbative corrections to effective field theory for an infalling observer are suppressed until times exponential in the black hole entropy, resolving a recently-raised puzzle in black hole physics. While the framework that we develop is exemplified in the toy-model of JT gravity, 
most of our analysis can be extended to higher dimensions
and, in particular, to generic spacetimes not admitting a conventional holographic description, 
such as cosmological universes or black hole interiors.

}

\date{\today}

\begin{document}

\maketitle

\parskip=3pt

\section{Introduction}
\label{sec:intro}

The description of cosmology and black hole interiors remains one of the most important open problems in quantum gravity. The last few decades have seen significant progress in our understanding of gravitational holography \cite{Susskind:1994vu,tHooft:1993dmi,Maldacena:1997re,Witten:1998qj,Aharony:1999ti}
and the gravitational path integral \cite{gibbons,Hartle:1983ai,Maldacena:2004rf,Lewkowycz:2013nqa,Faulkner:2013ana,Marolf:2020xie,Saad:2018bqo,Saad:2019lba,Saad:2021rcu,Saad:2021uzi,Penington:2019kki,Almheiri:2019qdq,Blommaert:2021fob,Blommaert:2022ucs,Boruch:2023trc,Iliesiu:2021are,Iliesiu:2022kny,Marolf:2022ybi,Colafranceschi:2023moh,Marolf:2024jze,Balasubramanian:2022gmo,Balasubramanian:2022lnw,Sasieta:2022ksu}.
These techniques are especially suited for answering global questions that can be clearly formulated from a boundary perspective. Examples include correlation functions between boundary points and entanglement entropies of subsystems in the dual theory \cite{Ryu2006a,Ryu2006b,Hubeny:2007xt,Engelhardt:2014gca}. On the other hand, it is often unclear how to describe the local experience of an observer that, like us, is part of the gravitational spacetime. This is particularly true for observers falling into black holes \cite{Almheiri:2012rt,Susskind:2012rm,Almheiri:2013hfa,Harlow:2013tf,Stanford:2022fdt,IliLev24,Blommaert:2024ftn} or observers in spacetimes that do not have asymptotic boundaries, for example in cosmology \cite{LevSha22, Maldacena:2004rf,McInnes:2004nx, Cooper:2018cmb,Antonini:2019qkt,Marolf:2020xie,VanRaamsdonk:2020tlr,VanRaamsdonk:2021qgv,Antonini:2022blk,Antonini:2022ptt,Antonini:2022fna,Sahu:2023fbx,Chakravarty:2024bna,Antonini:2024bbm,Betzios:2024oli,VanRaamsdonk:2024sdp,Antonini:2024mci,Sahu:2024ccg,Chandrasekaran:2022cip,Strominger:2001pn,Coleman:2021nor,dsds,Araujo-Regado:2022gvw,McFadden:2009fg}. This is an obstacle that we need to overcome if we are to achieve a description of quantum gravity that explains our own experience.

In recent years, progress has been made towards understanding the experience of observers in gravitating spacetimes 
\cite{Anninos:2011af, Anninos:2011zn,Leutheusser:2022bgi,Witten:2021unn,Chandrasekaran:2022eqq,Chandrasekaran:2022cip,Witten:2023qsv,Witten:2023xze,Gesteau:2023hbq,Kolchmeyer:2024fly,Chen:2024rpx,Kudler-Flam:2024psh,DeVuyst:2024pop} 
but these approaches have not incorporated non-perturbative quantum gravity effects into their descriptions of local physics experienced by the observer. 
These effects, which are captured by the gravitational path integral and holography, are critical to calculating various properties of the quantum gravity Hilbert space and the observables defined on it \cite{Penington:2019kki,Almheiri:2019qdq,Saad:2019lba,Stanford:2022fdt,IliLev24,Blommaert:2024ftn,Boruch:2024kvv,Balasubramanian:2022gmo,Balasubramanian:2022lnw,Antonini:2023hdh}. It is plausible that they also play a role in describing physics experienced by an observer, but no concrete proposal exists for a non-perturbative treatment of this issue. The goal of this paper is to present such a proposal. We use the gravitational path integral to determine properties of the non-perturbative quantum gravity Hilbert space needed to describe the experience of a gravitating observer and to study the observables that such an observer could measure.

Our proposal is motivated by several shortcomings of the current approach to the gravitational path integral when describing the experience of an observer. An important example is the physics of closed universes. The inner products between different closed universe states computed by the gravitational path integral have large fluctuations due to non-perturbative effects \cite{Marolf:2020xie,Antonini:2023hdh,Antonini:2024mci,Usatyuk:2024mzs, Usatyuk:2024isz}, signaling a drastic breakdown of effective field theory. Notably, even inner products between states in which an observer exists and states in which it does not are $O(1)$, a fact signaled by large statistical fluctuations in these inner-products. Therefore, the observer itself is not well-defined at all times.  This is troublesome because we need a well-defined notion of an observer not subject to large fluctuations in order to describe physics from its point of view. A further drastic consequence of these large fluctuations is that the quantum gravity Hilbert space for a closed universe is one-dimensional \cite{Marolf:2020xie,Usatyuk:2024mzs,McNamara:2020uza}. This result is problematic if our goal is to describe the experience of an observer living in a closed universe, because a local observer in a closed universe should experience non-trivial physics. In fact, experimental data do not rule out the possibility that we are living in a closed universe \cite{DiValentino:2019qzk,Handley:2019tkm,Planck:2018vyg}, and our local experience should be unaffected by the global properties of the universe. Therefore, an observer should have access to a much larger Hilbert space than the one-dimensional Hilbert space suggested by the gravitational path integral arguments. 

A non-perturbative description of an observer's experience is also necessary to address the long-standing issue regarding the fate of an observer falling into a black hole \cite{Almheiri:2012rt,Susskind:2012rm,Almheiri:2013hfa,Harlow:2013tf,Stanford:2022fdt,IliLev24,Blommaert:2024ftn}. This is particularly important for very old black holes, for which non-perturbative effects are expected to give large corrections to effective field theory.
Various proposals have attempted to quantify how effective field theory fails \cite{Stanford:2022fdt,IliLev24,Blommaert:2024ftn}, but none of them develop the framework that understands the problem in the Hilbert space relevant for describing the experience of the infalling observer.

We propose to compute inner-products and their moments using the gravitational path integral while imposing that the observer of interest always exists between the bra and the ket. Consider an initial and a final slice on which the boundary conditions associated with the bra and the ket are defined. In our proposal, only geometries for which there is a foliation continuously interpolating between these two slices in which the observer is present on every slice contribute to the gravitational path integral. We shall motivate our proposal by showing that, among other benefits, it excludes non-physical contributions to transition probabilities. This prescription allows us to compute transition probabilities in the presence of the observer as well as compute observables dressed to the observer's worldline. In practice, our proposal can be achieved by requiring that, when computing moments of a given overlap, the observer's worldline always connects a given bra to the corresponding ket, see e.g.~Figure \ref{fig:newbc}, whereas other matter worldlines are allowed to connect arbitrary bras and kets. This prescription should be contrasted with the usual rules for the global gravitational path integral, in which worldlines associated with all operator insertions, including those creating and annihilating an observer, can connect between arbitrary bras and kets.\footnote{Notice that this prescription for the gravitational path integral, which we will review in Section \ref{sec:review}, is precisely what causes large fluctuations in the inner product between closed universe states, consequently causing the notion of an observer to be ill-defined.}

\begin{figure}[h]
    \centering
    \includegraphics[width=\linewidth]{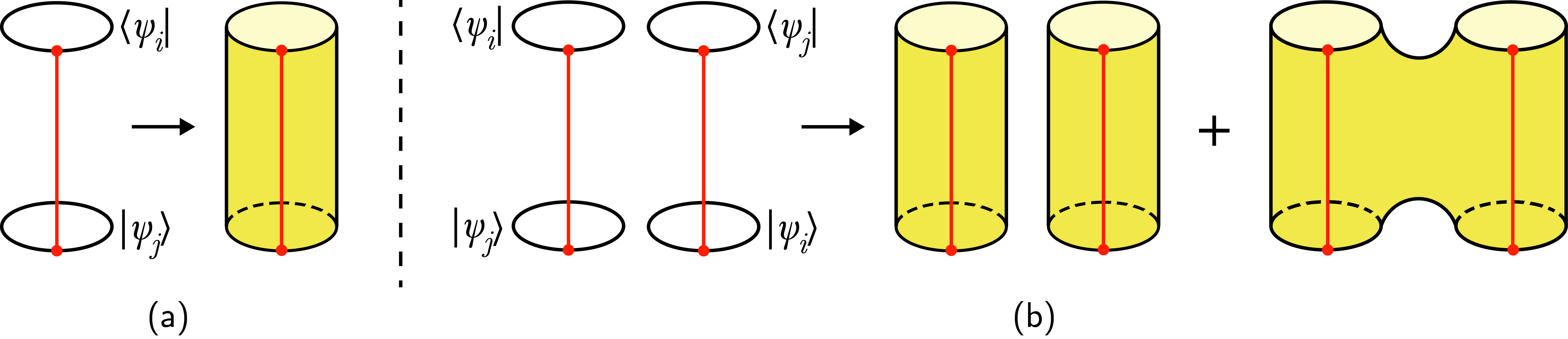}
    \caption{Our proposal for the gravitational path integral in the presence of an observer. The worldline of the observer (depicted in red) must connect a bra and the corresponding ket when computing moments of an overlap $\langle\psi_i|\psi_j\rangle^n$. This rule can be seen as an additional boundary condition to impose when summing over geometries. (a) The overlap $\langle\psi_i|\psi_j\rangle$ between two closed universe states in pure JT gravity is unaffected by our new rules. (b) The square of an overlap $|\langle\psi_i|\psi_j\rangle|^2$ between closed universe states. The gravitational path integral sums over all possible geometries satisfying the boundary conditions at the asymptotic boundaries and at the worldline of the observer. We depict here the leading disconnected contribution and the first subleading term, a genus-zero connected geometry.}
    \label{fig:newbc}
\end{figure}

We emphasize that our proposal is relevant when asking questions related to the experience of a given gravitating observer. We are building Hilbert spaces that describe physics \textit{with respect to} some subsystem. This justifies a different prescription for the observer's worldvolume compared to other operator insertions. We decide
which worldvolumes to treat in this way based on the quantities we are interested in computing. For example, 
to capture the experience of two bulk observers in 2D gravity, both of their worldlines would need to 
connect bras to their corresponding kets. In contrast, if we are interested in global questions that do not depend on the presence of a bulk observer---e.g. the entropy of Hawking radiation of a subsystem of the dual holographic theory, boundary correlation functions, etc.---
we should treat all worldlines and matter fields in the same way and use the usual rules for the gravitational path integral in the absence of an observer. 
This implies that most results obtained in recent years using the non-perturbative gravitational path integral remain unchanged.

As we will see, our proposal succeeds in describing non-trivial physics in a closed universe and in giving new, sensible answers to questions about the experience of observers infalling into black hole horizons. Furthermore, it shifts the attention away from the global properties of the spacetime, rather focusing on the non-perturbative physics experienced by a local, gravitating observer. For these reasons, we believe it represents a new, promising framework for the description of quantum cosmology and more generally, for the description of spacetimes without a conventional holographic description.

\subsection{Outline and summary of results}

In this paper, we focus our attention on two-dimensional Jackiw-Teitelboim (JT) gravity with a negative cosmological constant \cite{Jackiw:1984je,Teitelboim:1983ux,Saad:2018bqo}. This theory is known to be dual to a random matrix integral \cite{Saad:2019lba}. We will comment on the extension of our results to de Sitter-JT gravity and higher dimensions in Section \ref{sec:discussion}. 
This paper is structured as follows.

In Section \ref{sec:review}, we review the relevant properties of AdS-JT gravity (with and without matter) and the gravitational path integral in the absence of an observer. The results we obtain in this framework offer a ``global picture'' for the theory since we do not isolate the subsystem of the observer in our gravitational theory. 
In particular, we discuss the dimensions and properties of the perturbative and non-perturbative quantum gravitational Hilbert spaces in a closed universe and in a two-sided black hole. Readers familiar with JT gravity can skip to Section \ref{sec:reviewclosed}, where we give a new derivation, based on a resolvent calculation, of the fact that the non-perturbative Hilbert space of closed universes $\Hnonp$ is one-dimensional \cite{Marolf:2020xie,Usatyuk:2024mzs,McNamara:2020uza}.

In Section \ref{sec:setup}, we explain our proposal for the modified rules of the gravitational path integral in the presence of an observer.  In our setup, states belong to the tensor product space $\cH^{\text{rel}}\otimes\HO$, where states in $\HO$ describe the observer and states in the relational Hilbert space $\cH^{\text{rel}}$ describe the states on the spatial slices of the gravitational theory with which the observer interacts. We then show how the statistics of overlaps for such states are modified with the new rules in the gravitational path integral, and consequently define the perturbative and non-perturbative relational Hilbert spaces, $\Hrelp$ and $\Hr$,  which are specific instances of $\cH^{\text{rel}}$ with different choices of inner-product.
%
We discuss various possible bases of states for these Hilbert spaces, and explain how to define states labeled by the time read by the observer's clock, with respect to which observables can be dressed.

In Section \ref{sec:hilbert}, we compute the dimension of the Hilbert space $\Hr$ for a closed universe and a two-sided black hole,\footnote{There, we will assume that the gravitational theory is coupled to additional matter sources.} using a resolvent calculation \cite{Penington:2019kki,Hsin:2020mfa,Boruch:2023trc,Boruch:2024kvv,Balasubramanian:2022gmo,Balasubramanian:2022lnw,Antonini:2023hdh} 
and the new rules for the gravitational path integral. We find that the closed universe Hilbert space is non-trivial and has dimension
\be 
\dim\left(\Hr\right)=d^2 \quad \text{ as opposed to } \quad \dim\left(\Hnonp\right)=1
\label{eq:dimension-CU}
\ee 
where $d$ is determined by the density of states in JT gravity and is exponential in $G_N^{-1}$. 
Similarly, the Hilbert space of the two-sided black hole is also larger in the presence of an observer, 
\be 
\dim\left(\Hr\right)=d^4\quad  \text{ as opposed to } \quad \dim\left(\Hnonp\right)=d^2.
\label{eq:dimension-BH}
\ee 
Note that the entropy of the observer is bounded above by $\log(\dim(\HO))$, and, in particular, this limits how much information they can learn about the properties of $\Hr$.\footnote{To quantify the amount of information that the observer can learn about $\Hr$ one can use the maximum possible von Neumann entropy of the observer among all states within $\Hr\otimes \HO$. This is given by $\min(\log\dim(\HO),\ \log\dim (\Hr) )$ with $\dim\left(\Hr\right)$ given in \eqref{eq:dimension-CU} for the closed universe and \eqref{eq:dimension-BH} for the two-sided black hole.   Nevertheless, choosing different states in $\Hr \otimes \HO$ yields different results for the observables that the observer can measure, and, for this reason, it is valuable to be able to understand the full relational Hilbert space that is entangled with the observer.}

To further characterize the properties of $\Hr$, we show (details can be found in Appendix \ref{app:factorisation}) that in both cases, $\Hr$ factorises into a tensor product of two Hilbert spaces that, in two-dimensional gravity, capture the degrees of freedom on the left and right side of the observer (and two additional Hilbert spaces associated with the left and right boundaries in the two-sided black hole cases). 

In Section \ref{sec:nullstates}, we show that the inner product defined by our new rules is positive-semi-definite. Similar to the results of \cite{Penington:2019kki,Almheiri:2019qdq,IliLev24}, we find that the set of perturbatively-defined states is an overcomplete basis for $\Hr$. This implies the presence of null states under the inner product defined by the non-perturbative gravitational path integral. We discuss these null states and comment on the non-isometric map \cite{Akers:2021fut,Akers:2022qdl,Antonini:2024yif} between $\Hrelp$ and $\Hr$.

Sections \ref{sec:hilbert} and \ref{sec:nullstates} make manifest a puzzling feature of quantum gravity: by isolating a subsystem -- the observer -- within the gravitational theory, we consequently restrict the set of geometries that are included in the gravitational path integral, and actually \emph{increase} the dimension of the resulting Hilbert space. This phenomenon is, of course, at the heart of the Page curve calculations \cite{Almheiri:2019qdq, Penington:2019kki, Penington:2019npb, Hartman:2020khs}, but we are seeing it arise again in the context of the observer's Hilbert space, $\Hr$. Importantly, even though one might be tempted to think that the observer's Hilbert space can be isolated via a linear projection onto geometries in which an observer is present, Section \ref{sec:hilbert} shows that this cannot be the case; in fact, the Hilbert space from an observer's point of view is larger whereas a projection always reduces the Hilbert space dimension. Furthermore, these results imply that observables that act linearly (state-independently) on $\Hr$ can only be represented \emph{non-linearly} (state-dependently) on $\Hnonp$ \cite{Papadodimas:2012aq,Papadodimas:2013jku,Akers:2022qdl,Antonini:2024yif}. This is in agreement with arguments that the infalling observer's experience should not be fully describable by linear operators on $\Hnonp$ \cite{MarPol15,Papadodimas:2012aq,Papadodimas:2013jku}.

In Section \ref{sec:observables}, we study several examples of observables defined on $\Hr$ that are relevant to probe the experience of a gravitating observer. In the closed universe setup, we study correlation functions along the observer's worldline and find that they are non-trivial. In particular, away from the cosmological singularities, they do not receive large non-perturbative corrections and, therefore, agree with the perturbative result. This is in contrast with the global perspective, in which the one-dimensional Hilbert space for the closed universe implies that all observables are trivially multiples of the identity, and the perturbative result receives large non-perturbative corrections. We then study several interesting properties of the perturbative correlators and point out that the structure of their divergences could be used to probe the cosmological singularity.
In the two-sided black hole case, we study two observables relevant for describing the physics of an infalling observer: the length of the Einstein-Rosen bridge \cite{Stanford:2022fdt,IliLev24,Chen:2024rpx, StaSus14} and the center-of-mass collision energy of an observer with a shockwave past the horizon  \cite{IliLev24}. Both of these observables were studied in detail in \cite{IliLev24} in the global picture---namely, as observables on $\Hnonp$. For the length, we find that the result obtained in the global picture is substantially unmodified: non-perturbative corrections become important at times $t=O(e^{S_0})$, and cause the length to plateau. For the center-of-mass collision energy, we also find that non-perturbative corrections to the perturbative result are only important at times $t=O(e^{S_0})$. This is in contrast with recent puzzling results in the global picture \cite{IliLev24} (reviewed in Section \ref{sec:casimir}), where non-perturbative effects become important at linear times $t=O(S_0)$, suggesting a breakdown of effective field theory roughly at the Page time. Our result resolves this puzzle and signals that perturbative effective field theory remains valid up to exponential times, assuming our proposal is correct.

In Section \ref{sec:discussion}, we discuss our results, their generalization to higher dimensions, and other future directions. We discuss the relationship between our proposal and other frameworks for the description of physics with respect to a gravitating observer.
We also consider what lessons can be drawn from our results for the definition of a holographic dual theory able to capture the observer's experience and suggest that such a theory should live on the worldline of the observer. Finally, we explain how our results can be extended to describe the experience of an observer in de Sitter spacetime. 

Additional technical details relevant to the main sections of this paper can be found in the Appendices.

\section{Review: non-perturbative effects in quantum gravity}
\label{sec:review}
Non-perturbative corrections to the gravitational path integral have recently proven essential to the description of many quantum gravity phenomena. Described colloquially as ``wormhole corrections" after the work of \cite{Maldacena:2004rf}, these effects have been used to explain the Page curve \cite{Penington:2019kki,Almheiri:2019qdq},  factorization \cite{Saad:2021uzi,Blommaert:2021fob,Iliesiu:2021are} and factorisation \cite{Harlow:2020fpj,Boruch:2024kvv}, and black hole microstate counting \cite{Penington:2019kki,Almheiri:2019qdq,Iliesiu:2022kny,Balasubramanian:2022gmo,Balasubramanian:2022lnw,Boruch:2023trc}. In this Section, we provide a review of how wormholes can be used to build a fully non-perturbative Hilbert space for JT gravity with and without matter and how to compute the dimension of the Hilbert space in a closed universe. As explained in \cite{Marolf:2020xie,Usatyuk:2024mzs,McNamara:2020uza}, we find that the Hilbert space of closed universes is far too small to adequately capture states of interest.

This Section adopts and reviews the global perspective of previous literature. For our proposed modifications to these rules in the presence of an observer, see Section \ref{sec:setup}.

\subsection{The perturbative Hilbert space}
\label{sec:pert_h}

We take our theory to be JT gravity \cite{Jackiw:1984je,Teitelboim:1983ux,Mertens:2022irh} with a negative cosmological constant and matter \cite{Penington:2023dql,Kolchmeyer:2023gwa,Iliesiu:2024cnh}. This is a 2D theory of a metric $g$ coupled to a dilaton $\Phi$ and matter fields $\phi$ on a manifold $M$, possibly with boundary. In Euclidean signature the action takes the form 
\begin{equation}
    I[g,\Phi,\phi]=I_{EH}[g]+I_{\Phi}[g,\Phi]+I_{m}[g,\phi].
    \label{eq:action}
\end{equation}
where
\begin{equation}
    I_{EH}[g]=-S_0\chi(M) \qquad\text{ and }\qquad I_{\Phi}[g,\Phi]=-\frac{1}{2}\left(\int_M \Phi(R+2)+I_{\Phi,\partial M}[h,\Phi_b]\right),
\end{equation}
where $\chi(M)$
is the Euler characteristic,
$h$ the restriction of $g$ to $\partial M$, and $\Phi_b$ the value of $\Phi$ on $\partial M$. The boundary dilaton action depends on the theory in question. The action of the matter fields $I_{m}[g,\phi]$ is left unspecified except to demand that matter does not directly couple to the dilaton. We build up the Hilbert space associated with this action by considering each term separately. First is $I_{EH}$, the purely gravitational part of the action. This term is a topological invariant. For now, we take $S_0$ to be a tunable parameter and consider the limit $S_0\to\infty$. In more general theories of gravity, $S_0$ would be replaced by a dimensionless parameter scaling with $G_N^{-1}$, the inverse of Newton's constant. The resultant quantum mechanics can be made exact to all orders in loop corrections \cite{Yang:2018gdb,Kitaev:2018wpr} but will have topological fluctuations suppressed. We call this the \textit{perturbative} limit of the theory. In the \textit{non-perturbative} limit, we take $S_0$ to be a large but finite parameter. 

The leading contribution to the gravitational path integral with action \eqref{eq:action} has the topology of a disk. Let us consider the dilaton action with respect to this background, starting with the two-sided black hole. We set 
\begin{equation}
       I_{\Phi,\partial M}[h,\Phi_b]=2\int_{\partial M}\Phi_b (K-1),
\end{equation}
choosing Dirichlet boundary conditions for the dilaton and Neumann boundary conditions for the extrinsic curvature. We then have
\begin{equation}
    ds^2=h_{ij}dx^idx^j=\frac{1}{\epsilon^2}d\tau^2 \qquad\text{and}\qquad \Phi|_{\partial M}=\frac{\Phi_b}{\epsilon}
\end{equation} 
with $\tau\sim \tau+\beta$, where $\beta$ is the renormalized asymptotic boundary length, which remains finite as we take $\epsilon\to 0$. We also choose units such that $\Phi_b=1$.
What remains is a two-dimensional phase space. Semiclassically, these parameters are $\ell$, the geodesic length of a Cauchy slice connecting the left and right boundaries, and its conjugate momentum. Time evolution on the left and right boundaries are generated by the same ADM Hamiltonian
\begin{equation}\label{eq:constraint}
    H_L=H_R=-\frac{1}{2}\partial_{\ell}^2+2e^{-\ell}.
\end{equation}
The constraint $H_0=H_L-H_R=0$ is a consequence of the gauged $SL(2,\R)$ symmetry.  The resultant Hilbert space is familiar
\begin{equation}\label{eq:H0}
    \cH_0=\left\{\ket{\ell};\ \ell\in\R\text{ and }\braket{\ell}{\ell'}=\delta(\ell-\ell')\right\}=L^2(\R)
\end{equation}
and is manifestly the space of gauge invariants $H_0\ket{\ell}=0$.

We can build generic wavefunctions for any state $\ket{\psi}\in\cH_0$ by taking overlaps
\begin{equation}
    \psi(\ell)=\braket{\ell}{\psi} \qquad\text{and}\qquad \braket{\psi}{\psi}=\int d\ell\ |\psi(\ell)|^2.
\end{equation}
Let us introduce two other useful spanning sets of $\cH_0$. The states of either set can be considered Hartle-Hawking wavefunctions as they represent semiclassical spacetimes with specific boundary conditions. First is the fixed energy basis\footnote{
    We use $s$ and $E$ interchangeably throughout. The reader is encouraged to consider $s$ as the momentum associated with $E$, but in all cases, we simply have $s=\sqrt{2E}$.
}
\begin{equation}
 \cH_0=\left\{\ket{E}=\ket{s\equiv\sqrt{2E}};\ E\in\R^+\text{ and }\braket{E}{E'}=\frac{\delta(s-s')}{\rho_0(s)}\right\}   
\end{equation} 
defined by wavefunctions 
\begin{equation}\label{eqn:lengthwf}
    \Tilde{\varphi}_s(\ell)\equiv \braket{\ell}{E}=4K_{2is}(4e^{-\ell/2})
\end{equation}
that diagonalize $H_L$ and $H_R$. Here $K_\alpha(x)$ is the modified Bessel function of the second-kind. The relevant density of states is 
\begin{equation}
     \rho_0(s)=\frac{s}{2\pi^2}\sinh(2\pi s) \qquad\text{or}\qquad \rho_0(E)=\frac{\rho_0(s=\sqrt{2E})}{\sqrt{2E}}.
     \label{eq:rhonot}
\end{equation}
Of course, the true density of states is $\rho=e^{S_0}\rho_0$. Throughout this paper, we extract factors of $e^{S_0}$ from the density of states and include them explicitly in our calculations to clarify which terms are relevant at each order in the genus expansion. 

We also consider the set of states $\{\ket{\beta};\ \beta\in\R^+\}$ defined through fixed asymptotic length boundary conditions in the gravitational path integral 
\begin{equation}\label{eq:wvfn_beta_l}
    \varphi_\beta(\ell)\equiv e^{-S_0/2} \braket{\ell}{\beta}\equiv\int ds\rho_0(s)\ e^{-\frac{\beta s^2}{2}}\Tilde{\varphi}_s(\ell)=e^{-S_0}\times\inlinefig{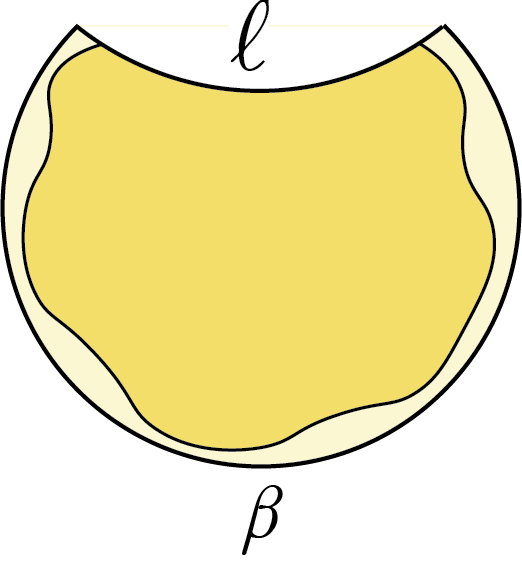}.
\end{equation}
The states in this basis are not orthogonal to each other. Their overlaps compute partition functions
\begin{equation}\label{eq:z}
    Z(\beta=\beta_1+\beta_2)=\braket{\beta_1}{\beta_2}=e^{S_0}\int ds\rho_0(s)\ e^{-\frac{(\beta_1+\beta_2)s^2}{2}} = \inlinefig{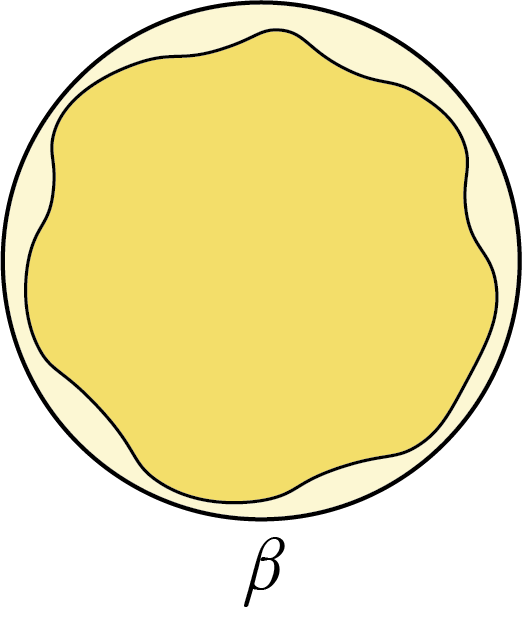}.
\end{equation}
for a spacetime of fixed asymptotic length $\beta=\beta_1+\beta_2$. Note the need to manually reintroduce factors of $e^{S_0}$ because of the convention for $\rho$ established above. The diagrammatic notation in equations \eqref{eq:wvfn_beta_l} and \eqref{eq:z} reflects the fact that the $\{\ket{\ell}\}$ and $\{\ket{\beta}\}$ bases are related by bulk evolution through the path integral. As we will discuss further in Section \ref{sec:setup}, since the Wheeler-DeWitt evolution which governs bulk dynamics is pure gauge \cite{Isham:1992ms}, this amounts to a change of basis rather than actual dynamics. For future reference we also define $Z_n(\beta_1,\cdots,\beta_n)$ as the generalization of $Z$ to a geometry with $n$-boundaries. At the perturbative level this is just the product of $n$-disk partition functions 
\be\begin{aligned}\label{eq:zn_pert}
    Z_{n}(\beta_1,\cdots,\beta_n) &= \inlinefig{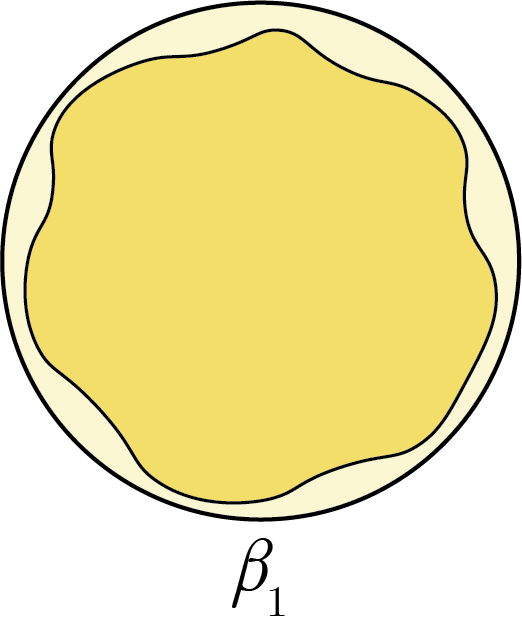}\times\cdots\times\inlinefig{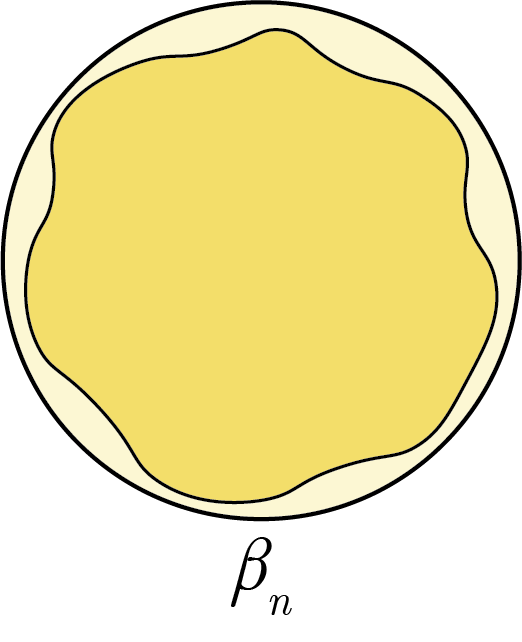} \quad , \\
    &=\prod_{i=1}^n Z(\beta_i).
\end{aligned}\ee

Let us now consider the Hilbert space of closed universes. We work with spacetimes whose spatial slices are compact \cite{Hartle:1983ai,Marolf:2020xie,Usatyuk:2024mzs}. Setting
\begin{equation}
    I_{\Phi,\partial M}[h,\Phi_b]=0
\end{equation}
gives closed Big Bang-Big Crunch cosmologies in Lorentzian signature. In Euclidean signature, this leads to wormholes connecting two asymptotic AdS boundaries, as shown in Figure \ref{fig:closeduniverse}. Again, the phase space is two-dimensional. The Hilbert space is spanned by states of the form
\begin{equation}
    \cH_0=\left\{\ket{b};\ b\in\R^+\text{ and }\braket{b}{b'}\equiv\inlinefig[6]{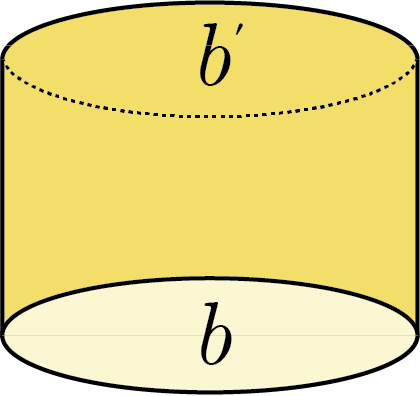}=\frac{1}{b}\delta(b-b')\right\}
\end{equation}
where $b$ is the length of the maximal slice of the closed universe in Lorentzian signature (or, equivalently, of the minimal slice of the wormhole in Euclidean signature). This slice is depicted in dark blue in Figure \ref{fig:closeduniverse}.  The diagrammatic notation captures that the inner product for closed universes is defined through the path integral. We can also define states of fixed asymptotic boundary length in Euclidean signature through the path integral 
\begin{equation}\label{eq:b_beta}
    \varphi_\beta(b)=\braket{b}{\beta}=\inlinefig[6]{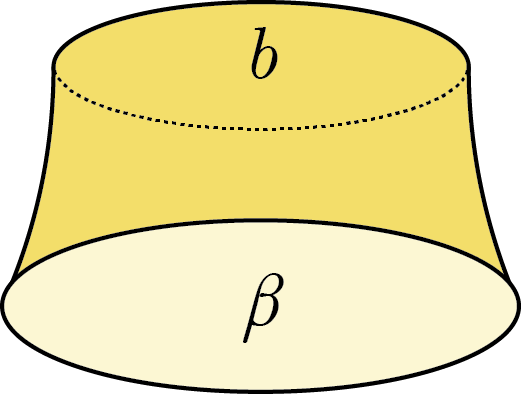}.
\end{equation}
In the absence of matter, such wavefunctions are dominated by vanishingly small values of $b$. 

\begin{figure}
    \centering
    \includegraphics[height=17\fontcharht\font`\B]{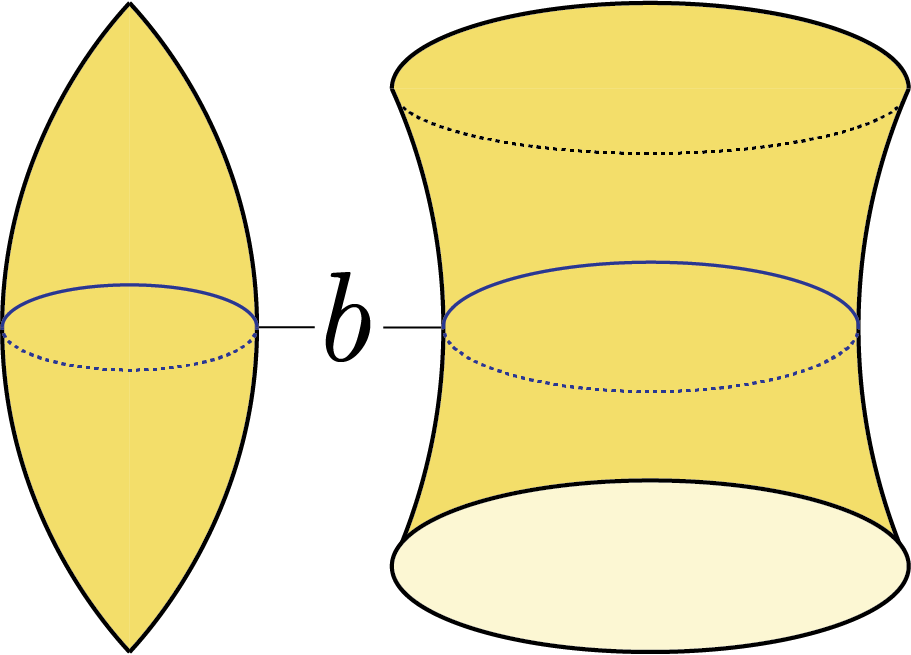}
    \caption{(Left) A Lorentzian Big Bang-Big Crunch closed universe with a Cauchy slice of maximal length $b$ (depicted in blue). (Right) Its Euclidean counterpart has two AdS asymptotic boundaries connected by a Euclidean wormhole. The geodesic slice of minimal length $b$ is identical to the maximal slice in Lorentzian signature. This length is vanishingly small on-shell when the wormhole is not supported by matter and is finite in the presence of matter.}
    \label{fig:closeduniverse}
\end{figure}

So far, we have analyzed pure JT gravity. Now we add matter starting with the two-sided black hole as the background. Semiclassically, we quantize the matter fields on a Cauchy slice of fixed length $\ell$. These fields come in representations of the background spacetime isometries, which in this case are given by the group $SL(2,\R)$. For now, we will suppress internal matter degrees of freedom such that it suffices to consider
\begin{equation}
    \cH_{m}=\C\oplus\bigoplus_{\Delta}\mathcal{D}^+_\Delta.
\end{equation}
Here $\C$ refers to the trivial representation of $SL(2,\R)$ and each $\mathcal{D}^+_\Delta$ corresponds to the block descendant from a ``primary" operator $\mathcal{O}_\Delta$ of scaling dimension $\Delta$ living on the boundary.\footnote{
    A representation theorist would call $\Delta$ a weight and $\mathcal{D}^+_\Delta$ the set of roots descendant from this weight. To a quantum mechanic, $\mathcal{D}^+_\Delta$ is simply the set of states with some total angular moment $\Delta$. Today, we are geometers.
} We ignore the trivial representation throughout this work. Let $H$, $P$, and $B$ be the generators of global time translations, spatial translations, and boosts in $SL(2,\R)$.\footnote{
    In terms of angular momentum operators, we have
    \begin{equation}
        H=iL_z \qquad P=\frac{L_--L_+}{2} \qquad B=\frac{iL_--iL_+}{2}
    \end{equation}
} The Casimir determining the representation is
\begin{equation}
    C_{\Delta} =H^2-P^2-B^2=\Delta(\Delta-1).
    \label{eq:reviewcasimir}
\end{equation}
Alternatively, $\Delta$ is the smallest $H$ eigenvalue in $\mathcal{D}^+_\Delta$. We choose the lowest weight representation so that the total energy is bounded from below.

Let us label different states in a representation $\mathcal{D}^+_\Delta$ by an integer $m$ corresponding to their $H$ eigenvalue. Since the dilaton does not interact with the matter sector, the new Hilbert space is just a tensor product 
\begin{equation}
\label{eq:matter_constraint}
    \Hp=\cH_{m}\otimes\cH_0 =\{\ket{\Delta;\ell,m};\ \braket{\Delta;\ell,m}{\Delta';\ell',m'}=\delta_{\Delta,\Delta'}\delta(\ell-\ell')\delta_{m,m'}\}.
\end{equation}
The constraints are also changed by the presence of matter. For one thing, it is no longer true that the difference of boundary ADM Hamiltonians annihilates physical states. Instead \cite{Iliesiu:2024cnh}
\begin{equation}
\begin{aligned}
    H_L &= \frac{1}{2}(-i\partial_\ell+\frac{1}{2}P)^2+(H-B)e^{-\ell/2}+2e^{-\ell},\\
    H_R &= \frac{1}{2}(-i\partial_\ell -\frac{1}{2}P)^2+(H+B)e^{-\ell/2}+2e^{-\ell},
\end{aligned}
\end{equation}
and the new $SL(2,\R)$ constraints depend on the action of the symmetries on the matter fields. One way to implement these constraints is to construct $\Hp$ as the space of co-invariants drawn from the tensor product of $\cH_m$ and the unconstrained pure JT Hilbert space \cite{Penington:2023dql,Held:2024rmg}. 

It turns out that we can rewrite \cite{Kolchmeyer:2023gwa} \begin{equation}
    \Hp = L^2(\R^+)_0\oplus\bigoplus_{\Delta}L^2(\R^+)_L\otimes L^2(\R^+)_R.    
\end{equation}
The first $L^2(\R^+)_0$ is the vacuum sector of the theory where we have no matter insertions. It is precisely $\cH_0$ when states are labeled by their energies. We'll be ignoring this sector so that we can focus our attention on the quantum matter. The left and right Hamiltonians defined above give a complete set of mutually commuting operators in a given representation
\begin{equation}
    [C,H_L]=[C,H_R]=[H_L,H_R]=0.
\end{equation} 
The direct sum is over representations where $L^2(\R^+)_L$ and $L^2(\R^+)_R$ refer to the energy on the left and right boundaries, respectively. In analogy with the vacuum sector, we label these states as
\begin{equation}\begin{split}\label{eq:hpert_matter}
    \Hp=\bigg\{\ket{\Delta;E_L,E_R}&=\ket{\Delta;s_L\equiv\sqrt{2E_L},s_R\equiv\sqrt{2E_R}}; \\
    &\braket{\Delta;E_L,E_R}{\Delta';E_L',E_R'}=\delta_{\Delta,\Delta'}\frac{\delta(E_L-E_L')}{\rho_0(E_L)}\frac{\delta(E_R-E_R')}{\rho_0(E_R)}\gamma_\Delta(E_L,E_R)\bigg\}
\end{split}\end{equation}
where the Hilbert space is defined up to a choice of normalization $\gamma_\Delta$ for our states. We choose the normalization such that overlaps calculate physically meaningful observables
\begin{equation}\label{eq:gamma}
    \gamma_\Delta(E_L,E_R)=\prod_{\pm,\pm}\frac{\Gamma(\Delta\pm i\sqrt{2E_L}\pm i\sqrt{2E_R})}{2^{2\Delta-1}\Gamma(2\Delta)}=\mel{E_L}{e^{-\Delta\ell}}{E_R}.
\end{equation}
In the first equality, we take a product over the four choices of $\pm$, and in the second, we calculate the overlap in $\cH_0$. We define a $\Delta$-sector specific wavefunction $\Tilde{\varphi}^\Delta_{E_L,E_R}(\ell,m)$ to act as the change-of-basis metric between our two representations of $\Hp$
\begin{equation}\label{eq:matter_basis_change}
    \ket{\Delta;E_L,E_R}=\sumint_{m}d\ell\  \Tilde{\varphi}^\Delta_{E_L,E_R}(\ell,m)\ket{\Delta;\ell,m}.
\end{equation} 

As in the case without matter, we can define a non-orthogonal set of Laplace conjugate states
\begin{equation}
    \ket{\Delta;\beta_L,\beta_R}=e^{S_0/2}\int ds_Lds_R\rho_0(s_L)\rho_0(s_R)\ e^{-\frac{\beta_Ls^2_L}{2}-\frac{\beta_Rs_R^2}{2}}\ket{\Delta;s_L,s_R}.
\end{equation}
With our choice of normalization, overlaps compute thermal two-point functions \cite{Mertens:2022irh}. That is 
\be\begin{aligned}\label{eq:pert_ip}
    & \braket{\Delta_i;\beta_{L}^{(1)},\beta_{R}^{(1)}}{\Delta_j;\beta_{L}^{(2)},\beta_{R}^{(2)}} = \delta_{\Delta_i,\Delta_j}\inlinefig{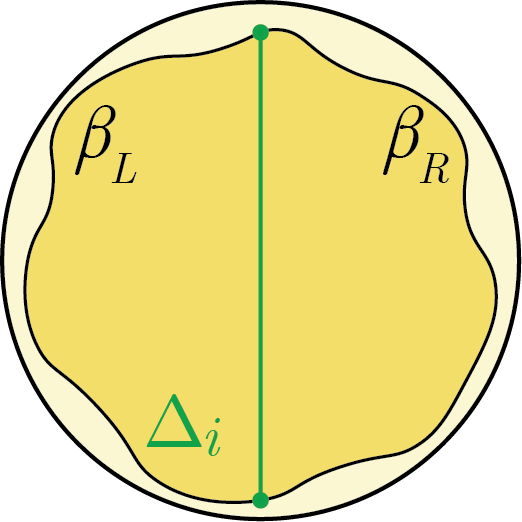}\\
    &=\delta_{\Delta_i,\Delta_j} e^{S_0}\int ds_Lds_R\rho_0(s_L)\rho_0(s_R)e^{-\frac{\beta_L s_L^2}{2}}e^{-\frac{\beta_Rs_R^2}{2}}\gamma_{\Delta}(s_L,s_R).
\end{aligned}\ee
where $\beta_L=\beta_{L}^{(1)}+\beta_{L}^{(2)}$ and $\beta_R=\beta_{R}^{(1)}+\beta_{R}^{(2)}$.
Alternatively, we can define 
\begin{equation}
    \braket{\Delta_i;\beta_{L}^{(1)},\beta_{R}^{(1)}}{\Delta_j;\beta_{L}^{(2)},\beta_{R}^{(2)}}=\delta_{\Delta_i,\Delta_j}\mel{\mathcal{O}_\Delta}{e^{-\beta_LH_L-\beta_RH_R}}{\mathcal{O}_\Delta}
\end{equation}
where $\ket{\mathcal{O}_\Delta}=\ket{\Delta;\beta_L=\beta_R=0}$ is prepared by acting on the maximally entangled state with a primary $\mathcal{O}_\Delta$.\footnote{
    This state is not normalizable unless we restrict ourselves to a finite energy window. States of this form will be useful when defining inner products in matrix theory, and we revisit them in Section \ref{sec:nullstates}.
} We can generalize the $n$-boundary partition function to the inclusion of matter. This is a multi-trace operator in the sense of \cite{Saad:2019lba} 
\be
\begin{aligned}\label{eq:zn_matter_pert}
    &Z_{n}^{i_1,i_1',\cdots,i_n,i_n'}(\beta_{1,L},\beta_{1,R},\cdots,\beta_{n,L},\beta_{n,R}) = \delta_{\Delta_{i_1},\Delta_{i_1'}}\inlinefig{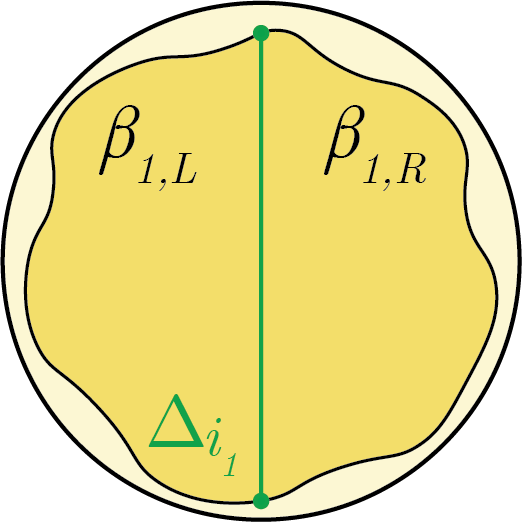}\times\cdots\times\delta_{\Delta_{i_n},\Delta_{i_n'}}\inlinefig{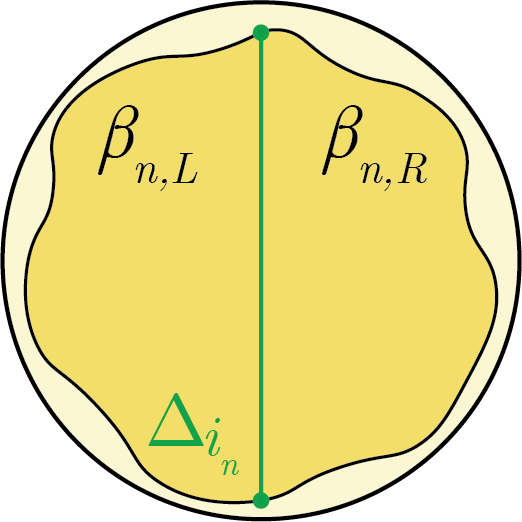},\\
    &=\prod_{k=1}^n \left[ \delta_{\Delta_{i_k},\Delta_{i_k'}} e^{S_0}\int ds_Lds_R\rho_0(s_L)\rho_0(s_R)e^{-\frac{\beta_{k,L} s_L^2}{2}}e^{-\frac{\beta_{k,R}s_R^2}{2}}\gamma_{\Delta_{i_k}}(s_L,s_R)\right].
\end{aligned}
\ee
The superscripts indicate the matter flavor indices $i_1,i_1',\cdots,i_n,i_n'$ and $\beta_{k,L}$, $\beta_{k,R}$ are the total left and right boundary lengths for each overlap. Here, we associate the flavor indices $i_k$ with different scaling dimensions $\Delta_{i_k}$ of the matter operators, but they could alternatively be associated with an internal symmetry of the matter sector. We choose matter fields whose one-point function vanishes so that we are forced to contract matter indices in all diagrams. Again, at the perturbative level, this is just the product of single-trace operators.

The analysis proceeds similarly for closed universes. We start by including matter insertions in the preparation of the state of the closed universe. Like the two-sided black hole, the new Hilbert space is a tensor product 
\begin{equation}
    \Hp=\cH_m\otimes\cH_0=\left\{\ket{\Delta;\beta};\ \braket{\Delta;\beta}{\Delta';\beta'}=\delta_{\Delta,\Delta'}\inlinefig{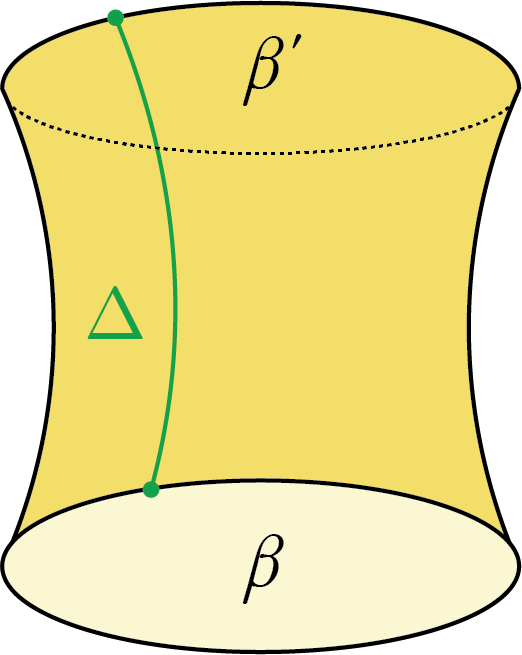}\right\}
\end{equation}
where, again, the inner product is defined through the path integral. Closed universe states can be labeled by a matter scaling dimension $\Delta$ and the asymptotic length of the boundary $\beta$. 
These states are normalized such that their overlaps compute the double-trace operator
\be\label{eq:z_closed}
    \braket{\Delta;\beta}{\Delta;\beta} =(\Tr e^{-\beta H}\mathcal{O}_\Delta )^2.
\ee
where $H$ generates asymptotic boundary time evolution. Unlike pure JT gravity, the Euclidean wormholes associated with closed universes are stabilized in the presence of matter, and wavefunctions are dominated by finite values of $b$ \cite{Usatyuk:2024mzs}. We discuss how these states can be represented on more general Cauchy slices (other than asymptotic boundaries) in Section \ref{sec:setup}.

We can explicitly compute these overlaps by slicing the path integrals open along the matter insertion. With fixed energy boundary conditions, a geodesic worldline ending on the insertion of a dimension $\Delta$ matter operator contributes a factor $e^{-\Delta\ell/2}\Tilde{\varphi}_s(\ell)$ to the path integral. 
We consider
\be\begin{aligned}
    \braket{\Delta;\beta'}{\Delta;\beta} &= \int d\ell e^{-\Delta\ell}\ \inlinefig{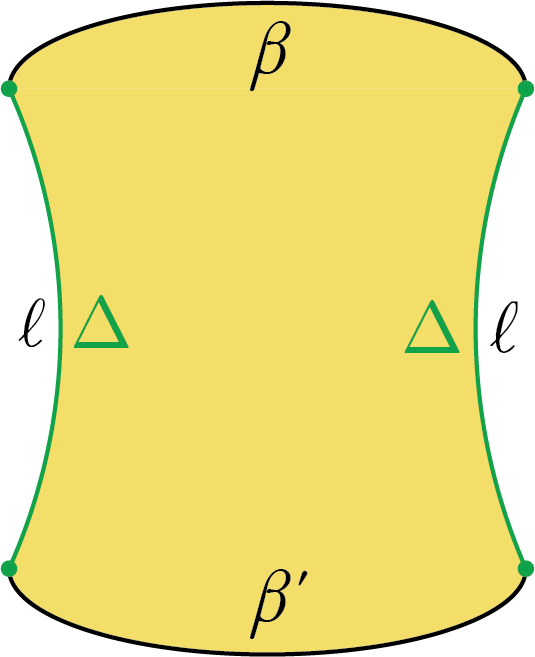} \\
    &= \int ds\rho_0(s)d\ell e^{-\Delta\ell}\  \Tilde{\varphi}_s(\ell)e^{-\beta s^2/2}\Tilde{\varphi}_s(\ell)e^{-\beta' s^2/2},
\end{aligned}\ee
where we have treated the exponential factors as a modification to the $\ell$-measure. In the first line, we cut open the path integral of equation \eqref{eq:z_closed} along the matter geodesic, and, in the second line, we reinterpreted the resulting diagram as an integral. We could alternatively depict the integral as a trace over the energy to emphasize how the boundary conditions form a closed cycle. This rewriting is conducive to the dual matrix integral \cite{Saad:2019lba,JafKol22} where the parameter $s$ is replaced by a random matrix.  We discuss the semiclassical interpretation of the matrix model below. 

\subsection{Non-perturbative effects}\label{sec:non-pert_h}

The perturbative Hilbert spaces we have discussed so far are infinite-dimensional. To see how this might change in the non-perturbative case, we consider an analogy inspired by \cite{Akers:2022qdl}. Suppose we are given a set of unit vectors $\{\ket{v_i}\}_{i=1}^K$ randomly pulled from an inner product space $\Hnonp$. The span of these vectors determines a new Hilbert space $\cH(K)$ 
with inner product inherited from $\Hnonp$. We are tasked with uncovering the size of $\Hnonp$ by conducting experiments in $\cH(K)$. If $\text{dim}(\Hnonp)\gg1$ then any two vectors are very likely orthogonal. That is 
\begin{equation}
    \overline{\braket{v_i}{v_j}}=\delta_{ij}.
\end{equation}
The overline $\overline{\cdot}$ denotes expectation values taken over many realizations of $\cH(K)$. We can quantify how clearly this expectation captures the underlying vector space statistics by considering higher-order moments. The second moment is
\begin{equation}
    \overline{\abs{\braket{v_i}{v_j}}^2}=\delta_{ij}+O\left(\frac{1}{\text{dim}(\cH(K))}\right)
\end{equation}
which means that the variance roughly computes the size of the Hilbert space
\begin{equation}\label{eq:dimH_estimate}
    \sigma^2=\overline{\abs{\braket{v_i}{v_j}}^2}-\left|\overline{\braket{v_i}{v_j}}\right|^2=O\left(\frac{1}{\text{dim}(\cH(K))}\right).
\end{equation}
This implies that we should not always expect to find $\text{dim}(\cH(K))=K$. Instead, $\text{dim}(\cH(K))$ saturates as $K$ approaches $\dim(\Hnonp)$ and the likelihood of encountering linear dependence between two vectors approaches $O(1)$. 

Non-perturbative quantum gravity produces a situation quite analogous to the toy model above. Let us focus our attention on an arbitrary but finite energy window. Quantum field theory on a fixed background teaches us to expect infinite-dimensional Hilbert spaces. This is the perturbative quantum gravity limit introduced in the previous section. Non-perturbative quantum gravity confounds this expectation by assigning a large but finite entropy $S_0$ to regions of spacetime. A basis in the perturbative Hilbert space must become overcomplete in the finite, non-perturbative Hilbert space.\footnote{
    We make a technical distinction between $\cH$ and $\Hp$. Since states in either space can be determined by analogous boundary conditions, we employ a slight abuse of notation and use identical labels for the states in either space, but the inner products on these spaces are distinct, so they are not equivalent as Hilbert spaces. In reality, there is a non-isometric map $V:\Hp\to\Hnonp$ that lets us elevate a subset $\cH\subset\Hp$ to a subspace $V(\cH)\subset\Hnonp$. This Section studies how the kernel of this map grows as a function of $K$.
} We have not yet developed the technology to compute exact overlaps in the non-perturbative Hilbert space of quantum gravity $\Hnonp$.\footnote{
    A notable exception is AdS/CFT, where overlaps in quantum gravity are given by well-understood overlaps in a dual CFT.
} But in the specific case of JT gravity, non-perturbative wormhole effects can be captured by the random matrix integral of \cite{Saad:2019lba}. Accordingly, wormhole corrections to overlaps computed using the gravitational path integral capture the statistical fluctuations of the inner product, viewed as a random variable drawn from the matrix ensemble \cite{Marolf:2020xie,Marolf:2024jze,Saad:2019lba,Saad:2018bqo,Saad:2021rcu,Saad:2021uzi,Stanford:2019vob}.

Let us consider some examples of wormhole corrections in the two-sided black hole case. Non-perturbative effects modify the inner product:\footnote{
One might hope that these effects are small enough that the states in the set $\{\ket{\ell}\}$ are still all linearly independent in $\Hnonp$. We will soon see that this is not true. The states span the full Hilbert space, but now the set is overcomplete. We explicitly discuss linear dependence and the emergence of null states in Section \ref{sec:nullstates}.
}
\be\begin{aligned}
    \overline{\braket{\ell}{\ell'}} &= \delta(\ell-\ell')+ \inlinefig[6]{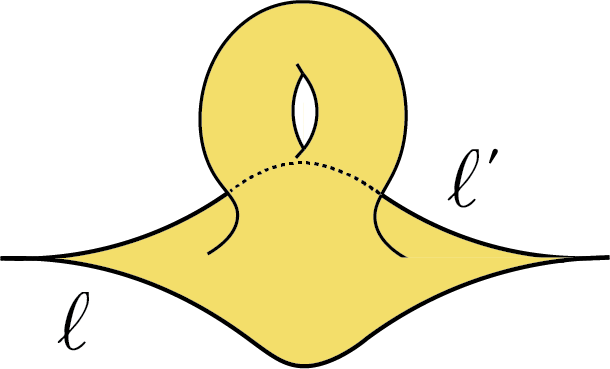} + O(e^{-4S_0}), \\
    \overline{\braket{b}{b'}} &= \frac{1}{b}\delta(b-b')+ \inlinefig[6]{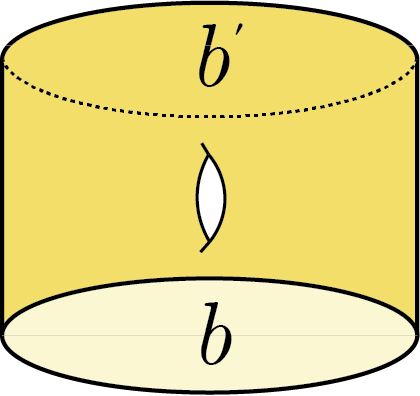} + O(e^{-4S_0})
\end{aligned}\ee
Moving forward, we exclusively use the overline $\overline{\cdot}$ to denote expectation values computed using the full gravitational path integral. We can also consider how non-perturbative effects modify $Z_n$. While the leading order term is still the product of $n$-disks given in equation \eqref{eq:zn_pert}, we also pick up connected terms like
\begin{equation}\label{eq:zn_nonpert}
    \overline{Z_n(\beta_1,\cdots,\beta_n)} \supset e^{(2-n)S_0}\left(\inlinefig{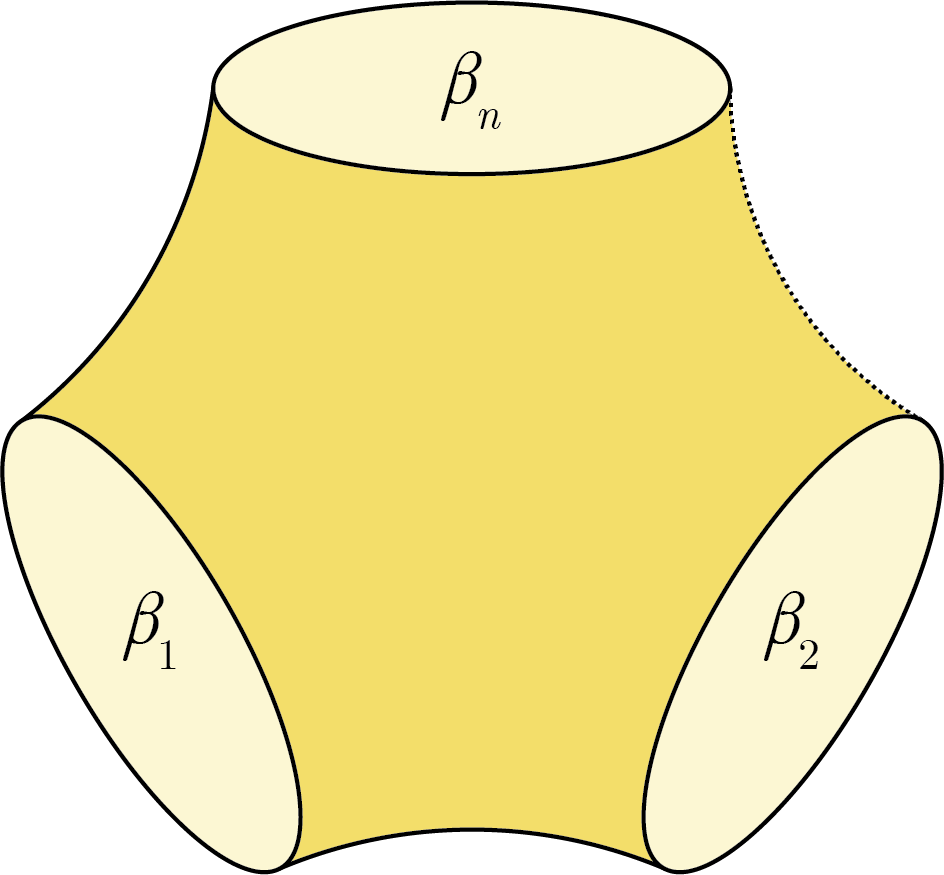}\right)
\end{equation}
where $\supset$ is used to denote summands in the genus expansion. These new terms are especially important after the inclusion of matter. The leading perturbative term in $Z_{n}^{i_1,i_1',\cdots,i_n,i_n'}$ appears with a product of $\delta$-functions. If $i_k\neq i_k'$ in equation \eqref{eq:zn_matter_pert}, then the leading bulk geometry must connect the $k$-th boundary to a disjoint boundary. If we contract the indices cyclically and $i_k\neq i_k'$ $\forall k$, then the leading contribution to the non-perturbative calculation is a fully connected geometry 
\be
\begin{aligned}\label{eq:zn_matter_nonpert}
    \overline{Z_{n}^{i_1,i_2,i_2,\cdots,i_n,i_1}(\beta_{1,L},\beta_{1,R},\cdots,\beta_{n,L},\beta_{n,R})} &\supset e^{(2-n)S_0}\left(\inlinefig{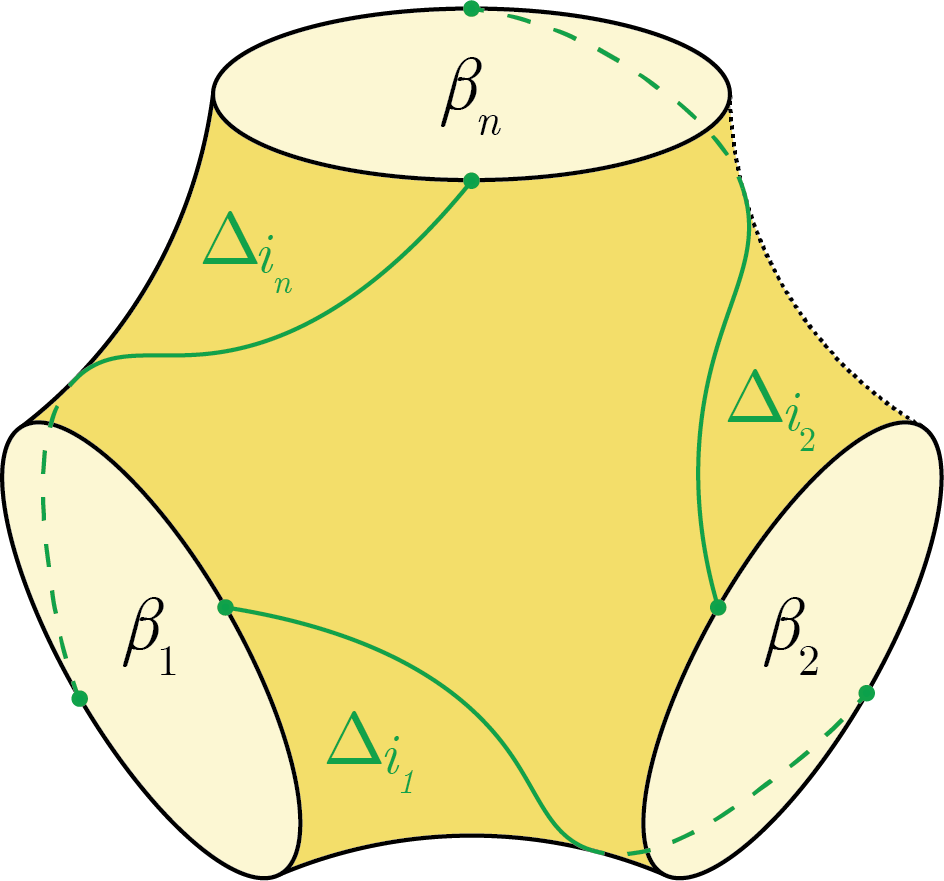}\right), \\
    &= e^{2S_0}\int ds_Lds_R\rho_0(s_L)\rho_0(s_R)\ \prod_{k=1}^n y_k(s_L,s_R) 
\end{aligned}
\ee
where
\begin{equation}
    y_k(s_L,s_R)=e^{-S_0}e^{-\frac{\beta_{k,L} s_L^2}{2}}e^{-\frac{\beta_{k,R} s_R^2}{2}}\gamma_{\Delta_{i_k}}(s_L,s_R).
\end{equation}

Moving forward, we'll refer to this path integral as the pinwheel geometry \cite{Boruch:2024kvv}. It is computed by noting that the geometry is made of two connected components separated by matter worldlines. These components are glued along the matter lines by the two energy integrals. The $y_k$ factors carry the contribution from the $k$-th boundary to the path integral. This consists of a factor of $e^{-S_0}$ needed to compute the Euler characteristic, factors of $e^{-\frac{\beta s^2}{2}}=e^{-\beta E}$ which specify the boundary length on either end of the matter insertion, and the normalization factor $\gamma_\Delta$ introduced in equation \eqref{eq:gamma}.
The pinwheel is dominant for a cyclic index contraction as long as matter indices on the same boundary differ (i.e., $i_k\neq i'_{k}$ $\forall k$). Otherwise, the leading contribution is the disconnected diagram given in equation \eqref{eq:zn_matter_pert}. In the perturbative limit $S_0\to\infty$, all of these corrections vanish, and we recoup equations \eqref{eq:H0}, \eqref{eq:zn_pert}, and \eqref{eq:zn_matter_pert}.

Equation \eqref{eq:dimH_estimate} gives us a way to estimate the size of the Hilbert space, but we can do better. Considering contributions from higher-order moments allows us to find an exact answer. This can be done systematically by introducing the resolvent \cite{Penington:2019kki,Boruch:2024kvv,Balasubramanian:2022gmo} 
\be
\begin{aligned}\label{eq:r}
    R_{ij}(\lambda) &\equiv \left(\frac{1}{\lambda \mathbbm{1}-M}\right)_{ij}, \\
    &= \frac{\delta_{ij}}{\lambda}+\frac{1}{\lambda}\sum_{n=1}^\infty \frac{(M^n)_{ij}}{\lambda^n},
\end{aligned}
\ee
where $M$ is the Gram matrix
\begin{equation}
    M_{ij}\equiv\braket{v_i}{v_j}
\end{equation}
and $\lambda\in\C$. The number of non-zero eigenvalues of the Gram matrix counts the dimension of the Hilbert space spanned by the set $\{\ket{v_i}\}$, $i=1,...,K$. Roughly speaking, an eigenvalue $z$ of $M$ manifests as a non-analyticity on $R_{ij}$ with multiplicity given by the residue $\text{Res}_{\lambda=z}(\sum_{i} R_{ii}(\lambda))$. Integrating $R$ on a contour which excludes $\lambda=0$ and no other eigenvalues counts the dimension of $\cH(K)=\text{Span}(\{\ket{v_i}\})$. If we also compute powers of the Gram matrix using the gravitational path integral, we get an expectation value for $\text{dim}(\cH(K))$ as a subset of the non-perturbative Hilbert space $\cH(K)\subset\Hnonp$. We review the analytic structure of the resolvent and how to arrive at this result from a replica trick in Appendix \ref{app:res_replica}. In summary, we have
\begin{equation}\label{eq:dimH_from_r}
    \overline{\text{dim}(\cH(K))}=\oint_C\frac{d\lambda}{2\pi i} R(\lambda)
\end{equation}
where the contour $C=C_0\sqcup C_\infty$ is defined in Figure \ref{fig:contour} and $R=\sum_i \overline{R_{ii}}$ is the trace of the resolvent.\footnote{
    Based on our earlier convention, we should really be referring to $\overline{R(\lambda)}$. We hope the reader will allow us this conceit to slightly unburden the notation moving forward.
} 

In Appendix \ref{app:reviewBH}, we use equation \eqref{eq:dimH_from_r} to find the size of a $K$-dimensional subspace of the perturbative Hilbert space of two-sided black holes under the non-perturbative inner product. In the double-scaling limit $K\to\infty$, $e^{2S_0}\to\infty$ with $K/e^{2S_0}=O(1)$, we find 
\be\begin{aligned}
    \overline{\dim(\Hnonp(K))}=\begin{cases}
        K &\qquad K<  d^2, \\
        d^2 &\qquad K>d^2,
    \end{cases}
\end{aligned}\ee
where 
\begin{equation}
    d\equiv e^{S_0}\int dE\rho_0(E)
    \label{eq:d}
\end{equation}
is given by an integral over an arbitrary but finite energy window. We reintroduced here the energy variable $E=s^2/2$, with the relationship between $\rho_0(E)$ and $\rho_0(s)$ given in equation \eqref{eq:rhonot}. Since this result is basis-independent, we conclude that the non-perturbative Hilbert space for an arbitrary but finite energy window has dimension $\dim(\Hnonp)=d^2$.

\begin{figure}
    \centering
    \includegraphics[height=20\fontcharht\font`\B]{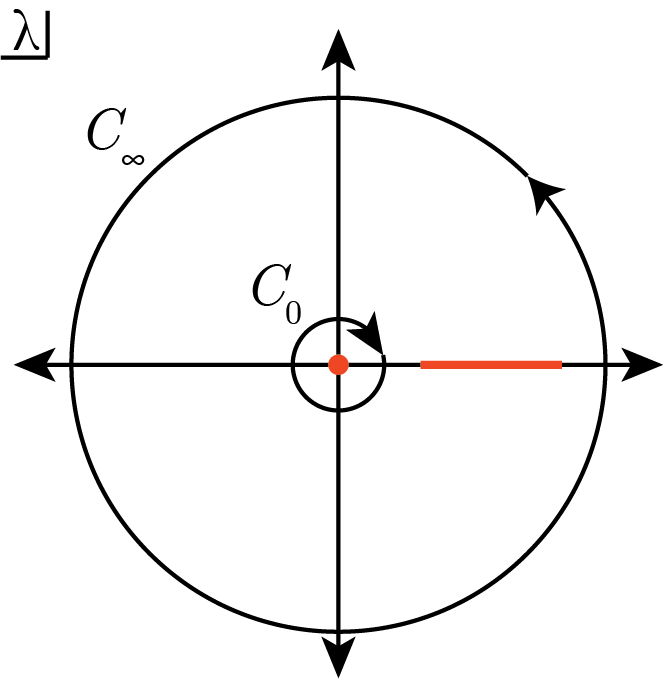}
    \caption{The integration contour needed to calculate the Hilbert space dimension consists of the large defining contour of the resolvent $C_\infty$ and a small contour $C_0$ used to exclude the contributions from zero eigenvalues of the Gram matrix.}
    \label{fig:contour}
\end{figure}

\subsection{The Hilbert space of closed universes}
\label{sec:reviewclosed}

Our goal in this Section is to derive the dimension of the non-perturbative Hilbert space of closed universes. For simplicity, we consider a fixed scaling dimension $\Delta$ and asymptotic length $\beta$. States are only differentiated by an internal flavor index $i$ 
\begin{equation}
    \Hnonp(K)=\text{Span} \left\{\ket{q_i};\ i\in\{1,\cdots, K\}\right\}.
\end{equation}
We establish these conventions strictly for computational convenience.\footnote{
    For example, one could alternatively define
    \begin{equation}
        \Hnonp(K)=\left\{\ket{q_i}=\ket{\Delta_i;\beta};\ i\in\{1,\cdots k\}\text{ and }\Delta_i\approx\Delta\right\}
    \end{equation}
    for similar constants $\Delta$ and $\beta$ or consider states labeled by different values of $\beta$. These sets are linearly independent from each other in the perturbative Hilbert space $\Hp$, but span the same space in $\Hnonp$. 
} We will show that this choice has no effect on our final result by proving that the $\{\ket{q_i}\}$ states are overcomplete in $\Hnonp$ if we take $K$ sufficiently large \cite{Penington:2019kki,Almheiri:2019qdq, Balasubramanian:2022gmo,Balasubramanian:2022lnw,Boruch:2023trc}.

The non-perturbative inner product is given by the topological expansion 
\begin{equation}
    \overline{\braket{q_i}{q_j}}=\inlinefig{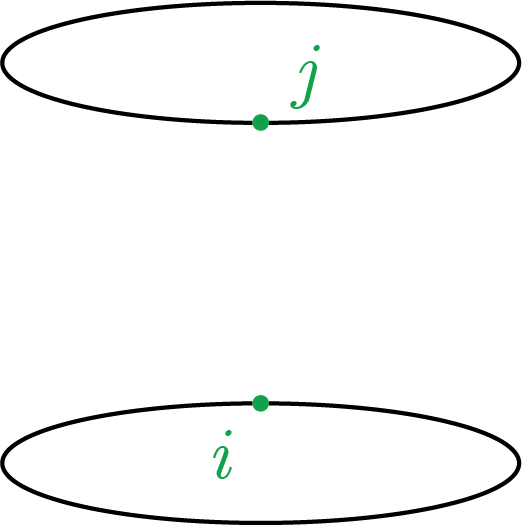}=\delta_{ij}\left(\inlinefig{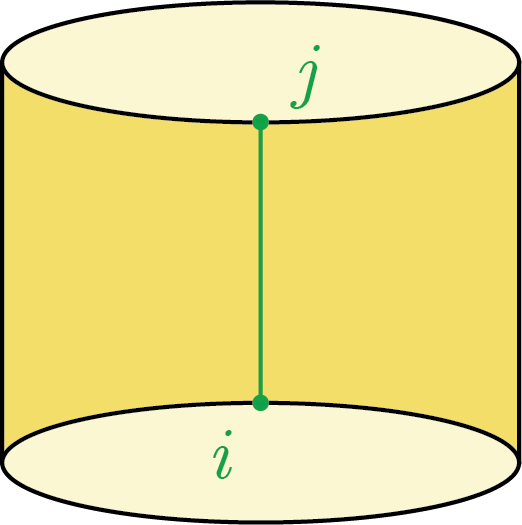}+O(e^{-2S_0})\right)
    \label{eq:overlapclosedreview}
\end{equation}
Note how the leading geometry appears at $O(1)$ when expanding in powers of $e^{S_0}$. The Kronecker delta ensures that lines that connect to each other carry the same flavor index.
The second moment of this overlap looks like 
\be\begin{aligned}\label{eq:closed_square}
    &\overline{|\braket{q_i}{q_j}|^2} = \\
    &\inlinefig{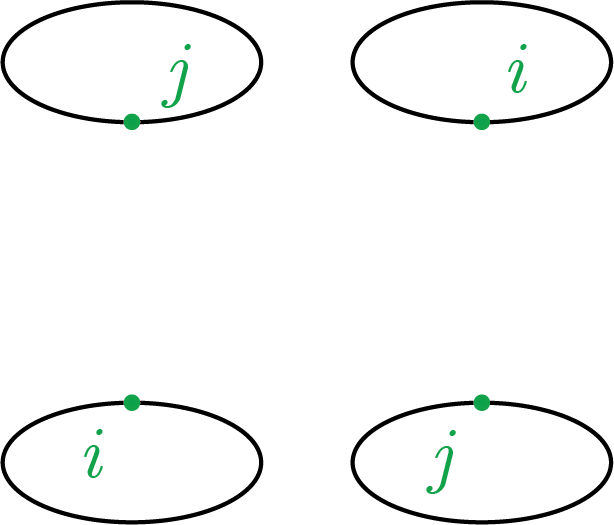} = \delta_{ij}\left(\inlinefig{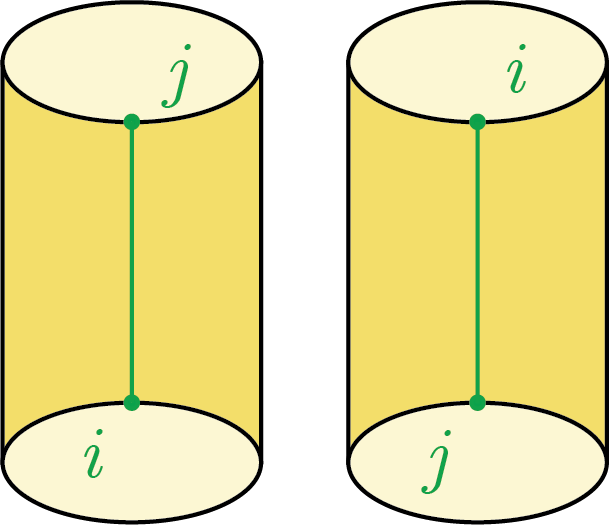}+\inlinefig{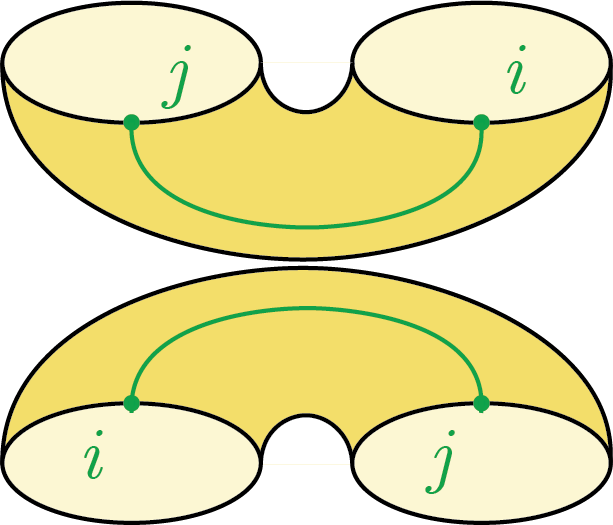}\right) + \inlinefig{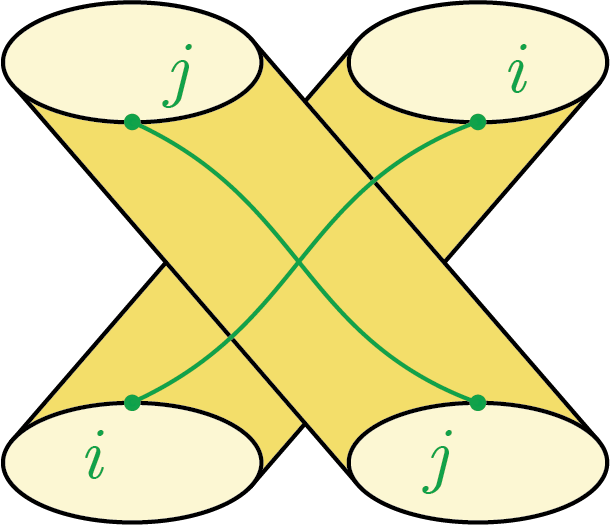} + O(e^{-2S_0}).
\end{aligned}\ee
 Equations \eqref{eq:overlapclosedreview} and \eqref{eq:closed_square} imply that the variance of the overlap is also $O(1)$ in powers of $e^{S_0}$. By equation \eqref{eq:variance}, we can then estimate $\dim(\Hnonp)=O(1)$. Like with the two-sided black hole, we see that the non-perturbative inner product greatly reduces the dimension of the Hilbert space. However, in this case, it seems as though there are very few states. This could not be true for the two-sided black hole because we could produce physically inequivalent states by evolving with the boundary ADM Hamiltonians $H_L$ and $H_R$. We have no such option here, and our result implies that trivial bulk ``evolution" generated by the Wheeler-DeWitt Hamiltonian is enough to span the Hilbert space. In other words, all bulk degrees of freedom are gauge.

To move past order-of-magnitude estimates, we need to include contributions from all moments as prescribed by equation \eqref{eq:dimH_from_r}. As a first step, let us compute the resolvent $R_{ij}$, defined in equation \eqref{eq:r}, using the gravitational path integral. For closed universes, each power of the Gram matrix $M$ necessitates the inclusion of two disconnected boundaries. For a given power $(M^n)_{ij}$, matter indices are assigned cyclically and summed over for all but two adjacent boundaries\footnote{\label{footnote:notation}
    Here, we explain the notation used in the diagrammatic expansion. Blue and green lines carry an index. Blue lines are ``bare propagators." Green lines denote bulk matter. The intersection of any two lines comes with a Kronecker delta, which imposes that they both carry the same index. Internal indices are summed over. In particular, we sum over the index in all loops. Blue lines not contracted with the insertion of a resolvent come with a factor of $1/\lambda$.
}
\be\begin{aligned}
    \inlinefig[5]{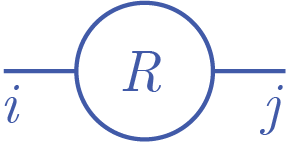} &= \inlinefig[5]{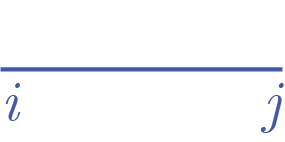}+\inlinefig[5]{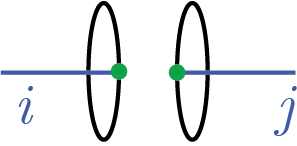}+\inlinefig[5]{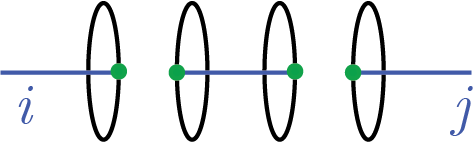}+\cdots \\
    &= \inlinefig[5]{Figures/App_A/dij300ppi.png}+\inlinefig[5]{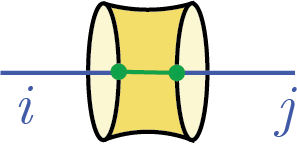}+\left(\substack{
        \inlinefig[8]{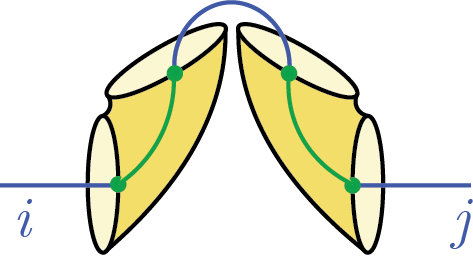} \\
        + \\
        \inlinefig[8]{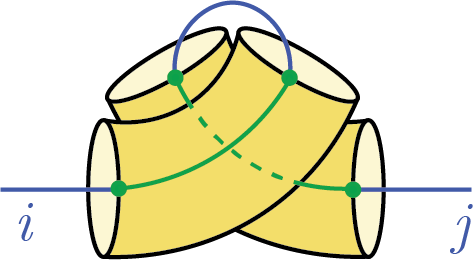} \\
        + \\
        \inlinefig[8]{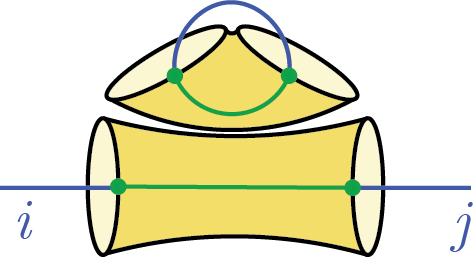}
    }\right)+\cdots
\end{aligned}\ee
The leading geometries are not fully connected. In the black hole case (c.f. Appendix \ref{app:reviewBH}), the pinwheel geometry allowed us to contract like indices at the cost of a higher Euler characteristic. Here, each matter insertion lives on a distinct boundary, meaning that we can form index loops without ever needing to connect more than two boundaries. We drop terms which are exponentially suppressed in $S_0$ 
so that the resolvent is fully calculable in powers of 
\be\begin{aligned}
     Z_1^{\text{closed}} &= \int ds\rho_0(s)\ e^{-\beta s^2}\gamma_{\Delta}(s,s)
     &= \inlinefig{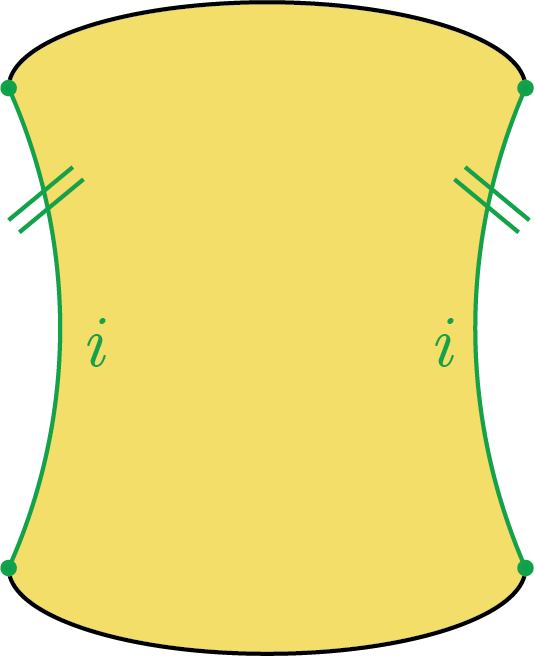}
\end{aligned}\ee
Note that while $Z_1^{\text{closed}}$ is topologically identical to $\overline{Z_2}$ introduced in equation \ref{eq:zn_matter_nonpert}, it has half as many matter insertions. 

In the double scaling limit $e^{2S_0}\to\infty$, $K\to\infty$ with $K/e^{2S_0}=O(1)$, at leading order in $K$ we get a single summand from each power $\Tr M^n$ where $n$ copies of $Z_1^{\text{closed}}$ are used to connect like indices
\begin{equation}
    \inlinefig[5]{Figures/App_A/rij300ppi.png} = \inlinefig[5]{Figures/App_A/dij300ppi.png}+\inlinefig[5]{Figures/Closed_Resolvent/closed_mij1300ppi.png}+\inlinefig[8]{Figures/Closed_Resolvent/closed_m2ij3300ppi.png}+\cdots
\end{equation}
After taking the trace on both sides, in equation form we have
\be\begin{aligned}
    \lambda R(\lambda) &= K+\sum_{n=1}^\infty\left(\frac{K Z_1^{\text{closed}}}{\lambda}\right)^n, \\
    &= K + \frac{KZ_1^{\text{closed}}}{\lambda - KZ_1^{\text{closed}}}.
\end{aligned}\ee
The resolvent has a single non-zero pole at $\lambda=KZ_1^{\text{closed}}$. Plugging into equation \eqref{eq:dimH_from_r} gives precisely $\dim(\Hnonp(K))=1$.

Before concluding that the Hilbert space is trivial, we should check how subleading corrections in $K$ affect our answer when we move away from the $K\to\infty$ limit.  Our approach is recursive. We start with the fact that $\Tr M=KZ_1^{\text{closed}}$ at leading order in $e^{S_0}$. Now suppose that we know $\Tr M^{n-1}$ and want to calculate
\begin{equation}
    \Tr M^n = (M^{n-1})_{i_1,i_n}(M)_{i_n,i_1}
\end{equation}
where repeated indices are summed over. The simplest contributions to the trace come from diagrams where the $i_n$ indices are contracted together through $Z_1^{\text{closed}}$. Such diagrams contribute as
\begin{equation}
    \Tr M^n \supset \Tr M^{n-1}\cdot \Tr M=KZ_1^{\text{closed}}\Tr M^{n-1}.
\end{equation}
Only considering such terms would yield the approximation studied above, but in general we can contract one of the $i_n$ indices with any of the $2(n-1)$ other $i_k$ indices. These contractions give factors of $\delta_{i_k,i_n}Z_1^{\text{closed}}$. We spend the delta by swapping the remaining $i_n$ index for an $i_k$ leaving another copy of $\Tr M^{n-1}$ 
\be\begin{aligned}
    \Tr M^n &\supset Z_1^{\text{closed}} (M^{k-2})_{i_1,i_{k-1}}M_{i_{k-1},i_n}M_{i_n,i_{k+1}}(M^{n-k-2})_{i_{k+1},i_{n-1}}M_{i_{n-1},i_1} \\
    &=Z_1^{\text{closed}}\Tr M^{n-1}.
\end{aligned}\ee
In total we have
\begin{equation}
    \Tr M^n=Z_1^{\text{closed}} (K+2(n-1))\Tr M^{n-1} = \left(2Z_1^{\text{closed}}\right)^n\left(\frac{K}{2}\right)^{(n)}
\end{equation}
where $x^{(n)}=\frac{\Gamma(x+n)}{\Gamma(x)}$ is the Pochhammer symbol.\footnote{
The authors of \cite{Usatyuk:2024mzs} computed $\Tr M$ using an integral over $\alpha$-sectors. Our result can be rewritten as
\begin{equation}
    \Tr M^n=\frac{1}{\Gamma\left(\frac{k}{2}\right)}\int dt\ (2Z_1^{\text{closed}})^nt^{K/2+n-1}e^{-t}
\end{equation}
which, after substituting $t=x/2$, gives
\begin{equation}
    \Tr M^n = (Z_1^{\text{closed}})^n\expval{x^n}
\end{equation}
with expectation values computed using
\begin{equation}
    p(x) = \frac{x^{K/2-1}e^{-x/2}}{2^{K/2}\Gamma\left(\frac{K}{2}\right)}
\end{equation}
as defined in equation (3.17) of \cite{Usatyuk:2024mzs}.
} 

We can now present a formula for the resolvent that is valid for any $K$. After taking the trace on both sides, equation \eqref{eq:r} becomes
\be\begin{aligned}
    \lambda R(\lambda) - K &= \sum_{n=1}^\infty \left(\frac{2Z_1^{\text{closed}}}{\lambda}\right)^n\left(\frac{K}{2}\right)^{(n)}, \\
    &= \frac{1}{\Gamma\left(\frac{K}{2}\right)}\sumint dt\ \left(\frac{2tZ_1^{\text{closed}}}{\lambda}\right)^n t^{K/2-1}e^{-t}, \\
    &= \frac{1}{\Gamma\left(\frac{K}{2}\right)}\int dt\ \frac{2t^{K/2}Z_1^{\text{closed}}e^{-t}}{\lambda-2tZ_1^{\text{closed}}}.
\end{aligned}\ee
The right hand side has poles on the entire real axis so we need a slightly different contour to evaluate equation \eqref{eq:dimH_from_r}. We calculate the principle value 
for a set of concentric contours $C_+^\epsilon$ wrapping $(\epsilon,\infty)\subset\R^+$. This gives
\be\begin{aligned}
    \overline{\dim(\Hnonp(K))} &= \lim_{\epsilon\to 0}\frac{1}{\Gamma\left(\frac{K}{2}\right)}\int_0^\infty dt\ 2t^{K/2}Z_1^{\text{closed}}e^{-t}\oint_{C_+^\epsilon}\frac{d\lambda}{2\pi i}\ \frac{1}{\lambda\left(\lambda-2tZ_1^{\text{closed}}\right)}, \\
    &= \lim_{\epsilon\to 0}\frac{1}{\Gamma\left(\frac{K}{2}\right)}\int_\epsilon^\infty dt\ t^{K/2-1}e^{-t} = 1,
\end{aligned}\ee
where the last line follows from the definition of the Gamma function. Since $\Hnonp(K)$ remains trivial for all values of $K$ we can now be certain that there is only one state in $\Hnonp$.

\section{Setup: the path integral from an observer's point of view}
\label{sec:setup}

In Section \ref{sec:review}, we found that, if we treat all operator insertions on the boundary in the same way, the gravitational path integral gives a non-perturbative Hilbert space of a closed universe that is one-dimensional and a non-perturbative Hilbert space of a two-sided black hole that has dimension $d^2$. As we pointed out in Section \ref{sec:intro}, this result is unsatisfactory if we are interested in the experience of a bulk observer. In fact, the inner product between states with and without an observer have fluctuations---which contribute at leading order in the closed universe case---meaning that the notion of a fixed observer is not well-defined. Moreover, a local observer in a closed universe should experience non-trivial physics, which requires access to a Hilbert space with dimension much larger than one.

Our goal in this Section is to introduce a modification to the rules of the gravitational path integral, which takes into account the presence of a bulk observer. 
To simplify the problem as much as possible, we will model the observer as a localized quantum mechanical system that propagates through the spacetime. Thus, the observer will travel along a bulk worldline starting and ending at the location where the observer is created in the states that we shall prepare using the gravitational path integral. Following the work of \cite{Chandrasekaran:2022eqq, Witten:2023qsv, Witten:2023xze, Kolchmeyer:2024fly}, we take the action of the observer along this worldline to be given by 
\be 
I_\text{observer} = \int d\tau\left[P \partial_{\tau} Q - Q \sqrt{g_{\tau \tau}}\right]\,,
\label{eq:action-of-observer}
\ee
where $Q$ is the Hamiltonian of the quantum system representing the observer, $P$ is its canonical conjugate, and $g_{\tau \tau}$ is the induced metric along the worldline of the observer, which is parametrized by $\tau$. The Hilbert space of the observer $\HO$ is spanned by eigenstates of $Q$, whose eigenergies,  $E^{(i)}_O$, are taken to be bounded from below and contained within some narrow range so that we can approximate $E^{(i)}_O \sim E_O$. We can label a basis of energy eigenstates of the observer by $\ket{\Delta^{(i)}_O} \in \HO$, where  $\Delta^{(i)}_O \sim \Delta_O$ are the AdS scaling dimensions associated to the energies $E^{(i)}_O$ contained within the narrow window. The observer can build a clock from the energy levels ${\Delta^{(i)}_O}$, as long as these energy levels are sufficiently dense.\footnote{The clock can in principle provide an accurate estimate of time for proper times along the worldline of the clock that are much shorter than the inverse of the typical spacing between neighboring energy levels.} We will sometimes dress observables to the time measured by this clock, in which case we will assume that the states of the clock showing different times are roughly orthogonal.

In this setup, states prepared by the gravitational path integral live in a tensor product space
\be
\ket{\psi_\text{grav},\psi_\text{obs}} \in \cH^\text{rel} \otimes \HO
\ee
where $\ket{\psi_\text{grav}}$ specifies the state of the metric and matter fields on a given slice subject to the boundary conditions at the worldline of the observer, which in turn depend on the scaling dimension $\Delta_O^{(i)}$ of the observer. The state of the observer $\ket{\psi_\text{obs}}$ lives in $\text{Span}(\{\ket{\Delta_O^{(i)}}\})$. For example, in the two-sided black hole setup at the perturbative level, we could take $\ket{\psi_\text{grav},\psi_\text{obs}}=\ket{\Delta_O^{(i)};\ell,m}\otimes\ket{\Delta_O^{(i)}}$. To make sense of the Hilbert space that such states belong to, we have to specify a positive definite inner product. As in Section \ref{sec:review}, we will always compute the inner-product between two such states using the gravitational path integral. 

To compute the inner-product $\braket{{\phi_\text{grav},\phi_\text{obs}} }{{\psi_\text{grav},\psi_\text{obs}} }$ at a perturbative level, one finds the leading gravitational saddle that interpolates between two spatial slices, one where the state of the observer is $\psi_\text{obs}$ and other field data is specified by $\psi_\text{grav}$, and one where the state of the observer is $\phi_\text{obs}$ and other field data is specified by $\phi_\text{grav}$. The inner-product is then computed by a perturbative expansion of the on-shell action around such a saddle. We will refer to the Hilbert space that results from this choice of inner-product as $\Hrelp \otimes \HO$ where we must impose the constraint that the energy $\Delta_O^{(i)}$ of the observer is the same in $\Hrelp$ and $\HO$.

We are interested in computing observables with respect to a given observer at a non-perturbative level. We, therefore, want to avoid large overlaps between states in which an observer is present and states in which it is not, and we want the observer to exist on any spatial slice in between a given bra and the corresponding ket when computing an overlap between any two states or moments of an overlap. We achieve this by imposing that the observer's worldline must connect the points at which it is inserted in the preparation of any given bra and the corresponding ket.  This guarantees that we can ask questions about the state of the observer on a given spatial slice or the state of the spatial slice given the state of the observer. In practice, this new prescription is equivalent to introducing an additional boundary at the location of the worldline of the observer and then integrating over all possible geometries in which the worldline of the observer can be included (see Figure \ref{fig:newbc}). The observer's worldline is allowed to fluctuate between the two points where the observer is inserted in the bra and ket, and we sum over all worldlines that are topologically inequivalent. However, we also need to quotient by the mapping class group in the gravitation path integral. We will often perform this quotient by only considering non-winding geodesics. 

We now present a simple example that helps motivate this modification to the usual rules of path integration. Suppose we would like to calculate the probability of transitioning from a tensor product state $\ket{\psi}$ describing the system and observer to a state $\ket{\phi}$. Such a state $\ket{\phi}$ could, for instance, be an eigenstate of an operator that the observer is trying to measure, given that the observer is in a certain state. As an example, they could measure the stress-tensor that they see along their worldline when their clock reads a given time. This probability is captured by the expectation value of a projector $\Pi_\phi$ on state $\ket{\phi}$. In the path integral, both $\ket{\psi}$ and  $\ket{\phi}$ states are defined by boundary conditions. The probability of the observer transitioning to a state $\ket{\phi}$ is given by the probability amplitude squared
\begin{equation}
    p(\phi) = |\braket{\phi}{\psi}|^2=\mel{\psi}{\Pi_\phi}{\psi} = \inlinefig{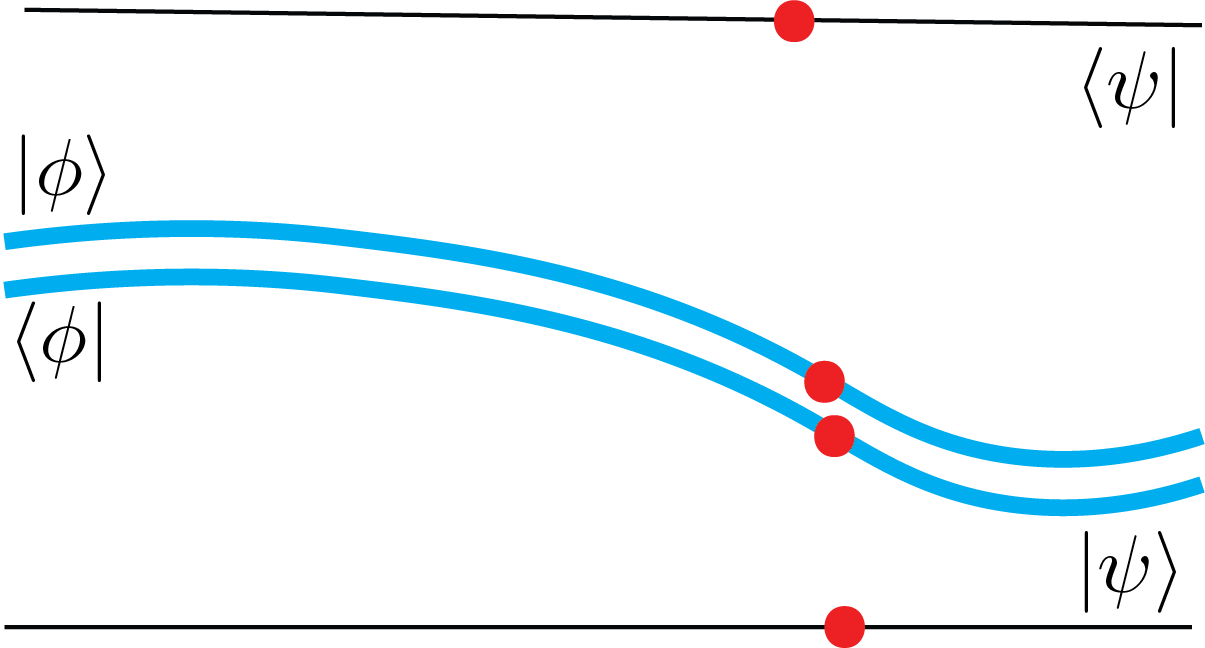},
\end{equation}
which is a path integral with four boundaries. The straight boundaries fix boundary conditions for the fields and metric to be in the state $\ket{\psi}$ while the curved boundaries have the fields and metric in state $\ket{\phi}$. The red dots stand for observer insertions and can be contracted against each other through a bulk worldline. The conventional rules of gravity tell us to consider the contributions of all saddles which satisfy the boundary conditions. Naively this gives
\begin{equation}
    p(\phi) = \inlinefig{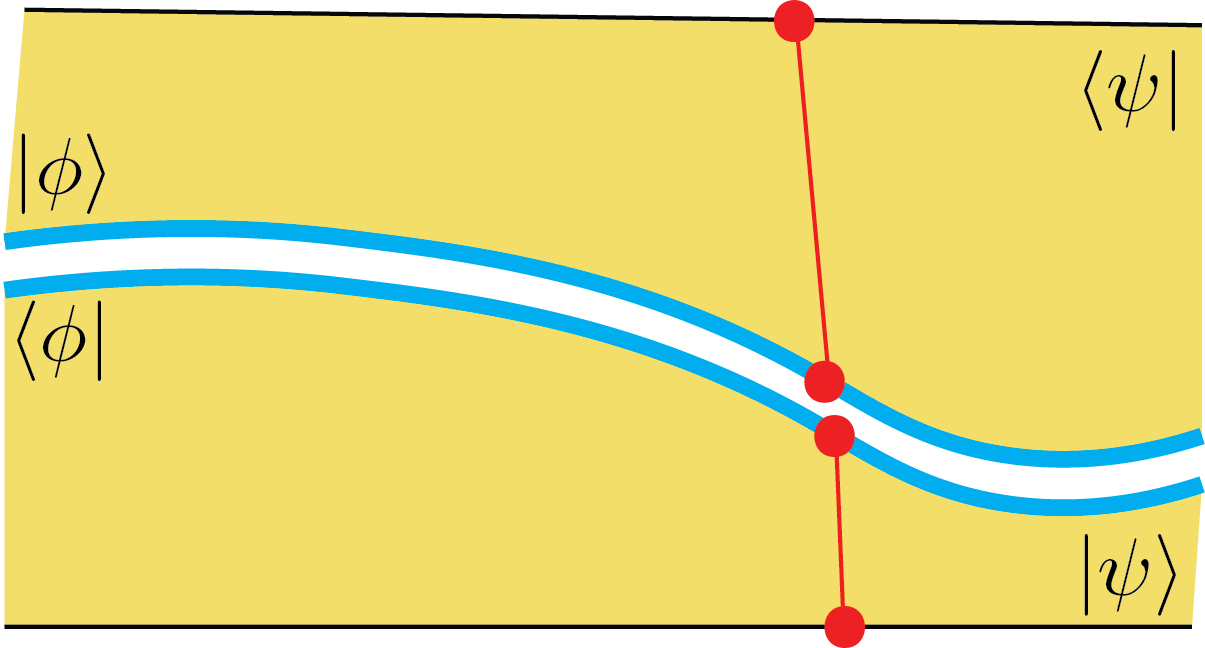}+\inlinefig{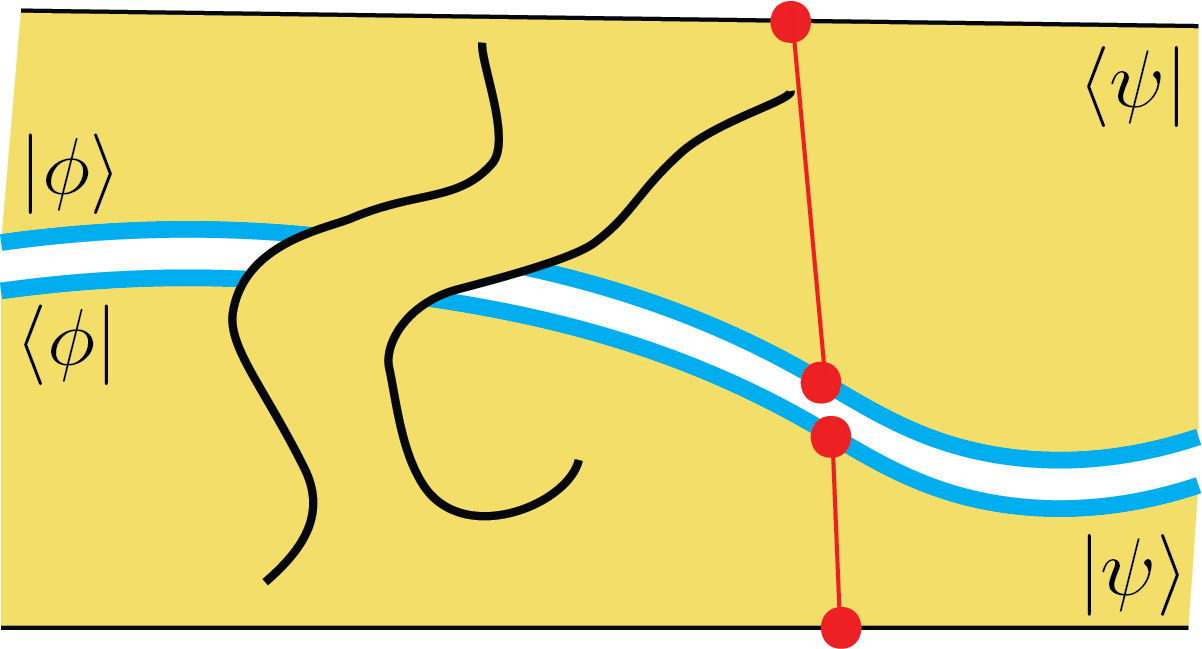}+\underbrace{\inlinefig{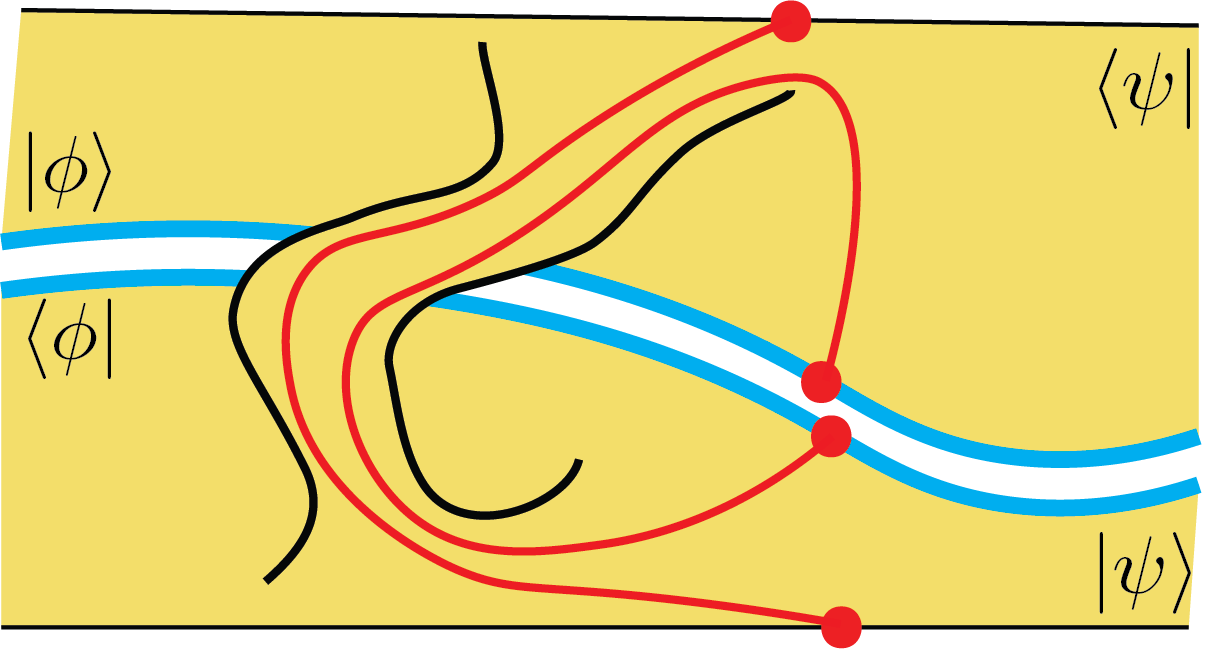}}_\text{Problematic contribution} + \cdots\,,
    \label{eq:transitionadam}
\end{equation}
where the $\cdots$ denotes other possible contributions to the GPI that we drop for simplicity. The third contribution above does not compute the probability that the bulk and observer are in a given state because the slice on which the projector is inserted never intersects the worldline of the observer whose state is prepared by the boundary conditions fixed by $\ket{\psi}$. In fact, the projector inserts a copy of the observer in $\bra{\phi}$, which is then measured in $\ket{\phi}$. 
This observer has nothing to do with the one inserted in $\ket{\psi}$ and is simply a consequence of the insertion of the projector.

As an example, let us consider the states of a closed universe when an observer is present. With the conventional rules of the gravitational path integral, 
\begin{equation}\label{eq:closedtransition}
    p(\phi) = \inlinefig{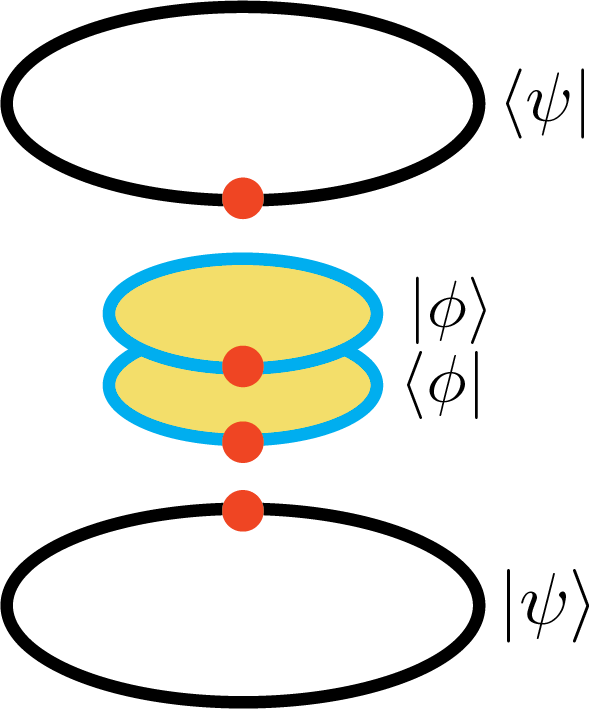} = \inlinefig{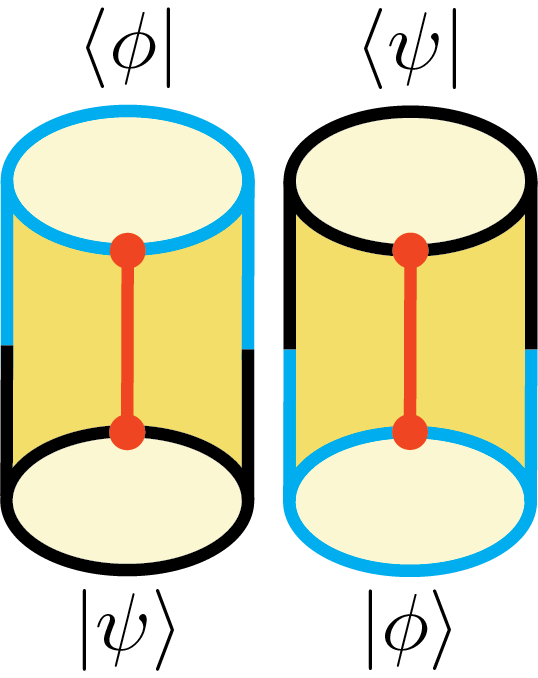} + \underbrace{\inlinefig{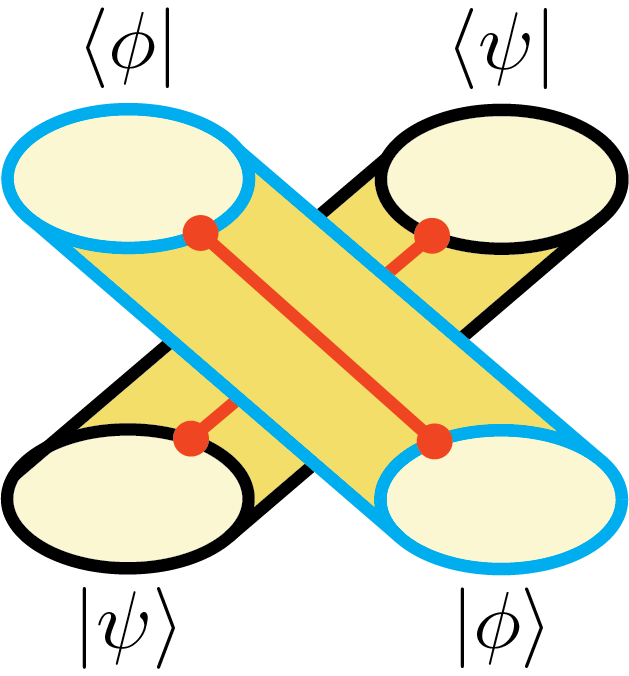}+\inlinefig{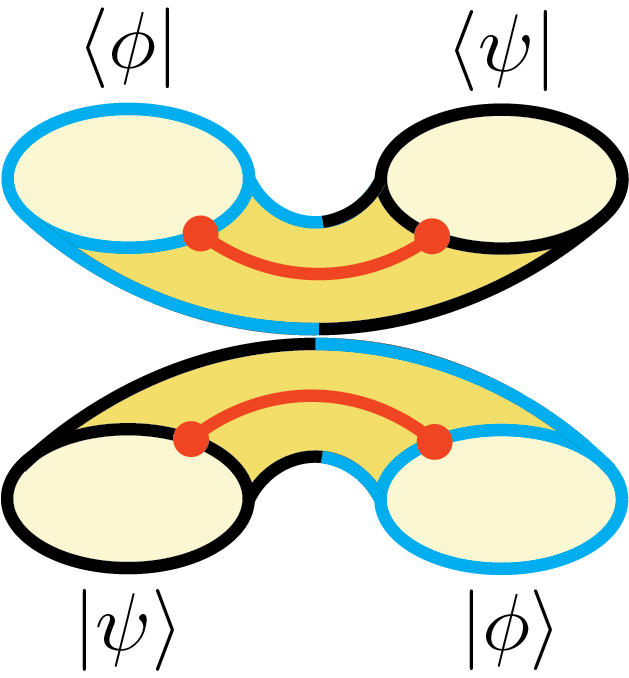}}_{\text{Problematic contributions}}
\end{equation} 
at leading order. Similar to the more general example above, the last two contributions do not compute a transition probability for the state of the system and observer from the initial state $\ket{\psi}$ to the desired state $\ket{\phi}$. In the second contribution, the observer connects cyclically within the projector, and the diagram simply computes the product of the norms of the two states $\ket{\psi}$ and $\ket{\phi}$. The third contribution showcases contractions between kets and bras, which are irrelevant for computing a transition probability. In both cases, the observer never encounters the slice where the measurement takes place. Our proposed rules for the gravitational path integral in the presence of an observer, in which an observer created in a given ket must be annihilated in the corresponding ket, naturally exclude the last contribution in equation \eqref{eq:transitionadam} as well as the second and third contributions in equation \eqref{eq:closedtransition}, leading to the correct transition probability for the observer and system.


We will call the Hilbert space that results from this choice of inner-product at the path integral level as $\Hr \otimes \HO$. Here, $\Hr$ is defined as the quantum gravity Hilbert space on which the operators that are relevant to describe the experience of any observer act linearly. We will find that at the non-perturbative level (see Section \ref{sec:nullstates}), states in $\Hr$ are not labeled by the observer's energy $\Delta_O^{(i)}$ which labels states in $\HO$. In particular, this implies that the dimension of the full Hilbert space is $\dim(\Hr)\times\dim(\HO)$ and that $\dim(\Hr)$ is independent of the size of the observer's Hilbert space $\HO$ (see Section \ref{sec:hilbert}). However, the amount of information a given observer can learn about the entire system is bounded by the von-Neumann entropy of the observer. The maximum value of this entropy is $\min(\dim(\HO), \dim(\Hr))$---which is $\dim(\HO)$ if the dimension of the observer's Hilbert space is smaller than that of the rest of the system and $\dim(\Hr)$ if, in likely unrealistic models, the dimension of the observer Hilbert space is greater than that of the rest of system.

We saw in Section \ref{sec:reviewclosed} for the case of closed universes,\footnote{And in Appendix \ref{app:reviewBH} for the case of two-sided black holes.} that allowing the worldlines associated with boundary operator insertions to connect between arbitrary bras and kets leads to a drastic reduction in the dimension of the Hilbert spaces, from that of the perturbative Hilbert space $\Hp$ to that of the non-perturbative Hilbert space $ \Hnonp$. This leads us to expect that our modification of the inner product described above will further modify this result, leading to differences in the properties of  $\cH_\text{non-pert}$, $\Hrelp$, and $\Hr$. To emphasize the drastic differences among these Hilbert spaces (for which the states in the gravitational path integral are defined by the same boundary conditions), 
we start by analyzing how moments of overlaps between closed universe states are modified by the new rules. For simplicity, to start, in this section, we will focus on pure JT gravity in the presence of an observer without any other matter insertions.

\subsubsection*{Moments of an overlap in the presence of an observer}

We consider asymptotic states of a closed universe $|\beta,\Delta_O^{(i)}\rangle$ labeled by the length $\beta$ of the asymptotic boundary and the scaling dimension $\Delta_O^{(i)}$ of the operator insertion associated with the observer. We will discuss alternative choices of basis states later in this Section. As shown in Figure \ref{fig:newbc} (a), an overlap of the form
\be 
\overline{\braket{\beta', \Delta_O^{(i)}}{\beta, \Delta_O^{(j)}}} = \delta^{ij}\int dE \rho_0(E) \gamma_{\Delta}(E,E) e^{-\left(\beta+\beta'\right)E} + O(e^{-2S_0})
\label{eq:overlapclosed}
\ee
is unchanged because there is only one boundary associated with the bra and one associated with the ket. The observer's worldline connects insertions of $\mathcal{O}_O^{(i)}$ in the bra to the corresponding ket regardless of whether we use the old or new rules for the gravitational path integral.
The first significant change comes when computing the square of the overlap $\overline{|\langle \beta',\Delta_O^{(i)}|\beta,\Delta_O^{(j)}\rangle|^2}$, see Figure \ref{fig:newbc} (b). The leading, disconnected geometry is given by a product of two cylinders, which is simply the square of equation \eqref{eq:overlapclosed}. The subleading contribution is given by the genus-zero wormhole depicted on the right of Figure \ref{fig:newbc} (b), which computes the variance of the overlap. This is given by 
\be 
\overline{\left|\braket{\beta', \Delta_O^{(i)}}{\beta, \Delta_O^{(j)}}\right|^2} &- \left|\overline{\braket{\beta', \Delta_O^{(i)}}{\beta, \Delta_O^{(j)}}}\right|^2\nn \\  &= e^{-2S_0} \delta^{ij} \int dE \overline{\rho(E) \rho(E')}_\text{conn.} \gamma_{\Delta_O}(E,E) \gamma_{\Delta_O}(E',E') e^{-\left(\beta+\beta'\right)(E+E')} 
+ O(e^{-4S_0}),
\label{eq:variance}
 \ee
 where $\overline{\rho(E) \rho(E')}_\text{conn.}$ is given by \cite{Saad:2019lba,Iliesiu:2021ari}\footnote{Notice that the connected correlator is between the full densities of states $\rho(s)=e^{S_0}\rho_0(s)$, and the connected part of the correlator is $O(1)$ \cite{Saad:2019lba,Iliesiu:2021ari}. This is why there is an additional factor of $e^{-2S_0}$ in front of \eqref{eq:variance}. 
 } 
\be
\overline{\rho(E) \rho(E')}_\text{conn.} = -\frac{1}{(2\pi)^2} \frac{E_1+E_2}{\sqrt{E_1 E_2} \left(E_1-E_2\right)^{2}}\,.
\label{eq:connectedcorrelator}
\ee
This result is different from what we obtained in Section \ref{sec:reviewclosed} using the old rules for the gravitational path integral.
In fact, if we allow the observer's worldline to connect between any bra/ket and any other bra/ket, the variance of the overlap is computed by the last two disconnected geometries depicted in equation \eqref{eq:closed_square} and is $O(1)$. This implies that the inner product receives corrections at leading order in the absence of an observer. This fact was the key mechanism that led to a one-dimensional Hilbert space for the closed universe in Section \ref{sec:reviewclosed}. With our new rules, the variance of the overlap \eqref{eq:variance} is $O(e^{-2S_0})$ instead. This suggests that the dimension of the non-perturbative quantum gravity relational Hilbert space $\Hr$ is much larger, $\dim\left(\Hr\right)=O(e^{2S_0})$. A similar result can be obtained for a two-sided black hole, where the Hilbert space is also enlarged in the presence of an observer. We explicitly compute the dimension of the Hilbert space in the presence of an observer in the closed universe and two-sided black hole in Section \ref{sec:hilbert}.

\subsubsection*{Geodesic bases, clock states, and the Wheeler-DeWitt constraint}

The bulk state that includes the observer can also be defined away from the asymptotic boundary of the closed universe. By using our proposal for the gravitational path integral, we can ask questions about what happens to the observer on time slices that probe the physics of the bulk. This is necessary if we want to understand what happens to the observer in a Lorentzian closed universe. The Lorentzian closed universe is an analytic continuation of the Euclidean spacetime starting from the minimal closed geodesic that wraps the Euclidean cylinder \cite{Maldacena:2004rf,Chen:2020tes,VanRaamsdonk:2021qgv,Antonini:2022blk,Usatyuk:2024mzs}; in Lorentzian signature this time slice becomes a maximal closed geodesic wrapping around the closed universe, see Figure \ref{fig:closeduniverse}. Specifically, it is the time-reflection symmetric slice of a closed Big Bang-Big Crunch universe which has no asymptotic future or past boundaries, see Figure \ref{fig:closeduniverse}. Consequently, in order to understand what happens to the observer in the Lorentzian spacetime, we need to study the overlaps between the states $\ket{\beta, \Delta_O^{(i)}}$ defined on asymptotic time-slices and the states defined on time-slices at finite Lorentzian time. 

\begin{figure}
\centering
\includegraphics[width=0.8\textwidth]{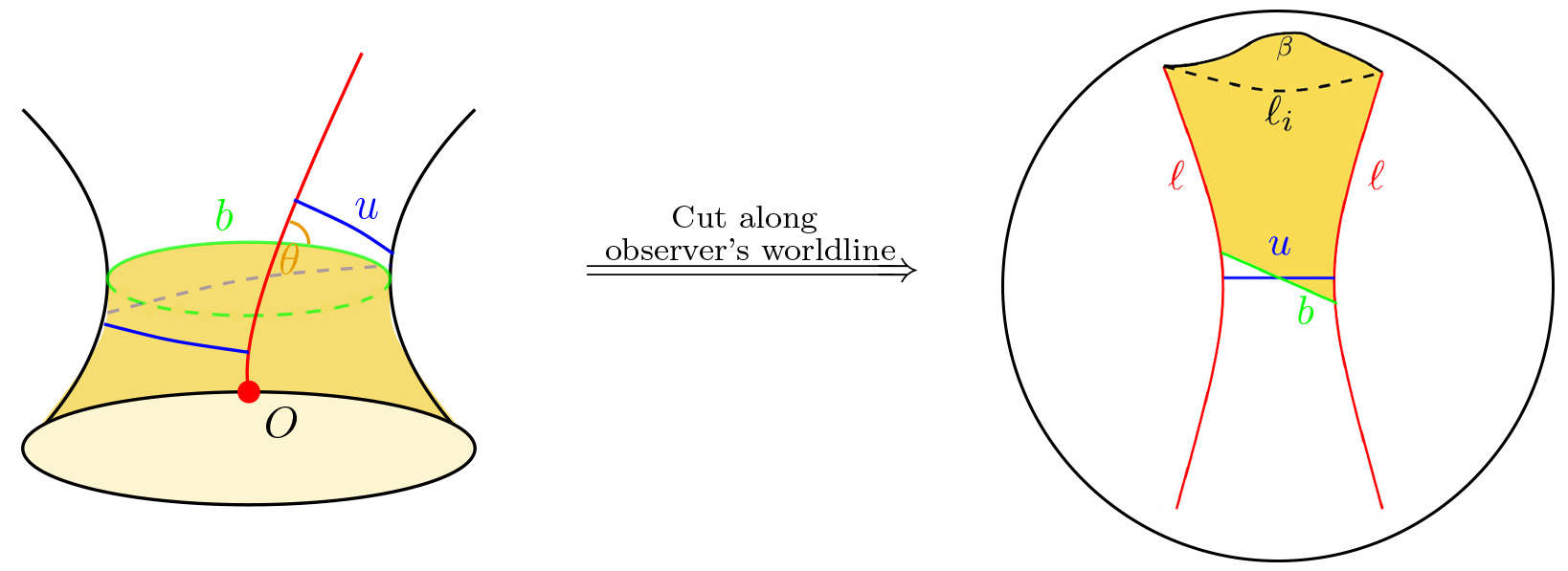}
    \caption{Overlap between an asymptotic closed universe state $|\beta,\Delta_O^{(i)}\rangle$ and a state defined on the closed minimal geodesic (depicted in green) $|b,\Delta_O^{(i)},u\rangle$. Left: the overlap on the Euclidean cylinder. The state on the minimal closed geodesic is labeled by the length $b$ of the closed minimal geodesic and the length $u$ of a geodesic segment wrapping the cylinder once and intersecting the observer's worldline perpendicularly at both ends. Right: the cylinder can be cut along the observer's worldline and embedded in the Poincar\'e disk (depicted in black), with the two red lines identified. The contribution of this geometry to the path integral can be decomposed into a Hartle-Hawking wavefunction $\phi_\beta(\ell_i)$, with $\ell_i$ the length of a geodesic starting and ending at the insertion of the observer on the asymptotic boundary, and a quadrilateral.}
    \label{fig:overlap}
\end{figure}

To start, we can compute the inner-product between a state defined on an asymptotic boundary with a proper length $\beta$ where the observer is inserted and the state defined along the closed minimal geodesic on the Euclidean cylinder (i.e., the closed maximal geodesic in the closed universe) whose length is $b$ (see Figure \ref{fig:overlap}). To define this inner-product we also have to specify how the worldline of the observer intersects this closed geodesic. One can parametrize this intersection in terms of the angle that the worldline of the observer makes (in the geodesic approximation) with the minimal closed geodesic. To simplify the computation, instead of using this angle, we will use the length $u$ of the geodesic segment that is perpendicular to the worldline of the observer at both ends of the segment and wraps the cylinder once; this is the minimal geodesic segment which wraps the cylinder once while starting and ending at some point on the worldline. The relation between this angle and the length $u$ is given by the hyperbolic law of cosines,
\be 
\label{eq:angle-worldline-b-and-u}
\cos(\theta) = \sqrt{\frac{\cosh(b)-\cosh(u)}2} \csch\left(\frac{b}2\right)\,.
\ee
We will denote the resulting state on the closed minimal geodesic by $|b, \Delta_O^{(i)}, u\rangle$. 

\enlargethispage{3\baselineskip}
 
As shown in Figure \ref{fig:overlap}, we can decompose the hyperbolic geometry associated with the overlap $\langle b,\Delta_O^{(i)},u|\beta,\Delta_O^{(j)}\rangle$ in terms of a Hartle-Hawking wavefunction $ \phi_\beta(\ell_i)$, with $\ell_i$ being the renormalized length of the geodesic that starts and ends at the intersection between the asymptotic boundary and the observer's worldline, and a quadrilateral that has two right angles on one of the edges with length $u$ that is opposite to the geodesic with renormalized length $\ell_i$.  Within this quadrilateral, all lengths and angles can be determined in terms of $\ell_i$, $u$, and $b$. Note that this quadrilateral has the same area as the quadrilateral bounded by $b$, $\ell_i$ and half of the observer's worldline, $\ell/2$. This is because the area is completely determined by the angles at the intersection between $\ell$ and the boundary particle. Let us call $Q_b$ the quadrilateral which includes $b$, and $Q_u$ the quadrilateral which includes $u$. In order to compute the overlap, we need the contribution of $Q_b$ to the path integral, which is just the area $A_{Q_b} = A_{Q_u}$. Computing $A_{Q_u}$, we find that the overlap between the asymptotic state and the state defined on the closed geodesic is 
\be 
\overline{\braket{b, \Delta_O^{(i)}, u}{\beta, \Delta_O^{(j)}}} &= \frac{  \delta^{ij} }{16} \int_{-\infty}^\infty d\ell_i \varphi_\beta(\ell_i)  e^{\Delta_O \log(e^{-\frac{\ell_i}2} \sinh \frac{u}2)} e^{-4e^{-\frac{\ell_i}2} \cosh(b/2)}  \frac{\coth^{\frac{1}2}\left(\frac{u}2\right) \sinh^{\frac{1}2}\left(\frac{b}2\right)}{2^{\frac{1}4} (\cosh(b)-\cosh(u))^{\frac{1}4}} \nn \\  & +O(e^{-2S_0})\,,
\label{eq:overlap-beta-and-b}
\ee
at leading order. The first exponential accounts for the weight of the worldline whose renormalized geodesic length is determined by $\ell_i$ and $u$ and, from a repeated application of the hyperbolic law of cosines, is given by $-\log(e^{-\ell_i} \sinh \frac{u}2)$. The second exponential comes from the two angle terms between the worldline of the observer and the geodesic of length $\ell_i$, which determine the area of the quadrilateral. The remaining $u$-dependent and $b$-dependent measure factors account for the fact that the overlap $\overline{\braket{\beta', \Delta_O^{(i)}}{\beta, \Delta_O^{(j)}}}$ in \eqref{eq:overlapclosed} is consistent with \eqref{eq:overlap-beta-and-b} if a resolution of the identity is inserted in the $\ket{b, \Delta_O^{(i)}, u}$ basis.\footnote{This measure factor can also be derived by starting with the Weil-Petersson measure in terms of $b$ and the twist $\tau$ between the bra and ket boundaries, $\Omega_{WP} = db \wedge d\tau$. If we say that the relative twist is zero when the observer's worldline crosses $b$ orthogonally, then more generally, the twist directly determines the angle at which the observer's geodesic meets $b$. Using \eqref{eq:angle-worldline-b-and-u}, we can then solve for the twist, $\tau = \tau(u,b)$, as a function of $b$ and $u$. Using hyperbolic trig-identities, we find the expression \vspace{-0.1cm}\begin{align}
    \tau(b,u) = \log \left({\sqrt{\left({\cosh b - \cosh u}\right)/{2}}+\sinh(b/2)} \right)- \log \left({\sinh(u/2)}\right)\quad.
\end{align} 
Plugging this into the Weil-Petersson measure and then absorbing a square root of this measure into the wavefunction $\overline{\braket{b, \Delta_O^{(i)}, u}{\beta, \Delta_O^{(j)}}}$, we get the same factor as in \eqref{eq:overlap-beta-and-b}.}

Here,  we have normalized the states $\ket{b, \Delta_O^{(i)}, u}$ as 
\be 
\overline{\braket{b, \Delta_O^{(i)}, u}{b', \Delta_O^{(j)}, u'}} = \delta^{ij}\delta(b-b') \delta(u-u')+ O(e^{-2S_0})
\ee
At a perturbative level, this implies that the state $\ket{\beta, \Delta_O^{(i)}}$ defined on the asymptotic boundary can be expressed as a linear combination of states defined on the minimal closed geodesic time slice:
\begin{align} 
\ket{\beta, \Delta_O^{(i)}}&=_\text{pert.}  \frac{1 }{16} \int_{-\infty}^\infty d\ell_i \varphi_\beta(\ell_i) \int_0^{\infty} d u \coth^{\frac{1}2}\left(\frac{u}2\right)   e^{\Delta_O \log(e^{-\frac{\ell_i}2} \sinh \frac{u}2)}  \nn  \\ &\quad  \times \int_{u}^{\infty} { db}\,  \frac{\sinh^{\frac{1}2}\left(\frac{b}2\right)}{2^{\frac{1}4} (\cosh(b)-\cosh(u))^{\frac{1}4}} \times e^{-4e^{-\frac{\ell_i}2} \cosh(b/2)}  \ket{b, \Delta_O^{(i)}, u}\,.
\label{eq:change-of-basis-1}
\end{align}
We can also do the reverse and express $\ket{b, \Delta_O^{(i)}, u}$ as a linear combination of asymptotic states. This can be done more easily by looking at asymptotic states with a fixed ADM energy, $\ket{E, \Delta_O^{(i)}}$, which are related to the states with fixed proper length by  $\ket{\beta, \Delta_O^{(i)}} = \int dE \rho_0(E) e^{-\beta E}\ket{E, \Delta_O^{(i)}} $.\footnote{As we explained in Section \ref{sec:review}, in this convention, the states with fixed ADM energy are normalized as $\braket{E', \Delta_O^{(j)}}{E, \Delta_O^{(i)}} = \frac{\delta(E-E')}{\rho_0(E)}  \gamma_{\Delta}(E,E) $.} Thus,  $\ket{b, \Delta_O^{(i)}, u}$ can be expressed as, 
\be 
 \ket{b, \Delta_O^{(i)}, u} &=_\text{pert.} \frac{ \coth^{\frac{1}2}\left(\frac{u}2\right) \sinh^{\frac{1}2}\left(\frac{b}2\right)}{16 \times 2^{\frac{1}4} (\cosh(b)-\cosh(u))^{\frac{1}4}}  \int dE \frac{\rho_0(E)}{\gamma_{\Delta_O^{(i)}}(E, E)} \int d\ell_i\varphi_E(\ell_i) e^{\Delta_O \log(e^{-\frac{\ell_i}2} \sinh \frac{u}2)}  \nn  \\ &\quad   \times e^{-e^{-4\frac{\ell_i}2} \cosh(b/2)}  \ket{E, \Delta_O^{(i)}}\,.
 \label{eq:change-of-basis-2}
\ee
Both equations \eqref{eq:change-of-basis-1} and \eqref{eq:change-of-basis-2} should be understood as a consequence of the Wheeler-DeWitt constraint, which in a closed universe relates states on all spatial slices related purely by bulk evolution (which are therefore equivalent under diffeomorphisms). They are the closed universe analog of equation \eqref{eq:wvfn_beta_l} which, in the context of two-sided black holes, related the states defined on asymptotic boundaries to those defined on geodesic slices, probing the black hole interior.

\begin{figure}
\centering
\includegraphics[width=0.3\textwidth]{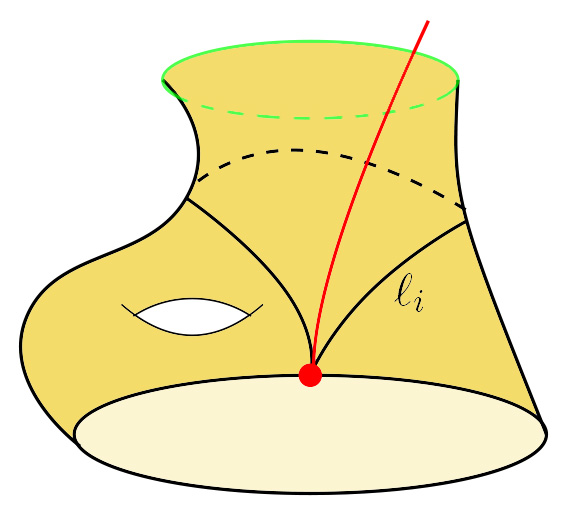}
\caption{Non-perturbative corrections to the overlap $\langle b,\Delta_O^{(i)},u|\beta,\Delta_O^{(j)} \rangle$. The geodesic slice of length $b$ is depicted in green. Any higher-genus geometry contributing to the overlap can be decomposed into a genus-zero quadrilateral similar to Figure \ref{fig:overlap} and a region with a higher genus homologous to the boundary. The contribution of geometries with higher topology is then encoded in the Hartle-Hawking wavefunction.}
\label{fig:nonoverlap}
\end{figure}

We now derive the non-perturbative corrections to equation \eqref{eq:overlap-beta-and-b}. This can easily be done by noticing that, for a given worldline of the observer, we can always choose the unique geodesic $\ell_i$ that starts and ends at the location of the observer on the asymptotic boundary in the preparation of the bra and is in the same homotopy class as the closed geodesic of length $b$ in the preparation of the ket (see Figure \ref{fig:nonoverlap}). Thus, the surface bounded by the geodesic of length $\ell_i$ and the closed geodesic of length $b$ has the topology of a cylinder. For this cylinder, we can use the same decomposition into a quadrilateral as the one used to compute \eqref{eq:overlap-beta-and-b} and a Hartle-Hawking wavefunction that takes into account the contribution of other topologies. All we need to do is replace  $\varphi_\beta(\ell_i) $ in \eqref{eq:overlap-beta-and-b} with
\be\label{eqn:cgwavefunction}
\overline{\varphi_\beta(\ell_i)}=  \int dE \overline{\rho_0(E)} \varphi_E(\ell_i) e^{-\beta E}\,.
\ee
where $\overline{\rho_0(E)}$ is given by the average density of states (multiplied by $e^{-S_0}$) with higher topology corrections included.  A similar logic applies both when computing the coarse-grained value of the overlap 
and when computing higher moments of such overlaps.\footnote{
For example, higher powers of the overlap $\braket{b, \Delta_O^{(i)}, u}{\beta, \Delta_O^{(j)}}$ are thus determined by 
\be \label{eqn:nonpertoverlapstatistics}
\overline{\braket{b, \Delta_O^{(i)}, u}{\beta, \Delta_O^{(j)}}^k} = \frac{  \delta^{ij} }{16} \int \left(\prod_{i=1}^k d\ell_i\right) \overline{\varphi_\beta(\ell_1) \dots \varphi_\beta(\ell_k) }   \prod_{i=1}^k \bigg(e^{\Delta \log(e^{-\frac{\ell_i}2} \sinh \frac{u}2)}  \frac{\coth^{\frac{1}2}\left(\frac{u}2\right)  \sinh^{\frac{1}2}\left(\frac{b}2\right)}{2^{\frac{1}4} (\cosh(b)-\cosh(u))^{\frac{1}4}} e^{-4e^{-\ell_i/2} \cosh(b/2)} \bigg)
\ee
where $ \overline{\varphi_\beta(\ell_1) \dots \varphi_\beta(\ell_k) } $ is determined by the spectral correlator  $\overline{\rho_0(E_1) \dots \rho_0(E_k)}$ that includes higher topology corrections.
} Because non-trivial contribution to all these statistics only involves the spectral correlators  $\overline{\rho_0(E_1) \dots \rho_0(E_k)}$, the statistics agree with coarse-graining or ensemble-averaging the spectrum of a single theory.  If, for a moment, one assumes that the statistics of these inner products come from an ensemble of theories that have a discrete spectrum $\{E_0,\, E_1, \dots\}$ (or by coarse-graining this spectrum in a single theory), then we can write the states defined on the closed geodesic in terms of the states $\ket{E_i, \Delta_O^{(i)}}$ that capture the possible discrete values of the ADM energy on the asymptotic boundary. For example, the non-perturbative version of equation \eqref{eq:change-of-basis-2} becomes 
\be 
 \ket{b, \Delta_O^{(k)}, u} &=_\text{non-pert} \frac{ \coth^{\frac{1}2}\left(\frac{u}2\right) \sinh^{\frac{1}2}\left(\frac{b}2\right)}{16 \times 2^{\frac{1}4} (\cosh(b)-\cosh(u))^{\frac{1}4}}   \sum_{E_i}\frac{1}{\gamma_\Delta(E_i, E_i)} \int d\ell_i\varphi_{E_i}(\ell_i) e^{-\ell_i}   e^{\Delta_O  \log(e^{-\frac{\ell_i}2} \sinh \frac{u}2)}  \nn  \\ &\quad   \times e^{-e^{-4\frac{\ell_i}2} \cosh(b/2)}  \ket{E_i, \Delta_O^{(k)}}\,.
 \label{eq:change-of-basis-non-pert}
\ee
This can be viewed as the non-perturbative generalization of the Wheeler-DeWitt constraint in the presence of an observer we discussed above. A consequence of equation \eqref{eq:change-of-basis-non-pert} is that the states $ \ket{b, \Delta_O^{(k)}, u}$ no longer form an orthogonal basis. Instead, just like the geodesic states in two-sided black holes, these states form an overcomplete set of states that still allow us to do concrete calculations to determine the state on a bulk slice in the presence of the observer. While the assumption of a discrete energy spectrum seems strong at first sight, in Section \ref{sec:hilbert} we will use the non-trivial statistics obtained from the gravitational path integral to determine the dimension of the Hilbert space $\Hr$ and show that it is in fact finite. 
Later, in Section \ref{sec:nullstates}, we will show that these statistics are consistent with a semi-positive definite inner-product in which one has to eliminate the null states -- i.e.,~states with a non-zero norm at the perturbative level that have zero norm once the inner-product is modified at the non-perturbative level -- to get a well-defined Hilbert space. There, we will also comment on how to construct such null states---for instance, by taking non-trivial combinations of the states  $\ket{b, \Delta_O^{(k)}, u}$.

Why stop at defining the states of the gravitational theory on a closed geodesic slice? For example, we might be interested in characterizing the spatial slice on which the observer resides at a Lorentzian time $t$ shown by their clock. Suppose we define the state of the observer's clock at time $t=0$ to be $\ket{O_{t=0}} = \sum_i c_i \ket{\Delta_i}$ and the state at time $t$ to consequently be $\ket{O_{t}} = \sum_i c_i e^{-i \Delta_i t} \ket{\Delta_i}$. We will assume that in the state $\ket{\beta, O_{t=0}}$ the observer is prepared in such a way that their clock shows $t=0$ on the minimal closed geodesic slice on the Euclidean cylinder, which in Lorentzian signature corresponds to the maximal, time-reflection symmetric slice of the crunching universe. We would like to once again characterize the state at time $t$ in terms of the proper length $a$ of the geodesic time slice on which they reside at time $t$ and the angle at which the worldline of the observer intersects this geodesic slice.  We will denote this geodesic timeslice by $\Sigma$. Though dressing to the observer's clock only makes sense in Lorentzian signature, for computational purposes we can perform the calculation in Euclidean signature first, where the slice $\Sigma$ is at a time $\tau$ along the observer's worldline, and then analytically continue our result to Lorentzian signature, by setting $\tau = -i t$, as shown in Figure \ref{fig:timestates}.

\begin{figure}
\centering
\includegraphics[width=0.8\textwidth]{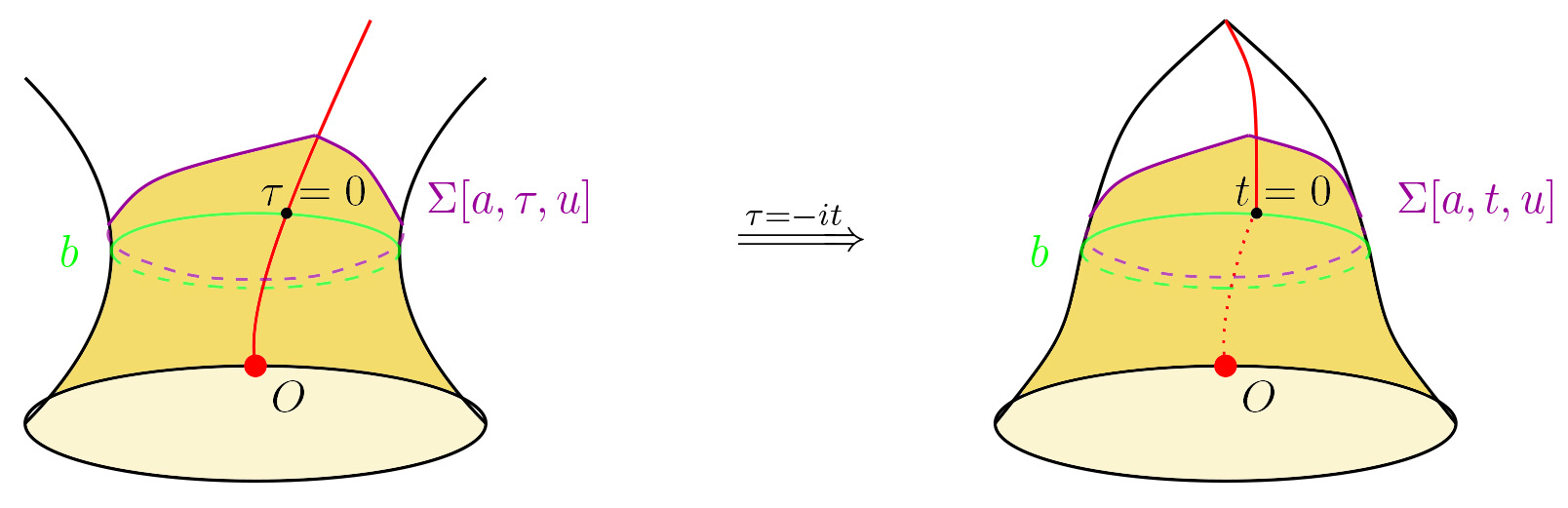}
\caption{
A state on a closed geodesic time slice $\Sigma$ homotopic to a minimal closed geodesic, whose length is $b$ (depicted in green), is characterized by its length $a$ and by the angle between $\Sigma$ and the observer's worldline. In general this geodesic is not smooth for $\tau\neq 0$ (Euclidean -- \textit{left side}) or $t\neq 0$ (Lorentzian -- \textit{right side}). The angle can again be specified in terms of the length of a geodesic segment wrapping around the universe once and intersecting the observer's geodesic perpendicularly at both ends. The state on $\Sigma$ is uniquely determined by the state on the minimal closed geodesic of length $b$. To make the worldline path real in Lorentzian signature, one has to complexify the path in Euclidean signature; for this reason, on the right side, we show the observer worldline by a dotted curve, signifying that that part of the trajectory is complexified. 
}
\label{fig:timestates}
\end{figure}

$\Sigma$ starts and ends at the location of the observer, wrapping the closed universe once. At $t=0$, $\Sigma$ is a minimal closed geodesic. For $t\neq 0$ this geodesic is not smooth; rather, it has a kink at the location of the observer. As we shall explain, the properties of $\Sigma$ are fully determined in terms of the state on the minimal closed geodesic that lies within the same homotopy class -- as above, we will take the length of this minimal closed geodesic to be $b$.\footnote{Note that for 2D surfaces with constant negative curvature, there always exists a closed smooth geodesic (whose length we call $b$) homotopic to $\Sigma$. } We can characterize how the worldline meets $\Sigma$ in terms of their intersection angle. This is once again determined by the length $u$ of the geodesic that is perpendicular to the worldline of the observer at both ends and that, together with the worldline segment in between the two ends, forms a closed curve that is also in the same homotopy class as $\Sigma$. As mentioned above, the value of $u$ also determines the angle at which the worldline of the observer meets the closed geodesic of length $b$. 

When going from Euclidean to Lorentzian signature there is an additional subtlety -- the angle $\theta$ at which the worldline of the observer intersects the minimal closed geodesic of length $b$ has to be complexified and needs to take the form $\frac{\pi}2 + i  \mathbb R$ in order for the observer's worldline to correspond to a real trajectory in Lorentzian signature. From \eqref{eq:angle-worldline-b-and-u}, we see that this can simply be achieved by taking $u >b$; while such a geometry does not make sense purely in Euclidean signature, it can be achieved in 
a complexified geometry. Thus, to prepare a state on the minimal geodesic slice in which the observer travels along a real trajectory in Lorentzian signature, we set $u>b$ in \eqref{eq:change-of-basis-non-pert}.

Having addressed this subtlety, we can now discuss how the state prepared on minimal geodesic slices is related to the state defined at finite time. If we keep $u$ fixed we can now easily express the length $a$ of $\Sigma$ in terms of the Lorentzian proper time $t$ measured along the worldline, with $t=0$ at the minimal closed geodesic slice, 
\be 
\cosh(a) = \cosh(b) - \sin(t)^2\left(\cosh(u)-1\right)\,.
\label{eq:length-of-geodesic-timeslice}
\ee
When $b=u$, the observer hits the time-reflection symmetric slice perpendicularly (i.e., it is at rest), and we see that the size of the geodesic time slice vanishes ($a=0$) at $t=\frac{\pi}2$; this is because the observer at rest hits the crunch singularity at this time. Setting $u>b$ the observer leaves the time-reflection symmetric slice at an angle and we see that the size of the geodesic time slice first vanishes at an earlier time. This is because the proper time that the observer experiences until they hit the singularity decreases once they are no longer at rest.  Past the time at which the observer hits the singularity, we see that $a<0$, and we can no longer make sense of the spacetime as corresponding to a real Lorentzian geometry. 

Since both the intersection angles with $\Sigma$ and the length $a$ of $\Sigma$ are fully determined in terms of $b$ and $u$, the geometry between $\Sigma$ and the closed geodesic of length $b$ is rigid. Therefore, by using \eqref{eq:length-of-geodesic-timeslice}, we can now express the state of the bulk and the observer $\ket{a, O_{t}, u}$, at a fixed proper Lorentzian time $t$, by replacing $b$ in \eqref{eq:change-of-basis-2} in terms of $a$, $t$ and $u$. Note that even though $t$ and $\Delta_O^{(i)}$ are conjugate to each other, we are considering a limit where the spacing of the eigen-energies of the clock is sufficiently dense and close to $\Delta_O$ such that we can keep $t$ fixed while replacing $\Delta_O^{(i)}$ by $\Delta_O$. We, therefore, obtained a new (over)complete basis of states, labeled by the time read by the observer's clock. It would be interesting to understand whether the overcompleteness of such states plays an important role in understanding the fate of the observer close to the crunch singularity.\footnote{We will further comment on this topic in Section \ref{sec:example-obs-closed-universe} where we discuss how correlation functions measured along the worldline of the observer can be used as a probe of this singularity.}

\section{The Hilbert space of quantum gravity from an observer's point-of-view}\label{sec:hilbert}

In this section, we compute the dimension of the relational Hilbert space $\Hr$ using the statistics of our modified inner product. As anticipated in Section \ref{sec:setup}, we find that the dimension of $\Hr$ is much larger than the dimension of $\Hnonp$ obtained in Section \ref{sec:review}. We start with the simple closed universe case in Section \ref{sec:closeduniverse} before moving on to the more involved two-sided black hole case in Section \ref{sec:twosided}.

We saw in Section \ref{sec:setup} that considering a finite window for the scaling dimension $\Delta_O$, and accordingly labeling states by a bulk geodesic length and the time read by the observer's clock is important when calculating correlation functions because operator insertions must be appropriately dressed to the observer's worldline. However, in the present section, we are only interested in computing properties of the Hilbert space $\Hr$ and not observables. For this purpose, we can simplify our analysis by working with the basis of asymptotic states discussed at the beginning of Section \ref{sec:setup}, which are entirely specified by their boundary conditions at the asymptotic boundary. In particular, we will take the scaling dimension $\Delta_O$ of the observer operator to be fixed for all states. We will go back to the study of dressed observables in Section \ref{sec:observables}. Moreover, we will focus on the generic case of JT gravity coupled to matter and restrict our attention to an arbitrary but finite energy window, for which the non-perturbative Hilbert space is finite-dimensional. We can then fix the length of the asymptotic boundary to be $\beta$ for all states and simply label states by the indices of matter insertions at the asymptotic boundary. 

In order to prove that $\Hr$ is, in fact, a Hilbert space, we need to show that the inner product computed by the gravitational path integral using our new rules is positive-definite. We will show in Section \ref{sec:nullstates} that the inner product is, in fact, positive semi-definite due to the presence of null states and $\Hr$ is a well-defined Hilbert space after quotienting out the set of null states. 
We can ignore these details for now because the resolvent calculation of $\dim\left(\Hr\right)$ is valid at the level of vector spaces and is thus independent of the positive-definiteness of the inner product.

\subsection{Closed universe}
\label{sec:closeduniverse}

Let us start by examining a closed universe in JT gravity coupled to matter. The most generic state is prepared by the path integral with one boundary of length $\beta$ where we insert an operator $\mathcal{O}_O$ associated with the observer and a matter operator $\mathcal{O}_{i}$ carrying an index $i$, see Figure \ref{fig:closedstate}. We will take the dimensions of the two operators to be $\Delta_O$ and $\Delta_m$, respectively.
We will label such a state by $\ket{\psi_i}$, where $i=1,...,K$ is a flavor index of the matter operator. We will take the two insertions to be antipodal on each boundary, namely $\beta_l=\beta_r=\beta/2$, where $\beta_l$ and $\beta_r$ are boundary lengths between the insertions of $\mathcal{O}_O$ and $\mathcal{O}_i$ to the left and right of the observer, respectively. As we will see, for sufficiently large $K$, this set of states spans the entire quantum gravity non-perturbative Hilbert space in an arbitrary but finite energy window. Considering states with arbitrary $\beta_l$ and $\beta_r$ would not enlarge the Hilbert space, but simply give us a different basis.\footnote{In terms of the infinite-dimensional, perturbative Hilbert space $\Hrelp$, fixing $\beta_l=\beta_r=\beta/2$ restricts us to a subspace of the full Hilbert space.}

\begin{figure}[h]
    \centering
    \includegraphics[width=0.5\linewidth]{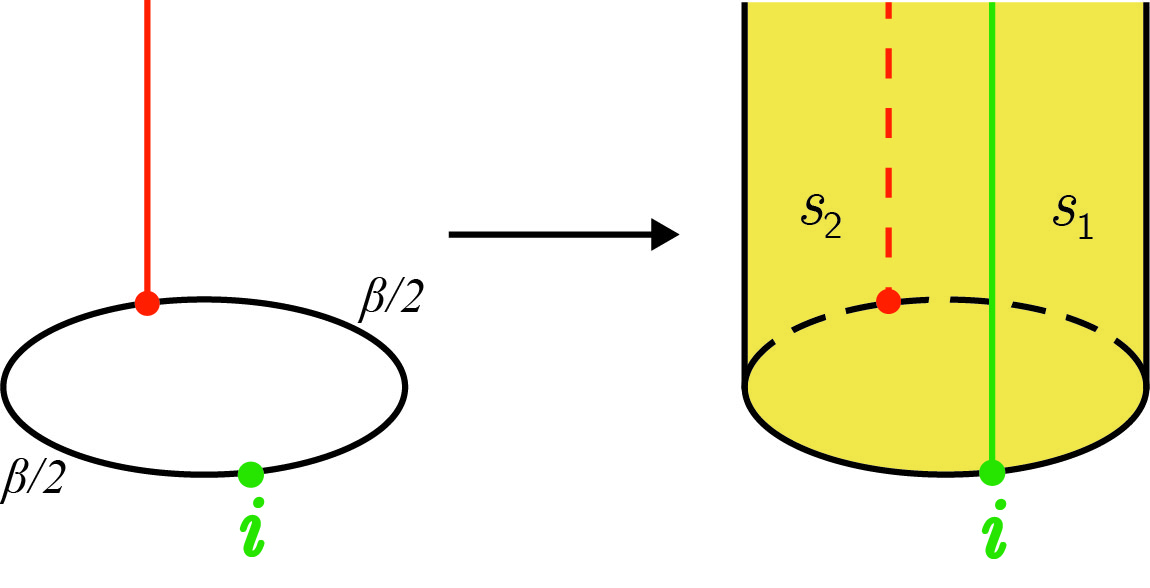}
    \caption{Path integral preparing a generic closed universe state $\ket{\psi_i}$ in the presence of an observer and a matter insertion labeled by $i$. Left: boundary conditions for the path integral, with a worldline for the observer (depicted in red) and the insertion of a matter operator with flavor $i$ and scaling dimension $\Delta_m$ at the asymptotic boundary. The insertions are at antipodal points. Right: the gravitational path integral prepares a state for the closed universe satisfying these boundary conditions. We indicate here the energies $s_1$ and $s_2$ of the patches to the left and right of the observer. We omit these labels in the rest of the Figures in this section, but the same convention is used in all geometries, leading to equations \eqref{eq:Zclosed} and \eqref{eq:tildex}.}
\label{fig:closedstate}
\end{figure}

We will also assume $\Delta_m,\Delta_O\gg S_0$. This assumption guarantees that we can safely neglect contributions to the path integral with intersecting geodesics because they are exponentially suppressed in $\Delta_m$, $\Delta_O$ and, therefore, subdominant with respect to higher genus contributions with no intersections. As we will see, this assumption greatly simplifies our computation of the dimension of $\Hr$. On the other hand, we expect the final answer for $\dim\left(\Hr\right)$ to be independent of the size of the scaling dimensions $\Delta_O,\Delta_m$, as we will discuss at the end of the present subsection.

Let us now compute the dimension of the non-perturbative, relational Hilbert space $\Hr$ for an observer in a closed universe using a resolvent calculation. Imposing that observer worldlines must connect between a bra and the corresponding ket, the leading contribution in the gravitational path integral to the connected $n$-th moment of an overlap,\footnote{The overline indicates that the quantity is computed using the gravitational path integral and therefore averaged over the dual matrix ensemble. Note that indices are not summed over in this formula.}  $\bar{\langle\psi_i|\psi_k\rangle\langle\psi_k|\psi_l\rangle...\langle\psi_m|\psi_j\rangle}_{conn}\approx Z_n^{\textrm{closed}}\delta_{ij}$ is given by a genus-zero pinwheel geometry with  $2n$ boundaries, see Figure \ref{fig:pinwheelclosed}.

\begin{figure}[h]
    \centering
    \includegraphics[width=0.3\linewidth]{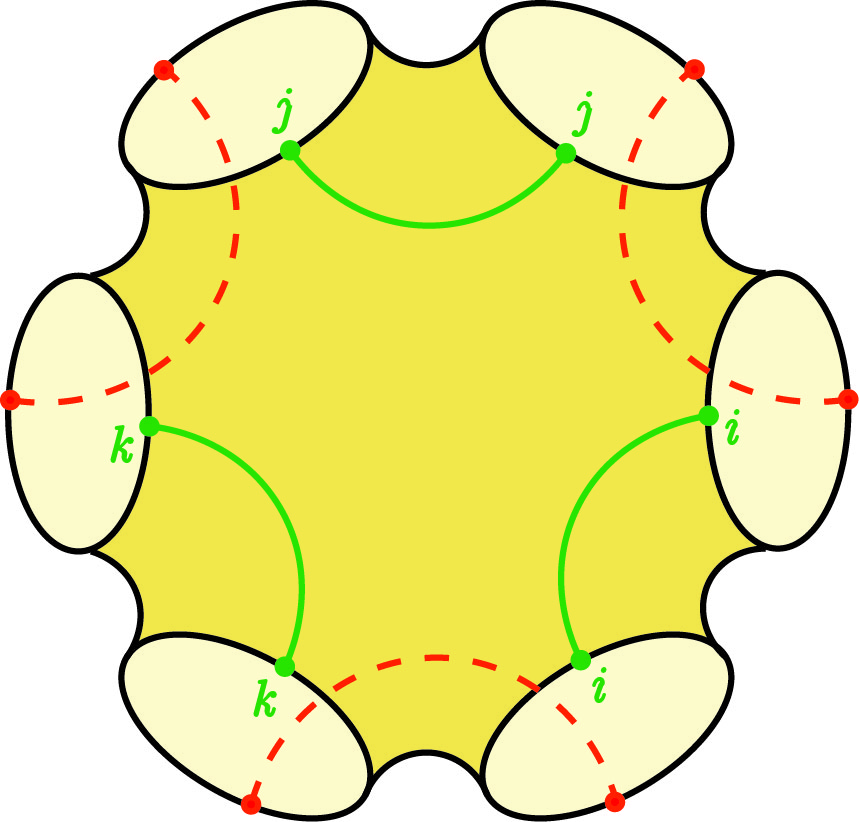}
    \caption{Genus-zero pinwheel geometry contributing to the $n=3$ moment of an overlap $\overline{\langle\psi_i|\psi_j\rangle\langle\psi_j|\psi_k\rangle\langle\psi_k|\psi_i\rangle}$ (no sum over indices). The observer's worldlines (depicted in red) must connect between a bra and the corresponding ket. Matter geodesics (depicted in green) can connect between arbitrary bras and kets.}
    \label{fig:pinwheelclosed}
\end{figure}

Notice that this geometry is analogous to that depicted in equation \eqref{eq:zn_matter_nonpert} 
and relevant for the calculation of $\dim\left(\Hnonp\right)$ in the two-sided black hole case (see Appendix \ref{app:reviewBH}). The only differences are that the geometry computing the $n$-th momentum now has $2n$ boundaries (instead of $n$ boundaries), and that the type of geodesic connecting the asymptotic boundaries now alternate between matter and observer geodesics. Just like in equation \eqref{eq:zn_matter_nonpert}, this surface has two connected patches separated from each other by matter and observer geodesics, and therefore two different energies to be integrated over.
The connected contribution $Z_n^{\textrm{closed}}$, which will play a central role in our resolvent calculation, thus takes the form\footnote{All the integrals over energies (i.e., over $s$) in this section are restricted to an arbitrary but finite energy window.}
\begin{equation}
\begin{aligned}
    Z_n^{\textrm{closed}}&=e^{(2-2n)S_0}\int ds_1ds_2\rho_0(s_1)\rho_0(s_2)\left[e^{-\frac{\beta s_1^2}{2}}e^{-\frac{\beta s_2^2}{2}}\gamma_{\Delta_O}(s_1,s_2)\gamma_{\Delta_m}(s_1,s_2)\right]^n\\
    &=e^{2S_0}\int ds_1ds_2\rho_0(s_1)\rho_0(s_2)\tilde{x}^n(s_1,s_2)
    \label{eq:Zclosed}
    \end{aligned}
\end{equation}
where we took into account that the pinwheel in Figure \ref{fig:pinwheelclosed} has genus zero and $2n$ boundaries, and in the second equality we defined
\begin{equation}
    \tilde{x}=e^{-2S_0}e^{-\frac{\beta s_1^2}{2}}e^{-\frac{\beta s_2^2}{2}}\gamma_{\Delta_O}(s_1,s_2)\gamma_{\Delta_m}(s_1,s_2).
    \label{eq:tildex}
\end{equation}
The normalization factors $\gamma_{\Delta_O},\gamma_{\Delta_m}$ associated respectively with the observer and matter worldlines are given in equation \eqref{eq:gamma}. Notice that each $\tilde{x}$ is associated with a pair of boundaries (namely, with an overlap), each boundary has length $\beta$, and we used $E=s^2/2$ like in Section \ref{sec:review}.

Let us now consider a set of $K$ states $\{\ket{\psi_i}\}$ with $i=1,...,K$, define
\begin{equation}
    d^2\equiv e^{2S_0}\int ds_1ds_2\rho_0(s_1)\rho_0(s_2)
    \label{eq:d2closed}
\end{equation}
similar to equation \eqref{eq:d},
and work in the regime $K\to\infty$, $e^{2S_0}\to\infty$ with $K/e^{2S_0}=O(1)$. Using the definition \eqref{eq:r} of the resolvent and computing moments of the overlap using the gravitational path integral, the resolvent takes the diagrammatic form
\begin{align}
\includegraphics[width=\textwidth]{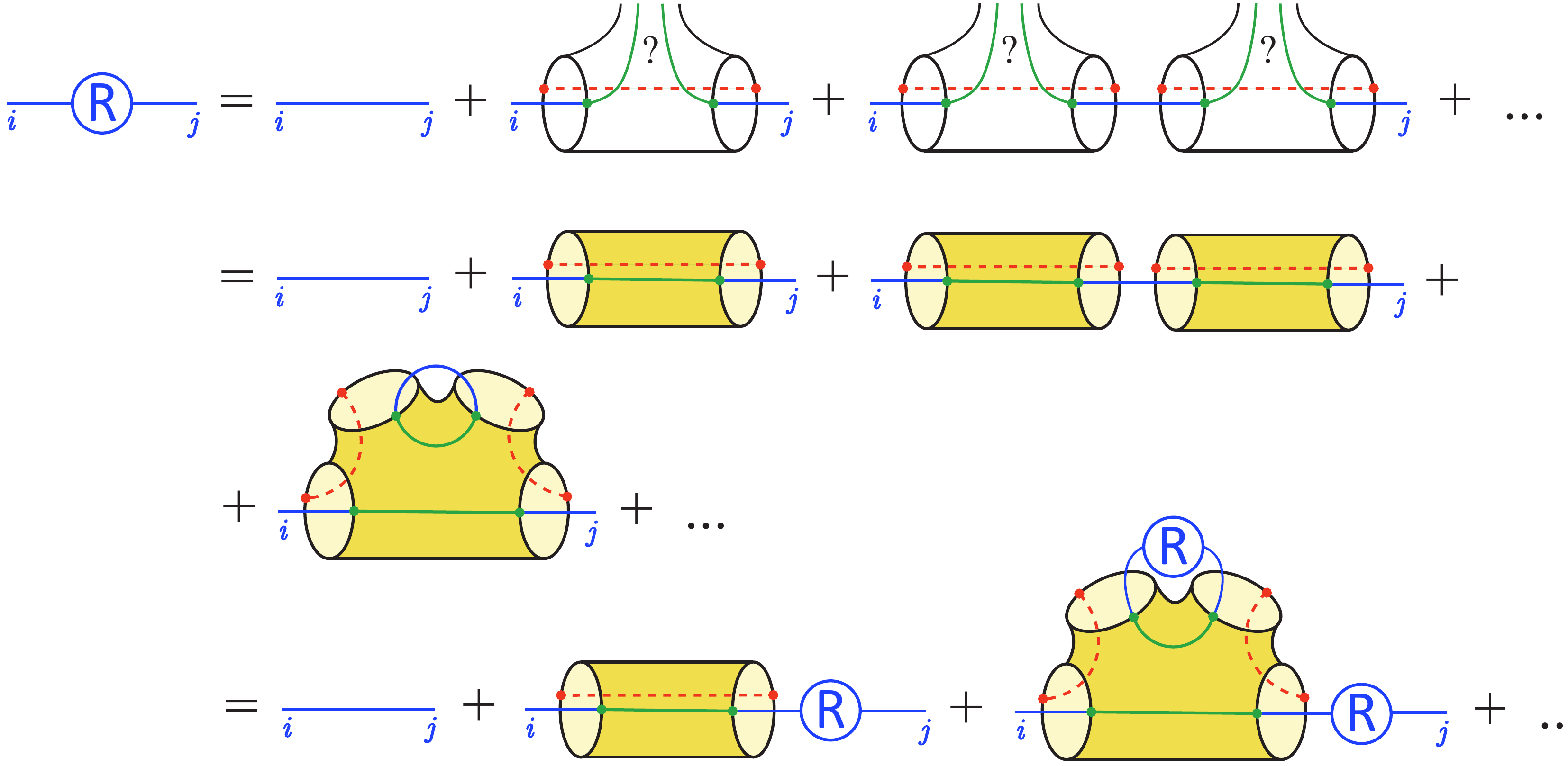}
\end{align}
In the last equality, we rearranged the expansion in terms of the number of boundaries to which the first boundary is connected. This gives us a Schwinger-Dyson equation:
\begin{equation}
    \overline{R_{ij}(\lambda)}=\frac{\delta_{ij}}{\lambda}+\frac{1}{\lambda}\sum_{n=1}^\infty Z_n^{\textrm{closed}}R^{n-1}(\lambda)\overline{R_{ij}(\lambda)}
\end{equation}
where $Z_n^{\textrm{closed}}$ is given in equation \eqref{eq:Zclosed}, we introduced the overline to indicate that the resolvent is computed using the gravitational path integral (and therefore averaging over the dual ensemble of random matrices), and we defined the trace of the resolvent $R(\lambda)=\sum_i\overline{R_{ii}(\lambda)}$. Taking the trace on both sides, using equation \eqref{eq:Zclosed}, and performing the sum over $n$, we obtain
\begin{equation}
    R(\lambda)=\frac{K}{\lambda}+\frac{e^{2S_0}}{\lambda}\int ds_1ds_2\rho_0(s_1)\rho_0(s_2)\frac{R(\lambda)\tilde{x}(s_1,s_2)}{1-R(\lambda)\tilde{x}(s_1,s_2)}.
    \label{eq:closedresolvent}
\end{equation}

Given $R(\lambda)$ in equation \eqref{eq:closedresolvent}, we can compute the rank of the Gram matrix $M$ of overlaps, i.e., the dimension of the Hilbert space $\Hr(K)$ spanned by the $K$ states using equation \eqref{eq:dimH_from_r}. Notice that the analytic structure of the resolvent is identical to that encountered in Appendix \ref{app:reviewBH} when computing $\dim\left(\Hnonp\right)$ in the two-sided black hole case. In particular, $R(\lambda)$ has a branch cut on the positive real axis, it has no pole for $K<d^2$, it has a pole at $\lambda=0$ for $K>d^2$, and it behaves as $R(\lambda)\sim K/\lambda$ for $\lambda\to\infty$ (see \cite{Boruch:2024kvv} and Appendix \ref{app:res_replica} for details). Therefore, for $K<d^2$, the integral over the contour $C=C_0\cup C_\infty$ depicted in Figure \ref{fig:contour} receives a contribution only from $C_\infty$ (the counterclockwise contour at infinity), and we obtain
\begin{equation}
    \dim\left(\Hr(K)\right)=\frac{K}{2\pi i}\oint_{C}\frac{d\lambda}{\lambda}=K, \quad\quad\quad K<d^2.
    \label{eq:dimclosedsmallk}
\end{equation}
On the other hand, for $K>d^2$, $R(\lambda)$ has a pole at $\lambda=0$, and the integral receives contributions from both $C_\infty$ and the clockwise contour $C_0$ around the pole at the origin. The integral over $C_\infty$ is the same as equation \eqref{eq:dimclosedsmallk} and gives a contribution equal to $K$. For the integral over $C_0$, the presence of the pole in $R(\lambda)$ implies that the integrand in the second term of equation \eqref{eq:closedresolvent} is independent of $\lambda$ near $\lambda = 0$, and we obtain
\begin{equation}
    \frac{1}{2\pi i}\oint_{C_0}d\lambda R(\lambda)=\frac{K-d^2}{2\pi i}\oint_{C_0}\frac{d\lambda}{\lambda}=d^2-K,
\end{equation}
where $d^2$ is defined in equation \eqref{eq:d2closed}.
We thus find $\dim\left(\Hr(K)\right)=d^2$ for $K>d^2$. In summary, the dimension of the non-perturbative Hilbert space spanned by $K$ matter states and relevant to describe relational dynamics with respect to an observer in the closed universe is given
\begin{equation}
    \dim\left(\Hr(K)\right)=\begin{cases}
        K \quad\quad\quad K<d^2\\
        d^2 \quad\quad\quad K>d^2
    \end{cases}
    \label{eq:dimclosed}
\end{equation}

A few comments are in order. First, the result \eqref{eq:dimclosed} is similar to that obtained for the Hilbert space of a black hole spanned by microstates of bulk matter \cite{Penington:2019kki,Almheiri:2019qdq,Balasubramanian:2022gmo,Balasubramanian:2022lnw}. In particular, we found that different matter states that are orthogonal to each other at the perturbative level are not orthogonal non-perturbatively due to spacetime wormhole corrections to the inner product. If we consider $K<d^2$ matter states, they span a non-perturbative Hilbert space of dimension $K$, and there are no null states. But if we choose $K>d^2$, the non-trivial overlaps between matter states lead to the existence of null states signaled by the pole of $R(\lambda)$ at $\lambda=0$. The residue $K-d^2$ gives the number of null states among the $K$ basis states considered. The $K$ matter states are overcomplete 
\cite{Penington:2019kki,Almheiri:2019qdq,Balasubramanian:2022gmo,Balasubramanian:2022lnw} in the non-perturbative gravitational Hilbert space $\Hr$, which has dimension $\dim\left(\Hr\right)=d^2$.\footnote{Notice that the dimension of $\Hr$ is finite only when restricting to an arbitrary but finite energy window. From equation \eqref{eq:rhonot} we see that the integral in equation \eqref{eq:d2closed} diverges if we consider an infinite energy window.} We will comment further about null states in this setup in Section \ref{sec:nullstates}. Second, notice that, unlike the global non-perturbative Hilbert space $\Hnonp$ of the closed universe in the absence of an observer considered in Section \ref{sec:reviewclosed}, the non-perturbative relational Hilbert space is non-trivial. 
Remember that $\Hr$ should be regarded as the Hilbert space of gravitational degrees of freedom that an observer in the bulk can interact with, and therefore, the one relevant for the description of the experience of the bulk observer.  Third, it is easy to show that if instead of considering JT gravity coupled to matter, we considered pure JT gravity, we would have obtained a quantum gravity relational Hilbert space of dimension $d$. This is due to the additional constraint that the energies on the two sides of the observer must be equal (i.e., $s_1=s_2$) because there is a single connected bulk patch wrapping around the closed universe.

Let us also comment on the assumption $\Delta_m,\Delta_O\gg S_0$ that allowed us to neglect contributions to the path integral in which worldlines intersect. Relaxing this assumption and taking intersecting geodesics into account would certainly modify the resolvent and its trace $R(\lambda)$. However, our calculation of the dimension of $\Hr$ only relied on the asymptotic behavior $R(\lambda)\sim K/\lambda$ as $\lambda\to\infty$ and on its analytic structure---namely that the only pole is at $\lambda=0$ when $K>d^2$ and the only branch cut on the positive real axis \cite{Boruch:2024kvv}. Although including corrections from intersecting geodesics modifies the resolvent, we expect the analytic properties in the neighborhood of $\lambda=0$, and therefore the result \eqref{eq:dimclosed}, to remain unchanged.\footnote{The universality of the residue of the resolvent at $\lambda=0$ is also manifest if we make the scaling dimensions of all operator insertions different from each other or if we change other details in the preparation of the states $\ket{\psi_i}$ whose statistics are captured by $R(\lambda)$ \cite{Boruch:2024kvv}.}

Finally, we remark that one interesting feature of the full quantum gravity relational Hilbert space $\Hr$ in the presence of matter is that it factorises into a Hilbert space to the right of the observer and a Hilbert space to the left of the observer:
\begin{equation}
    \Hr=\Hrl\otimes \Hrr.
    \label{eq:factorclosed}
\end{equation}
This result, which we explicitly derive in Appendix \ref{app:closedfactorisation}, also holds in the presence of an arbitrary number of matter insertions. This contrasts with pure JT gravity, where $\dim\left(\Hr\right)=d$ and the Wheeler-DeWitt constraint guarantees the existence of a single, non-factorised Hilbert space to the right and left of the observer. The latter case is analogous to the non-factorised, $d$-dimensional Hilbert space for a two-sided black hole in the absence of an observer or matter, and it is due to the additional constraint that the energies to the left and right of the observer are equal. As we will argue in Section \ref{sec:holography}, the result \eqref{eq:factorclosed} suggests that, in the JT gravity with matter setup, a putative holographic dual description living on the worldline of the observer would consist of two entangled, non-interacting copies of a holographic theory.

\subsection{Two-sided black holes }
\label{sec:twosided}

Let us now examine the case of a two-sided black hole in JT gravity coupled to matter. In the presence of the worldline of an observer, the 
most generic state one can consider is prepared by the insertion of a matter operator $\mathcal{O}_{i}$ carrying an index on each side of the observer, see Figure \ref{fig:2sidedstate}. We will label such a generic state by $\ket{\psi_{ii'}}$, where $i$ and $i'$ are the indices of the matter insertions on the right and left side of the observer, respectively.\footnote{Notice that the ordering of the insertions is important, because it determines the specific gluing of the asymptotic boundary and observer worldline when computing overlaps. For example, in the $\Delta_m,\Delta_O\gg 1$ limit of our interest, $\langle\psi_{12}|\psi_{21}\rangle \approx 0$ perturbatively, because the geodesics connecting the matter insertions to each other on the disk need to cross the observer worldline, and are therefore exponentially suppressed in $\Delta_m$ and $\Delta_O$.} 
\begin{figure}[t!]
    \centering
    \includegraphics[width=0.8\linewidth]{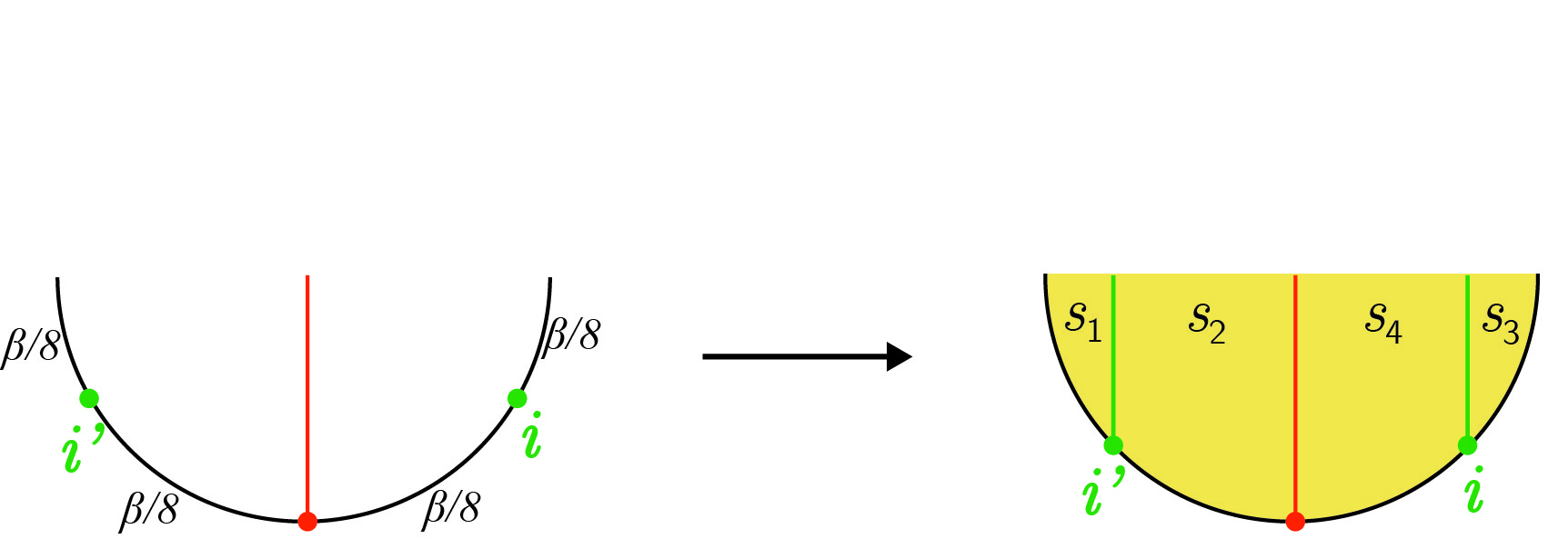}
    \caption{Path integral preparing a state $\ket{\psi_{ii'}}$ for the two-sided black hole in the presence of an observer and matter insertions to the right and left of the observer labeled by $i$ and $i'$. Every boundary segment in the preparation of a state is taken to have length $\beta/8$. Left: boundary conditions for the path integral, with a worldline for the observer (depicted in red) and the insertion of matter operators with flavors $i$ and $i'$ and scaling dimension $\Delta_m$ at the asymptotic boundary. Right: the gravitational path integral prepares a state for the two-sided black hole satisfying these boundary conditions. We indicate here the energies $s_1$, $s_2$, $s_3$, $s_4$ of the four patches. We omit these labels in the rest of the Figures in this section, but the same convention is used in all geometries, leading to equations \eqref{eq:BHZn} and \eqref{eq:tildey}.}
    \label{fig:2sidedstate}
\end{figure}
Similar to the closed universe case, throughout this section, we will take each boundary segment in the preparation of the state to have the same length $\beta/8$, where $\beta$ is the length of a full boundary, see Figure \ref{fig:2sidedstate}.\footnote{Like in the closed universe case, this specific set of states spans only a subspace of the perturbative relational Hilbert space $\Hrelp$, but is an overcomplete basis for the non-perturbative relational Hilbert space $\Hr$.}
We will further assume that the dimension of the matter operators is the same for all insertions---i.e., $\Delta_{i}\equiv \Delta_m$ for all $i$---and that, just like in the closed universe case, the dimension of all the operators considered is large, $\Delta_m,\Delta_O\gg S_0$ and we can therefore neglect contributions to the path integral where geodesics intersect. As we have discussed in Section \ref{sec:closeduniverse}, we do not expect this assumption to affect the analytic structure of the resolvent and, therefore, the computation of the dimension of $\Hr$.

Let us now compute the dimension of the Hilbert space seen by an observer using a resolvent calculation. Imposing that observer worldlines connect a bra to the corresponding ket and that no geodesics intersect, the leading contribution to the connected $n$-th moment of an overlap is given by a surface with $n$ boundaries and genus $n-1$, see Figure \ref{fig:BHsurfaces}. 
\begin{figure}[h]
    \centering
    \includegraphics[width=0.7\linewidth]{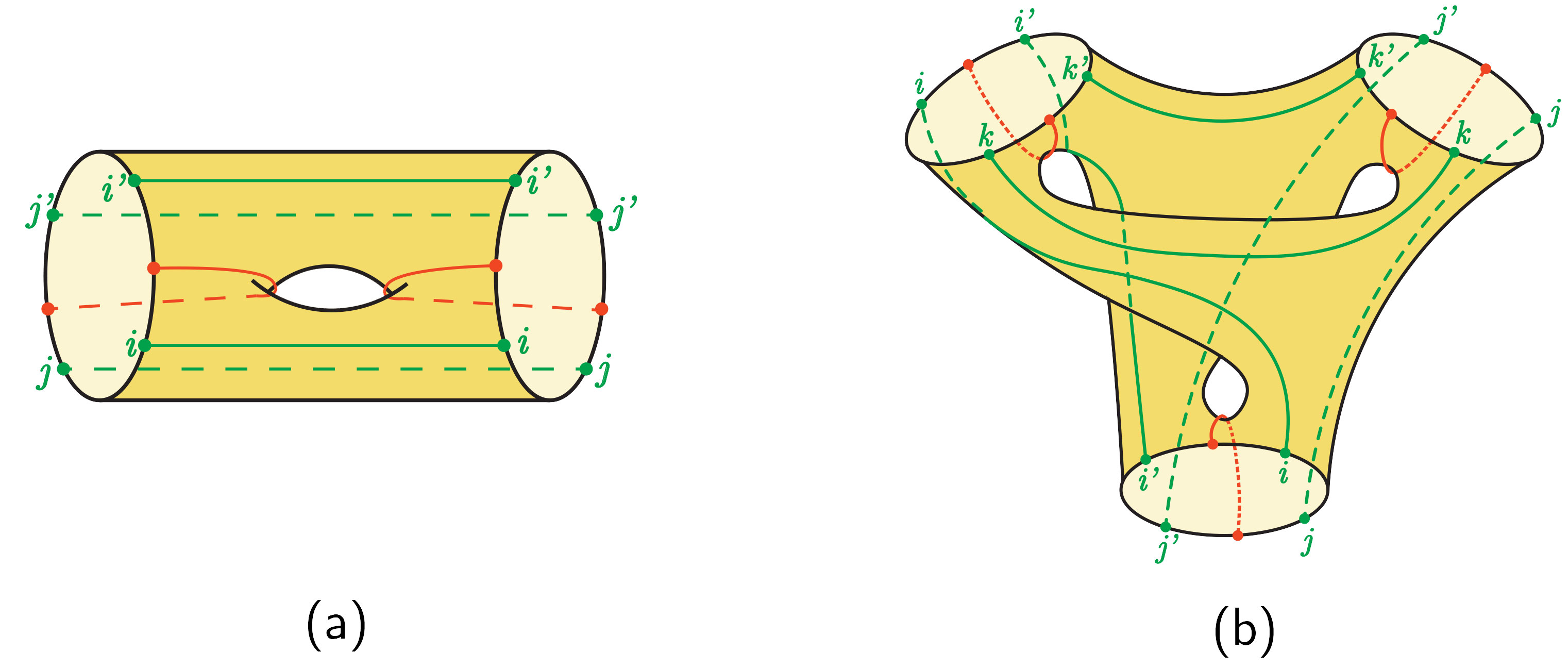}
    \caption{Genus $n-1$ geometries with $n$ boundaries contributing to the $n$-th moment of an overlap in the two-sided black hole in the presence of an observer. The observer's worldlines (depicted in red) must connect between a bra and the corresponding ket. Matter geodesics (depicted in green) can connect between arbitrary bras and kets. Geodesics do not intersect in these leading, connected contributions. (a) Genus $g=1$, 2-boundary geometry contributing to the square of an overlap $\overline{|\langle\psi_{jj'}|\psi_{ii'}\rangle|^2}$. (b) Genus $g=2$, 3-boundary geometry contributing to the $n=3$ moment of an overlap $\overline{\langle\psi_{jj'}|\psi_{ii'}\rangle\langle\psi_{ii'}|\psi_{kk'}\rangle\langle\psi_{kk'}|\psi_{jj'}\rangle}$ (no sum over indices).}
    \label{fig:BHsurfaces}
\end{figure}
This geometry can be built by considering two copies of a $n$-boundary pinwheel geometry where each boundary is composed by a segment associated with an asymptotic boundary and a segment associated with the worldline of an observer, see Figure \ref{fig:gluing}. The genus $n-1$ surface is then obtained by gluing the two pinwheels along the observer worldlines.
\begin{figure}[t!]
    \centering
    \includegraphics[width=0.9\linewidth]{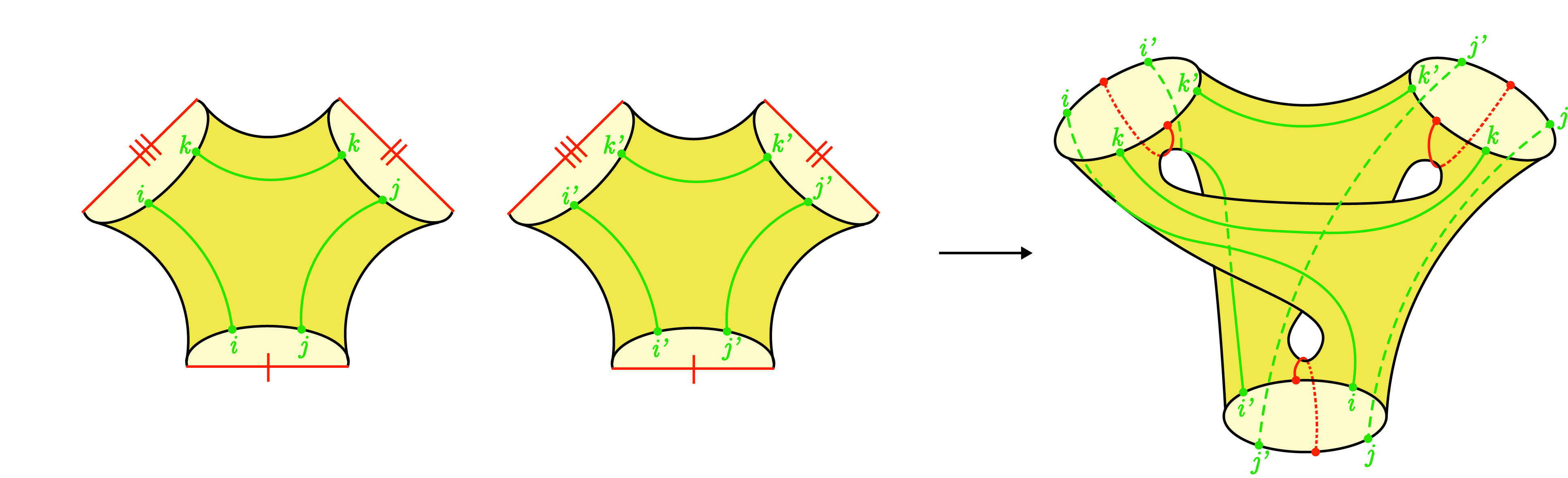}
    \caption{The genus $n-1$, $n$-boundary geometry contributing to the $n$-th moment of an overlap can be obtained by gluing together two genus-0 pinwheel geometries with $n$ boundaries. Here we depict the pinwheels and the resulting glued geometry for $n=3$. Each boundary of the pinwheels is composed of a segment of asymptotic boundary and a segment of observer's worldline. The two pinwheels are glued together along the observer's worldlines as depicted.}
    \label{fig:gluing}
\end{figure}
It is then clear that such a geometry has four different connected patches (two for each pinwheel) separated from each other by the observer and matter geodesics, and therefore four different energies to be integrated over.
The leading contribution to the connected $n$-th moment of an overlap\footnote{Notice that indices are not being summed over in this formula.} $Z_n^{\textrm{BH}}\approx \overline{\langle \psi_{ii'}|\psi_{jj'}\rangle...\langle\psi_{mm'}|\psi_{ii'}\rangle}$ is then given by 
\begin{equation}
\begin{aligned}
    Z_n^{\textrm{BH}}\equiv 
    e^{[2-2(n-1)-n]S_0}&\int ds_1ds_2ds_3ds_4\rho_0(s_1)\rho_0(s_2)\rho_0(s_3)\rho_0(s_4)
   \\
    &\times\left[e^{-\frac{\beta\left(s_1^2+s_2^2+s_3^2+s_4^2\right)}{8}}\gamma_{\Delta_m}(s_1,s_2)\gamma_{\Delta_m}(s_3,s_4)\gamma_{\Delta_O}(s_2,s_4)\right]^n\\
    =e^{4S_0}\int ds_1ds_2&ds_3ds_4\rho_0(s_1)\rho_0(s_2)\rho_0(s_3)\rho_0(s_4)\tilde{y}^n(s_1,s_2,s_3,s_4)
    \end{aligned}
    \label{eq:BHZn}
\end{equation}
where we used $\chi=(2-2g-n)$ for the Euler characteristic (our geometry has genus $g=n-1$ and $n$ boundaries) and in the second equality we defined
\begin{equation}
    \tilde{y}(s_1,s_2,s_3,s_4)=e^{-3S_0}e^{-\frac{\beta\left(s_1^2+s_2^2+s_3^2+s_4^2\right)}{8}}\gamma_{\Delta_m}(s_1,s_2)\gamma_{\Delta_m}(s_3,s_4)\gamma_{\Delta_O}(s_2,s_4).
\label{eq:tildey}
\end{equation}
$\gamma_{\Delta_m}$, $\gamma_{\Delta_O}$ are the normalization factors defined in equation \eqref{eq:gamma}.

Let us consider a set of $K^2$ states $\{\ket{\psi_{ii'}}\}$ where $i,i'=1,...,K$, and denote the overlap between two states by $M_{ii',jj'}=\langle\psi_{ii'}|\psi_{jj'}\rangle$. We will work in the regime $K\to \infty$, $e^{2S_0}\to\infty$ with $K/e^{2S_0}=O(1)$. Similar to equations \eqref{eq:d} and \eqref{eq:d2closed}, let us define
\begin{equation}
    d^4\equiv e^{4S_0}\int ds_1ds_2ds_3ds_4\rho_0(s_1)\rho_0(s_2)\rho_0(s_3)\rho_0(s_4).
\end{equation}
We assume once again that the integrals over energies are to be restricted to an arbitrary but finite energy window. The resolvent is now defined as\footnote{The identity in this space is given by $\mathbbm{1}_{ii',jj'}=\delta_{ij}\delta_{i'j'}$.}
\begin{equation}
    R_{ii',jj'}(\lambda)=\left(\frac{1}{\lambda \mathbbm{1}- M}\right)_{ii',jj'}=\frac{\delta_{ij}\delta_{i'j'}}{\lambda}+\frac{1}{\lambda}\sum_{n=1}^\infty \frac{\left(M^n\right)_{ii',jj'}}{\lambda^n}.
    \label{eq:BHresolventdef}
\end{equation}
Using the gravitational path integral to compute moments of the overlap as explained above, we can rewrite equation \eqref{eq:BHresolventdef} in terms of spacetime geometries:\footnote{The notation used in the diagrammatic expansion for the 4-index resolvent is completely analogous to that explained in Footnote \ref{footnote:notation} for the 2-index resolvent.}
\begin{align}
    \includegraphics[width=\textwidth]{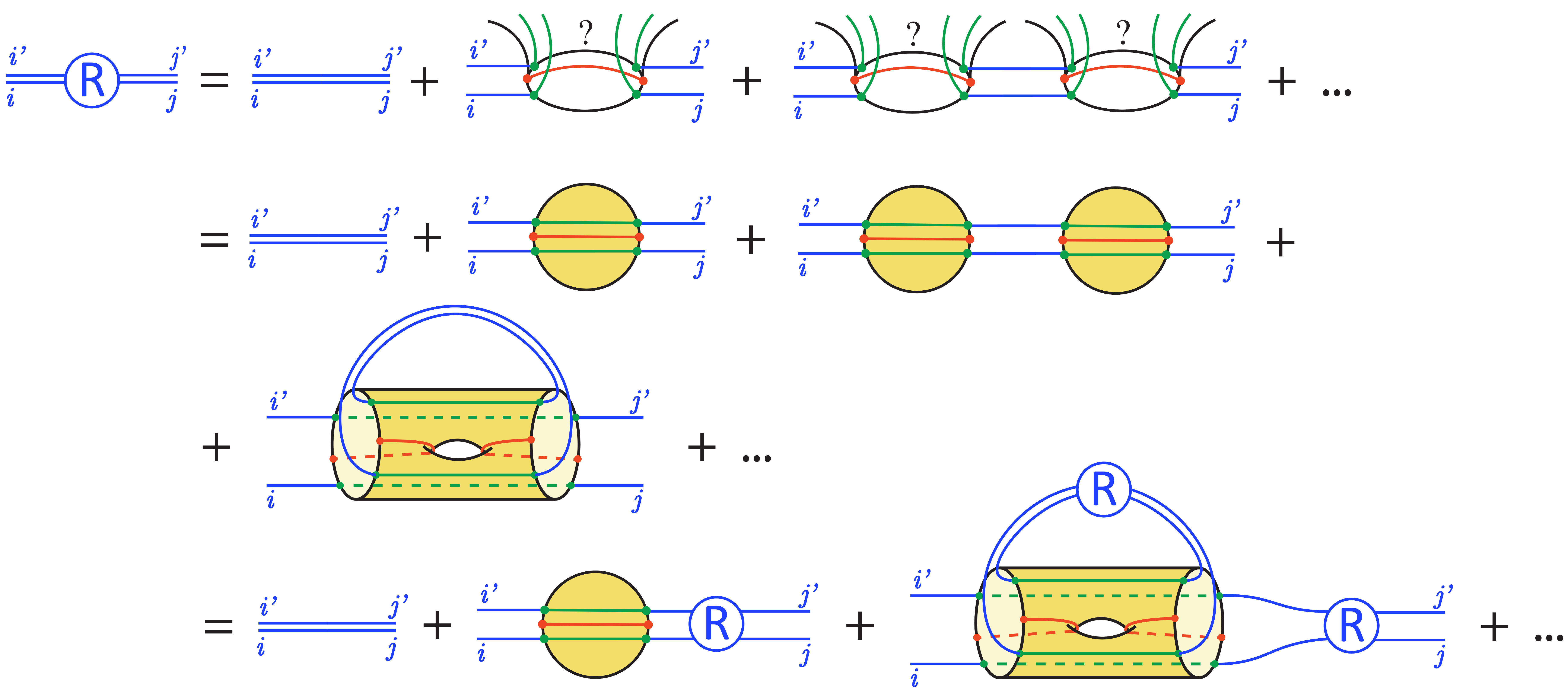}
\end{align}
In the last equality, we rearranged the expansion in terms of the number of boundaries to which the first boundary is connected. The corresponding Schwinger-Dyson equation takes the form:
\begin{equation}
    \overline{R_{ii',jj'}(\lambda)}=\frac{\delta_{ij}\delta_{i'j'}}{\lambda}+\frac{1}{\lambda}\sum_{n=1}^\infty Z_n^{\textrm{BH}} R^{n-1}(\lambda) \overline{R_{ii',jj'}(\lambda)}
    \label{eq:BHSD}
\end{equation}
where $Z_n^{\textrm{BH}}$ is given by equation \eqref{eq:BHZn}, we introduced again the overline to indicate that the resolvent is computed using the gravitational path integral, and $R(\lambda)=\sum_{i, i'=1}^K \overline{R_{ii',ii'}(\lambda)}$ is the trace of the resolvent. Taking the trace of equation \eqref{eq:BHSD} and performing the sum, we obtain
\begin{equation}
    R(\lambda)=\frac{K^2}{\lambda}+\frac{e^{4S_0}}{\lambda}\int ds_1ds_2ds_3ds_4\rho_0(s_1)\rho_0(s_2)\rho_0(s_3)\rho_0(s_4)\frac{R(\lambda)\tilde{y}(s_1,s_2,s_3,s_4)}{1-R(\lambda)\tilde{y}(s_1,s_2,s_3,s_4)}.
    \label{eq:BHR}
\end{equation}
With this result in hand, we can now calculate the dimension of the Hilbert space $\Hr(K)$ spanned by the $K$ states, which is related to the trace of the resolvent by equation \eqref{eq:dimH_from_r}.

Notice that the structure of the Schwinger-Dyson equation \eqref{eq:BHR} is analogous to that obtained in the closed universe case in Section \ref{sec:closeduniverse}, and in the computation of $\dim\left(\Hnonp\right)$ in Appendix \ref{app:reviewBH}. 
Therefore, the analytic properties of $R(\lambda)$ are also similar (see \cite{Boruch:2024kvv} and Appendix \ref{app:res_replica} for details). The calculation of $\dim\left(\Hr(K)\right)$ is then completely analogous to the closed universe case, and we will skip the details here. 
In the end, we find that the dimension of the Hilbert space spanned by the $K$ matter states and relevant to describe relational dynamics with respect to a bulk observer in the two-sided black hole case is
\begin{equation}
    \dim\left(\Hr(K)\right)=\begin{cases}
        K^2 \hspace{2cm} K^2<d^4\\
        d^4 \hspace{2.16cm} K^2>d^4
    \end{cases}
    \label{eq:dimBH}
\end{equation}
which shows that the underlying non-perturbative, relational quantum gravity Hilbert space has dimension $\dim\left(\Hr\right)=d^4$. Notice that this is again much larger than the 
Hilbert space of the two-sided black hole $\Hnonp$ in the absence of an observer computed in Appendix \ref{app:reviewBH}, which had dimension $\dim\left(\Hnonp\right)=d^2$. 

Computing the dimension of the Hilbert space in pure JT gravity (in the absence of matter) in the presence of an observer gives $\dim\left(\Hr\right)=d^2$ as opposed to $\dim\left(\Hnonp\right)=d$ in the absence of an observer. This is due to the additional constraint $s_1=s_2$ and $s_3=s_4$, namely that there is only one connected patch of spacetime to the left of the observer, and one to the right of the observer. Just like in the closed universe case, we expect that relaxing the technical assumption that $\Delta_m,\Delta_O\gg S_0$ does not affect the analytic structure of the resolvent, and therefore the result \eqref{eq:dimBH}.

Finally, the non-perturbative quantum gravity relational Hilbert space also factorises. In the presence of matter, we have (this result is derived in Appendix \ref{app:BHfactorisation})
\begin{equation}
    \Hr=\HrR\otimes \Hrr\otimes \HrL\otimes \Hrl,
    \label{eq:factorBH}
\end{equation}
where $\HrR$, $\HrL$ are associated with the right and left asymptotic boundaries, and $\Hrr$, $\Hrl$ with the right and left sides of the observer, respectively. For pure JT gravity without matter, the $d^2$-dimensional Hilbert space $\Hr$ factorises into two Hilbert spaces: one associated with the region between the observer and the left asymptotic boundary, and the other with the region between the observer and the right asymptotic boundary. This is because of the additional constraints that $s_1=s_2$ and $s_3=s_4$ arise in the absence of matter. As we will discuss in Section \ref{sec:holography}, the result \eqref{eq:factorBH} for factorisation also suggests that a putative dual holographic description of the two-sided black hole with matter in the presence of an observer consists of four non-interacting holographic theories: two living on the two asymptotic boundaries, and two on the worldline of the observer.

\section{Positivity of the inner product 
  and null states}
\label{sec:nullstates}

In the previous section, we computed properties of different Hilbert spaces in the presence of an observer. Technically speaking, however, we still did not show that the vector spaces $\Hrelp$ and $\Hr$ are Hilbert spaces, because we did not prove that the modified inner product in the presence of an observer introduced in Section \ref{sec:setup} is positive semi-definite. We turn now to filling this logical gap, and, as a by-product, we will find a simple description of the null states in the various non-perturbative Hilbert spaces in the presence of an observer.

\subsection{Inner product is positive semi-definite}

We will discuss the inner product in the closed universe case, leaving the two-sided black hole case, discussed in Section \ref{sec:twosided}, as an exercise for the reader since the discussion is quite similar. We saw in Section \ref{sec:setup} that we can describe the set of states with an observer by their asymptotic boundary conditions, with possible matter insertions along the asymptotic boundary. We could also describe the states in terms of the geometry and the observer state on a closed geodesic slice. In this section, we will find it more convenient to work with asymptotic states, although everything we say can be extended to closed geodesic slices by using the transform in \eqref{eq:change-of-basis-1} (or its generalization when matter is also present). Such asymptotic states are labeled by the two asymptotic lengths, $\beta_l$ and $\beta_r$, which lie between the observer and the matter insertion. The matter insertions will also be labeled by the flavor index, which runs over $i = 1,...,K$. We denote these states by $\ket{\beta_l, \beta_r, i}$. Note that this set of states spans the whole perturbative Hilbert space. This is unlike our choice of states in Section \ref{sec:hilbert}, where we restricted to the subspace of the perturbative Hilbert space where $\beta_l = \beta_r$ for convenience since, as far as the non-perturbative Hilbert space is concerned, both sets of states are vastly over-complete and span the full space. 

The non-perturbative inner-product in the presence of an observer is defined by a path integral in the presence of boundary conditions given by
\begin{align}\label{eqn:innerproductbcs}
    \bra{\beta_l, \beta_r, \Delta_i}\ket{\beta_l', \beta_r', \Delta_j}_H \quad = \quad \int d\ell\ e^{-\Delta_O \ell}\quad
    \begin{tikzpicture}[scale=.8, baseline={([yshift=-0.1cm]current bounding box.center)}]
    \coordinate (A) at (0,1);
    \coordinate (B) at (-1.4,1);
    \coordinate (C) at (-1.4,-1);
    \coordinate (D) at (0,-1);
    \coordinate (E) at (-.75,1.3);
    \coordinate (F) at (-.75,-1.3);
    \draw[thick] (0,1) arc[start angle=45,end angle=135,radius=1cm];
    \draw[thick] (0,-1) arc[start angle=-45,end angle=-135,radius=1cm];
    \fill[red] (A) circle (2pt);
    \fill[red] (B) circle (2pt);
    \fill[red] (C) circle (2pt);
    \fill[red] (D) circle (2pt);
    \fill[green] (E) circle (2pt);
    \fill[green] (F) circle (2pt);
    \draw[thick, red] (A) .. controls (-.5, .5) and (-.5,-.5) .. (D);
     \draw[thick, red] (B) .. controls (-.9, .5) and (-.9,-.5) .. (C);
    \node[right] at (-.2,0) {$\ell$};
    \node[left] at (-1.2,0) {$\ell$};
    \node[above] at (-1.2,1.2) {$\beta_l$};
    \node[above] at (-.2,1.2) {$\beta_r$};
    \node[above] at (E) {$i$};
    \node[below] at (F) {$j$};
    \node[below] at (-1.2,-1.2) {$\beta_l'$};
    \node[below] at (-.2,-1.2) {$\beta_r'$};
\end{tikzpicture}\quad .
\end{align}
As in \cite{IliLev24}, and as we mentioned in Section \ref{sec:review} and the discussion around \eqref{eqn:cgwavefunction}, we will find it useful to use the matrix integral perspective on these overlaps. In equations \eqref{eqn:nonpertoverlapstatistics} and \eqref{eq:change-of-basis-non-pert}, the corrections due to higher topology of the overlaps between the closed geodesic state $\ket{b, \Delta_O^{(i)}, u}$ and the asymptotic state $\ket{\beta, \Delta_O^{(j)}}$ were interpreted as coming from a random average over possible boundary Hamiltonians, $H$. Similarly, here, we interpret the non-perturbative corrections to the overlaps in \eqref{eqn:innerproductbcs} as also coming from a random average over possible boundary theories (i.e., a matrix integral). The subscript $H$ on the left-hand side of \eqref{eqn:innerproductbcs} is then to denote that these boundary conditions define an operator in the matrix integral, which depends upon a draw of the Hamiltonian $H$ from the ensemble. To leading order in the genus expansion, this inner product reads 
\begin{align}\label{eqn:disklevelinnerproduct}
    &\overline{\bra{\beta_l, \beta_r, \Delta_i}\ket{\beta_l', \beta_r', \Delta_j}_H} \quad = \quad \int d\ell\ e^{-\Delta_O \ell}\quad
    \begin{tikzpicture}[scale=0.8, baseline={([yshift=-0.1cm]current bounding box.center)}]
    \coordinate (A) at (0,1);
    \coordinate (B) at (-1.4,1);
    \coordinate (C) at (-1.4,-1);
    \coordinate (D) at (0,-1);
    \coordinate (E) at (-.75,1.3);
    \coordinate (F) at (-.75,-1.3);
    \fill[outsideyellow] (0,1) arc[start angle=45,end angle=135,radius=1cm] 
        -- (B) .. controls (-.9,.5) and (-.9,-.5) .. (C) 
        (C) arc[start angle=-135,end angle=-45,radius=1cm]
        -- (D) .. controls (-.5, -.5) and (-.5,.5) .. (A);    
    \draw[thick] (0,1) arc[start angle=45,end angle=135,radius=1cm];
    \draw[thick] (0,-1) arc[start angle=-45,end angle=-135,radius=1cm];
    \fill[red] (A) circle (2pt);
    \fill[red] (B) circle (2pt);
    \fill[red] (C) circle (2pt);
    \fill[red] (D) circle (2pt);
    \fill[green] (E) circle (2pt);
    \fill[green] (F) circle (2pt);
    \draw[thick, red] (A) .. controls (-.5, .5) and (-.5,-.5) .. (D);
    \draw[thick, red] (B) .. controls (-.9, .5) and (-.9,-.5) .. (C);
    \draw[thick, green] (E) -- (F);
    \node[right] at (-.2,0) {$\ell$};
    \node[left] at (-1.2,0) {$\ell$};
    \node[above] at (-1.2,1.2) {$\beta_l$};
    \node[above] at (-.2,1.2) {$\beta_r$};
    \node[above] at (E) {$i$};
    \node[below] at (F) {$j$};
    \node[below] at (-1.2,-1.2) {$\beta_l'$};
    \node[below] at (-.2,-1.2) {$\beta_r'$};
\end{tikzpicture}\quad \nonumber \\
& = \int ds_lds_r \, \rho_0(s_l) \rho_0(s_r) \,e^{-(\beta_l +\beta_l')\frac{s_l^2}{2} - (\beta_r+\beta_r')\frac{s_r^2}{2}} \times \gamma_{\Delta_O}(s_l,s_r) \times \left(\delta^{ij} \gamma_{\Delta_i}(s_l,s_r)\right) \quad .
\end{align}
As in previous sections, the overline denotes that we are computing quantities using the gravitational path integral or, equivalently, we are averaging the quantities over the dual ensemble. We can re-arrange this equation for the inner product into bras and kets to make the inner product structure more manifest. Namely we can define wave-functions
\begin{align}
    \bra{s_l,s_r,\Delta_i}\ket{\beta_l, \beta_r,\Delta_i} \equiv e^{-\beta_l \frac{s_l^2}{2} -\beta_r \frac{s_r^2}{2}} \delta^{ij} \sqrt{\gamma_{\Delta_i}(s_l,s_r)}.
\end{align}
Then equation \eqref{eqn:disklevelinnerproduct} may be written suggestively as 
\begin{align}\label{eqn:pertinnerproductmatrix}
    &\overline{\bra{\beta_l, \beta_r, \Delta_i}\ket{\beta_l', \beta_r', \Delta_j}_H}\nonumber \\
    &=\sum_{i'j'}\int ds_l ds_r\, ds_l'ds_r' \bra{\beta_l, \beta_r,\Delta_i}\ket{s_l,s_r,\Delta_{i'}}\ g^{i'j'}_{\Delta_O}(s_l,s_r;s_l',s_r')\ \bra{s_l',s_r',\Delta_{j'}}\ket{\beta_l', \beta_r',\Delta_j}.
\end{align}
where
\begin{align}
    g^{ij}_{\Delta_O}(s_l,s_r;s_l',s_r')\equiv \delta^{ij} \delta(s_l-s_l')\delta(s_r - s_r') \rho_0(s_l) \rho_0(s_r) \gamma_{\Delta_O}(s_l,s_r).
\end{align}
We, therefore, see that the perturbative inner product expressed in the $\ket{s_l,s_r,\Delta_i}$ basis is diagonal and furthermore has positive eigenvalues.  Therefore, we confirmed that $\Hrelp$ is a well-defined Hilbert space. This result was expected because the perturbative inner product on the subspace of perturbative states with an observer is unmodified by our new rules for the gravitational path integral.

To include higher genus effects and study the non-perturbative inner product, we now rely on the mapping between this model and a dual matrix integral, which captures all non-perturbative effects. Indeed, since we are ignoring all worldline crossings, all possible relevant contractions on higher-genus surfaces (including in calculations of higher moments of the inner product matrix) are captured by the JT matrix model of \cite{Saad:2019lba} together with $K$ operators $\mathcal{O}^i$ for $i = 1,..., K$ which model the matter. As pointed out in Section 4 of \cite{JafKol22}, these matter operators have matrix elements that are drawn from a Gaussian ensemble of mean zero and variance 
\begin{align}\label{eqn:mattervar}
\overline{\bra{E_{a_1}}\mathcal{O}^i\ket{E_{a_2}}\bra{E_{a_3}}\mathcal{O}^j\ket{E_{a_4}}} \equiv \overline{\mathcal{O}^i_{a_1a_2}\mathcal{O}^j_{a_3a_4}} = \delta^{ij}\delta_{a_1a_4}\delta_{a_2a_3} \gamma_{\Delta_i}(E_{a_1},E_{a_2}).
\end{align}
As discussed in \cite{IliLev24} and at the end of Section \ref{sec:pert_h}, when the matter fields on a geodesic slice are in their $\Delta = m =0$ ground state, fixed $\ell$ boundary conditions correspond non-perturbatively to an insertion of $\Tilde{\varphi}_{\sqrt{2H}}(\ell)$.\footnote{Here $\Tilde{\varphi}_{\sqrt{2E}}(\ell)$ is the length wavefunction given in \eqref{eqn:lengthwf} and $\Tilde{\varphi}_{\sqrt{2H}}(\ell)$ is the corresponding function of the Hamiltonian.} 
The following diagram can then be computed to all orders in the genus expansion by inserting into the matrix integral the operator 
\begin{align}\label{eqn:fixedellmatter}
    \begin{tikzpicture}[scale=1.0, baseline={([yshift=-0.1cm]current bounding box.center)}]
    \coordinate (A) at (0,1);
    \coordinate (B) at (-1.4,1);
    \coordinate (C) at (-1.4,-1);
    \coordinate (D) at (0,-1);
    \coordinate (E) at (-.75,1.3);
    \coordinate (F) at (-.75,-1.3);
    \draw[thick] (0,1) arc[start angle=45,end angle=135,radius=1cm];
    \draw[thick] (0,-1) arc[start angle=-45,end angle=-135,radius=1cm];
    \fill[red] (A) circle (2pt);
    \fill[red] (B) circle (2pt);
    \fill[red] (C) circle (2pt);
    \fill[red] (D) circle (2pt);
    \fill[green] (E) circle (2pt);
    \fill[green] (F) circle (2pt);
    \draw[thick, red] (A) .. controls (-.5, .5) and (-.5,-.5) .. (D);
     \draw[thick, red] (B) .. controls (-.9, .5) and (-.9,-.5) .. (C);
    \node[right] at (-.2,0) {$\ell$};
    \node[left] at (-1.2,0) {$\ell$};
    \node[above] at (-1.2,1.2) {$\beta_l$};
    \node[above] at (-.2,1.2) {$\beta_r$};
    \node[above] at (E) {$i$};
    \node[below] at (F) {$j$};
    \node[below] at (-1.2,-1.2) {$\beta_l'$};
    \node[below] at (-.2,-1.2) {$\beta_r'$};
\end{tikzpicture}\quad = \quad \text{Tr} \left(\Tilde{\varphi}_{\sqrt{2H}}(\ell) e^{-\beta_l' H} \mathcal{O}^j e^{-\beta_r'H} \Tilde{\varphi}_{\sqrt{2H}}(\ell) e^{-\beta_r H} \mathcal{O}^i e^{-\beta_l H} \right),
\end{align}
where we take a trace since the boundary conditions form a closed cycle. To write this expression in a form more akin to \eqref{eqn:pertinnerproductmatrix}, it is convenient to introduce a doubled Hilbert space so that operators, such as $\mathcal{O}^i$, get mapped to entangled states as
\begin{align}\label{eqn:matterstate}
    \mathcal{O}^i = \sum_{a,b} \mathcal{O}^i_{ab} \ket{E_a}\bra{E_b} \to \sum_{a,b} \mathcal{O}^i_{a,b} \ket{E_a}_l \ket{E_b}_r \equiv \ket{\mathcal{O}^i}_{lr}.
\end{align}
One can then write expression \eqref{eqn:fixedellmatter} more suggestively as
\begin{align}
    \begin{tikzpicture}[scale=1.0, baseline={([yshift=-0.1cm]current bounding box.center)}]
    \coordinate (A) at (0,1);
    \coordinate (B) at (-1.4,1);
    \coordinate (C) at (-1.4,-1);
    \coordinate (D) at (0,-1);
    \coordinate (E) at (-.75,1.3);
    \coordinate (F) at (-.75,-1.3);
    \draw[thick] (0,1) arc[start angle=45,end angle=135,radius=1cm];
    \draw[thick] (0,-1) arc[start angle=-45,end angle=-135,radius=1cm];
    \fill[red] (A) circle (2pt);
    \fill[red] (B) circle (2pt);
    \fill[red] (C) circle (2pt);
    \fill[red] (D) circle (2pt);
    \fill[green] (E) circle (2pt);
    \fill[green] (F) circle (2pt);
    \draw[thick, red] (A) .. controls (-.5, .5) and (-.5,-.5) .. (D);
     \draw[thick, red] (B) .. controls (-.9, .5) and (-.9,-.5) .. (C);
    \node[right] at (-.2,0) {$\ell$};
    \node[left] at (-1.2,0) {$\ell$};
    \node[above] at (-1.2,1.2) {$\beta_l$};
    \node[above] at (-.2,1.2) {$\beta_r$};
    \node[above] at (E) {$i$};
    \node[below] at (F) {$j$};
    \node[below] at (-1.2,-1.2) {$\beta_l'$};
    \node[below] at (-.2,-1.2) {$\beta_r'$};
\end{tikzpicture} \quad = \quad \bra{\beta_l, \beta_r, \mathcal{O}^i} \psi_{H_{l}}(\ell) \otimes \psi_{H_r}(\ell) \ket{\beta_l', \beta_r', \mathcal{O}^j}_{lr}\quad ,
\end{align}
where $\ket{\beta_l, \beta_r, \mathcal{O}^i} \equiv e^{-\beta_l H_l -\beta_r H_r} \ket{\mathcal{O}^i}$. Integrating against $e^{-\Delta_O \ell}$, we find
\begin{align}\label{eqn:npip}
    \bra{\beta_l, \beta_r, \Delta_i}\ket{\beta_l', \beta_r', \Delta_j}_H = \bra{\beta_l, \beta_r, \mathcal{O}^i} g_{\Delta_O}(H_l, H_r) \ket{\beta_l', \beta_r', \mathcal{O}^j}_{lr}
\end{align}
where\footnote{We remind the reader that we are viewing $\gamma_{\Delta_O}(H_l,H_r)$ as a matrix function of the two operators $H_l \equiv H \otimes 1$ and $H_r \equiv 1 \otimes H$ which act on the doubled Hilbert space.}
\begin{align}\label{eqn:gnp}
    g_{\Delta_O}(H_l,H_r) = \gamma_{\Delta_O}(H_l,H_r).
\end{align}
Since $g$ is a positive definite matrix on the tensor-product Hilbert space $\Hr=\Hrl\otimes \Hrr$, then the non-perturbative modified inner product in the presence of an observer is also positive semi-definite. The reason this inner product is only positive semi-definite and not positive definite is because of the presence of null states, which we describe in the next subsection. 

First, however, we briefly point out what the operator $g_{\Delta_O}$ would be if we had treated the observer like yet another matter operator, with matrix elements pulled from the Gaussian ensemble. In that case, the boundary conditions would just be two disconnected circles, each with two green dots---one for the observer insertion and one for the matter insertion. We can still introduce the left and right Hilbert spaces, and after doing so, it is not hard to see that, as a consequence of the Gaussian statistics of the operator, $g_{\Delta_O}(H_l,H_r)$ as defined in \eqref{eqn:gnp} should be replaced by the projector
\begin{align}\label{eqn:ipreplacement}
    g_{\Delta_O}(H_l,H_r) \to \ketbra{\mathcal{O}_{O}}{\mathcal{O}_{O}}_{lr}.
\end{align}
This inner product tells us that all perturbative states in $\cH_{lr}$ which are orthogonal to $\ket{\mathcal{O}_{O}}_{lr}$ become null, and so the non-perturbative Hilbert space is one-dimensional. Equation \eqref{eqn:ipreplacement} clearly explains how replacing the observer with a ``regular" matter particle reduces the non-perturbative Hilbert space to one dimension. 
The replacement in \eqref{eqn:ipreplacement} also makes clear what one needs to do in order to return to the inner product in \eqref{eqn:gnp}: average over the matrix elements of the observer's wavefunction coefficients, $\mathcal{O}_{O}(E_l,E_r)$, used to define $\ket{\mathcal{O}_{O}} = \sum_{E_l,E_r} \mathcal{O}_O(E_l,E_r)\ket{E_l} \ket{E_r}$. Namely, using the formula \eqref{eqn:mattervar}, we see that 
\begin{align}
    \overline{\ketbra{\mathcal{O}_{O}}{\mathcal{O}_{O}}}_{lr} = \gamma_{\Delta_O}(H_l,H_r).
\end{align}
The inner product in the presence of an observer is, in this sense, an average over non-perturbative Hilbert spaces without an observer. This result also clarifies how to compute quantities relevant to describing the experience of an observer from a matrix integral point of view. First, we identify a given Gaussian matrix $\mathcal{O}_O$ describing matter to represent our observer. Second, we compute the observable of interest in a single realization of the Hamiltonian and matter ensembles. Finally, we average over the $\mathcal{O}_O$ ensemble. Notice that we average only over $\mathcal{O}_O$ and not the Hamiltonian or other matter matrix ensembles. This corresponds to the special treatment reserved to the observer in our rules for the gravitational path integral.

\subsection{Null states}

Armed with formula \eqref{eqn:npip}, for the non-perturbative inner product in the presence of an observer, we can now simply read off the null states. By ``null states," we have in mind states in the perturbative Hilbert space that have zero norm with respect to the non-perturbative inner product. States in the perturbative Hilbert space in the presence of an observer form a continuum, labeled by the continuous labels $\beta_l$ and $\beta_r$ (or $E_l$ and $E_r$), along with the discrete index labeling the flavor of matter particle. For analyzing null states, since the inner product operator $g_{\Delta_O}(H_l,H_r)$ is a function of just $H_l$ and $H_r$, it is convenient to work with the basis of states labeled by continuous energies $E_l$ and $E_r$. We can then write a general state, $\ket{\psi}$, in the perturbative Hilbert space as a linear superposition of such states as
\begin{align}\label{eqn:generalpertstate}
\ket{\psi} = \sum_k \int dE_l dE_r \ \hat{\psi}_{\Delta_k}(E_l,E_r)\  \ket{E_l,E_r,\Delta_k}.
\end{align}
For a given draw of the ensemble of JT Hamiltonians $H$ together with the Gaussian-distributed matrix elements for the matter, the spectrum of $H_l$ and $H_r$ will be discrete. Thus, when we compute the norm of the general state in \eqref{eqn:generalpertstate} using the inner product in \eqref{eqn:npip}, we find 
\begin{align}\label{eqn:npnorm}
    \braket{\psi} = \sum_{k,k',a,b} \left(\hat{\psi}_{\Delta_k}(E_a,E_b)\right)^*\hat{\psi}_{\Delta_{k'}}(E_a,E_b) \times \left(\mathcal{O}_{ab}^{\Delta_k}\right)^* \mathcal{O}_{ab}^{\Delta_{k'}}  \times g_{\Delta_O}(E_a,E_b).
\end{align}
Clearly, null states can be formed just by demanding that $\hat{\psi}_{\Delta_k}$ have no support on the discrete energy spectrum of either $H_l$ or $H_r$. 
Furthermore, if we allow the matter flavor to vary, we see that there is an even larger multiplicity of null states due to the fact that $g_{\Delta_O}(H_l,H_r)$ does not have a label for the matter index. 

This pattern of null states should be directly compared to what was found for a two-sided black hole in the presence of matter but without an observer in \cite{IliLev24}. There, again, the perturbative Hilbert space is labeled by three numbers $\ket{E_l, E_r,\Delta_i}$, where $E_l$ and $E_r$ are continuous labels at the disk level, as reviewed in Section \ref{sec:pert_h}. Non-perturbatively, however, the Hilbert space becomes labeled by a (random) set of discrete numbers, $\ket{E_l^i, E_r^j}$, and states are not labeled by the flavor/matter index. This means the null states in the two-sided black hole with matter but without an observer are of the same structural form as the closed universe with matter and an observer. Of course, this is not a coincidence. By the formula in equation \eqref{eqn:innerproductbcs}, we see that we can think of the closed universe with an observer as a two-sided black hole where the two asymptotic boundaries are replaced by geodesic worldlines and then glued together via the integral in \eqref{eqn:innerproductbcs}. Thus, the non-perturbative Hilbert space structure of the closed universe with an observer is inherited from that of the two-sided black hole without an observer.

Note that although the presence of flavor indices means that, in a sense, there are ``more" null states than in the case without matter, it does not mean that any two states that differ only in their flavor indices are the same state non-perturbatively. This is due to the fact that the states $\ket{E_l,E_r,\Delta_i}$ and $\ket{E_l,E_r,\Delta_j}$ for $E_l,E_r$ in the non-perturbative spectrum of the theory are only the same up to the matter wavefunctions $\mathcal{O}^{i,j}(E_l,E_r)$ in \eqref{eqn:matterstate}. Thus, superpositions of such states over different energies will not be proportional to each other. This explains how, in Section \ref{sec:closeduniverse}, we found that a collection of states with varying flavor index can still span the full non-perturbative Hilbert space. In Section \ref{sec:closeduniverse}, we worked with states at fixed temperatures $\beta_l = \beta_r$ but varying flavor index. In that context, we found that null states only became important with a flavor index number $K$ of at least $e^{2S_0}$.  

Given the null states described above, we then have a non-isometric map, of the type discussed in \cite{Akers:2023fqr,IliLev24}, connecting the perturbative Hilbert space of a closed universe with an observer $\Hrelp$ to the non-perturbative Hilbert space of a closed universe with an observer $\Hr$. Furthermore, there is yet another map connecting the perturbative Hilbert space with an observer $\Hrelp$ to the non-perturbative Hilbert space without one $\Hnonp$. In other words, we get two non-isometric maps
\begin{align}\label{eqn:nonisomaps}
&\Hrelp \xlongrightarrow[]{V_{\text{Obs.}}} \Hr, \nonumber \\
&\Hrelp \xlongrightarrow[]{V_{\text{no-Obs.}}} \Hnonp,
\end{align}
where $V_{\text{Obs.}}$ is implemented by the inner product $g_{\Delta_O}(H_l,H_r)$ in \eqref{eqn:gnp} and $V_{\text{no-Obs.}}$ is implemented by the projector onto the observer's state $\ketbra{\mathcal{O}_O}{\mathcal{O}_O}_{lr}$ as in \eqref{eqn:ipreplacement}. 
An exactly analogous structure as in \eqref{eqn:nonisomaps} exists for the case of the two-sided black hole, except in that case, obviously, the dimensions of $\Hr$ and $\Hnonp$ are different, namely $e^{4S_0}$ and $e^{2S_0}$ respectively. Note that in both cases, we can also find a non-isometric map between $\Hr$ and $\Hnonp$. In the context of two-sided black holes, one can consequently use the ideas of \cite{Akers:2022qdl} to rewrite operators on $\Hr$ as non-linear, state-dependent observables on $\Hnonp$. It would be interesting to work out explicit formulae for such a non-linear reconstruction.

\section{{Observables along the worldline}}
\label{sec:observables}

We will now analyze the possible observables that act on the Hilbert space $\Hr$ that a gravitating observer can measure. Before that, we should address a basic point. The evolution of a gravitating observer is, in principle, fully determined through the Wheeler-deWitt constraint, and the observer cannot make any measurements not determined by this constraint. The only quantities that we can consequently calculate are the transition probabilities between the different states of the observer as in Section \ref{sec:setup}. We remark that we have used an oversimplified model of the observer: in the action \eqref{eq:action-of-observer}, we have, for instance, neglected the coupling between the observer's worldline and other matter fields in the theory. By computing correlation functions of operators constructed from these fields and dressed to the worldline of the observer, we can understand what would happen to the observer in more realistic models and quantify when non-perturbative effects severely modify the transition probabilities between different states in $\Hr \otimes \HO$.\footnote{Specifically, such correlation functions determine the transition amplitudes when the action \eqref{eq:action-of-observer} is modified by source terms. For example, we can modify the action by coupling a field $\phi(x)$ that exists everywhere in the bulk to sources, $I_\text{source} = \int du \sqrt{h} j(u) \phi(u)$. }     In the case of a closed universe, we will see that, in contrast to the global picture, correlation functions of matter operators dressed to the worldline are not affected by non-perturbative effects, at least far from the Big Bang and Big Cruch singularities. For two-sided black holes, we will see that for some observables, such as the length of the Einstein-Rosen bridge seen by an observer inserted at late times, the global picture and the observer's point of view agree; in both cases, non-perturbative effects become important at times $t=O(e^{S_0})$ and cause the length of the Einstein-Rosen bridge to plateau. For other observables, which capture the experience of an infalling observer more accurately, the two pictures prove drastically different. For example, the center-of-mass (CM) collision energy between the observer and a perturbation behind the horizon of the black hole receives large non-perturbative corrections at times $t=O(S_0)$---i.e., close to the Page time---in the global picture. On the contrary, the perturbative calculation of this collision energy proves reliable until times $t=O(e^{S_0})$ when describing physics from the point-of-view of the observer using our proposal.   

\subsection{Examples of what an observer sees in a closed universe}
\label{sec:example-obs-closed-universe}

To contrast the difference between observables measured in the global Hilbert space and from the point of view of a gravitating observer, let us compute a correlation function for a matter field $\phi$ measured along the observer's worldline, with the operator insertions dressed to the state of the clock that the observer carries. For simplicity, we will take this correlation function to be a two-point function and assume that $\phi$ interacts with the observer only gravitationally. As in Section \ref{sec:setup}, we will prepare the state of the observer $\ket{\beta, O_{t=0}}$ on the asymptotic boundary so that the clock reads $t=0$ on the smooth closed geodesic slice and keep track of the time along the worldline with respect to this slice. While the notion of the clock of the observer only makes sense in Lorentzian signature, for computational purposes it is useful to first consider the Euclidean two-point function $\bra{\beta, O_{t=0}} \phi(\tau_1) \phi(\tau_2) \ket{\beta, O_{t=0}}$. The relevant geometry is shown in Figure \ref{fig:two-pt-function-closed-universe}.
\begin{figure}
\centering
\includegraphics[width=0.9\textwidth]{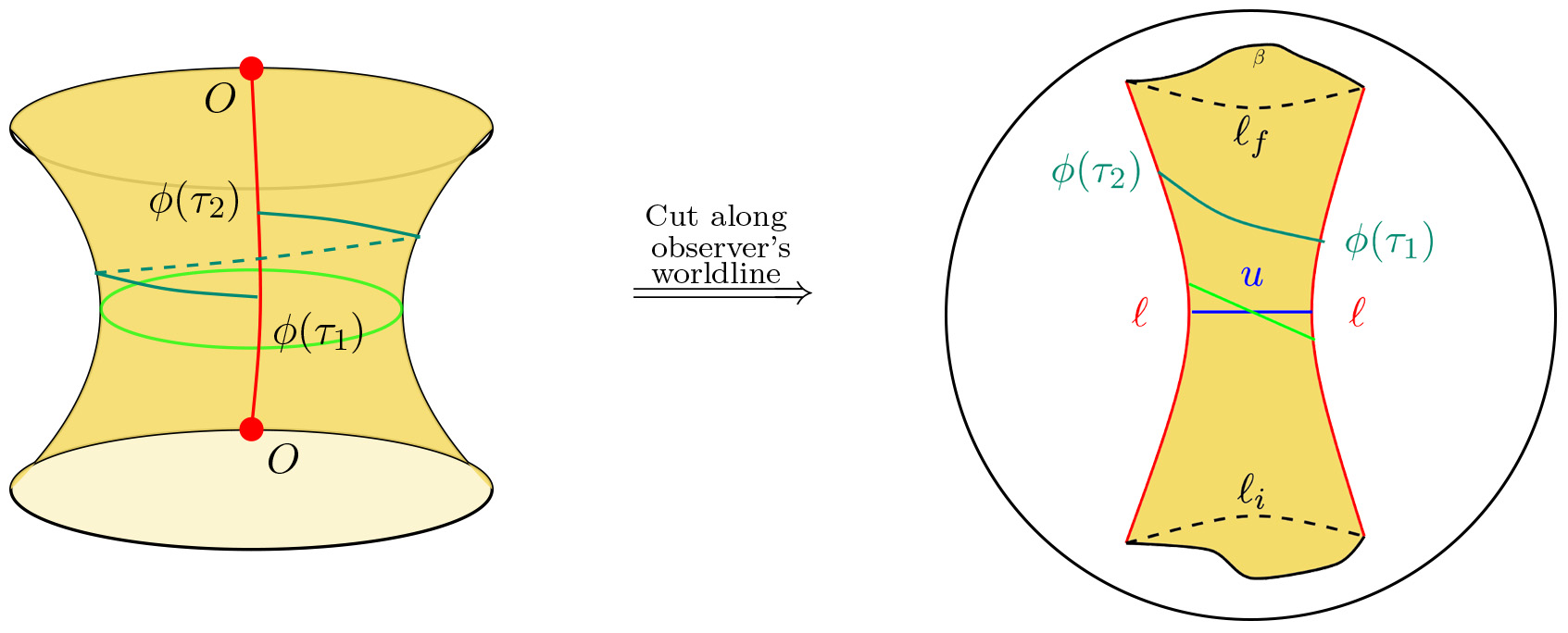}
\caption{\label{fig:two-pt-function-closed-universe} Euclidean diagrams showing the two-point function of a scalar field $\phi$, dressed to the worldline of the observer. The figure on the right is obtained by cutting the figure on the left along the observer's worldline, embedding the resulting geometry in the Poincar\'e disk (depicted in black), and then identifying the two red geodesics of length $\ell$ in the right diagram. We give an example of the worldline connecting the two operator insertions that winds once around the cylinder, which we show in dark green in both diagrams. The light green circle (left) and line (right) represent the shortest closed geodesic on the geometry, while the blue line (only shown on the right) is the shortest geodesic (not closed) perpendicular to the worldline of the observer on both ends.  To evaluate the two-point function in the Lorentzian closed universe we perform the analytic continuation $\tau_{1,2} = -i t_{1,2}$. }
\end{figure}
We will then analytically continue $\tau_{1,2} = -i t_{1,2}$ to obtain the Lorentzian two-point function.

\subsection*{The global picture}

To start, let us briefly discuss the global picture, in which the observer is treated as an ordinary matter field. In this case, it is impossible to come up with a well-defined clock for the observer. Because the Hilbert space $\Hnonp$ is one-dimensional, the overlap between states defined on spatial slices with the observers' clocks reading different times $t_1$ and $t_2$ always has large fluctuations. In fact,
\be 
|\braket{\beta, O_{t=t_1}}{\beta ', O_{t=t_2}}|^2 = 1
\ee
if the states are unit-normalized. Similarly, the two-point function when $\phi$ is inserted at different times has a large standard-deviation and can only differ by a phase, 
\be 
\frac{|\bra{\beta, O_{t=0}} \phi(t_1) \phi(t_2) \ket{\beta,O_{t=0}}|^2}{|\bra{\beta, O_{t=0}} \phi(t_1') \phi(t_2') \ket{\beta,O_{t=0}}|^2 } = 1\,,
\ee
regardless of the times $t_{1,2}$ and $t_{1,2}'$. This is again because, in the one-dimensional Hilbert space $\Hnonp$, unit-normalized states labeled by any time can only differ from each other by a phase of the closed universe, and operators can only be multiples of the identity. Of course, this simply means that physics is trivial. This is not what we expect an observer to see in a closed universe. To obtain a sensible result, we will recompute the same two-point function in the same closed universe state following our new proposed rules. 

\subsection*{Correlation functions from an observer's point of view}

When computing correlation functions using the gravitational path integral from an observer's point of view, the perturbative answer dominates.\footnote{We believe this happens at least when the observer is not parametrically close (in $e^{-S_0}$) to the singularity.} This is again because the disconnected diagrams that connect the different bras and different kets in the global picture (see Section \ref{sec:reviewclosed}) no longer contribute once we impose that the observer's worldline connects each bra to the corresponding ket. Since even the perturbative answer is physically interesting in a closed universe, we will focus on finding the two-point function at leading order in $e^{-S_0}$, leaving all computations of non-perturbative corrections for future work.

Before discussing the computation of the two-point function $\bra{{\beta,O_{t=0}}} \phi(\tau_1) \phi(\tau_2) \ket{{\beta,O_{t=0}}}$, it will be useful to rewrite the inner-product between two asymptotic states by gluing patches bounded by geodesic segments whose length is fixed. The Euclidean version of a closed universe can be decomposed into a quadrilateral with geodesic boundaries -- with two opposite sides that are equal --  and two Hartle-Hawking disk wavefunctions, $\varphi_\beta(\ell)$; see the right diagram in Figure \ref{fig:two-pt-function-closed-universe}. We will denote the path integral contribution of the quadrilateral by $I_4(\ell_i, \ell, \ell_f, \ell) $, 
where $\ell_{i}$ and $\ell_f$ are the renormalized lengths of the geodesics wrapping the cylinder that start and end at the location of the observer insertion at the asymptotic boundary and $\ell$ is the renormalized length of the observer's worldline between the two asymptotic boundaries. The inner-product between two asymptotic states $\ket{\beta,O_{t=0}}$ can then be rewritten as
\be
\braket{{\beta,O_{t=0}}}{\beta',O_{t=0}} = \int d\ell_i d \ell_f d \ell \varphi_{\beta}(\ell_i) \varphi_{\beta'}(\ell_f) e^{-\Delta_O \ell} I_4(\ell_i, \ell, \ell_f, \ell),
\ee
where
\be 
I_4(\ell_i, \ell, \ell_f, \ell) = e^{S_0}\int dE \rho_0(E) \varphi_E(\ell_i)\varphi_E(\ell)\varphi_E(\ell)\varphi_E(\ell_f)\,.
\ee
In this decomposition of the geometry, it is also useful to note that the distance $u$ between the two geodesics whose renormalized lengths are $\ell$ in the quadrilateral, i.e., the length of the geodesic that wraps the cylinder once and is perpendicular to the worldline of the observer at both ends, is fixed to be\footnote{This can be computed using the fact that three lengths in a quadrilateral with two right angles determine the fourth length. We will use this property for two quadrilaterals, one involving the geodesic of length $u$ and renormalized length $\ell_i$ and the other involving the geodesic of length $u$ and renormalized length $\ell_f$ (see Figure \ref{fig:two-pt-function-closed-universe}). In the quadrilateral involving the geodesic of length $u$ and the geodesic of proper length $\tilde \ell_i$ we have, 
 \be 
 \sinh^2\left(\frac{\tilde \ell_i}2\right) = \sinh^2\left(\frac{u}2\right) \cosh(\tilde \ell_1)\cosh(\tilde \ell_2) + \sinh^2\left(\frac{\tilde \ell_1-\tilde \ell_2}2\right) \Rightarrow {e^{\ell_i}} \approx  \sinh^2\left(\frac{u}2\right) e^{\ell_1+\ell_2}\,,
 \label{eq:right-quadrilateral-relation}
 \ee
 where $\tilde \ell_1$ and $\tilde \ell_2$ are the proper lengths for the side edges of the quadrilateral with two right angles and $ \ell_1$, $ \ell_2$ and $ \ell$ are renormalized lengths. In the quadrilateral involving the geodesic of length $u$ and the geodesic of length $\tilde \ell_f$ we have 
 \be 
 \sinh^2(\frac{\tilde \ell_f}2) = \sinh^2\left(\frac{u}2\right) \cosh(\tilde \ell_3)\cosh(\tilde \ell_4) + \sinh^2\left(\frac{\tilde \ell_3-\tilde \ell_4}2\right) \Rightarrow {e^{\ell_f}}{} \approx   \sinh^2\left(\frac{u}2\right) e^{\ell_3+\ell_4}\,,
 \ee
where, once again, $\tilde \ell_3$ and $\tilde \ell_4$ are the proper lengths for the side edges of the quadrilateral with two right angles and $ \ell_3$, $ \ell_4$ and $ \ell_f$ are renormalized lengths. Here, we have made use of the fact that all proper lengths diverge in the same way, which lets us ignore the exponential factors of their differences. Using the fact that $\ell = \ell_3+\ell_1 = \ell_2+\ell_4$ ,  \eqref{eq:modulus-length} follows.}
\be
\label{eq:modulus-length}
\sinh^2\left(\frac{u}2\right)  = e^{\frac{\ell_i + \ell_f}2-\ell} \,.
\ee
When the observer is heavy ($\Delta_O \gg 1$), the relative twist between the observer insertions on the two asymptotic boundaries is minimized, and the geodesic of length $u$ agrees with the closed smooth geodesic that wraps the cylinder, whose length we denote by $b$. In this limit, the observer enters the Lorentzian spacetime perpendicular to the time-reflection symmetric slice, and we don't have to deal with the subtlety of the observer not traveling on a real Lorentzian path as discussed in Section \ref{sec:setup}. We take this limit, which will, therefore, greatly simplify the calculation. As mentioned in Section \ref{sec:setup}, when performing the analytic continuation $\tau \to -i t$ to Lorentzian signature, there is a big bang singularity located at $t=-\frac{\pi}2$ and a crunch singularity located at $t=\frac{\pi}2$ along the worldline of the observer when the observer worldline is perpendicular to the time-reflection symmetric slice.

To account for the two-operator insertion, we now simply need to compute the lengths of all possible geodesic connecting the two points where we insert the operators $\phi$ on the worldline of the observer. The Euclidean proper times for these points are $\tau_{1}$ and $\tau_2$, where $\tau=0$ is fixed to be on the minimal closed geodesic. The shortest such geodesic follows the observer's worldline. However, in a closed universe setup, there are additional geodesics connecting the two points that wrap around the Euclidean cylinder (or, in Lorentzian signature, the closed universe). In terms of $u$, the lengths $\ell_{12}^{(n)}$ of such geodesics in Euclidean signature are given by\footnote{This once again follows from the fact that three lengths in a quadrilateral with two right angles determine the fourth length -- e.g., see \eqref{eq:right-quadrilateral-relation}.  }
\be 
\sinh^2 \left(\frac{\ell_{12}^{(n)}}2\right) = \sinh^2\left(\frac{n\,u}2\right) \cosh(\tau_1 )\cosh(\tau_2 ) + \sinh^2\left(\frac{\tau_1-\tau_2}2\right),
\ee
where $n$ is the number of times that a given geodesic wraps around the universe (with $n=0$ corresponding to the shortest geodesic). Analytically continuing to Lorentzian signature ($\tau_{1,2} \to -i t_{1,2}$) we find that the geodesic distances $\ell_{12}^{(n)}$ are given by 
\be 
\sinh^2 \left(\frac{\ell_{12}^{(n)}}2\right) = \frac{1}2 \left(-1 + \frac{1}{\xi_{12}^{(n)}}\right), \label{eq:geodesic-distance}
\ee
where $\xi_{12}^{(n)}$ are the chordal distances
\be
\xi_{12}^{(n)} = \frac{1}{\sin(t_1) \sin(t_2) + \cosh(n u) \cos(t_1) \cos(t_2)}\,. \label{eq:chordal-distance}
\ee
\begin{figure}[t]
    \centering
    \includegraphics[width=0.28\textwidth]{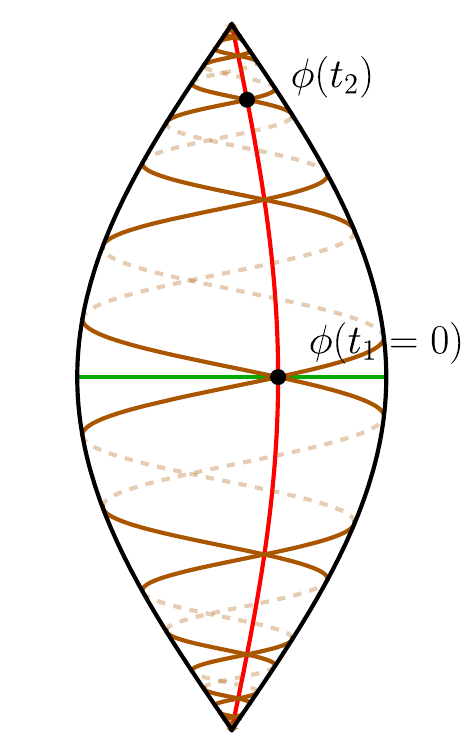}
    \hspace{1.25cm}
    \includegraphics[width=0.28\textwidth]{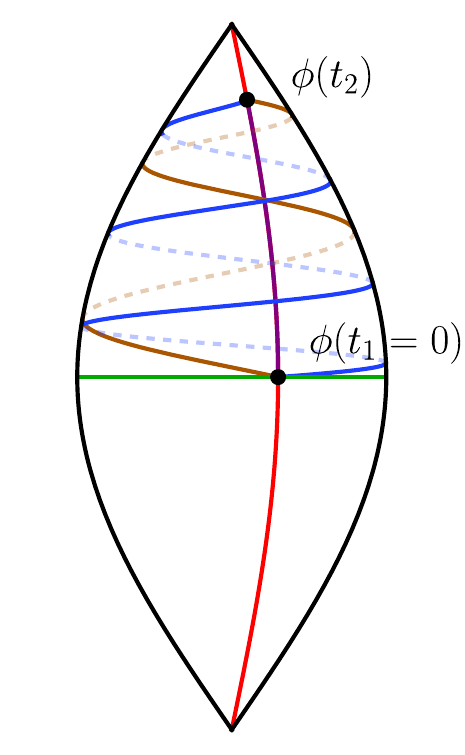}
    \caption{Two frontal views of the closed universe in Lorentzian signature. The observer's worldline is depicted in red and the maximal Cauchy slice of length $b$ is depicted in green. We insert a matter insertion at time $t_1=0$, defined by where the observer intersects this Cauchy slice. Another matter insertion is place sometime later along the observer's worldline. (\textit{Left}) Left and right null rays emanating from the first matter insertion are shown in brown. Each time the second operator insertion is placed at the intersection of these null rays with the worldline of the observer, the two-point function on a fixed background diverges. (\textit{Right}) Generically, an infinite number of winding geodesics contribute to the two-point function. A timelike geodesic along the observer's worldline is drawn in purple, a winding spacelike geodesic is drawn in blue, and a null winding geodesic is shown in brown. }
\label{fig:closedUniverseCorrelator}
\end{figure}

When the RHS of equation \eqref{eq:geodesic-distance} is negative and $\ell_{12}^{(n)}$ is purely imaginary, the geodesic connecting the two operator insertions is timelike ($\xi_{12}^{(n)}>1$); when it vanishes, the geodesic is null  ($\xi_{12}^{(n)}=1$); when it is positive and $\ell_{12}^{(n)}$ is purely real, the geodesic is spacelike  ($0\leq \xi_{12}^{(n)}<1$). If we set the first point at $t_1 = 0$,  then the two points become null separated whenever $\cos(t_2) = 1/\cosh(nu)$. Thus, the worldline of the observer intersects its own past lightcone that starts at the $t=0$ point. As the crunch singularity (or, similarly, the bang singularity in the past) is approached, the number of lightcone intersections grows, and the observer intersects this past lightcone an infinite number of times as they approach the singularity. This is shown in Figure \ref{fig:closedUniverseCorrelator}. In terms of the chordal distance \eqref{eq:chordal-distance}, the Wightman propagator in AdS$_2$ for the field $\phi$, with scaling dimension $\Delta$, is given by \cite{DHoker:2002nbb}
\be 
G_\Delta(\xi) = \frac{2^{-\Delta} \Gamma(\Delta)}{(2\Delta-1)\pi^{\frac{1}2} \Gamma\left(\Delta-\frac{1}2\right)} \xi^{\Delta} F\left(\frac{\Delta}2, \frac{\Delta}2+\frac{1}2 ; \Delta + \frac{1}2; \xi^2 \right)\,.
\ee 
This propagator has a divergence in the coincident point limit or whenever the points become null-separated.
Thus, we expect that the observer will see a divergence in the two-point function whenever they insert the second operator at a point on the future lightcone of the first operator insertion. As we have seen, the number of such points---and therefore the number of divergences---itself diverges as the observer approaches the singularity.

To obtain the Wightman propagator in the fluctuating spacetime, we need to sum the contribution from all the geodesics that wind around the closed universe and write this propagator in terms of the geodesic lengths in the quadrilateral with length $\ell_i,\, \ell, \,\ell,$ and $\ell_f$. The resulting (un-normalized) Wightman propagator is given by 
\be 
 \bra{{\beta,O_{t=0}}} \phi(t_1) \phi(t_2) \ket{{\beta,O_{t=0}}} = \sum_{n=-\infty}^\infty \int d\ell_i d \ell_f d \ell \varphi_{\beta}(\ell_i) \varphi_{\beta'}(\ell_f) e^{-\Delta_O \ell} I_4(\ell_i, \ell, \ell_f, \ell) G_\Delta(\xi^{(n)}_{12}),
 \label{eq:Wightman-propagator-closed-universe}
\ee
where $\xi_{12}^{(n)}$ is given by equation \eqref{eq:chordal-distance} with $u$ given by equation \eqref{eq:modulus-length}. Because the propagator has only a mild logarithmic divergence when the points of the two operator insertions become null separated (i.e., when $\xi^{(n)}_{12} \to 1$), this divergence is smoothed out by the fluctuation in the lengths $\ell$, $\ell_i$ and $\ell_f$ that are captured by the integral in equation \eqref{eq:Wightman-propagator-closed-universe}. Thus, the two-point function of $\phi$ becomes well defined along the entire worldline of the observer due to quantum gravity fluctuations. While this might be a desirable feature, we do not see a contradiction with such divergences appearing when computing observables in the Hilbert space of closed universes. This is because even though the Hilbert space we computed in Section \ref{sec:closeduniverse} was finite, there, we restricted to asymptotic boundaries whose ADM energy (labeled by $E = s^2/2$) was within a given energy window that could be arbitrarily large but finite. When acting with the operator $\phi$ on a time slice, the resulting state might have support outside of the space spanned by the states specified by the asymptotic boundary condition whose ADM energy is within a finite window. In fact, we can explicitly see that more serious divergences (for example, those appearing in the two-point function of $\partial^k_t \phi(t)$ for $k>1$) persist even after performing the integral over metric fluctuations.\footnote{One can cure such divergences by smearing the insertions of $\phi$ along (or in the neighborhood of) the worldline. We also expect such a smearing to decrease the contributions of intermediate states (in between the two operator insertions) that have large energies, providing an explanation for why such divergences are cured.}  Interestingly, the accumulation of such divergences provides a way to probe the crunch singularity of the closed universe. Even more interestingly, when probing this accumulation of divergences, the relevant Euclidean geodesic in the calculation of the propagator \eqref{eq:Wightman-propagator-closed-universe} winds a very large number of times around the closed universe and has an arbitrarily long length, $n u$. Typically, the appearance of such long geodesics -- such as the long geodesics probing the length of the Einstein-Rosen bridge at times $t\sim e^{S_0}$ that we shall analyze in the next subsection -- signals the fact that non-perturbative effects become important and could drastically alter the perturbative calculation presented above. It would be interesting to perform such a non-perturbative analysis in future work.

\subsection{Examples of what an observer sees when falling into a black hole (I): \\ 
the length of an ER bridge}

Our computation of correlators dressed to the worldline of the observer can easily be generalized for observers probing two-sided black hole states. However, for two-sided black holes, such observables are less interesting than their closed universe analog: the winding geodesics discussed in Section \ref{sec:example-obs-closed-universe} that appear even at the perturbative level are absent in the two-sided black hole setting. We do expect non-perturbative corrections to affect such correlators even for two-sided black holes when the time separation between operator insertions along the worldline of the observer becomes very large. While in JT gravity such large time separations are possible, for realistic black holes we expect the observer to hit the black hole singularity after a finite proper time, making such large time separations impossible. Instead, for two-sided black holes, we will discuss two other observables -- how the length of the Einstein-Rosen bridge changes with time (in this subsection) and the center-of-mass collision energy between the observer and an arbitrary matter perturbation coming from the opposite side (in Section \ref{sec:casimir}) -- that probe what happens to the observer after they pass the horizon.

Following the work of \cite{StaSus14, Sus15, Sus20}, we ask what happens to an observer who jumps into one side (say the right side) of a two-sided black hole. Unless the black hole was perfectly in the unperturbed thermofield double state, there will be matter excitations that may have fallen into the other side of the black hole.\footnote{There may also have been excitations which fell in on the side of the observer, but for simplicity, we restrict ourselves to the situation where all excitations other than the observer fall in from the left.} 
\begin{figure}[t!]
\begin{center}
\[\begin{tikzpicture}[baseline={([yshift=-0cm]current bounding box.center)}]
\fill[lightgray, opacity=0.3] (-2,-2) rectangle (2,2);
\fill[lightgray, opacity=0.3] (-2,-2) -- (0,-4) -- (2,-2);
\fill[lightgray, opacity=0.3] (-2,2) -- (0,4) -- (2,2);
\draw[dashed] (-2,-2) -- (2,2);
\draw[dashed] (-2,2) -- (2,-2);
\draw[black,very thick] plot [smooth, tension=1.3] coordinates {(-2,-2) (-1.7,0) (-2,2)};
\draw[black,very thick] plot [smooth, tension=1.3] coordinates {(2,-2) (1.7,0) (2,2)};
\draw[blue, thick, dotted] (-2,-1.8) -- node[above] {WH shrinking} (2,-1.8) ;
\fill[blue] (-1.95,-1.8) circle (0.1) node[left] {$V$};
\draw[blue,->] (-1.95,-1.8) -- (-1.95+1,-1.8+1);
\fill[red] (1.95,-1.8) circle (0.1) node[right] {$O$};
\draw[red,->] (1.95,-1.8) -- (1.95-1,-1.8+1);
\node at (0,0) {{\bf ouch!}};
\end{tikzpicture} 
\hspace{0.4cm}
\begin{tikzpicture}[baseline={([yshift=-0cm]current bounding box.center)}]
\fill[lightgray, opacity=0.3] (-2,-2) rectangle (2,2);
\fill[lightgray, opacity=0.3] (-2,-2) -- (0,-4) -- (2,-2);
\fill[lightgray, opacity=0.3] (-2,2) -- (0,4) -- (2,2);
\draw[dashed] (-2,-2) -- (2,2);
\draw[dashed] (-2,2) -- (2,-2);
\draw[black,very thick] plot [smooth, tension=1.3] coordinates {(-2,-2) (-1.7,0) (-2,2)};
\draw[black,very thick] plot [smooth, tension=1.3] coordinates {(2,-2) (1.7,0) (2,2)};
\draw[red, thick, dotted] (-2,1.8) -- node[below] {WH growing} (2,1.8) ;
\fill[blue] (-1.95,1.8) circle (0.1) node[left] {$V$};
\draw[blue,->] (-1.95,1.8) -- (-1.95+1,1.8+1);
\fill[red] (1.95,1.8) circle (0.1) node[right] {$O$};
\draw[red,->] (1.95,1.8) -- (1.95-1,1.8+1);
\end{tikzpicture} \simeq
\begin{tikzpicture}[baseline={([yshift=-0cm]current bounding box.center)}]
\fill[lightgray, opacity=0.3] (-2,-2) rectangle (2,2);
\fill[lightgray, opacity=0.3] (-2,-2) -- (0,-4) -- (2,-2);
\fill[lightgray, opacity=0.3] (-2,2) -- (0,4) -- (2,2);
\draw[dashed] (-2,-2) -- (2,2);
\draw[dashed] (-2,2) -- (2,-2);
\draw[black,very thick] plot [smooth, tension=1.3] coordinates {(-2,-2) (-1.7,0) (-2,2)};
\draw[black,very thick] plot [smooth, tension=1.3] coordinates {(2,-2) (1.7,0) (2,2)};
\draw[red, thick, dotted] plot [smooth, tension=.3] coordinates {(-1.95+.25,0) (-1.95+.55,0) (-1.95+1.2,.2) (0.9,1.8) (1.7,2) (2,2)};
\fill[blue] (-1.95+.25,0) circle (0.1) node[left] {$V$};
\draw[blue,->] (-1.95+.25,0) -- (-1.95+.25+1,0+1);
\fill[red] (2,2) circle (0.1) node[right] {$O$};
\draw[red,->] (2,2) -- (2-1,2+1);
\end{tikzpicture} \]
\end{center}
\caption{\label{fig:collision}  The left-most diagram illustrates how if the observer and particle jump in both at early times, the geodesic slice for the wormhole is shrinking. This leads to a large relative boost between the observer and the particle and a correspondingly high energy collision. If both insertions happen far in the future (center diagram), then the observer and particle experience a growing wormhole length and only a minor collision. In the right-most diagram, we illustrate how the boost symmetry of the thermofield double can be used to move one of the insertions down to $T=0$. In this work we choose to fix $V$ at $T=0$ on the left and vary the time that the observer jumps in. The above figure is reproduced from Figure 1 of \cite{IliLev24}. }
\end{figure}
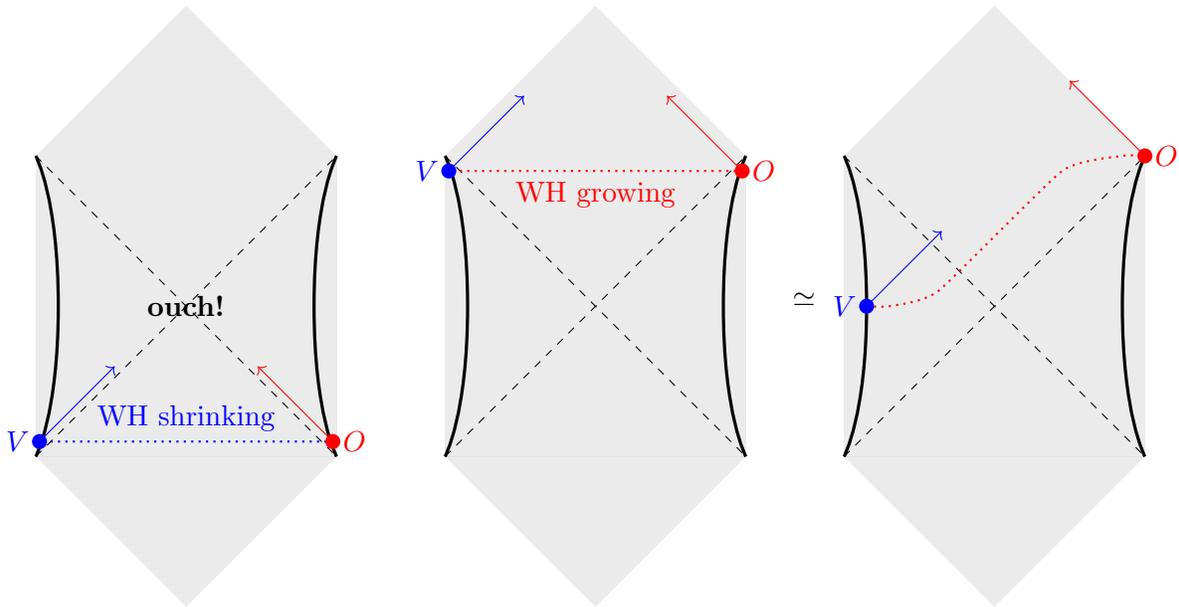
For simplicity, we will model all the matter excitations on the left as a single operator insertion on top of the thermofield double. As illustrated in Figure \ref{fig:collision}, if both the observer and matter fall in at early times, then there will be a large relative blueshift between the two, which leads to a very high collision energy. On the other hand, if they fall in at late times, this relative blueshift becomes a redshift, and the observer feels very little. In \cite{StaSus14,Stanford:2022fdt}, it was pointed out that a gauge invariant way of measuring blueshift/redshift is to compute the change in time of the length $\ell$ of the geodesic slice connecting the observer and matter particle, denoted by $p = \dot{\ell}$. If $p<0$, the system is in the so-called ``white hole" phase, and the observer and excitation collide with high energy. If $p>0$, the system is in the ``black hole" phase, and there is no collision. In \cite{Stanford:2022fdt}, it was observed that non-perturbative wormhole effects could cause a black hole to tunnel into a white hole at very late times. Thus, even though semiclassically an observer jumping into a very old black hole would expect to see no collision, a non-perturbative, quantum tunneling could lead the observer to see one. 

In \cite{IliLev24,Stanford:2022fdt,Blommaert:2024ftn}, the length of the geodesic slice and its time dependence were computed in various ways, and non-perturbative effects were taken into account. Since $p(t)$ and $\ell(t)$ are just properties of the geodesic slice, these computations can be done in pure JT gravity without matter. When non-perturbative effects are present, there may be more than one geodesic connecting the same two boundary points---a phenomenon that does not occur at disk level in JT gravity. There is, therefore, a non-perturbative ambiguity about which geodesic slice is used to compute $\ell(t)$. In \cite{Stanford:2022fdt}, a proposal was made for how to identify a unique geodesic slice whose properties could be analyzed, but it was not clear whether this proposal corresponded to a linear operator (on \emph{any} Hilbert space), and so it was not clear how to compute probabilities in a principled manner. In \cite{IliLev24}, an alternate proposal was made for an infinite family of length operators, which all act linearly on the non-perturbative Hilbert space of the two-sided black hole, $\Hnonp$, and agree with the disk-level predictions for the length operator. 

We now ask how the computations of \cite{IliLev24} get modified in the presence of an observer. Namely, we want to investigate length operators which act linearly on $\Hr$. Note that such operators can \emph{not} correspond to linear operators on $\Hnonp$, in agreement with arguments that observers measure non-linear operators on the boundary Hilbert space \cite{MarPol15,Papadodimas:2012aq}. We can start with a disk-level computation. To compute how the length of the geodesic slice connecting two boundary times changes in the presence of an observer, we want to compute the diagram
\begin{align}\label{eqn:lengthoppert}
   \bra{\mathcal{O}_{O}(T)} \hat{\ell}_{AB} \ket{\mathcal{O}_{O}(T)} \quad = \inlinefig[20]{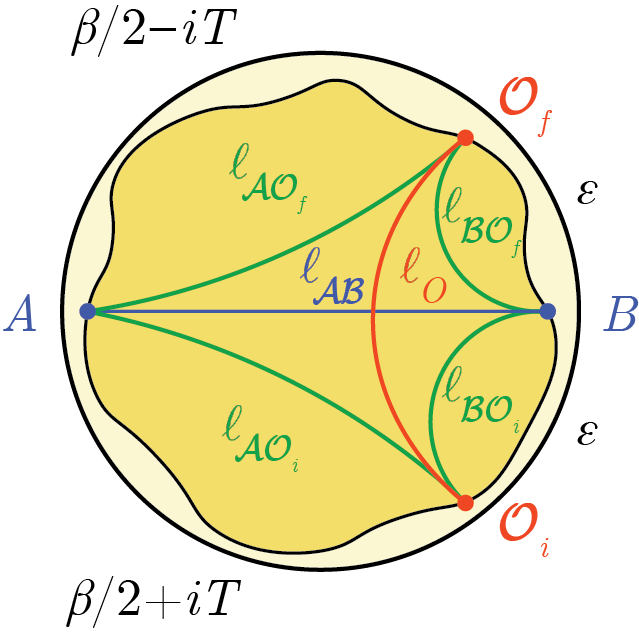}
\end{align}
where we keep in mind that all inner products are computed in $\Hrelp$. Here, we insert the observer at Lorentzian time $T$ and smear its insertion by $\varepsilon$ in Euclidean time to make the state normalizable. Since at the disk-level the observer worldline cannot connect with any other operators, the diagram in \eqref{eqn:lengthoppert} could be viewed as a calculation in either $\Hp$ or $\Hrelp$. 

This diagram can be computed by decomposing the disk into two triangular patches bounded by geodesics, with the observer's geodesic $\ell_{O}$ common to both. These triangles have vertices at $A,\ \mathcal{O}_i,\ \mathcal{O}_f$ and $B,\ \mathcal{O}_i,\ \mathcal{O}_f$ in the Figure in equation \eqref{eqn:lengthoppert}. Sewing along $\ell_{O}$ with the weighting $e^{-\Delta_O \ell_{O}}$, we find\footnote{Notice that here we normalize the state for convenience.}
\be 
& \bra{\mathcal{O}_{O}(T)} \hat{\ell}_{AB} \ket{\mathcal{O}_{O}(T)} = \frac{e^{S_0}}{\langle \cO_O, \frac{\beta}2- iT, \epsilon | \cO_O, \frac{\beta}2+ iT, \epsilon \rangle }\int ds_{i1} d{s_{i2}} ds_{f1} d{s_{f2}}\, \rho_0(s_{i1}) \dots \rho_0(s_{f2}) \nn \\ & \times \int d\ell_O d\ell_{A \cO_i} d\ell_{A \cO_f} d\ell_{B \cO_i} d\ell_{B \cO_f}\  \ell_{AB} \ e^{-\Delta_O \ell_O} I_3(\ell_{A\cO_i}, \ell_{A\cO_f}, \ell_O)I_3(\ell_{B\cO_i}, \ell_{B\cO_f}, \ell_O) \nn \\ &\times \Tilde{\varphi}_{s_{i1}}(\ell_{A\cO_i}) \Tilde{\varphi}_{s_{f1}}(\ell_{A\cO_f})\Tilde{\varphi}_{s_{i2}}(\ell_{\cO_i B})  \Tilde{\varphi}_{s_{f2}}(\ell_{\cO_f B})  e^{-\varepsilon\frac{s_{i2}^2}2 -  \varepsilon \frac{s_{f2}^2}2-(\beta/2+iT)\frac{s_{i1}^2}2-(\beta/2-iT)\frac{s_{f1}^2}2}
\ee
where $\ell_{AB}$ is a function of all five lengths in the problem $\ell_{AB} = \ell_{AB} \left( \ell_{O}, \ell_{A \cO_i}, \ell_{A \cO_f}, \ell_{B \cO_i}, \ell_{B \cO_f}\right)$. The factors of $I_3$ correspond to the path integral of JT gravity over triangular regions of the hyperbolic disk
\begin{align}
    I_3(\ell_1,\ell_2,\ell_3) \quad = \quad\begin{tikzpicture}[scale=0.7, baseline={([yshift=-0.1cm]current bounding box.center)}]
    \coordinate (A) at (0, 0);
    \coordinate (B) at (4, 0);
    \coordinate (C) at (2, 3);
    \fill[outsideyellow] (A) to[bend left=30] (B) -- (B) to[bend left=30] (C) -- (C) to[bend left=30] (A);
    \draw[thick] (A) to[bend left=30] node[midway, below] {$\ell_1$} (B);
    \draw[thick] (B) to[bend left = 30] node[midway, above right] {$\ell_2$} (C);
    \draw[thick] (C) to[bend left=30] node[midway, above left] {$\ell_3$} (A);
    \fill[red] (A) circle (2pt);
    \fill[red] (B) circle (2pt);
    \fill[red] (C) circle (2pt);
\end{tikzpicture} \quad = \quad e^{S_0} \int dE\rho_0(E) \ \Tilde{\varphi}_{s}(\ell_1) \Tilde{\varphi}_{s}(\ell_2) \Tilde{\varphi}_{s}(\ell_3),
\end{align}
where in the last equality, we used the results of \cite{Yang:2018gdb}. We remind the reader that the parameter $\varepsilon$ is present to ``smear" the state of the observer so that it is normalizable. In the limit of small $\varepsilon$, the $\ell$-integrals are dominated by very large, negative $\ell_{B\cO_i}$, $\ell_{B\cO_f}$ and $\ell_{O}$. The integrals over these three lengths then decouple from the remaining $\ell_{A \cO_i}$ and $\ell_{A\cO_f}$ integrals. Furthermore, in the limit of very negative $\ell_{O}$, $I_3(\ell_{A\cO_i},\ell_{A\cO_f},\ell_{O}) \sim \delta(\ell_{A\cO_i}-\ell_{A\cO_f})$ up to order one constants which are canceled by the normalization $\langle \cO_O, \frac{\beta}2- iT, \epsilon | \cO_O, \frac{\beta}2+ iT, \epsilon \rangle$. Additionally, in the small $\varepsilon$ limit $\ell_{A\cO_i} \approx \ell_{AB}$, and so all-in-all we get
\be
\bra{\cO_{O}(T)} \hat{\ell}_{AB} \ket{\cO_{O}(T)} \approx  \frac{e^{S_0}}{Z(\beta)}&\int ds_{i} ds_{f}\, \rho_0(s_{i}) \rho_0(s_{f})   \left(\int d\ell_{AB} \, \ell_{AB}\, \Tilde{\varphi}_{s_{i}}(\ell_{AB}) \Tilde{\varphi}_{s_{f}}(\ell_{AB})\right) \nn \\ &\quad \times e^{-(\beta/2+iT)E_{i} -(\beta/2-iT)E_{f} }.
\label{eq:616}
\ee
Since $\varepsilon$ controls how far away from the black hole the observer jumps from, in the $\varepsilon \to 0$ limit, we see that the presence of the observer does not change the expectation value $\bra{\cO_{O}(T)} \hat{\ell}_{AB} \ket{\cO_{O}(T)}$ to leading order in $\varepsilon$. Thus, the disk-level answer for the growth of the geodesic slice's length is approximately unchanged relative to what was found in \cite{IliLev24}. Indeed, this simplification actually occurs to all orders in the genus expansion. The only modification to equation \eqref{eq:616} when including higher genus corrections is to replace $\rho_0(s_i)\rho_0(s_f)$ with the non-perturbative correlator $\overline{\rho_0(s_i)\rho_0(s_f)}$. In other words, the only higher genus effects that are important at leading order in $\varepsilon$ come from wormholes connecting the two patches in the Figure of \eqref{eqn:lengthoppert} that are bounded by $T$-dependent boundary lengths. Thus, at least when the observer jumps in very far from the black hole (when $\varepsilon$ is small), the calculation of the growth of the length reduces to that discussed in \cite{IliLev24} to all orders in the genus expansion. It would be interesting to analyze such non-perturbative effects when $\varepsilon$ is not so small, but we leave this analysis for future work.

\subsection{Examples of what an observer sees when falling into a black hole (II):\\ collision energy with a shockwave past the horizon}
\label{sec:casimir}

In \cite{IliLev24}, the authors studied non-perturbative corrections to the center-of-mass (CoM) energy of the collision between the matter particle, inserted on the left, and the observer, inserted on the right. The collision energy is a natural observable since it can be measured locally by the observer -- by simply measuring the components of the stress tensor along their worldline -- as opposed to the change in length of the geodesic slice which cannot. Furthermore, it was shown in \cite{IliLev24} that statistics of the CoM energy can be computed analytically by studying the distribution of Casimir energies of the matter-plus-observer state. The definition of the Casimir energy will be reviewed shortly. Surprisingly, even though this is a low complexity observable, the authors found that the non-perturbative corrections became important for this observable at times that are polynomial in the entropy of the black hole as opposed to exponential in the entropy. The calculations in \cite{IliLev24} seemed to suggest that the notion of an observer can cease to make sense at times much earlier than exponential in the black hole entropy, and so EFT can break down at times of order the Page time. In this section, we revisit the CoM observable and modify it to be an operator acting on the relational Hilbert space $\Hr$. First, we review the discussion of \cite{IliLev24}.

The observable discussed in \cite{IliLev24} was the center-of-mass energy of a collision experienced by an observer falling into one side of a black hole in the thermofield double state when a particle has been dropped into the black hole on the other side. We return now to this observable, but compute what the observer experiences using our new rules for the inner product on the non-perturbative relational Hilbert space. Consider the state $\ket{V\Oo(-T)}$, where $V$ is inserted on the left at boundary time $t=0$ and $\Oo$ on the right at boundary time $t=-T$. Perturbatively, we can understand the resulting state as living in the Hilbert space of JT gravity with matter \cite{IliLev24, Penington:2023dql,Kolchmeyer:2023gwa}. As we reviewed in Section \ref{sec:review}, the matter states can be organized into representations of the background SL(2,$\mathbb{R}$)  symmetry, with ordinary matter falling into discrete series representations. These representations are labeled by a Casimir value $C_{\Delta}$ related to the mass/dimension $\Delta$ of the field via $C_{\Delta} = \Delta (\Delta -1)$. 

When two operators labeled by masses $\Delta_O$ and $\Delta_V$ are inserted at the boundary, the full bulk state, including the perturbative gravitational fluctuations, can then be expanded as a superposition over different representations by using ``gravitationally-dressed blocks" \cite{Jafferis:2022wez}. The operator that measures the SL(2, $\mathbb{R}$) Casimir in each representation appearing in this expansion can be directly related to the operator which measures the center-of-mass collision energy experienced by the two particles \cite{IliLev24}. By measuring the Casimir, we are measuring the collision energy experienced by the observer.

We now explain these words in more detail. In \cite{IliLev24, Jafferis:2022wez}, the gravitationally-dressed blocks can be represented by the diagram
\def\red{\color{red}}
\def\blue{\color{blue}}
\begin{align}
   P^{\Delta_V, \Delta_O}_n(s;s_L,s_R) \quad = \quad   \begin{tikzpicture}[scale=0.7, baseline={([yshift=-0.1cm]current bounding box.center)}]
    \node at (-1.1,0) {$s_L$} ;
    \node at (1.1,0) {$s_R$} ;
    \node at (0,1.5) {$s$} ;
 	\node at (-1.4,1.8) {$\blue \Delta_V$} ;
 	\node at (1.4,1.8) {$\red \Delta_O$} ;
 	\draw[thick, blue] (-1,1.5) -- (0,0.5); 
 	\draw[thick, red] (1,1.5) -- (0,0.5);
 	\draw[thick] (0,0.5) -- (0,-0.72) node[below] {$\Delta_V + \Delta_O + n$};
\end{tikzpicture}\,,\label{P}
\end{align}
where the observer and particle have masses $\Delta_O$ and $\Delta_V$ respectively. Like in previous sections, the $s$-labels are related to the energy of the boundary particle in that region by $E = s^2/2$. The resulting Casimir value, $C_n$, after fusing these two operators is that of a particle with mass $\Delta_n = \Delta_V + \Delta_O +n$, namely $C_n = \Delta_n (\Delta_n -1)$. Furthermore, these blocks obey an orthogonality relation\footnote{Note that due to a factor of 2 difference between $\rho_0$ in this work compared to \cite{Jafferis:2022wez}, our $P^{\Delta_V,\Delta_O}_n$'s differ by a factor of $\sqrt{2}$ from \cite{Jafferis:2022wez} so that equation \eqref{eqn:Portho} still holds.}
\begin{align}\label{eqn:Portho}
    \int ds \  \rho_0(s) P^{\Delta_V,\Delta_O}_n(s;s_L,s_R)P^{\Delta_V,\Delta_O}_m(s;s_L,s_R) = \delta_{nm}
\end{align}
and a completeness relation
\begin{align}
    \frac{\delta(s-s')}{\rho_0(s)}=\sum^{\infty}_{n=0} P_n^{\Delta_V, \Delta_O}(s;s_L, s_R) P_n^{\Delta_V, \Delta_O}(s';s_L,s_R).    \label{schannelIdentity}
\end{align}
This completeness relation can be summarized in the following schematic \cite{Jafferis:2022wez} 
\begin{align}
\begin{tikzpicture}[scale=0.7, baseline={([yshift=-0.1cm]current bounding box.center)}]
    \node at (-2.2,1.8) {$\blue \Delta_V$} ;
    \node at (2.2,1.8) {$\red \Delta_O$} ;
    \node at (-2.2,-1.6) {$\blue \Delta_V$} ;
    \node at (2.2,-1.6) {$\red \Delta_O$} ;
    \node at (-2.2,0) {$s_L$} ;
    \node at (2.2,0) {$s_R$} ;
    \node at (0,1.5) {$s$} ;
    \node at (0,-1.5) {$s'$} ;
 	\draw[thick, blue] (-2,1.3) -- (-1,0); 
 	\draw[thick, blue] (-2,-1.3) -- (-1,0); 
 	\draw[thick, dashed] (-1,0) -- (1,0) node[midway,above] {$\mathbf{1}$} ;
 	\draw[thick, red] (1,0) -- (2,-1.3);
 	\draw[thick, red] (1,0) -- (2,1.3);
\end{tikzpicture}  \quad 
=  \quad \sum_{n=0}^\infty 
\begin{tikzpicture}[scale=0.7, baseline={([yshift=-0.1cm]current bounding box.center)}]
    \node at (-1.2,0) {$s_L$} ;
    \node at (1.2,0) {$s_R$} ;
    \node at (0,1.5) {$s$} ;
    \node at (0,-1.5) {$s'$} ;
 	\node at (-1.4,1.8) {$\blue \Delta_V$} ;
 	\node at (1.4,1.8) {$\red \Delta_O$} ;
 	\node at (-1.4,-1.6) {$\blue \Delta_V$} ;
 	\node at (1.4,-1.6) {$\red \Delta_O$} ;
 	\draw[thick, blue] (-1,1.5) -- (0,0.5); 
 	\draw[thick, red] (1,1.5) -- (0,0.5);
 	\draw[thick] (0,0.5) -- (0,-0.5) node[midway, right] {$n$};
 	\draw[thick, blue] (0,-0.5) -- (-1,-1.5);
 	\draw[thick, red] (0,-0.5) -- (1,-1.5);
\end{tikzpicture}\label{identityBlock}
\end{align}
Using these relations, we can define states of definite Casimir by integrating over states defined by asymptotic boundary conditions as\footnote{From this equation forward, we will often drop the super-script $\Delta_V, \Delta_O$ on the gravitational-blocks $P^{\Delta_V,\Delta_O}_n$'s since we will always be considering the same two operators, $V$ and $\cO_O$.} 
\begin{align}\label{eqn:casimirstates}
    \ket{n, s_L, s_R}_{LR} = e^{-S_0/2} \int ds\, \rho_0(s) \frac{P_n(s;s_L, s_R)}{\sqrt{\gamma_{\Delta_V}(s_L, s) \gamma_{\Delta_O}(s_R,s)}} \ket{s_L, s_R, V \Oo s}_{LR}
\end{align}
where the normalization factors $\gamma$ are given in equation \eqref{eq:gamma}. By construction, the eigenvalue of the state $\ket{n, s_L, s_R}$ under the Casimir operator $\widehat{C}$ is given by $\widehat{C}\ket{n,s_L,s_R} = C_n \ket{n, s_L,s_R}$. The state $\ket{s_L, s_R, V \Oo s}_{LR}$ is a bulk state defined by asymptotic boundary conditions, pictorially represented as
\begin{align}
\ket{s_L, s_R, V\Oo s}_{LR} \quad = \quad
\begin{tikzpicture}[scale=1, baseline={([yshift=-0.1cm]current bounding box.center)}]
    \coordinate (A) at (0, 0);
    \coordinate (B) at (1.414,-1.414);
    \coordinate (C) at (-1.414,-1.414);
    \coordinate (D) at (1,0);
    \coordinate (E) at (-1,0);
    \draw[thick, fill=outsideyellow] (2,0) arc[start angle=0,end angle=-180,radius=2cm];
    \fill[red] (B) circle (2pt);
    \fill[red] (C) circle (2pt);
    \draw[thick, red] (C) .. controls (-1.2, -1) and (-1,-.4) .. (E);
    \draw[thick, red] (B) .. controls (1.2, -1) and (1,-.4) .. (D);
    \node[below right] at (B) {$\Oo$};
    \node[below left] at (C) {$V$};
    \node[above] at (0,-1) {$s$};
    \node[above right] at (-1.8,-.75) {$s_L$};
    \node[above left] at (1.8,-.75) {$s_R$};
\end{tikzpicture}
\quad .
\end{align}
From this definition, we have
\begin{align}
    \leftindex_{LR}{\bra{s_L, s_R, V \Oo s}}\ket{s_L', s_R',V \Oo s' }_{LR} =e^{S_0} \frac{\delta(s_L - s_L') \delta(s_R - s_R')\delta(s -s')}{\rho_0(s_L)\rho_0(s_R)\rho_0(s)} \gamma_{\Delta_V}(s_L, s)\gamma_{\Delta_O}(s_R,s)
    \label{eq:VOnorm}
\end{align}
Therefore, the norm of the states in \eqref{eqn:casimirstates} is\footnote{The states $\ket{n,s_L,s_R}$ are the same (up to a normalization constant) as the states $\ket{\Delta,E_L,E_R}$ introduced in equation \eqref{eq:hpert_matter}, where $\Delta=C_n$.}
\begin{align}
\bra{n, s_L, s_R} \ket{m,s_L',s_R'} &=\frac{\delta(s_L - s_L') \delta(s_R - s_R') }{\rho_0(s_L) \rho_0(s_R)} \times \nonumber \\
&\int ds ds' \rho_0(s) \rho_0(s') P_n(s; s_L, s_R) P_m(s'; s_L, s_R) \frac{\delta(s - s')}{\rho_0(s')} = \frac{\delta_{nm} \delta(s_L - s_L') \delta(s_R - s_R') }{\rho_0(s_L) \rho_0(s_R)},
\end{align}
where we used \eqref{eqn:Portho}. We then see that the $\ket{n,s_L,s_R}$ states are also plane-wave normalizable, and form a complete basis of states in which we can expand a given state. To make these normalizable in the standard sense, we should imagine integrating $s_L$ and $s_R$ over a small window of size $\delta s$ centered around $\bar{s}_{L,R}$. To normalize the states we divide by $\sqrt{\rho_0(\bar{s}) \delta s}$ so that the normalized states are 
\begin{align}\label{eqn:smearedcasimirstates}
   \ket{n, \bar{s}_L, \bar{s}_R} =  \int^{\bar{s}_{L,R} + \delta s/2}_{\bar{s}_{L,R}-\delta s/2} ds_L ds_R \frac{\rho_0(s_L)\rho_0(s_R)}{\delta s \sqrt{\rho_0(\bar{s}_L) \rho_0(\bar{s}_R)}} \ket{n, s_L, s_R}.
\end{align}
We will often leave this smearing over $s_L$ and $s_R$ implicit below unless necessary.

\subsubsection*{Perturbative probability of detecting a fixed Casimir}
To understand the distribution of collision energies discussed in \cite{IliLev24}, we would like to expand the state with an observer inserted at Lorentzian boundary time $t=-T$ in the past on the right and a particle inserted at boundary time $t=0$ on the left, denoted $\ket{V(0) \OR(-T)}$ in the basis of Casimir states we just introduced. To make the state $\ket{V(0) \OR(-T)}$ normalizable we move the two operators off the $\tau=0$ slice in Euclidean time by an amount $\varepsilon$. The overlap between the two states is 
\be
\begin{aligned}\label{eqn:overlapbwcasandstate}
 \Big\langle n,s_L, s_R\Big |&V(0)\OR(-T)\Big\rangle \\
 &=  e^{-S_0/2} \int ds \rho_0(s) \frac{P_n(s; s_L, s_R)}{\sqrt{\gamma_{\Delta_V}(s_L, s) \gamma_{\Delta_O}(s_R,s)}}\bra{s_L, s_R, V\OR s}\ket{V(0)\OR(-T)}
\end{aligned}
\ee
At the disk level, the inner product inside the integral is just
\be
\begin{aligned}
\Big\langle s_L, s_R, V\OR s\Big|&V(0)\OR(-T)\Big\rangle \\
&= e^{S_0} e^{-\frac{\varepsilon s_L^2}{2} - \frac{\varepsilon s_R^2}{2}-(\beta/2-iT)\frac{s^2}2} \times\sqrt{\gamma_{\Delta_V}(s_L,s) \gamma_{\Delta_V}(s_L,s)\gamma_{\Delta_O}(s_R,s) \gamma_{\Delta_O}(s_R,s)}
\end{aligned}
\ee
and so 
\be
\begin{aligned}
\Big\langle n,s_L,s_R\Big |&V(0)\OR(-T)\Big\rangle \\
&= e^{S_0/2}\int ds \rho_0(s) P_n(s; s_L, s_R) e^{-\frac{\varepsilon s_L^2}{2} - \frac{\varepsilon s_R^2}{2}-(\beta/2-iT)\frac{s^2}2 } \times 
\sqrt{\gamma_{\Delta_V}(s_L,s)\gamma_{\Delta_O}(s_R,s) }.
\label{eq:624}
\end{aligned}
\ee
With these overlaps in hand, we can then compute the expectation value in the $\ket{V \cO_O}$ state of the projector onto a fixed $n$ eigenspace 
\begin{align}\label{eqn:nproj}
    \Pi_n \equiv \int ds_Lds_R\,\rho_0(s_L) \rho_0(s_R) \ketbra{n,s_L,s_R}
\end{align}
We thus get that the distribution of $n$-eigenvalues takes the form
\begin{align}\label{eqn:fulldistributionmain}
    p(n) =& \frac{e^{S_0}}{\braket{V\OR}} \int_0^{\infty} \prod_{j=1}^4 \left( ds_j \rho_0(s_j) e^{-\tau_j \frac{s_j^2}2}\right)\times \nonumber \\
& \left[\gamma_{\Delta_V}(s_1,s_2)\gamma_{\Delta_O}(s_2,s_3)\gamma_{\Delta_O}(s_3,s_4)\gamma_{\Delta_V}(s_4,s_1)\right]^{1/2} P_n^{\Delta_V \Delta_O}(s_4;s_1, s_3) P_n^{\Delta_V \Delta_O}(s_2;s_1,s_3),
\end{align}
where here $\lbrace \tau_1, \tau_2, \tau_3, \tau_4 \rbrace = \lbrace 2\varepsilon, \beta/2-iT, 2\varepsilon, \beta/2+iT\rbrace$ and $s_1 = s_L,\ s_3 = s_R,\ s_2=s$ and $s_4 = s'$. This is the Casimir distribution quoted in \cite{IliLev24}. By using equation \eqref{schannelIdentity} together with the fact that
\begin{align}\label{eqn:statenorm}
    \braket{V\OR} = e^{S_0}\int_0^{\infty} \prod_{j=1}^3 \left( ds_j \rho_0(s_j) e^{-\tau_j \frac{s_j^2}2}\right) \gamma_{\Delta_V}(s_1,s_2)\gamma_{\Delta_O}(s_2,s_3), \quad \lbrace \tau_1, \tau_2, \tau_3 \rbrace = \lbrace 2\varepsilon, \beta, 2\varepsilon\rbrace,
\end{align}
we see that this distribution is normalized in $n$.

Without Schwarzian, gravitational effects, this distribution is sharply peaked for early times about $n \sim \sqrt{\Delta_V \Delta_O} e^{\pi T/\beta}$.\footnote{For an explanation of this point, see Sec. 7.2 of \cite{IliLev24}.} Thus, the collision energy experienced by the observer grows exponentially. To understand gravitational effects, it is then warranted to expand the gravitational blocks above at large $n$. As explained in Appendix B of \cite{IliLev24}, at large $n$ the blocks $P_n(s;s_L,s_R)$ take the form
\begin{align}\label{eqn:Pnlargen}
P_n(s;s_L,s_R) =\sqrt{\frac{4 \Gamma(\Delta_V \pm is_L \pm is) \Gamma(\Delta_O \pm is_R \pm is)}{n}} \left(  \frac{\Gamma(2is)e^{2is \log n}}{\Gamma(\Delta_V \pm is_L + is) \Gamma(\Delta_O \pm is_R + is)} + c.c.\right).
\end{align}
Furthermore, at large $\Delta_V/\varepsilon$ and $\Delta_O/\varepsilon$, it is natural to work with re-scaled variables $s_{1,3} = \Delta_{V,O}/\varepsilon \sigma_{1,3}$ in the integral in \eqref{eqn:fulldistributionmain}. Expanding the Gamma-functions in \eqref{eqn:Pnlargen} at large $\Delta_{V,O}/\varepsilon$ and plugging the result into \eqref{eqn:fulldistributionmain}, we can then perform the $\sigma_1$ and $\sigma_3$ integrals by saddle point, and indeed one finds saddle point values at $\sigma_{1,3} \sim \mathcal{O}(1)$. Furthermore, in this limit, the $\sigma_{1,3}$ integrals decouple from the $s_2,s_4$ integrals in \eqref{eqn:fulldistributionmain}, and, to leading order in $\Delta/\varepsilon$, the $\sigma_{1,3}$ integrals just give a factor of $\braket{V\OR}/Z(\beta)$ with $\braket{V\OR}$ evaluated in the $\Delta/\varepsilon \to \infty$ limit of \eqref{eqn:statenorm}. 

When the dust settles, one finds that the full normalized distribution in \eqref{eqn:fulldistributionmain} is \cite{IliLev24} 
\begin{align}\label{eqn:totalprobdist}
& p(n) = \frac{4}{n Z(\beta)}  \int_{-\infty}^{\infty} ds_2 ds_4 \rho_0(s_2) \rho_0(s_4)\exp \left(-\frac{\beta}4 (s_2^2 + s_4^2) + i\frac{T}2 (s_2^2 - s_4^2)  \right)   \times \nonumber \\
&\Gamma(2i s_4)\Gamma(-2is_2) \exp \left(i2 (s_2 - s_4) \log \frac{n}{4 \frac{\Delta}{\varepsilon}}\right)  + \mathcal{O}(\varepsilon/(\Delta)).
\end{align} 
Finally changing to sum and difference variables, $s_{\pm} = s_2 \pm s_4$, we see that the integral is dominated by $s_- \ll s_+$. The resulting integral becomes a sharply peaked Gaussian in $\log n/n_0$ with mean $\sim \pi T/\beta$ and width growing linearly in $T$. Changing from $n$ to $C_n \sim n^2$ variables, where $C_n$ is the eigenvalue under $\widehat{C}$, one finds the Casimir distribution quoted in \cite{IliLev24}
\begin{align}\label{eqn:casimirdistributionpert}
p(C) = \frac{\beta^{3/2}\Phi_b^{\frac{1}{2}}\log \frac{C}{C_0}}{C \sqrt{8\pi^3}T^2} \exp\left(-\frac{\beta\Phi_b}{2T^2}\left(\log \frac{C}{C_0} - 2\pi T/\beta\right)^2\right),
\end{align}
where $C_0 = \frac{16 \Delta^2 \Phi_b^2}{\varepsilon^2}$. Here we reintroduced the Schwarzian coupling $\Phi_b$, with $\Phi_b\gg \beta,T$ in the semiclassical limit which we are interested in. Integrating over $C$, one can check that this distribution is indeed normalized. The important point about this distribution is just that near the peak at $C \sim C_0 e^{2\pi T/\beta}$,\footnote{The distribution \eqref{eqn:casimirdistributionpert} is only valid for $C\gg 1$, $C\gg C_0$. The peak is at $C \sim C_0 e^{2\pi T/\beta}$ only in the semiclassical limit $\gamma\gg\beta,T$.} the distribution has a height which is exponentially decaying in time, due to the factor of $1/C$ in front of the Gaussian. In other words, at times of order $S_0$, the distribution is dominated by values which are exponentially small in the entropy and so are susceptible to $e^{-S_0}$-suppressed corrections. We turn now to discussing such non-perturbative corrections and how they are suppressed in the presence of an observer. 

 \subsubsection*{Wormhole corrections to the CM collision energy: the global picture}
So far, we have computed the disk-level/semi-classical contribution to the distribution of Casimir energies in the state $\ket{V(0) \OR(-T)}$. This can be visualized as computing the contribution to the gravitational path integral from the spacetime
\begin{align}
\begin{tikzpicture}[baseline={([yshift=0cm]current bounding box.center)}, scale=.8]
    \draw[thick] (0,0) circle (2.0);
     \draw[thick, fill=outsideyellow]
        plot [smooth cycle, tension=1] coordinates {
            (1.9, 0)
            (1.75, 0.5)
            (1.25, 1.4)
            (0.8, 1.6)
            (0.4, 1.8)
            (0, 1.9)
            (-0.5, 1.7)
            (-1.1, 1.5)
            (-1.7, 1)
            (-1.9, 0.3)
               (-1.95, 0)
            (-1.8, -0.4)
            (-1.5, -1.1)
            (-1.2, -1.6)
            (-0.7, -1.8)
            (-0.2, -1.9)
            (0.5, -1.7)
            (1.1, -1.3)
            (1.6, -0.8)
            (1.8, -0.2)
        };
    \draw[red, thick] ({2*cos(10)-0.2},{-2*sin(10)}) arc (270:90:0.348623);
    \node[black] at (0,-2.4) {$\beta+i t$};
     \node[black] at (0,2.4) {$\beta-i t$};
    \node[purple] at (0,0.25) {\Large $\Pi_n$} ;
    \draw[thick, purple, dashed] (-1.95,0)--(1.9,0) ;
    \fill[red] ({2*cos(10)-0.2},{-2*sin(10)}) circle (0.1);
     \node[below left] at ({2*cos(10)-0.2},{-2*sin(10)}) {$O_i$}; 
    \fill[red] ({2*cos(10)-0.2},{2*sin(10)}) circle (0.1);
         \node[above left] at ({2*cos(10)-0.2},{2*sin(10)}) {$O_f$}; 
    \fill[blue] ({-(2*cos(10)-0.15)},{-2*sin(10)}) circle (0.1);
     \node[blue, below left] at ({-(2*cos(10)-0.15)},{-2*sin(10)}) {$V$}; 
    \fill[blue] ({-(2*cos(10)-0.15)},{2*sin(10)}) circle (0.1);
          \node[blue, above left] at ({-(2*cos(10)-0.15)},{2*sin(10)}) {$V$}; 
              \draw[blue, thick] ({-(2*cos(10)-0.15)},{-2*sin(10)}) arc (-90:90:0.348623);=
    \end{tikzpicture},
    \end{align} where we have inserted the projector onto fixed $n$ states, $\Pi_n$, defined in \eqref{eqn:nproj}. Now we would like to compute the contribution from higher genus wormholes. The contribution we will be most interested in is when the projector, $\Pi_n$, ``threads" the wormhole as follows:\footnote{One can check that contributions in which the projector does not thread the wormhole are exponentially suppressed in $S_0$ and have the same $1/n$-suppression as the perturbative result.}

\begin{align}\label{eqn:caswh}
\begin{tikzpicture}[scale=.7,baseline={([yshift=-0.0cm]current bounding box.center)}]
\fill[blue] (2.2,2.4) circle (1.5pt) node[below left] {$P_2$};
\draw[thick,  blue, fill = outsideyellow] (2, 0) ellipse (1 and 3);
\fill[blue] (1,0) circle (0.1) node[below right] {$V$};
\fill[blue] (1.1,1) circle (0.1) node[above right] {$V$};
\fill[red] (2.75,2.05) circle (0.1) node[above right] {$O$};
\fill[red] (2.9,1.2) circle (0.1) node[above right] {$O$};
 \draw[thick, purple, dashed]
        plot [smooth, tension=1] coordinates {
            (1.05,0.5)
            (2,0)
            (3,0)
        };
\clip (2.25,-1.51) rectangle (5.6,1.51);
\draw[smooth, thick, blue, fill = outsideyellow] (0,0.5) to[out=0,in=0] (2,0.5) to[out=0,in=210] (3,1) to[out=30,in=150] (5,1) to[out=-30,in=30] (5,-1) to[out=210,in=-30] (3,-1) to[out=150,in=0] (2,-0.5) to[out=0,in=0] (0,-0.5) to[out=150,in=-150] (0,0.5);
\fill[smooth, white] (3.52,0) .. controls (3.8,-0.2) and (4.2,-0.2) .. (4.48,0);
\draw[smooth, blue, fill = white] (3.5,0) .. controls (3.8,0.2) and (4.2,0.2) .. (4.5,0);
\draw[smooth, blue] (3.4,0.1) .. controls (3.8,-0.25) and (4.2,-0.25) .. (4.6,0.1);
 \draw[thick, purple, dashed]
        plot [smooth, tension=1] coordinates {
            (1.05,0.5)
            (2,0)
            (3.52,0)
        };
  \draw[thick, purple, dashed]
        plot [smooth, tension=1] coordinates {
            (2.9, 1.7)
            (2.6, 1)
            (2.8,0.85)
        };       
    \node[purple] at (3.85,-0.5) {\Large $\Pi_n$} ;      
\end{tikzpicture} \quad.
\end{align}

Now by equation \eqref{eqn:nproj}, computing the expectation value of $\Pi_n$ requires us to compute the square of $\bra{n, s_L, s_R}\ket{V(0)\OR(-T)}$. Accordingly, we can cut the wormhole in \eqref{eqn:caswh} along the dashed line and view it as a bra-ket to bra-ket wormhole between the two copies of $\bra{n, s_L, s_R}\ket{V(0)\OR(-T)}$ in the expectation value of $\Pi_n$. Via the definition of the states in \eqref{eqn:casimirstates}, we can then replace the $\ket{n,s_L,s_R}$ states with an integral over $s$ of $\ket{s_L, s_R, V\OR s}$ states. In this way, we can view the wormhole in \eqref{eqn:caswh} as (an integral of) the cylinder between two asymptotic boundaries. These two steps of cutting the wormhole in \eqref{eqn:caswh} and then replacing $\ket{n,s_L,s_R}$ boundaries with asymptotic boundaries are depicted in Figure \ref{fig:cutwormhole}.\footnote{Note that the factor of $e^{-S_0/2}$ in \eqref{eqn:overlapbwcasandstate} accounts for the $e^{S_0}$ difference between the geometry in \eqref{eqn:caswh} (genus one) and the geometries in Figure \ref{fig:cutwormhole} (genus zero).}

\begin{figure}[h]
    \centering
    \begin{subfigure}[c]{.45\textwidth}
        \includegraphics[width=\linewidth]{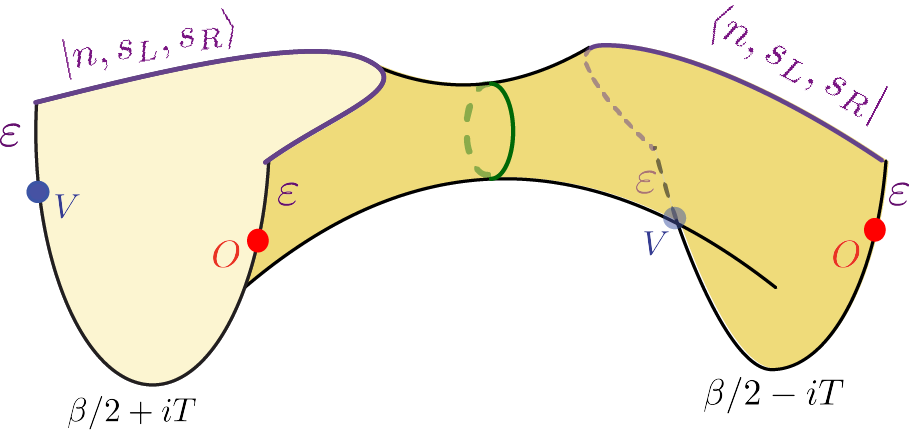}
    \end{subfigure}
    \hspace{1cm}
    \begin{subfigure}[c]{.45\textwidth}
        \includegraphics[width=\linewidth]{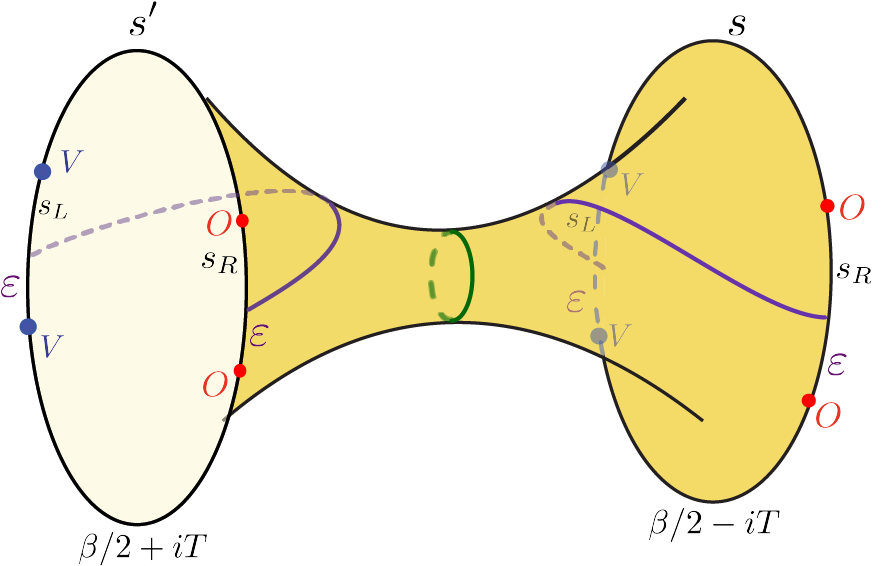}
    \end{subfigure}
    \caption{The left figure illustrates the geometry obtained upon cutting the wormhole in \eqref{eqn:caswh} along the dashed line. The resulting geometry corresponds to the bra-ket to bra-ket wormhole contributing to the square of $\bra{n, s_L, s_R}\ket{V(0)\OR(-T)}$. In the right figure, the $\ket{n,s_L,s_R}$ boundary conditions are replaced by (an integral of) asymptotic boundary segments given by the state $\ket{s_L,s_R, V \OR s}$.}
    \label{fig:cutwormhole}
\end{figure}

We can then compute the two-boundary wormhole in Figure \ref{fig:cutwormhole} by summing over different contractions of the various boundary operators. Without any restrictions on which contractions we sum over, the dominant contraction will be the one where each operator pairs up with its corresponding operator on the opposite boundary. Since the two-boundary wormhole can be constructed by doing an identification of global AdS$_2$, the relevant contraction can be visualized as a diagram on global AdS$_2$ as 
\def\red{\color{red}}
\def\blue{\color{blue}}
\def\yellow{\color{yellow}}
\begin{align}\label{eqn:contractionclass3}
\begin{tikzpicture}[scale=1, baseline={([yshift=-0.1cm]current bounding box.center)}]
    \coordinate (T1) at (-1.5, 1.5);
    \coordinate (T2) at (0, 1.5);
    \coordinate (T3) at (1, 1.5);
    \coordinate (T4) at (2.5, 1.5);
    \coordinate (T5) at (4, 1.5);
    \coordinate (B1) at (-1.5, -1.5);
    \coordinate (B2) at (0, -1.5);
    \coordinate (B3) at (1, -1.5);
    \coordinate (B4) at (2.5, -1.5);
    \coordinate (B5) at (4, -1.5);
    \coordinate (slashl) at (-1.5,0);
    \coordinate (slashr) at (4, 0);
    \fill[fill=red!20] (B4) -- (T4) -- (T5)--(B5)--cycle;
    \fill[fill=green!20] (B1) -- (T1) -- (T2)--(B2) -- cycle;
    \fill[fill=yellow!20] (B3) -- (T3) -- (T4)--(B4) -- cycle;
    \draw[thick, dashed] (T1) -- (B1);
    \draw[thick, red] (T5) -- (B5);
    \draw[thick, red] (T4) -- (B4);
    \fill[blue!20] (T2) -- (B2) -- (B3) -- (T3) -- cycle;
    \draw[thick, blue] (T2) -- (B2);
    \draw[thick, blue] (T3) -- (B3);
    \draw[thick] (T1) -- (T2) -- (T3) -- (T4) -- (T5);
    \draw[thick] (B1) -- (B2) -- (B3) -- (B4) -- (B5);
    \node[above] at (.5, -1.58) {$s_L$};
    \node[below] at (.5, 1.58) {$s_L$};
    \node[above] at (3.25, -1.58) {$s_R$};
    \node[below] at (3.25, 1.58) {$s_R$};
    \node[above] at (1.75, -1.58) {\scriptsize $\beta/2 + iT$};
    \node[below] at (1.75, 1.58) {\scriptsize $\beta/2 - iT$};
    \node[above] at (-.75, -1.58) {$s$};
    \node[below] at (-.75, 1.58) {$s'$};
    \node at (slashl) {$=$};
    \node at (slashr) {$=$};
    \node[above] at (T2) {$V$};
    \node[above] at (T3) {$V$};
    \node[above] at (T4) {$\Oo$};
    \node [above right] at (T5) {$\Oo$};
    \node[below] at (B2) {$V$};
    \node[below] at (B3) {$V$};
    \node[below] at (B4) {$\Oo$};
    \node [below right] at (B5) {$\Oo$};
    \node at (-7,0) {\includegraphics[width=8cm]{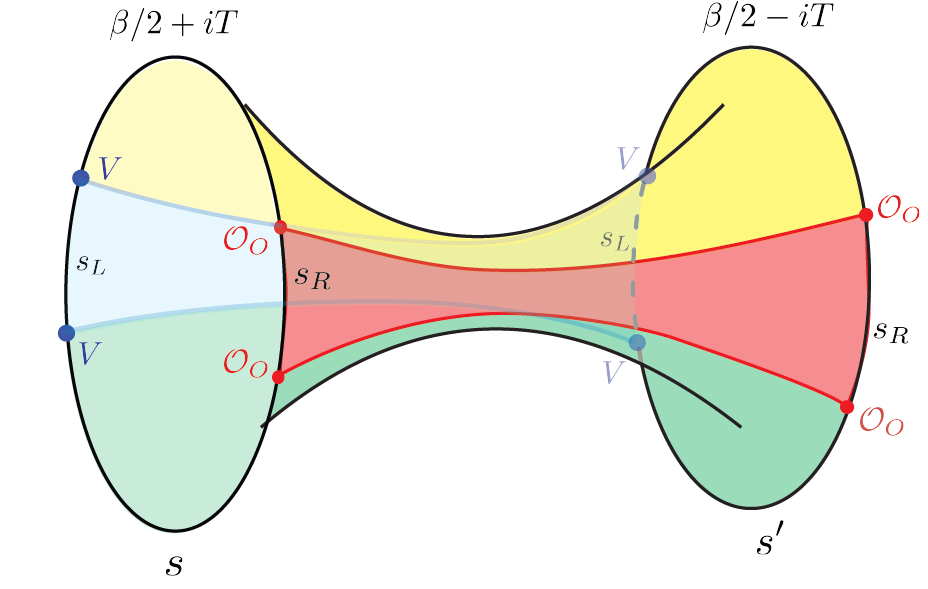}};
    \node at (-2.5,0) {$=$};
    \fill[red] (T1) circle (2pt);
    \fill[blue] (T2) circle (2pt);
    \fill[blue] (T3) circle (2pt);
    \fill[red] (T4) circle (2pt);
    \fill[red] (T5) circle (2pt);
    \fill[red] (B1) circle (2pt);
    \fill[blue] (B2) circle (2pt);
    \fill[blue] (B3) circle (2pt);
    \fill[red] (B4) circle (2pt);
    \fill[red] (B5) circle (2pt); 
\end{tikzpicture}\ .
\end{align}
From now on, below, we will directly draw the diagram on the right that schematically shows the operator contraction.  Such diagrams can be computed using standard rules derived in \cite{Jafferis:2022wez}. The rules state that for each intersection of a matter/observer worldline with an asymptotic boundary particle worldline, the diagram receives a $\sqrt{\gamma_{\Delta}(s_1,s_2)}$ vertex factor. Furthermore, asymptotic boundary regions that are shaded the same color share the same energy or $s$-variable. Putting these rules together, this diagram gives
\begin{align}\label{eqn:domcontraction}
    &\int d\lambda\, \rho_0(\lambda)\, e^{-\varepsilon (s_L^2 + s_R^2)-\beta\frac{\lambda^2}2} \sqrt{\gamma_{\Delta_V}(s_L,s)\gamma_{\Delta_O}(s_R,s)\gamma_{\Delta_V}(s_L,s')\gamma_{\Delta_O}(s_R,s')}  \gamma_{\Delta_V}(s_L,\lambda) \gamma_{\Delta_O}(s_R, \lambda) \frac{\delta(s - s')}{\rho_0(s)}.
\end{align}
Integrating this expression against $P_n(s)$ and $P_n(s')$ to project onto fixed Casimir states as in \eqref{eqn:casimirstates}, we see that the $\gamma$ factors inside the square-root in \eqref{eqn:domcontraction} cancel against those in \eqref{eqn:casimirstates}. Plugging into \eqref{eqn:casimirstates} and using the orthogonality relations for the $P_n$'s in \eqref{eqn:Portho} gives
\begin{align}\label{eqn:integrateddomcontraction}
&|\bra{n,s_L,s_R}\ket{V_L(0)\cO_O(-T)} |^2 \supset e^{-S_0} \int d\lambda\, \rho_0(\lambda)\, e^{-\varepsilon (s_L^2 + s_R^2)-\beta\frac{\lambda^2}2}  \gamma_{\Delta_V}(s_L,\lambda) \gamma_{\Delta_O}(s_R, \lambda).
\end{align}
To compute the contribution of this overlap to the probability distribution in $n$, we need to compute the expectation value of the projector in \eqref{eqn:nproj}, which involves integrating \eqref{eqn:integrateddomcontraction} over $s_L,s_R$. Doing so, we see that \eqref{eqn:integrateddomcontraction} just contributes an amount equal to the norm of the state $\ket{V_L(0) \cO_O(-T)}$ given in \eqref{eqn:statenorm} up to a factor of $e^{-2S_0}$. Dividing (the integral over $s_L,s_R$ of) \eqref{eqn:integrateddomcontraction} by the norm of the state \eqref{eqn:statenorm} to compare with the disk distribution, we arrive at a wormhole contribution to the probability distribution $p(n)$ which is $n$-independent but suppressed by an amount $e^{-2S_0}$. At the peak of the distribution in \eqref{eqn:casimirdistributionpert}, this wormhole correction will become dominant at times $T/\beta$ of order $S_0$. Note that such a correction, in fact, leads to a non-normalizable distribution for $n$ (or the Casimir). This is just reflective of the fact that the $\ket{n, s_L, s_R}$ states are over-complete with respect to the non-perturbative inner product.

 \subsubsection*{Wormhole corrections to the CM collision energy: the observer's point of view}

We will now discuss how this effect is eliminated when using our modified inner product in the presence of an observer. Using our rules for the quantum gravity relational Hilbert space, we should always contract the operator $\Oo$ with a partner in the same bra-ket pair. In other words, when using the non-perturbative inner product in the presence of an observer, the observer's worldline should never cross the wormhole. This property disturbs the contraction that led to the factor of $\delta(s -s')$ in \eqref{eqn:domcontraction}. For other matter particle contractions, the use of the orthogonality relation gets spoiled, leaving behind $P_n$'s, which are suppressed by factors of $1/\sqrt{n}$ due to \eqref{eqn:Pnlargen}. This means that the wormhole corrections in the non-perturbative observer Hilbert space are of the same order in $n$-counting as the perturbative answer \eqref{eqn:casimirdistributionpert}, but they are suppressed in $e^{S_0}$-counting.

We now explain these words in more detail by examining the form of other contractions and how they contribute to $|\bra{n,s_L,s_R}\ket{V(0)\OR(-T)}|^2$. In our modified inner product rules, the observer operator inside the definition of the state $\ket{n,s_L,s_R}$ in \eqref{eqn:casimirstates} will always connect with the observer operator in the corresponding bra/ket. Assuming that no particle worldlines cross (as we have done throughout this work), we are just left with two classes of diagrams, depending upon whether the worldline of $V_L$ traverses the wormhole or not. \footnote{We are focusing here again on diagrams for which the projector threads the wormhole. We will find that these diagrams are suppressed by $1/C$ and $e^{S_0}$; we remark that diagrams in which the projector does not thread the wormhole are also suppressed by $1/C$ and $e^{S_0}$ and therefore give corrections to the perturbative answer at the same order. However, unlike one of the diagrams we will compute, they are time-independent and will therefore not be enhanced at exponential times.} For the class of diagrams where the $V_L$ operator's worldline crosses the wormhole, there are two possible contractions which we need to evaluate to compute $|\bra{n,s_L,s_R}\ket{V(0)\OR(-T)}|^2$. The important contractions can be illustrated as 
\def\red{\color{red}}
\def\blue{\color{blue}}
\begin{align}\label{eqn:contractionclass1}
    \begin{tikzpicture}[scale=1, baseline={([yshift=-0.1cm]current bounding box.center)}]
    \coordinate (T1) at (-1.5, 1.5);
    \coordinate (T2) at (0, 1.5);
    \coordinate (T3) at (1, 1.5);
    \coordinate (T4) at (2.5, 1.5);
    \coordinate (T5) at (4, 1.5);
    \coordinate (B1) at (-1.5, -1.5);
    \coordinate (B2) at (0, -1.5);
    \coordinate (B3) at (1, -1.5);
    \coordinate (B4) at (2.5, -1.5);
    \coordinate (B5) at (4, -1.5);
    \coordinate (slashl) at (-1.5,0);
    \coordinate (slashr) at (4, 0);
    \fill[fill=green!20] (B1) -- (T1) -- (T5)--(B5)-- cycle;
    \fill[fill=blue!20] (B2) -- (T2) -- (T3)--(B3)--cycle;
    \draw[thick, red, fill=red!20] (T4) .. controls (3.0, 1.0) and (3.5, 1.0) .. (T5) -- cycle;
    \draw[thick, red, fill=red!20] (B4) .. controls (3.0, -1.0) and (3.5, -1.0) .. (B5) -- cycle;
    \draw[thick] (T1) -- (T2) -- (T3) -- (T4) -- (T5);
    \draw[thick] (B1) -- (B2) -- (B3) -- (B4) -- (B5);
    \draw[thick, dashed] (T1) -- (B1);
    \draw[thick, blue] (T2) -- (B2);
    \draw[thick, blue] (T3) -- (B3);
    \draw[thick, dashed] (T5) -- (B5);
    \node[above] at (.5, -1.58) {$s_L$};
    \node[below] at (.5, 1.58) {$s_L$};
    \node[above] at (3.25, -1.58) {$s_R$};
    \node[below] at (3.25, 1.58) {$s_R$};
    \node[above] at (1.75, -1.58) {\scriptsize $\beta/2 + iT$};
    \node[below] at (1.75, 1.58) {\scriptsize $\beta/2 - iT$};
    \node[above] at (-.75, -1.58) {$s$};
    \node[below] at (-.75, 1.58) {$s'$};
    \node at (slashl) {$=$};
    \node at (slashr) {$=$};
    \node[above] at (T2) {$V$};
    \node[above] at (T3) {$V$};
    \node[above] at (T4) {$\Oo$};
    \node [above right] at (T5) {$\Oo$};
    \node[below] at (B2) {$V$};
    \node[below] at (B3) {$V$};
    \node[below] at (B4) {$\Oo$};
    \node [below right] at (B5) {$\Oo$};
    \fill[red] (T1) circle (2pt);
    \fill[blue] (T2) circle (2pt);
    \fill[blue] (T3) circle (2pt);
    \fill[red] (T4) circle (2pt);
    \fill[red] (T5) circle (2pt);
    \fill[red] (B1) circle (2pt);
    \fill[blue] (B2) circle (2pt);
    \fill[blue] (B3) circle (2pt);
    \fill[red] (B4) circle (2pt);
    \fill[red] (B5) circle (2pt); 
\end{tikzpicture}
+ \quad 
    \begin{tikzpicture}[scale=1, baseline={([yshift=-0.1cm]current bounding box.center)}]
    \coordinate (T1) at (-1.5, 1.5);
    \coordinate (T2) at (0, 1.5);
    \coordinate (T3) at (1, 1.5);
    \coordinate (T4) at (2.5, 1.5);
    \coordinate (T5) at (4, 1.5);
    \coordinate (B1) at (-1.5, -1.5);
    \coordinate (B2) at (0, -1.5);
    \coordinate (B3) at (1, -1.5);
    \coordinate (B4) at (2.5, -1.5);
    \coordinate (B5) at (4, -1.5);
    \coordinate (VL) at (-1.5, -.5);
    \coordinate (VR) at (4, -.5);
    \fill[fill=green!20] (B1) -- (T1) -- (T5) -- (B5) -- cycle;
    \fill[fill=blue!20]  (B2) .. controls (-.6, -.9) and (-.9,-.75)  .. (VL) -- (T1) -- (T2) -- (B3) -- cycle;
    \fill[fill=blue!20]  (VR) .. controls (.9, 0) and (1.3, 1) .. (T3) --(T5) -- cycle;
    \draw[thick, red, fill=red!20] (T4) .. controls (3.0, 1.0) and (3.5, 1.0) .. (T5) -- cycle;
    \draw[thick, red, fill=red!20] (B4) .. controls (3.0, -1.0) and (3.5, -1.0) .. (B5) -- cycle;
    \draw[thick] (T1) -- (T2) -- (T3) -- (T4) -- (T5);
    \draw[thick] (B1) -- (B2) -- (B3) -- (B4) -- (B5);
    \draw[thick, dashed] (T1) -- (B1);
    \draw[thick, dashed] (T5) -- (B5);
    \draw[thick, blue] (B2) .. controls (-.6, -.9) and (-.9,-.75)  .. (VL);
    \draw[thick,blue] (B3) -- (T2);
    \draw[thick, blue] (VR) .. controls (.9, 0) and (1.3, 1) .. (T3);
    \node[above] at (.5, -1.58) {$s_L$};
    \node[below] at (.5, 1.58) {$s_L$};
    \node[above] at (3.25, -1.58) {$s_R$};
    \node[below] at (3.25, 1.58) {$s_R$};
    \node[above] at (1.75, -1.58) {\scriptsize $\beta/2 + iT$};
    \node[below] at (1.75, 1.58) {\scriptsize $\beta/2 - iT$};
    \node[above] at (-.75, -1.58) {$s$};
    \node[below] at (-.75, 1.58) {$s'$};
    \node at (slashl) {$=$};
    \node at (slashr) {$=$};
    \node[above] at (T2) {$V$};
    \node[above] at (T3) {$V$};
    \node[above] at (T4) {$\Oo$};
    \node [above right] at (T5) {$\Oo$};
    \node[below] at (B2) {$V$};
    \node[below] at (B3) {$V$};
    \node[below] at (B4) {$\Oo$};
    \node [below right] at (B5) {$\Oo$};
    \fill[red] (T1) circle (2pt);
    \fill[blue] (T2) circle (2pt);
    \fill[blue] (T3) circle (2pt);
    \fill[red] (T4) circle (2pt);
    \fill[red] (T5) circle (2pt);
    \fill[red] (B1) circle (2pt);
    \fill[blue] (B2) circle (2pt);
    \fill[blue] (B3) circle (2pt);
    \fill[red] (B4) circle (2pt);
    \fill[red] (B5) circle (2pt);
\end{tikzpicture}
.
\end{align}
As in \eqref{eqn:domcontraction}, to get to the overlap $|\bra{n,s_L,s_R}\ket{V(0)\OR(-T)}|^2$, we need to integrate these diagrams over $s,s'$ against the kernel in \eqref{eqn:casimirstates}.
Again following the Feynman diagram-like rules discussed in \cite{Jafferis:2022wez} for computing the contribution of such spacetimes, the left diagram gives
\begin{align}\label{eqn:subdomcontraction1}
    &|\bra{n,s_L,s_R}\ket{V(0)\OR(-T)}|^2 \supset e^{-S_0} \int ds \ \rho_0(s) e^{-\varepsilon (s_L^2 + s_R^2)-\beta \frac{s^2}2}  \gamma_{\Delta_V}(s_L,s)\gamma_{\Delta_O}(s_R,s) P_n^2(s,s_L,s_R),
\end{align}
The time dependence of this contribution has canceled out because the two asymptotic regions whose length depends on $T$ share the same energy variable $s$. Importantly, the presence of the Boltzmann factor has spoiled the enhancement in the integral over $s$ that the diagram computed in \eqref{eqn:integrateddomcontraction} exhibited. One way to see that the enhancement is gone is that in the calculation that leads to equation \eqref{eqn:integrateddomcontraction}, the piece of the $s$-integral that scales as $n$ comes from the large $s$ region, where the integrand oscillates sinusoidally. In \eqref{eqn:subdomcontraction1}, the large-$s$ region of the integral is damped by the factor of $e^{-\beta s^2}$, which was not present in the $s$-integral in \eqref{eqn:domcontraction}. 
The second diagram gives 
\begin{align}
    &|\bra{n,s_L,s_R}\ket{V(0)\OR(-T)}|^2 \supset e^{-S_0} e^{-\varepsilon (s_L^2 + s_R^2) -\beta \frac{s_L^2}2} \gamma_{\Delta_V}(s_L,s_L)\gamma_{\Delta_O}(s_R,s_L) P^2_n(s_L;s_L,s_R),
\end{align}
where the $s$ and $s'$ integrals have been localized to $s,s' = s_L$. Again, this diagram is time-independent and decays like $1/n$ at large $n$ from the factors of $P_n(s_L;s_L,s_R)$.

In the second class of diagrams, the matter operator $V$ does not traverse the wormhole. Now, there are three diagrams to consider, which look like
\def\red{\color{red}}
\def\blue{\color{blue}}
\begin{align}\label{eqn:contractionclass2}
    \begin{tikzpicture}[scale=.75, baseline={([yshift=-0.1cm]current bounding box.center)}]
    \coordinate (T1) at (-1.5, 1.5);
    \coordinate (T2) at (0, 1.5);
    \coordinate (T3) at (1, 1.5);
    \coordinate (T4) at (2.5, 1.5);
    \coordinate (T5) at (4, 1.5);
    \coordinate (B1) at (-1.5, -1.5);
    \coordinate (B2) at (0, -1.5);
    \coordinate (B3) at (1, -1.5);
    \coordinate (B4) at (2.5, -1.5);
    \coordinate (B5) at (4, -1.5);
    \coordinate (slashl) at (-1.5,0);
    \coordinate (slashr) at (4, 0);
    \fill[fill=green!20] (B1) -- (T1) -- (T5)--(B5)-- cycle;
    \draw[thick] (T1) -- (T2) -- (T3) -- (T4) -- (T5);
    \draw[thick] (B1) -- (B2) -- (B3) -- (B4) -- (B5);
    \draw[thick, dashed] (T1) -- (B1);
    \draw[thick, dashed] (T5) -- (B5);
    \draw[thick, red, fill=red!20] (T4) .. controls (3,.7) and (3.5,.7) .. (T5) -- cycle;
    \draw[thick, red, fill=red!20] (B4) .. controls (3,-.7) and (3.5,-.7) .. (B5) -- cycle;
    \draw[thick, blue, fill=blue!20] (T2) .. controls (.25, 1.0) and (.75, 1.0) .. (T3) -- cycle;
    \draw[thick, blue, fill=blue!20] (B2) .. controls (.25, -1.0) and (.75, -1.0) .. (B3) -- cycle;
    \node[above] at (.5, -1.58) {$s_L$};
    \node[below] at (.5, 1.58) {$s_L$};
    \node[above] at (3.25, -1.5) {$s_R$};
    \node[below] at (3.25, 1.5) {$s_R$};
    \node[above] at (1.75, -1.58) {\scriptsize $\beta/2 + iT$};
    \node[below] at (1.75, 1.58) {\scriptsize $\beta/2 - iT$};
    \node[above] at (-.75, -1.58) {$s$};
    \node[below] at (-.75, 1.58) {$s'$};
    \node at (slashl) {$=$};
    \node at (slashr) {$=$};
    \node[above] at (T2) {$V$};
    \node[above] at (T3) {$V$};
    \node[above] at (T4) {$\Oo$};
    \node [above right] at (T5) {$\Oo$};
    \node[below] at (B2) {$V$};
    \node[below] at (B3) {$V$};
    \node[below] at (B4) {$\Oo$};
    \node [below right] at (B5) {$\Oo$};
    \fill[red] (T1) circle (2pt);
    \fill[blue] (T2) circle (2pt);
    \fill[blue] (T3) circle (2pt);
    \fill[red] (T4) circle (2pt);
    \fill[red] (T5) circle (2pt);
    \fill[red] (B1) circle (2pt);
    \fill[blue] (B2) circle (2pt);
    \fill[blue] (B3) circle (2pt);
    \fill[red] (B4) circle (2pt);
    \fill[red] (B5) circle (2pt); 
\end{tikzpicture}+
    \begin{tikzpicture}[scale=.75, baseline={([yshift=-0.1cm]current bounding box.center)}]
    \coordinate (T1) at (-1.5, 1.5);
    \coordinate (T2) at (0, 1.5);
    \coordinate (T3) at (1, 1.5);
    \coordinate (T4) at (2.5, 1.5);
    \coordinate (T5) at (4, 1.5);
    \coordinate (B1) at (-1.5, -1.5);
    \coordinate (B2) at (0, -1.5);
    \coordinate (B3) at (1, -1.5);
    \coordinate (B4) at (2.5, -1.5);
    \coordinate (B5) at (4, -1.5);
    \fill[fill=red!20] (B1) -- (T1) -- (T5) -- (B5) -- cycle;
    \fill[green!20] (B4) .. controls (2.2, 0) and (-1.2, 0) .. (B1) --cycle;
     \fill[yellow!20] (T4) .. controls (2.2, 0) and (-1.2, 0) .. (T1) --cycle;
    \draw[thick] (T1) -- (T2) -- (T3) -- (T4) -- (T5);
    \draw[thick] (B1) -- (B2) -- (B3) -- (B4) -- (B5);
    \draw[thick, dashed] (T1) -- (B1);
    \draw[thick, dashed] (T5) -- (B5);
    \draw[thick, red] (T4) .. controls (2.2, 0) and (-1.2, 0) .. (T1);
    \draw[thick, red] (B4) .. controls (2.2, 0) and (-1.2, 0) ..(B1); 
    \draw[thick, blue, fill=blue!20] (T2) .. controls (.25, 1.0) and (.75, 1.0) .. (T3) -- cycle;
    \draw[thick, blue, fill=blue!20] (B2) .. controls (.25, -1.0) and (.75, -1.0) .. (B3) -- cycle;
    \node[above] at (.5, -1.58) {$s_L$};
    \node[below] at (.5, 1.58) {$s_L$};
    \node[above] at (3.25, -1.5) {$s_R$};
    \node[below] at (3.25, 1.5) {$s_R$};
    \node[above] at (1.75, -1.58) {\scriptsize $\beta/2 + iT$};
    \node[below] at (1.75, 1.58) {\scriptsize $\beta/2 - iT$};
    \node[above] at (-.75, -1.58) {$s$};
    \node[below] at (-.75, 1.58) {$s'$};
    \node at (slashl) {$=$};
    \node at (slashr) {$=$};
    \node[above] at (T2) {$V$};
    \node[above] at (T3) {$V$};
    \node[above] at (T4) {$\Oo$};
    \node [above right] at (T5) {$\Oo$};
    \node[below] at (B2) {$V$};
    \node[below] at (B3) {$V$};
    \node[below] at (B4) {$\Oo$};
    \node [below right] at (B5) {$\Oo$};
    \fill[red] (T1) circle (2pt);
    \fill[blue] (T2) circle (2pt);
    \fill[blue] (T3) circle (2pt);
    \fill[red] (T4) circle (2pt);
    \fill[red] (T5) circle (2pt);
    \fill[red] (B1) circle (2pt);
    \fill[blue] (B2) circle (2pt);
    \fill[blue] (B3) circle (2pt);
    \fill[red] (B4) circle (2pt);
    \fill[red] (B5) circle (2pt);
\end{tikzpicture}
+ \quad  \begin{tikzpicture}[scale=.75, baseline={([yshift=-0.1cm]current bounding box.center)}]
    \coordinate (T1) at (-1.5, 1.5);
    \coordinate (T2) at (0, 1.5);
    \coordinate (T3) at (1, 1.5);
    \coordinate (T4) at (2.5, 1.5);
    \coordinate (T5) at (4, 1.5);
    \coordinate (B1) at (-1.5, -1.5);
    \coordinate (B2) at (0, -1.5);
    \coordinate (B3) at (1, -1.5);
    \coordinate (B4) at (2.5, -1.5);
    \coordinate (B5) at (4, -1.5);
    \coordinate (VL) at (-1.5, -.5);
    \coordinate (VR) at (4, -.5);
    \coordinate (UL) at (-1.5,.5);
    \coordinate (UR) at (4, .5);
    \fill[fill=blue!20] (B1) -- (T1) -- (T5) -- (B5) -- cycle;
    \fill[fill=yellow!20]  (T1) -- (T2) .. controls (-.4, .9) and (-.9, .6)  .. (UL)--cycle;
    \fill[fill=green!20]  (B1) -- (B2) .. controls (-.4, -.9) and (-.9, -.6)  .. (VL)--cycle;
    \fill[fill=yellow!20]  (T5) -- (UR) .. controls (2.5, .5) and (1.2,.9)  .. (T3)--cycle;
    \fill[fill=green!20]  (B5) -- (VR) .. controls (2.5, -.5) and (1.2,-.9)  .. (B3)--cycle;
    \draw[thick] (T1) -- (T2) -- (T3) -- (T4) -- (T5);
    \draw[thick] (B1) -- (B2) -- (B3) -- (B4) -- (B5);
    \draw[thick, dashed] (T1) -- (B1);
    \draw[thick, dashed] (T5) -- (B5);
    \draw[thick, red, fill=red!20] (T4) .. controls (3,.7) and (3.5,.7) .. (T5) -- cycle;
    \draw[thick, red, fill=red!20] (B4) .. controls (3,-.7) and (3.5,-.7) .. (B5) -- cycle;
    \draw[thick, blue] (B2) .. controls (-.4, -.9) and (-.9,-.6)  .. (VL);
    \draw[thick, blue] (VR) .. controls (2.5, -.5) and (1.2,-.9)  .. (B3);
    \draw[thick, blue] (T2) .. controls (-.4, .9) and (-.9, .6)  .. (UL);
    \draw[thick, blue] (UR) .. controls (2.5, .5) and (1.2,.9)  .. (T3);
    \node[above] at (.5, -1.58) {$s_L$};
    \node[below] at (.5, 1.58) {$s_L$};
    \node[above] at (3.25, -1.5) {$s_R$};
    \node[below] at (3.25, 1.5) {$s_R$};
    \node[above] at (1.8, -1.58) {\scriptsize $\beta/2 + iT$};
    \node[below] at (1.8, 1.58) {\scriptsize $\beta/2 - iT$};
    \node[above] at (-.75, -1.58) {$s$};
    \node[below] at (-.75, 1.58) {$s'$};
    \node at (slashl) {$=$};
    \node at (slashr) {$=$};
    \node[above] at (T2) {$V$};
    \node[above] at (T3) {$V$};
    \node[above] at (T4) {$\Oo$};
    \node [above right] at (T5) {$\Oo$};
    \node[below] at (B2) {$V$};
    \node[below] at (B3) {$V$};
    \node[below] at (B4) {$\Oo$};
    \node [below right] at (B5) {$\Oo$};
    \fill[red] (T1) circle (2pt);
    \fill[blue] (T2) circle (2pt);
    \fill[blue] (T3) circle (2pt);
    \fill[red] (T4) circle (2pt);
    \fill[red] (T5) circle (2pt);
    \fill[red] (B1) circle (2pt);
    \fill[blue] (B2) circle (2pt);
    \fill[blue] (B3) circle (2pt);
    \fill[red] (B4) circle (2pt);
    \fill[red] (B5) circle (2pt);
\end{tikzpicture} 
\end{align}
These diagrams are a bit more subtle due to the issue that we need to mod-out by the action of the mapping-class-group (MCG) on the geometry. When matter geodesics traversed the wormhole, geometries related by MCG transformations were really different since the traversing geodesic could be used to distinguish the two geometries by measuring its winding number. Without such a geodesic, MCG transformations are pure gauge. The effect of doing this quotient by the MCG is just to insert a factor of the density-density correlator, $\overline{\rho(s) \rho(s')}_{\text{conn.}}$, into the calculations. This connected piece of the density-density correlator can be computed in JT gravity to be \cite{Saad:2019lba}\footnote{Notice that the connected correlator is between the full densities of states $\rho(s)=e^{S_0}\rho_0(s)$, and the connected part of the correlator is $O(1)$ \cite{Saad:2019lba,Iliesiu:2021ari}.}
\begin{align}\label{eqn:rrcorr}
    \overline{\rho(s) \rho(s')}_{\text{conn.}} = \frac{-1}{\pi^2}\frac{s^2 + s'^2}{(s^2-s'^2)^2} \quad .
\end{align}

Therefore, after inserting \eqref{eqn:contractionclass2} into the square of \eqref{eqn:overlapbwcasandstate}, the first term gives
\begin{align}\label{eqn:class2term1}
    e^{-S_0} \int ds ds' & e^{-\varepsilon (s_L^2 + s_R^2) - (\beta/2+iT)\frac{s^2}{2} - (\beta/2-iT)\frac{s'^2}{2}} \nonumber \\
    & \times \overline{\rho(s) \rho(s')}_{\text{conn.}}  P_n(s;s_L, s_R) P_n(s';s_L, s_R) \sqrt{\gamma_{\Delta_V}(s_L,s) \gamma_{\Delta_V}(s_L,s') \gamma_{\Delta_O}(s_R,s) \gamma_{\Delta_O}(s_R,s')}.
\end{align}
As before, it is convenient to go to $s_{\pm} = s \pm s'$ variables. The connected correlator of $\rho$'s in \eqref{eqn:rrcorr} has a double pole at $s_- = 0$ and so, at large $T$, this integral is dominated by the $s_-=0$ region. Again, however, the presence of the Boltzmann factors spoils the enhancement coming from the large $s_+$ region, and so at large $n$, this diagram goes like $1/n$. Furthermore, after insertion into the square of \eqref{eqn:overlapbwcasandstate}, the second diagram in \eqref{eqn:contractionclass2} gives
\begin{align}
    e^{-S_0}&\int ds ds' e^{-\varepsilon (s_L^2 + s_R^2) - (\beta/2+iT)\frac{s^2}{2} - (\beta/2-iT)\frac{s'^2}2} \nonumber \\
    & \times \overline{\rho(s_R) \rho(s_R)}_{\text{conn.}}  P_n(s;s_L, s_R) P_n(s';s_L, s_R) \sqrt{\gamma_{\Delta_V}(s_L,s) \gamma_{\Delta_V}(s_L,s') \gamma_{\Delta_O}(s_R,s) \gamma_{\Delta_O}(s_R,s')}.
\end{align}
Naively this looks divergent since we are evaluating $\left( \rho(s) \rho(s') \right)_{\text{conn.}}$ directly on its double pole, but we should remember that we defined our $\ket{n,s_L,s_R}$ states to involve a smearing over $s_L$ and $s_R$ as in \eqref{eqn:smearedcasimirstates} so that $\ket{n,\bar{s}_L, \bar{s}_R}$ is normalizable. The effect of this smearing over $s_L$ is just to replace $\overline{\rho(s) \rho(s')}_{\text{conn.}}$ with a factor of $1/\delta s$. As long as $\delta s$ is not exponentially small in $S_0$, again, this term will be suppressed by both $e^{-2S_0}$ and $1/n$ relative to the leading perturbative answer in \eqref{eqn:casimirdistributionpert} once we account for the normalization of the state $\ket{V\OR}$. Finally, this same analysis for the second diagram in \eqref{eqn:contractionclass2} applies directly to the third diagram in \eqref{eqn:contractionclass2} but with $\overline{\rho(s) \rho(s')}_{\text{conn.}}$ evaluated at $s = s' = s_R$.

In summary, when using the inner product in the presence of an observer, all genus-one corrections---such as \eqref{eqn:caswh}---to $|\bra{n,s_L,s_R}\ket{V(0)\OR(-T)}|^2$ are suppressed at large $n$ by $1/n$ like the disk answer, and they are further suppressed by  $e^{-S_0}$ relative to the disk answer. One could also wonder about higher-genus contributions to these observables. For example, in the various diagrams of \eqref{eqn:contractionclass1} and \eqref{eqn:contractionclass2}, we could consider geometries with handles connecting the various regions of different or the same color. The essential point is that the large $s$ behavior in all of these diagrams will be damped by the same Boltzmann factors, so the large $n$ behavior will always be $1/n$. Therefore, the higher-genus contributions are even more suppressed than the geometries studied above.

The upshot of our analysis is then that the probability distribution in the presence of the observer does not get corrected at times of order $S_0$. Note, however, that there are contributions to the probability distribution which are exponentially suppressed in $S_0$ and suppressed in $1/n$ but have non-trivial time dependence. For example, the first contraction in \eqref{eqn:contractionclass2} has non-trivial time-dependence as exhibited in \eqref{eqn:class2term1}. At large enough time, this term gets a contribution that grows linearly in $T$. This is just the same linear growth in time exhibited by the spectral form factor. In other words, we expect that the Casimir distribution, defined using the inner product in the presence of an observer, agrees with the perturbative answer in \eqref{eqn:casimirdistributionpert} until times of order $e^{S_0}$. Therefore, at exponential times, non-perturbative corrections to the collision energy experienced by an observer become important, and perturbative effective field theory breaks down. We remark that this result is substantially different from the puzzling results of \cite{IliLev24}, in which effective field theory breaks down at linear times, namely around the Page time.


\section{Discussion}
\label{sec:discussion}

In this paper, we have proposed a modification of the rules of the non-perturbative gravitational path integral to take into account the presence of a gravitating observer. In our proposal, we require the observer to always be present on any spatial slice when computing overlaps between quantum gravity states and their moments. Specifically, a worldline for the observer must always connect a bra and the corresponding ket. These rules should be applied when asking questions about the experience of the gravitating observer.

Using our new proposal in the context of two-dimensional JT gravity with a negative cosmological constant, we have computed the dimension of the Hilbert space $\Hr$ relevant for the description of an observer's experience in a closed universe and in a two-sided black hole, and found it is much larger than the corresponding global Hilbert spaces $\Hnonp$ computed with the old rules for the gravitational path integral. We then studied the properties of this Hilbert space, including its factorisation, the positivity of the inner product defined by the path integral with our new rules, and the presence of null states. 

We also studied various relational observables probing the experience of the observer. In the closed universe setting, we computed correlation functions between points on the observer's worldline. In the two-sided black hole setup, we studied the length of the Einstein-Rosen bridge and the center-of-mass collision energy between an observer and a shockwave behind the horizon. The behavior of these observables (except the length of the Einstein-Rosen bridge) is substantially modified in our new framework with respect to the global picture for the gravitational path integral in the absence of an observer: non-perturbative corrections 
are typically smaller from the point of view of the observer than in the global perspective. Nevertheless, such non-perturbative effects can become important at leading order for the observer, giving in-principle testable corrections to their experience, when some scale in the setup scales exponentially with $S_0$. Such corrections can be critical even when the proper time that the observer travels in the spacetime is $O(1)$. For example, in the case of an observer falling into a two-sided black hole, the observer travels an $O(1)$ proper time until they cross the horizon. Once they cross the horizon of a an old black hole, they could see (with tragic consequences) that both the length of the Einstein-Rosen bridge and the center-of-mass collision energy with a shockwave sent in from the other side receive large non-perturbative corrections when the age of the black hole scales as $e^{S_0}$. It would be interesting to see if we also encounter large non-perturbative corrections to the two-point function measured at $O(1)$ proper times along the observer's worldline in the closed universe. In this case, contributions from arbitrarily long winding geodesics become important close to the singularity leading to an accumulation of divergences; whether this accumulation of divergences can be resolved by non-perturbative corrections remains to be seen. 

The results of this paper open up several new avenues of research. A first important task is to generalize our results beyond the two-dimensional toy-model discussed here to higher-dimensional theories of gravity. We remark that, although they also immediately hold for pure JT gravity, all our results were obtained in JT gravity with matter. In this context, most of the wormhole geometries considered in this paper, e.g., in all our resolvent calculations, are saddles of the gravitational path integral \cite{Stanford:2020wkf,Hsin:2020mfa}. Similar calculations that do not rely on off-shell contributions to the gravitational path integral can also be carried out in higher dimensions, see, e.g., \cite{Balasubramanian:2022gmo,Balasubramanian:2022lnw,Antonini:2023hdh,deBoer:2023vsm,Climent:2024trz}. We, therefore, expect the extension of our results to these setups to be rather straightforward. 

In the rest of this section, we will explore additional future directions. Specifically, in Section \ref{sec:relational}, we discuss the relationship between our proposal and other approaches to relational dynamics in quantum gravity and draw an analogy between the quantization of gravity in the presence of an observer and Chern-Simons theory in the presence of sources. In Section \ref{sec:holography}, we comment on the consequences of our results for the holographic description of the experience of a gravitating observer. Finally, in Section \ref{sec:desitter}, we explain how to extend our results to describe the experience of an observer in de Sitter space.

\subsection{Relational dynamics and the gravitational path integral}
\label{sec:relational}

An interesting future direction is to study the relationship between our proposal and other formalisms for describing relational dynamics in gravitating systems. As it was recently discussed in detail \cite{Held:2024rmg}, physical states in quantum gravity can be represented in two different ways. The first one is in terms of co-invariants, i.e., equivalence classes of states under the action of the gravitational constraint. In this approach, relational dynamics emerges naturally as gauge evolution within a given equivalence class \cite{Held:2024rmg}. This is the approach used in \cite{Chandrasekaran:2022cip}. The second one is to define physical states to be annihilated by the gravitational constraint operators \cite{wheeler,dewitt,Isham:1992ms}. These states are sometimes called ``invariants''. In this case, relational dynamics can be studied by selecting a suitable subsystem to play the role of a clock (i.e., an observer) and describing physics with respect to this clock using the Page-Wootters formalism. \cite{pagewooters,Hoehn:2019fsy,DeVuyst:2024pop}\footnote{We thank Elliott Gesteau and Ronak Soni for emphasizing this point to us.}. At least at the perturbative level, the two approaches describe the same physics

We expect that our gravitational path integral approach to be compatible with these relational approaches at the perturbative level, and, in particular, it implements the co-invariant approach, which arises naturally when defining states in terms of boundary conditions on geodesic boundaries \cite{Held:2024rmg}. On the other hand, our prescription is also well-defined non-perturbatively and should therefore represent a non-perturbative extension of these ideas. We leave a more precise characterization of the role of our proposal in the description of relational dynamics in quantum gravity to future work. 

Before moving on to discuss implications for holography, we would like to point out an analogy between our proposal for the gravitational path integral from the point of view of an observer and Chern-Simons theory in the presence of a source.\footnote{A similar comparison between Chern-Simons and the quantum gravity of closed universes was made in \cite{Sha21,Sha23}.}
Exact quantization in the presence of a detector is an old and well-studied problem in quantum field theory. In the course of this work, we found it useful to draw analogies to models where this quantization can be accomplished exactly. One such case is Chern-Simons (CS) theory in the presence of sources. CS is a topological 3D gauge theory of a Lie algebra-valued connection \cite{Witten:1988hf,Frohlich:1989gr}.\footnote{We provide a more complete review of CS in Appendix \ref{app:chernsimons} and summarize some of the motivating similarities here.}

Pure JT gravity and pure CS are both theories with constraints. In JT, the most commonly remarked upon constraint is $H_0$, defined in equation \eqref{eq:constraint}, which involves the left and right boundary Hamiltonians in a two-sided black hole in the absence of matter. In CS, the constraint dictates that the field strength vanishes, or, equivalently, that the connection is flat. We have two options when attempting to implement these constraints \cite{Witten:1988hf}. We could ``quantize then constrain" (also known as Dirac quantization \cite{Isham:1992ms}) by quantizing the phase space of field configurations before implementing the constraint as an operator whose kernel determines physical states. Alternatively, we could ``constrain then quantize" by restricting to the subspace of configurations satisfying the constraint before quantization. When we ``constrain then quantize", the constraint is a manifestly vanishing operator on physical states, and it is natural to refer to this as the Hilbert space of invariants, as described above. The states in the kernel of the constraint in the ``quantize then constrain" approach are, by definition, invariants. However, this approach is most rigorously defined with respect to co-invariants to account for cases where the kernel of the constraint contains non-normalizable states.\footnote{
    It is not generally known whether these procedures result in identical Hilbert spaces \cite{Isham:1992ms,guillemin1982geometric}, but these toy models are simple enough that we will neglect this concern. Specifically, in the case of JT, one can quantize, constrain, and define invariants at the cost of division by the infinite volume of the gauge group \cite{Penington:2023dql}.
}

The addition of matter modifies the constraint. That modification is determined by how the underlying symmetry acts on the coupled matter.  This is a gauge symmetry for a generic quantum field theory, and the background isometry in the specific case of gravity. We saw part of how this modification manifests gravitationally in equation \eqref{eq:matter_constraint}. In CS, the modification is the addition of a source current density. Regardless of our quantization procedure, the Hilbert space in the presence of matter will be different because the definition of a state has changed. However, simply adding matter is not enough to make observations. Observables in gravity need to be dressed to a clock. Similarly, observables in CS are Wilson lines, holonomies of the gauge connection traced in a particular representation. Just as gravity need not provide us a clock, \textit{a priori} CS does not dictate a representation \cite{Murayama:1989we}. A ``detector" is a special type of matter field that is localized to a worldline $\Gamma$ and tells us how to calculate observables. In gravity, that detector is the observer, and in CS, it is the source charge.

The similarities between these formalisms are abundant. CS assigns the source a Hilbert space $\HO$ that lives on their worldline \cite{Alekseev:1994nzg,Elitzur:1989nr}. The field configuration space is manifestly the direct product $\HO\times\Hr$ where $\Hr$ is the space of allowed gauge connections. Observables are calculated via a path integral with $\Gamma$ acting as an additional boundary. Our new rules are justified by this order of operations. Choosing a location for the detector before path integration dictates that this degree of freedom is not subject to the usual dynamics of the theory. A source so-defined is not influenced by external fields, just as the observer worldline must connect bra and ket independent of topology change. As we have seen, these rules mean that the Hilbert space can change in the presence of a detector. A once trivial theory unveils hidden dynamics in the presence of an observer. We give additional details of this analogy in Appendix \ref{app:chernsimons}.

\subsection{Lessons for holography}
\label{sec:holography}

One particularly interesting future direction is understanding the holographic dual theory that arises from our proposal and describes the experience of the observer. In AdS/CFT---our best-understood example of holography---the dual theory lives on the asymptotic AdS boundary. For instance, in the two-sided black hole case, it is given by two entangled CFTs living on the two asymptotic boundaries \cite{Maldacena:2001kr}. Despite its many successes, this class of setups presents some limitations. Most notably, it is unclear how to holographically describe the experience of a bulk observer in generic settings \cite{Antonini:2024mci}. This is particularly true when the dynamics experienced by a bulk observer are not directly related to the dynamics of the boundary CFT. 
Paradigmatic examples of this fact are bulk observers in the interior of a black hole at late times or moving within a single Wheeler-DeWitt patch of a pure AdS spacetime. In both cases, the experience of the observer is not uniquely determined from the holographic dual theory.  
In the black hole case, this is a consequence of the non-isometricity of the bulk-to-boundary map \cite{Akers:2022qdl,Antonini:2024yif}. The most extreme example of these ambiguities is provided by an isolated closed universe, in which there is no boundary where to define a dual theory at all.

\begin{figure}[t!]
    \centering
    \includegraphics[width=0.85\linewidth]{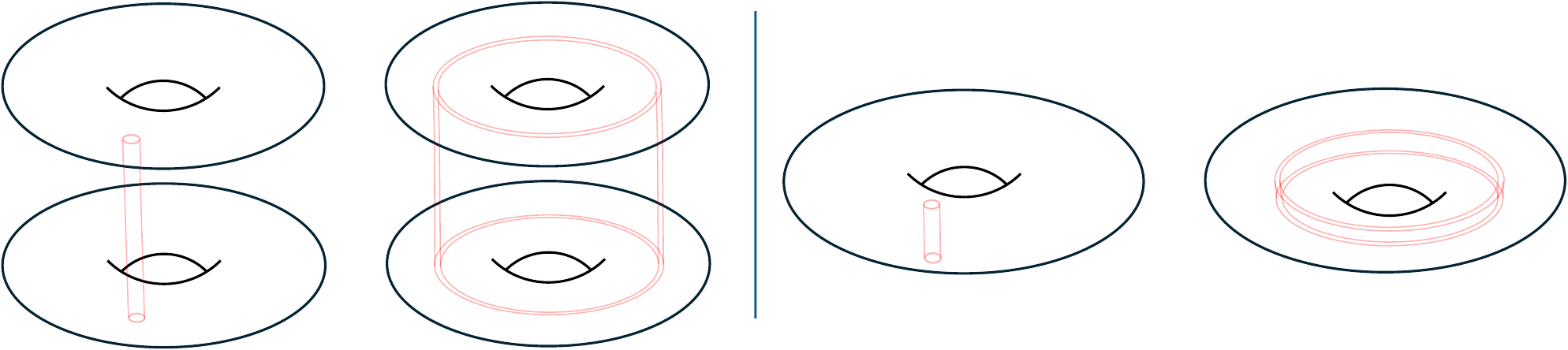}
    
    \caption{Examples of two-dimensional boundary configurations relevant for studying dressed observables in a closed universe (left figures) or in a two-sided black hole (right figures). The undeformed CFT is placed on all tori (asymptotic boundaries); a future task is to find the precise deformation of the CFT to be placed on the red cylinders (boundaries of the observer worldvolumes). The first figure in each set assumes that the worldvolume of each observer is bounded by the thin red cylinder, while the second figure in each set assumes that the observer is circularly symmetric and their worldvolume is bounded by the two red cylinders. The latter case is related through dimensional reduction to the analysis performed in this paper.}
    \label{fig:lessons-for-holography}
\end{figure}

Our proposed modification of the gravitational path integral in the presence of an observer suggests a path to address this issue. As we have discussed throughout this paper, the worldline of the observer is treated as an additional timelike boundary in the gravitational path integral and thus provides a natural location for the holographic dual theory. In less idealized settings, the worldline should be replaced by a codimension-0 worldvolume for the observer, which defines a codimension-1 timelike boundary in the spacetime where the dual theory would be defined, see Figure \ref{fig:lessons-for-holography}. The results of Section \ref{sec:hilbert} about the dimension and factorisation of the quantum gravity relational Hilbert space $\Hr$ (see also Appendix \ref{app:factorisation}) are compatible with this intuition. 

In the closed universe case, $\dim\left(\Hnonp\right)=1$ in the absence of an observer, which can be interpreted as the absence of a dual theory (whose Hilbert space must be isomorphic to the non-perturbative quantum gravity Hilbert space). On the other hand, $\dim\left(\Hr\right)=d^2$ when we fix the worldline (or, more generally, the worldvolume) of the observer. The fact that $\Hr$ in the presence of matter factorises into two Hilbert spaces of dimension $d$ to the left and right of the observer (see Appendix \ref{app:closedfactorisation} for details) signals that the dual description could consist of a pair of non-interacting holographic theories living on the left and right boundaries of the observer's worldvolume (which in the two-dimensional case under examination would simply be a thin strip). Similarly, in the two-sided black hole case in the presence of matter, $\dim\left(\Hnonp\right)=d^2$ without an observer and $\dim\left(\Hr\right)=d^4$ when fixing the worldline of the observer. $\Hr$ now factorizes into a tensor product of four Hilbert spaces (see Appendix \ref{app:BHfactorisation}). This is compatible with having two additional dual theories on the boundaries of the observer's worldvolume besides those living on the two asymptotic AdS boundaries. A generic two-sided black hole state, in this case, would be dual to an entangled state of four CFTs with a specific four-partite entanglement structure determined by the operator insertions associated with the observer and matter.

What could this holographic theory look like? One important observation is that the worldline of the observer, or, more accurately, the boundary of their worldvolume, should not be regarded as an asymptotic boundary; rather, it is a timelike boundary embedded into the bulk spacetime. Therefore, it seems plausible that the holographic theory living on such a boundary is not an ordinary CFT but some deformed theory, perhaps similar to $T\bar{T}$-deformed CFTs \cite{Zamolodchikov:2004ce,Smirnov:2016lqw,Cavaglia:2016oda,Gorbenko:2018oov,Araujo-Regado:2022gvw} relevant for the holographic description of spacetimes in the presence of timelike boundaries. It would be interesting to further explore this possibility.

As we have discussed in this paper, the Hilbert space $\Hr$ computed by the modified gravitational path integral is the Hilbert space relevant for the description of relational dynamics with respect to the observer. In other words, it is the Hilbert space needed to capture the observer's experience in the bulk. The Hilbert space of the holographic theory living on the boundary of the observer's worldvolume (plus those on any asymptotic boundary) would clearly be isomorphic to $\Hr$. This implies that, unlike ordinary AdS/CFT, the experience of the bulk observer would certainly be encoded in such a dual theory. This would, therefore, be a consistent step forward in the description of local bulk physics in holography. Notice that this holographic approach could also help us extend our results beyond the toy model of this paper to higher dimensional settings in which the gravitational path integral is not as well understood.

\subsection{Observers in de-Sitter space}
\label{sec:desitter}

Our results for the gravitational path integral in the presence of an observer can be extended to de Sitter JT gravity. Unlike the AdS-JT case discussed in this paper, there are no asymptotic boundaries in Euclidean signature in dS-JT gravity. Nonetheless, a generic state\footnote{In this discussion, we will consider the pure JT gravity case for simplicity, but the inclusion of matter does not alter our conclusions.} $|a,K,\Delta_O^{(i)}\rangle$ can be specified in terms of the length $a$ and extrinsic curvature $K$ of an initial slice of Euclidean dS${}_2$ where we insert the operator $\mathcal{O}_O^{(i)}$ of scaling dimension $\Delta_O^{(i)}$ associated with the observer. The resulting state---which is prepared on the reflection-symmetric slice of the Euclidean geometry computing the norm of the state $\langle a,K,\Delta_O^{(i)}|a,K,\Delta_O^{(i)}\rangle$---can then be evolved in Lorentzian time, yielding two-dimensional Lorentzian de Sitter, see Figure \ref{fig:desitterstate}. We remark that this set of states $|a,K,\Delta_O^{(i)}\rangle$ includes (but is not limited to) the Hartle-Hawking no-boundary state \cite{Hartle:1983ai,Ivo:2024ill}.\footnote{The Hartle-Hawking no-boundary state is obtained through the fact that on slices of constant $K$, $K$ determines $a$ on $S^2$ with 
\be 
a=\frac{2\pi}{\sqrt{1+K^2}}\,,
\ee 
for spheres with $R=2$. States that do not satisfy this relation between $K$ and $a$ are not defined on a smooth $S^2$ and are therefore different from the no-boundary state. An analogous property holds in AdS$_2$ where on surfaces that have the topology of a disk, $K$ also determines the proper length of the curve \cite{Goel:2020yxl}.  } In fact, for generic boundary conditions on a given time slice, the Euclidean time evolution to the past of that slice does not yield a smooth no-boundary geometry but rather a punctured sphere.

\begin{figure}[t!]
    \centering
    \includegraphics[width=0.25\linewidth]{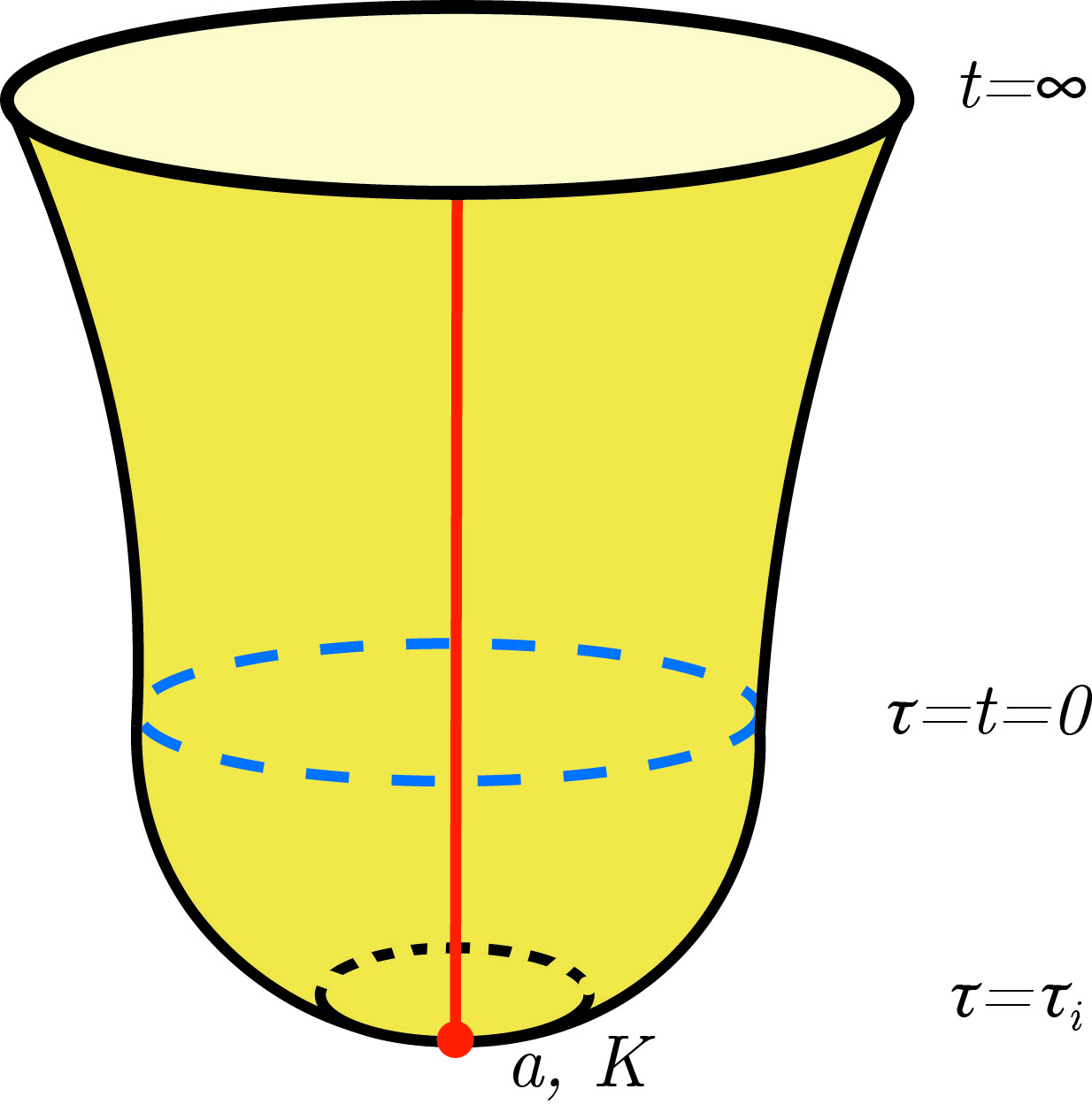}
    
    \caption{The Euclidean preparation (bottom half) and subsequent Lorentzian evolution (top half) of a de Sitter state $|a,K,\Delta_O^{(i)}\rangle$ specified by the length $a$ and extrinsic curvature $K$ of an initial slice at Euclidean time $\tau=\tau_i$ where the observer is inserted. The Lorentzian part of the evolution from the Euclidean sphere's equator at $\tau=0$ (depicted in blue) yields Lorentzian dS${}_2$, with a spacelike asymptotic boundary at Lorentzian time $t\to\infty$. Our set of states $|a,K,\Delta_O^{(i)}\rangle$ includes (but is not limited to) the smooth Hartle-Hawking no-boundary state.}
    \label{fig:desitterstate}
\end{figure}

The inner product between two states can be computed using our rules for the gravitational path integral in the presence of an observer (see Figure \ref{fig:desitter} (a)). Similarly, we can compute higher moments of an overlap. With our rules, the leading order correction to the product of two overlaps is given by the connected geometry depicted on the right of Figure \ref{fig:desitter} (b), which is a sphere with four holes. This geometry computes the variance of the overlap, which is then 
\be 
\overline{|\langle a',K',\Delta_O^{(j)}|a,K,\Delta_O^{(i)}\rangle|^2} - \left|\overline{\langle a',K',\Delta_O^{(j)}|a,K,\Delta_O^{(i)}\rangle}\right|^2=\delta_{ij}O(e^{-2S_0}).
\ee

\begin{figure}[t!]
    \centering
    \includegraphics[width=\linewidth]{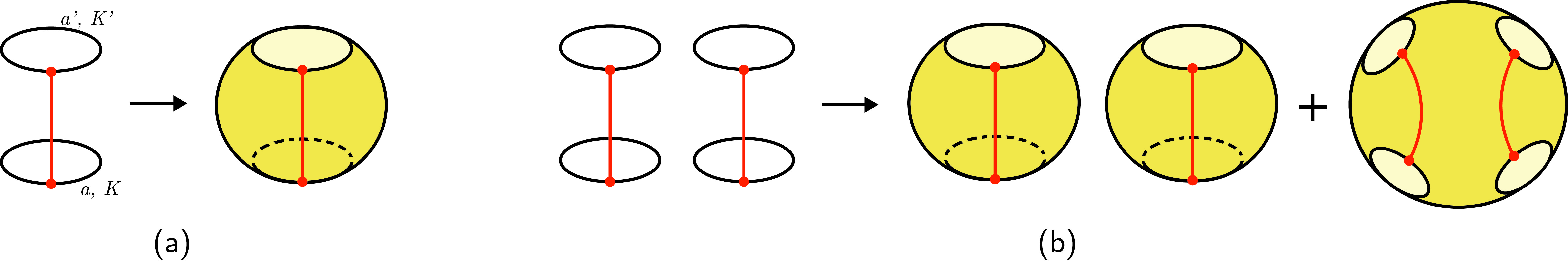}
    \caption{The gravitational path integral in the presence of an observer for de Sitter spacetime. A state is defined in terms of the length $a$ and extrinsic curvature $K$ of an initial slice where we insert the observer with scaling dimension $\Delta_O$. The worldline of the observer (depicted in red) must connect a bra and the corresponding ket when computing moments of an overlap. (a) The overlap $\overline{\langle a',K'|a,K\rangle}$ between two states. (b) The square of an overlap $\overline{|\langle a',K'|a,K\rangle|^2}$ in the presence of an observer. The variance of the overlap is given by the connected geometry on the right, namely a sphere with four holes, and is therefore $O(e^{-2S_0})$. This suggests a Hilbert space dimension $\dim\left(\Hr\right)=O(e^{2S_0})$. We plan to explicitly compute this dimension in future work.  }
    \label{fig:desitter}
\end{figure}
\noindent Following the discussion at the beginning of Section \ref{sec:non-pert_h}, this suggests a dimension of the quantum gravity relational Hilbert space \be 
\dim\left(\Hr\right)=O(e^{2S_0})
\ee
for global de Sitter. Similar to the closed universe in AdS-JT discussed at the beginning of Section \ref{sec:setup}, this result should be contrasted with the $O(1)$ variance and one-dimensional Hilbert space we would obtain with the old rules by allowing the observer's worldline to connect between all bras and kets in the gravitational path integral.

It would be interesting to explore these results in more detail, determine the precise dimension of the Hilbert space $\Hr$ in dS-JT gravity, and study relational observables in de Sitter using our proposal. In fact, a prescription similar to our proposal for how to treat an observer was discussed in the context of dS-JT gravity in \cite{LevSha22}. There, it was shown that treating the observer as a non-perturbatively well-defined boundary condition in the path integral resolved apparent contradictions between the GPI and the no-cloning theorem. The generalization to higher dimensional setups is also of great interest and could set the stage for a non-perturbative generalization of the results obtained in \cite{Chandrasekaran:2022cip}. A holographic realization of this setup similar to that suggested in Section \ref{sec:holography} could represent a new framework for de Sitter holography. 

Finally, our results do not depend on specific global properties of the spacetime; instead, they are centered on the experience of a gravitating observer. Therefore, they could be generalized to describe quantum gravity in generic settings, including realistic cosmological spacetimes. We leave the investigation of these intriguing new avenues to future work.

\section*{Acknowledgements}

We would like to thank Jan Boruch, Raphael Bousso, Elliott Gesteau, Daniel Harlow, Jorrit Kruthoff, Guanda Lin, Geoff Penington, Pratik Rath, Arvin Shahbazi Moghaddam, Edgard Shaghoulian, Ronak Soni, Misha Usatyuk and Ying Zhao for useful discussions. AA is supported by the NSF Graduate Research Fellowship Program under grant no. DGE 2146752. S.A. is supported by the U.S.
Department of Energy through DE-FOA-0002563. LVI is supported by the DOE Early Career Award DE-SC0025522 and by the DOE QuantISED Award DE-SC0019380. AL is supported by the Heising-Simons foundation under grant no. 2023-4430 and the Packard Foundation.

\appendix

\section{Resolvents and replicas}\label{app:res_replica}

In this appendix we derive the analytic properties of the resolvent used in the main text. Along the way, we also give a slightly different derivation of the Hilbert space dimension starting from the Schwinger-Dyson equation for the resolvent. We start by returning to the toy model which gave equation \eqref{eq:dimH_estimate} and considering the maximally mixed state on $\cH$\footnote{
We temporarily make use of boldface for operators to distinguish them from their traces.
}
\begin{equation}
    \bm{\rho}\equiv\sum_{i=1}^K \ket{v_i}\bra{v_i}.
\end{equation}
We suspect that $\rho$ is not full rank due to wormhole corrections and seek to calculate
\begin{equation}
    \overline{\text{dim}(\cH)}=\overline{\text{rank}(\bm{\rho})}
\end{equation}
by using replicas
\begin{equation}\label{eq:replica_trick}
    \overline{\text{dim}(\cH)}=\lim_{n\to 0}\overline{\Tr_{\cH}{\bm{\rho}}^n}.
\end{equation}
The trace is computed using the non-perturbative inner product and we assume $\overline{\Tr_{\cH}}\approx\Tr_{\Hnonp}$.\footnote{
The equality is approximate due to coarse-graining. Taking $\Tr_{\Hnonp}$ requires specifying a UV theory. Instead, we calculate $\overline{\Tr_{\cH}}$ which calculates an average over an ensemble of theories \cite{Marolf:2020xie,Saad:2019lba,Marolf:2024jze}. 
}

We emphasize some technical details before proceeding to the calculation. Null states are introduced into $\cH$ by taking $S_0$ finite (c.f. Section \ref{sec:nullstates}). We see that the replica limit does not commute with the choice to send $S_0\to\infty$ by considering 
\begin{equation}\label{eq:zero_to_zero}
    0=\lim_{n\to 0^+}\overline{\braket{\lambda}^n}\neq\overline{\lim_{n\to 0^+}\braket{\lambda}^n}=1
\end{equation}
for a null state $\ket{\lambda}$. The first equality is essential to counting the dimension of the non-perturbative Hilbert space. Otherwise, we would always find a space of dimension $K$. With this operational understanding, equation \ref{eq:replica_trick} counts the number of non-zero eigenvalues of $\bm{\rho}$. 

To perform this calculation, we will need to calculate arbitrary powers of the Gram matrix
\begin{equation}
    M_{ij}\equiv\braket{v_i}{v_j}
\end{equation}
which is just the metric on $\cH$. The rank of the Gram matrix is the same as that of $\bm{\rho}$. 

The negative powers of the Gram matrix are not well defined since the matrix is singular. We introduce the pseudo-inverse to remedy this. For an operator $\bm{A}$ acting on $\cH$ the Moore-Penrose pseudo-inverse $\bm{A}^{-1}$ is defined such that
\begin{equation}
    \bm{A}\bm{A}^{-1}\bm{A}=\bm{A}
\end{equation}
and $(\bm{A}^{-1})^{-1}=\bm{A}$. Although unnecessary for our purposes, it can be shown that $\bm{A}^{-1}$ is unique in the space of observables. Furthermore, we have
\begin{claim}\label{thm:inv_rank}
    $\Tr_{\cH}{\bm{A}^{-1}\bm{A}}=\Tr_{\cH}{\bm{A}\bm{A}^{-1}}=\text{rank}(\bm{A})
    $
\end{claim}
\begin{proof}
    Let $\bm{I}=\bm{A}^{-1}\bm{A}$. By construction, $\bm{A}$ and $\bm{A}^{-1}$ share a kernel so $\bm{I}$ is of the same rank as $\bm{A}$. Moreover, this operator is a projection
    \begin{equation}
        \bm{I}^2=\bm{A}^{-1}\bm{A}\bm{A}^{-1}\bm{A}=\bm{A}^{-1}\bm{A}=\bm{I}
    \end{equation}
    so it admits a rank factorization into some $\bm{I}=\bm{L}\bm{R}$ where $\bm{L}$ is left invertible and $\bm{R}$ is right invertible. Because $\bm{I}$ is a projection, $\bm{RL}$ must be the identity matrix of dimension $\text{rank}(\bm{I})$. By cyclicity of the trace
    \begin{equation}
        \Tr_{\cH}{\bm{I}}=\text{rank}(\bm{I}).
    \end{equation}
    The argument follows similarly for $\bm{I}=\bm{A}\bm{A}^{-1}$.
\end{proof}
\noindent As a corollary we have that $\bm{M}^{-1}\bm{M}\neq\bm{\rho}$ since
\begin{equation}
    \overline{\text{rank}(\bm{M})}=\overline{\Tr_{\cH}{\bm{M}^{-1}\bm{M}}}\neq \overline{\Tr_{\cH}{\bm{\rho}}}=K
\end{equation}
where the third equality holds assuming normalized $\{\ket{v_i}\}$. This clarifies that $\bm{\rho}$ is never a resolution of the identity if we endow $\cH$ with the non-perturbative inner product.

The correct limit of the Gram matrix is
\begin{equation}\label{eq:dimH_from_m}
    \overline{\text{dim}(\cH)}=\lim_{n\to 0}\overline{\Tr_{\cH}{\bm{\rho}^n}}=\lim_{n\to -1}\overline{\Tr_{\cH}{\bm{M}^{n}\bm{M}}}.
\end{equation}
We will compute this limit using the resolvent of the Gram matrix
\be\begin{aligned}\label{eq:rb}
    R_{ij}(\lambda) &\equiv \left(\frac{1}{\lambda\mathbbm{I}-\bm{M}}\right)_{ij}, \\
    &= \frac{1}{\lambda}\left(\delta_{ij}+\frac{1}{\lambda}(\bm{M})_{ij}+\frac{1}{\lambda^2}(\bm{M}^2)_{ij}+\cdots\right)
\end{aligned}\ee
The power series is taken literally when $|\lambda|\gg||\bm{M}||$ and the analytic continuation is assumed outside of the radius of convergence. 

Given a complex function $f(\lambda)$ that is analytic in the interior of some contour, we can define 
\begin{equation}\label{eq:hol_calc}
    f(\bm{A})=\oint \frac{d\lambda}{2\pi i}\ \frac{f(\lambda)}{\lambda\mathbbm{I}-\bm{A}}
\end{equation}
to be its generalization to square matrices. For holomorphic $f$, this definition is independent of the choice of contour and matches the standard matrix function. The pseudo-inverse is not analytic on general operators so we cannot use equation $\eqref{eq:hol_calc}$ to calculate it. Instead, we consider functions of the form
\begin{equation}
    f(\bm{A})\bm{M}=\oint\frac{d\lambda}{2\pi i}\ \frac{f(\lambda)\bm{M}}{\lambda\mathbbm{I}-\bm{A}}
\end{equation}
which are analytic on operators $\bm{A}$ with $\ker(\bm{A})=\ker(\bm{M})$. The contour is irrelevant for holomorphic $f(\lambda)$. Conversely, we can extend to non-holomorphic $f$ by a choice of contour. The contour which correctly assigns equation \eqref{eq:zero_to_zero} for $f(\lambda)=\lim_{n\to -1}\lambda^{n}\cdot\lambda$ is $C=C_\infty\sqcup C_0$ (c.f. Figure \ref{fig:contour}). Equation \eqref{eq:dimH_from_m} is of this form so we aim to calculate
\begin{equation}\label{eq:dimH_from_rm}
    \overline{\text{dim}(\cH)} = \overline{\Tr\left[\oint_{C}\frac{d\lambda}{2\pi i\lambda}\ R_{ij}(\lambda)M_{ji}\right]}=\overline{\Tr\left[\oint_{C}\frac{d\lambda}{2\pi i}\ R_{ij}(\lambda)\right]}
\end{equation}
where repeated indices are summed over. The second equality follows since $f(\lambda)$ is holomorphic in the interior of $C$. 

The integrand can be rewritten as
\begin{align}\label{eq:contour_integrand}
    (\bm{R}\bm{M})_{ij} &= \frac{1}{\lambda}(\bm{M})_{ij}+\frac{1}{\lambda^2}(\bm{M}^2)_{ij}+\cdots, \\
    &= \lambda R_{ij}-\delta_{ij}.
\end{align}
In Appendix \ref{sec:review} and Section \ref{sec:hilbert}, we find that the trace of the right-hand-side can be rewritten using a SDE of the form
\be\begin{aligned}\label{eq:sde_general}
    \lambda R = K+\int \mu(\Vec{E})\ \frac{y(\Vec{E}) R}{1-y(\Vec{E})R}
\end{aligned}\ee
where $\Vec{E}$ and $\mu(\Vec{E})$ generalize the energy bases and measures introduced in Sections \ref{sec:review} and \ref{sec:hilbert}. We also generalize by defining
\begin{equation}
    D = \int\mu(\Vec{E}).
\end{equation} 
Several useful properties follow. Proofs are given at a physicist's level of rigor.
\begin{claim}
    The resolvent is finite on $\C/0$. 
\end{claim}
\begin{proof}
    It suffices to show that $\lambda R(\lambda)<\infty$ on $\C$. Suppose otherwise. Then there is some nonzero $\lambda$ such that $y R=1$ meaning $y$ vanishes. This can only happen if $\gamma_\Delta$ vanishes but this is impossible by the properties of the $\Gamma$ function (c.f. equation \eqref{eq:gamma}).
\end{proof}
\noindent As a corollary we have that the function $f_y(\lambda)=\frac{y R}{1-y R}$ has no poles for $y\in\R^+$.
\begin{claim}
    The resolvent has the following structure as $\lambda\to 0$
    \begin{equation}
        R(\lambda)\sim\begin{cases}
            -R_0(K)&\qquad\text{if } K\leq D \\
            \frac{K-D}{\lambda}, &\qquad\text{otherwise}
        \end{cases}
    \end{equation}
    where $R_0(K)$ is a finite positive function on $0\leq K\leq D$.
\end{claim}
\begin{proof}
    By definition, if $R$ is finite at $\lambda=0$ then it is real and non-positive. Suppose that $-\infty<R(0)<0$ and let $R_0=-R(0)$. The SDE becomes 
    \begin{equation}
        \int\mu(\Vec{E})\ \frac{y R_0}{1+y R_0}=K.
    \end{equation}
    Since $yR_0>0$, the integrand is less than one and this ansatz can only work for $K\leq D$. In this domain we take this as a definition for $R_0(K)$. Now suppose that we encountered a singularity so that
    \begin{equation}
        \lim_{\lambda\to 0}R\sim\frac{R_{0,n}}{\lambda^n}
    \end{equation}
    with $R_{0,n}\neq 0$. By the SDE
    \begin{equation}
        R_{0,n} \sim \lambda^{n-1}\left(K+\int\mu(\Vec{E})\frac{y R_{0,n}}{\lambda^n-y R_{0,n}}\right)
    \end{equation}
    so at most $R$ has a simple pole at $\lambda=0$ with residue
    \begin{equation}
        R_{0,n} = \begin{cases}
            K-D &\qquad\textit{if } n=1, \\
            0 &\qquad\textit{otherwise}.
        \end{cases}
    \end{equation}
\end{proof}
\noindent As a corollary we have that
\begin{equation}
    f_y(0) = \begin{cases}
        \frac{-yR_0(K)}{1+yR_0(K)} &\qquad\textit{if } K\leq D, \\
        -1 &\qquad\textit{otherwise}.
    \end{cases}
\end{equation}
\noindent For completeness, we also verify
\begin{claim}
    The resolvent has no branch points for $\lambda\in\{0,\infty\}$.
\end{claim}
\begin{proof}
    Differentiating the SDE gives
    \begin{equation}
        R(\lambda) = \frac{dR}{d\lambda}\left(\int\mu(E)\frac{y}{(1-yR(\lambda))^2} - \lambda\right).
    \end{equation}
    At a branch point, we must have
    \begin{equation}
        \lambda = \int\mu(E)\frac{y}{(1-yR(\lambda))^2}.
    \end{equation}
    For $\lambda=0$ this requires $y$ to vanish which is ruled out by equation \eqref{eq:gamma}. As $\lambda\to\infty$ this requires $yR=1$. This is ruled out by the first claim.
\end{proof}
\noindent As a corollary we have that $f_y$ has no branch points for $\lambda\in\{0,\infty\}$. We will assume without proof that $f_y$ has no branch cuts in a neighborhood of these points. This is not always true (e.g., a closed universe without observers) and should be checked case-by-case.

Plugging the trace of \eqref{eq:contour_integrand} into equation \eqref{eq:dimH_from_rm} using equation \eqref{eq:sde_general} gives
\begin{equation}
    \overline{\dim(\cH)} = \int \mu(\Vec{E})\ \oint_{C}\frac{dz}{2\pi i\lambda}\ f_y(\lambda) 
\end{equation}
where $C=C_\infty\sqcup C_0$ is a contour defined in a neighborhood without branch points. We pick up no contribution from the $C_\infty$ integral since the integrand goes as $\frac{1}{\lambda}R\sim\frac{1}{\lambda^2}$ for large $\lambda$. The residue at $\lambda=0$ is $-f_y(0)$ since $C_0$ wraps clockwise. At last, performing the phase space integral gives
\begin{equation}
    \overline{\dim(\cH)} = \begin{cases}
        K &\qquad\textit{if }K\leq D, \\
        D &\qquad\textit{otherwise}
    \end{cases}
\end{equation}
which confirms the results of the main text.

\section{The Hilbert space of two-sided black holes}
\label{app:reviewBH}

We now move on to studying the non-perturbative Hilbert space of two-sided black holes. We work in the limit $K\to\infty$, $e^{2S_0}\to\infty$ with $K/e^{2S_0}=O(1)$. We study spaces spanned by states carrying an unspecified matter index $i$, a fixed scaling dimension $\Delta$, and a fixed asymptotic boundary length $\beta_L=\beta_R=\beta/4$ to the left and right of the matter insertion
\begin{equation}
    \Hnonp(K)=\text{Span}\left\{\ket{q_i};\ i\in\{1,\cdots,K\}\right\}.
\end{equation}
As discussed in the main text, these conventions are established purely for notational convenience.\footnote{
For example, one could equivalently define
\begin{equation}
    \Hnonp(K)=\text{Span} \left\{\ket{q_i}=\ket{\Delta_i;\beta_L=\beta_R=\beta/4};\ i\in\{1,\cdots, K\}\text{ and } \Delta_i\approx\Delta\right\}
\end{equation}
for fixed $\Delta$ and $\beta$. Basis independence is discussed in Appendix B of \cite{Boruch:2024kvv}.
}

Inner products are calculated using the topological expansion
\begin{equation}
    \overline{\braket{q_i}{q_j}} = \inlinefig{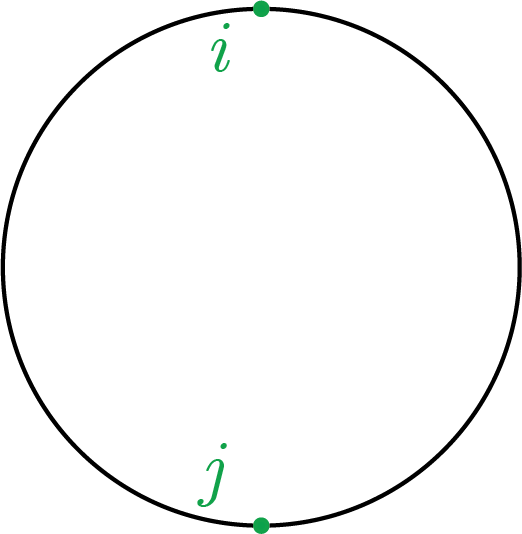} = \delta_{ij}\left(\inlinefig{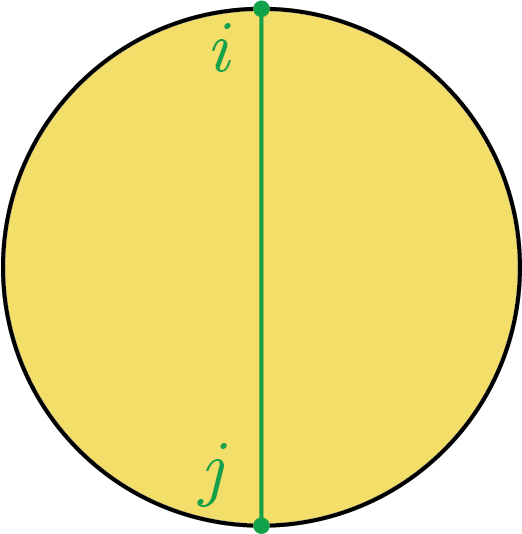} + \inlinefig{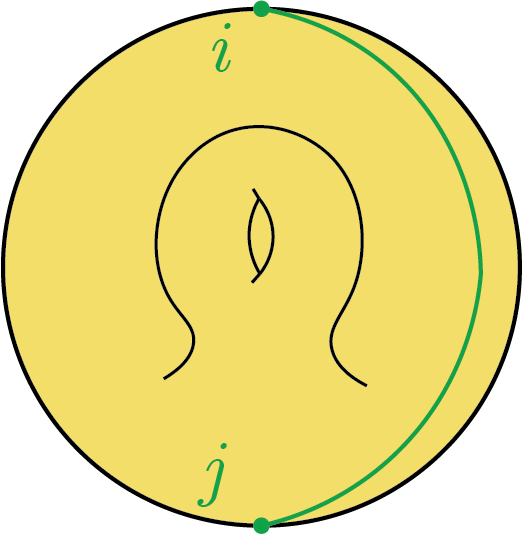} + \cdots\right)
\end{equation}
Note how the asympotic boundaries are forced to connect while the matter insertions float freely until the geometry is filled in with the gravitational path integral. The second moment is calculated similarly
\be\begin{aligned}
    \overline{|\braket{q_i}{q_j}|^2} = \inlinefig{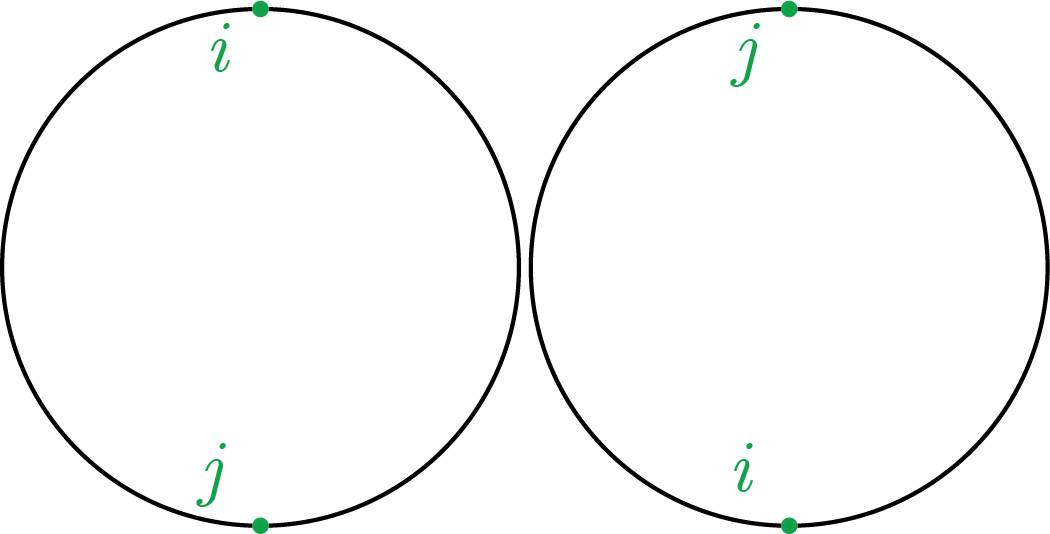}&=\delta_{ij}\left(\inlinefig{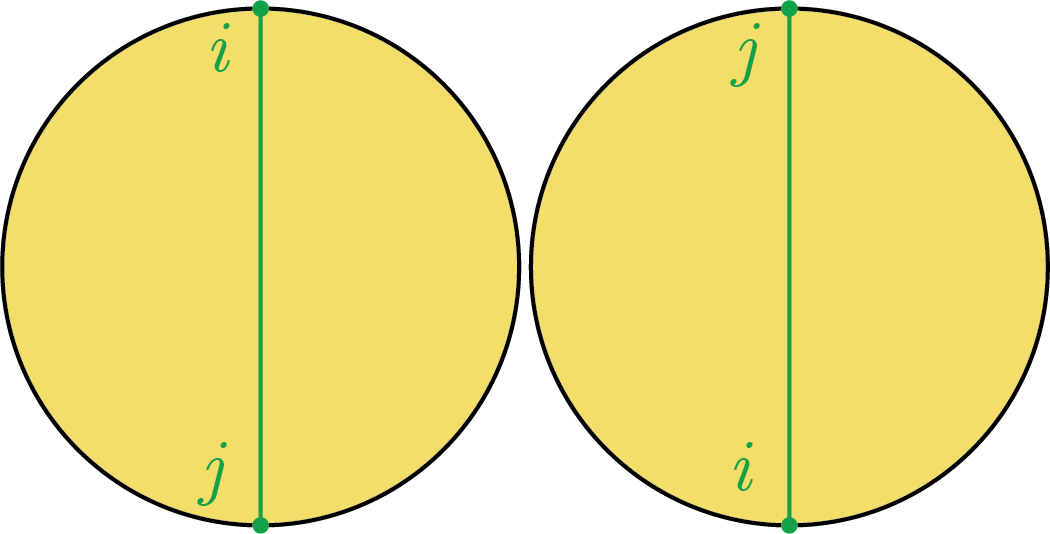}+O(e^{-2S_0})\right)  \\
    &+e^{-2S_0}\left(\inlinefig{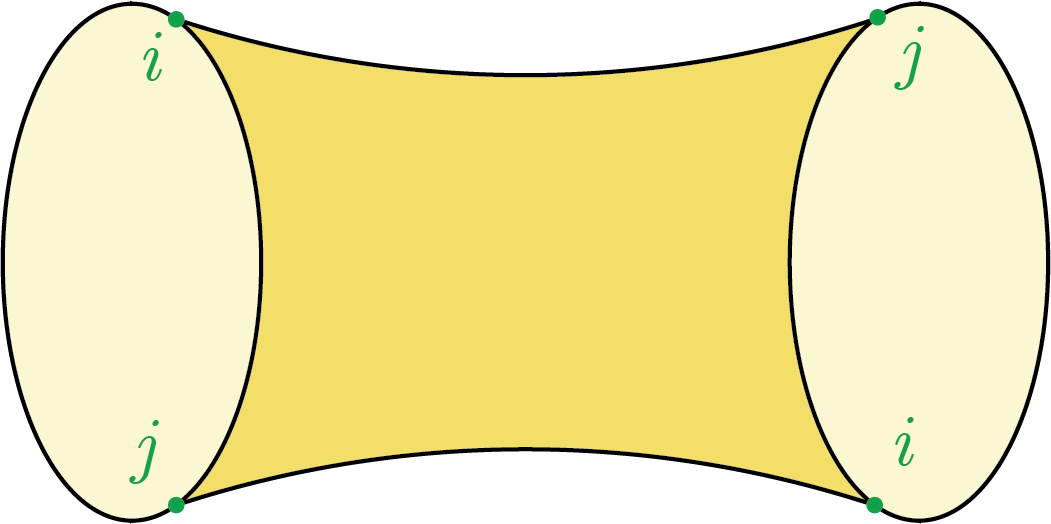}+O(e^{-2S_0})\right).
    \label{eq:squareoverlap}
\end{aligned}\ee
The pinwheel \eqref{eq:zn_matter_nonpert} is the generalization to the $n$-th moment
\begin{equation}
    \overline{\braket{q_{i_1}}{q_{1_2}}\cdots\braket{q_{i_n}}{q_{i_1}}}=\overline{Z_n}=\int ds_Lds_R\rho_0(s_L)\rho_0(s_R)\ y_\Delta^n(s_L,s_R)
\end{equation}
where 
\begin{equation}\label{eq:y_d}
    y_\Delta(s_L,s_R)=e^{-S_0}e^{-\frac{\beta s_L^2}{4}}e^{-\frac{\beta s_R^2}{4}}\gamma_{\Delta}(E_L,E_R).
\end{equation}
We are free to drop indices and lengths when defining the pinwheel thanks to our choice of states $\{\ket{q_i}\}$. 

Already, we can see why the Hilbert space dimension is reduced by the non-perturbative inner product. The first moment is proportional to $\delta_{ij}$ while higher moments are nonzero even when $i\neq j$. These extra terms correspond exclusively to wormhole corrections. We ignored them in Section \ref{sec:pert_h} because they are non-perturbatively suppressed by factors of $e^{-S_0}$. Undoing that choice is precisely what gives us a finite-dimensional Hilbert space. We introduce
\begin{equation}
    d\equiv e^{S_0} \int dE\rho_0(E)
\end{equation}
as the relevant parameter for discussions of the Hilbert space dimension, where $E=s^2/2$ and the relationship between $\rho_0(E)$ and $\rho_0(s)$ is given in equation \eqref{eq:rhonot}. For this parameter to be finite we need to work in an arbitrary but finite energy window $E\in (E_{min},E_{max})$. To leading order in $e^{S_0}$, the variance is given by the first connected wormhole geometry in equation \eqref{eq:squareoverlap}
\begin{equation}
    \sigma^2 = \overline{|\braket{q}{q_j}|^2}-\left|\overline{\braket{q_i}{q_j}}\right|^2 = O(e^{-2S_0}). 
\end{equation}
Using equation \eqref{eq:variance} gives us an estimate $\dim(\Hnonp)=O(d^2)$. This result is in agreement with our semiclassical understanding of this system as a two-sided black hole where each horizon carries an entropy $S_0$.\footnote{
In non-perturbative pure JT gravity, the constraint $H_0$ \eqref{eq:constraint} reduces the dimension of the Hilbert space by imposing that the left and right horizons are highly entangled. This gives $\dim(\Hnonp)=O(e^{S_0})$ \cite{IliLev24}. We saw in equation \eqref{eq:matter_constraint} that the inclusion of matter breaks the constraint, allowing us to explore the full Hilbert space of the two-sided black hole.
} 

We once again refer to a resolvent for a more detailed analysis. The $n$-th power of the Gram Matrix $(M^n)_{ij}$ is given by pinwheel boundary conditions where all but two adjacent indices are identified and summed over. This means that we can write the resolvent diagrammatically as 
\be\begin{aligned}
    \inlinefig[5]{Figures/App_A/rij300ppi.png} &= \inlinefig[5]{Figures/App_A/dij300ppi.png}+\inlinefig[5]{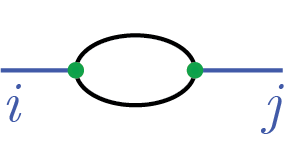}+\inlinefig[5]{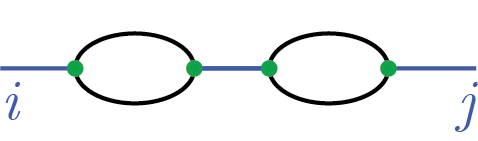}+\cdots \\
    &= \inlinefig[5]{Figures/App_A/dij300ppi.png}+\inlinefig[5]{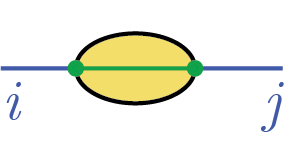}+ \left(\substack{
        \inlinefig[3.75]{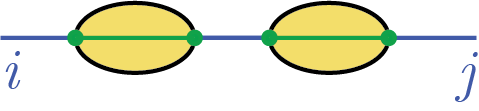} \\
        + \\
        \inlinefig[8]{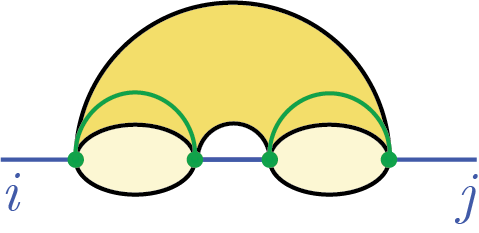} \\
        + \\
        \inlinefig[8]{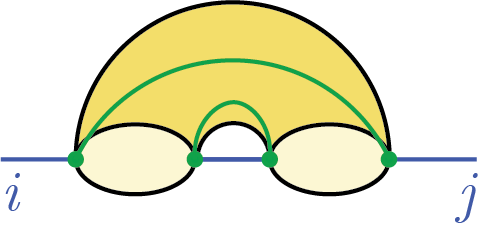}
    }\right)
    +\cdots
\end{aligned}\ee
We can reorganize this sum recursively based on how many boundaries the first boundary connects to. This yields the following Schwinger-Dyson equation (SDE)
\be\begin{aligned}
    \inlinefig[5]{Figures/App_A/rij300ppi.png} &=\inlinefig[5]{Figures/App_A/dij300ppi.png}+\inlinefig[5]{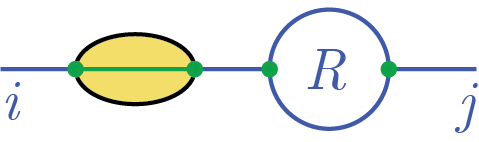}+\\
    &\qquad \inlinefig[13]{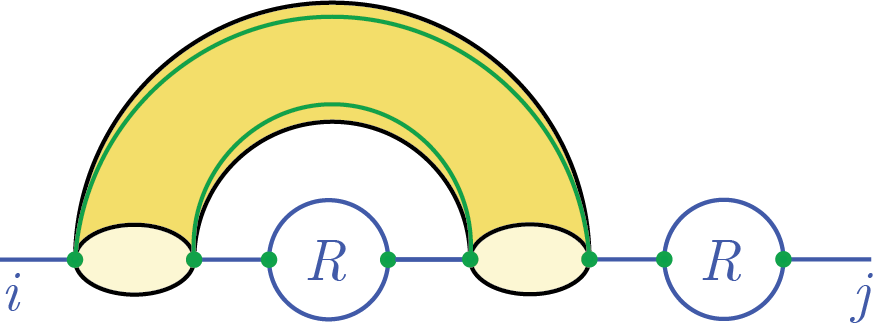}+\cdots
\end{aligned}\ee
For a fixed topology, we are ignoring terms like
\begin{equation}
    \inlinefig{Figures/App_A/bh_m2ij2300ppi.png}
\end{equation}
because they are subleading in $K$ compared to a cyclic contractions of the indices. This means that internal factors of the resolvent are traced over. 

The diagrams translate to the following equations
\be\begin{aligned}
    \lambda \overline{R_{ij}(\lambda)} &= \delta_{ij}+\sum_{n=1}^\infty Z_nR^{n-1}(\lambda)R_{ij}(\lambda), \\
    &= \delta_{ij}+e^{2S_0}\sumint ds_Lds_R\rho(s_L)\rho(s_R)\ y_\Delta^n(s_L,s_R) R^{n-1}(\lambda)\overline{R_{ij}(\lambda)}.
\end{aligned}\ee
Tracing over the external indices gives
\be\begin{aligned}\label{eq:sde_algebraic}
    \lambda R &= K+e^{2S_0}\sumint ds_Lds_R\rho(s_L)\rho(s_R)\ y_\Delta^n(s_L,s_R) R^{n}(\lambda), \\
    &= K+e^{2S_0}\int ds_Lds_R\rho(s_L)\rho(s_R)\ \frac{y_\Delta R(s_L,s_R)}{1-y_\Delta(s_L,s_R) R(\lambda)}.
\end{aligned}\ee
This SDE is of the form described in equation \ref{eq:sde_general}. By the results of Appendix \ref{app:res_replica}, we obtain\footnote{
    We could also compute equation \eqref{eq:dimH_from_r} using the integration contour given in Figure \ref{fig:contour}. This calculation is very similar to that explicitly carried out in Section \ref{sec:closeduniverse}, to which we refer for details.
} 
\begin{equation}
    \overline{\dim(\Hnonp(K))}=\begin{cases}
        K &\qquad \text{if }K\leq  d^2, \\
        d^2 &\qquad otherwise.
    \end{cases}
    \label{eq:two-sided-BH-dimension}
\end{equation}
For $K<d^2$, the $\{\ket{q_i}\}$ states span a $K$-dimensional subspace of $\Hnonp$. When $K>d^2$, the states become overcomplete and the subspace dimension saturates at $\dim(\Hnonp(K))=d^2$. Since this result is independent of basis, we conclude that $\dim(\Hnonp)=d^2$. The difference between the number of states and the dimension of the span is $K-d^2$, the residue of the resolvent at $\lambda=0$. Equivalently, the Gram matrix $M$ has eigenvalue $\lambda=0$ with multiplicity $K-d^2$. We call a state in $\ker(M)$ a null state \cite{Penington:2019kki,IliLev24}. A detailed description of these states is provided in Section \ref{sec:nullstates}.

\section{Factorisation of the Hilbert space in the presence of an observer}
\label{app:factorisation}

In this Appendix, we will show that the non-perturbative quantum gravity relational Hilbert space $\Hr$ factorises as discussed at the end of Sections \ref{sec:closeduniverse} and \ref{sec:twosided}.

\subsection{Closed universe}
\label{app:closedfactorisation}

Given the resolvent analysis carried out in Section \ref{sec:closeduniverse}, we can now show that the non-perturbative, relational Hilbert space $\Hr$ for JT gravity coupled to matter in a closed universe factorises: 
\be
\Hr=\Hrl\otimes\Hrr.
\ee
Here $\Hrl$ and $\Hrr$ denote the Hilbert spaces to the left and right of the observer. To show this, let us follow \cite{Boruch:2024kvv} and show that\footnote{Technically, we will show this relationship using the gravitational path integral, and therefore in terms of averages over the dual random matrix ensemble.}
\begin{equation}
    \tr_{\Hr}\left(k_lk_r\right)=\tr_{\Hrl}(k_l)\tr_{\Hrr}(k_r),
\end{equation}
where $k_l$ and $k_r$ are operator insertions to the left and right of the observer. In particular, we will consider insertions of $k_l$ and $k_r$ on an asymptotic boundary, and focus on the simple case in which $k_l=e^{-\beta_l H_l}$ and $k_r=e^{-\beta_r H_r}$, see Figure \ref{fig:kinsertions}, where $H_l$ and $H_r$ are the generators of translations along the asymptotic boundary on the left and right of the observer. Our results can be generalized to arbitrary operators $k_l$, $k_r$, but this requires dressing the operators to the observer's worldline, as we discussed in Section \ref{sec:observables}.\footnote{In general, when considering a correlation function $\langle\psi_i|k_lk_r|\psi_j\rangle$, we need to specify where the operators are acting with respect to the time read by a clock along the worldline. For instance, we must specify whether they are acting on the boundary associated with the bra or the ket. We will clarify this procedure in Section \ref{sec:observables}. The operators $e^{-\beta H}$ we consider here do not require dressing for our purposes because changing the length of an asymptotic boundary in a bra or the corresponding ket yields the same result for the overlap between the two states and its moments, which is all we will need. This is manifest from equation \eqref{eq:tildexk}, which only depends on the sum of the lengths of the bra and ket asymptotic boundaries to the left and right of the observer.} If we start with one of the states $\ket{\psi_i}$, $i=1,..,K$ considered above---for which the boundary length is fixed to be $\beta$---the boundary associated with the resulting state $|\psi_{i,\tilde{\beta}}\rangle=k_lk_r\ket{\psi_i}$
will then have length $\tilde{\beta}=\beta+\beta_l+\beta_r$. 

\begin{figure}[h]
    \centering
    \includegraphics[width=0.7\linewidth]{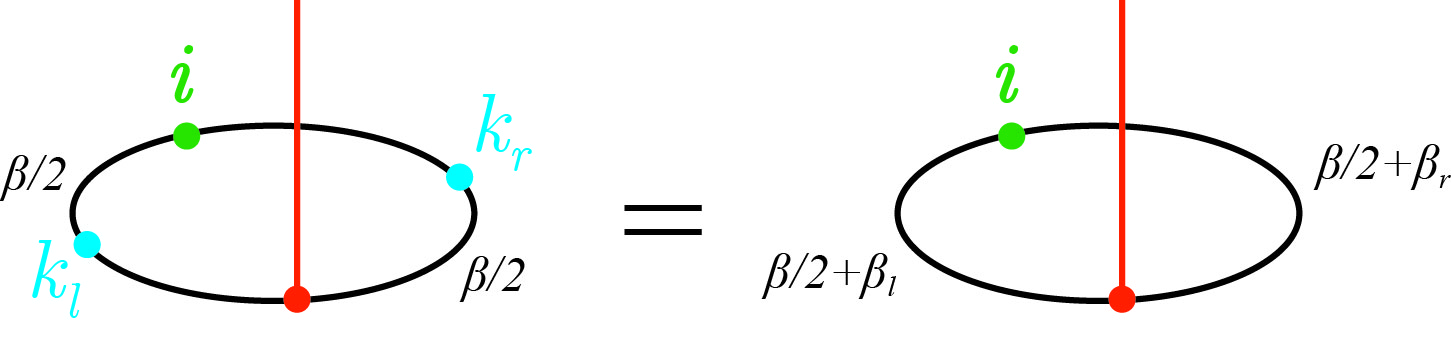}
    \caption{Boundary conditions for a closed universe state prepared by the insertion of the observer, a matter operator of flavor $i$, and operators $k_l=e^{-\beta_lH_l}$ and $k_r=e^{-\beta_rH_r}$ to the left and right of the observer. The insertion of $k_l$ and $k_r$ in a state with boundary length $\beta/2$ to the left and right of the observer (left picture) is equivalent to a state with boundary length $\beta/2+\beta_l$ to the left of the observer and $\beta/2+\beta_r$ to the right of the observer (right picture).}
    \label{fig:kinsertions}
\end{figure}

The trace over the non-perturbative, relational Hilbert space spanned by $K$ matter basis states $\Hr(K)$ is given by \cite{Boruch:2024kvv}
\begin{equation}
    \tr_{\Hr(K)}(k_lk_r)=\lim_{n\to -1} \sum_{ij}\left(M^n\right)_{ij}\langle \psi_j|k_lk_r|\psi_i\rangle,
\end{equation}
where $M_{ij}$ is again the Gram matrix of overlaps and $M^{-1}$ is its generalized inverse.
To relate the trace to the resolvent, we can contract both sides of equation \eqref{eq:r} with $\langle\psi_j|k_lk_r|\psi_k\rangle$ and obtain
\begin{equation}
    \sum_jR_{ij}(\lambda)\langle\psi_j|k_lk_r|\psi_k\rangle=\frac{1}{\lambda}\sum_{n=0}^\infty \frac{\sum_j\left(M^n\right)_{ij}\langle\psi_j|k_lk_r|\psi_k\rangle}{\lambda^n}.
    \label{eq:traceres}
\end{equation}
The trace can then be obtained by the residue integral
\begin{equation}
    \tr_{\Hr(K)}(k_lk_r)=\lim_{n\to -1}\frac{1}{2\pi i}\oint_{C} d\lambda\lambda^n \sum_{ij}R_{ij}(\lambda)\langle\psi_j|k_lk_r|\psi_i\rangle.
    \label{eq:traceresolventclosed}
\end{equation}
The integral here is on the same contour $C=C_0\cup C_\infty$ depicted in Figure \ref{fig:contour}.\footnote{Also in this case, the contour must exclude the origin because we must exclude any vanishing eigenvalues when defining the generalized inverse.} 
The integrand in equation \eqref{eq:traceresolventclosed} can be obtained by a diagrammatic expansion of equation \eqref{eq:traceres}. This is very similar to the one used above for the resolvent. The main difference is that geometries contributing the $n$-th term in the sum have $2n+2$ boundaries, with $2n+1$ boundaries of length $\beta$ and one boundary of length $\tilde{\beta}$.
The leading connected contribution is given by a genus-zero pinwheel with an observer insertion and a matter insertion on each of the $2n+2$ boundaries, depicted in Figure \ref{fig:kpinwheelclosed}.
\begin{figure}
    \centering
    \includegraphics[width=0.4\linewidth]{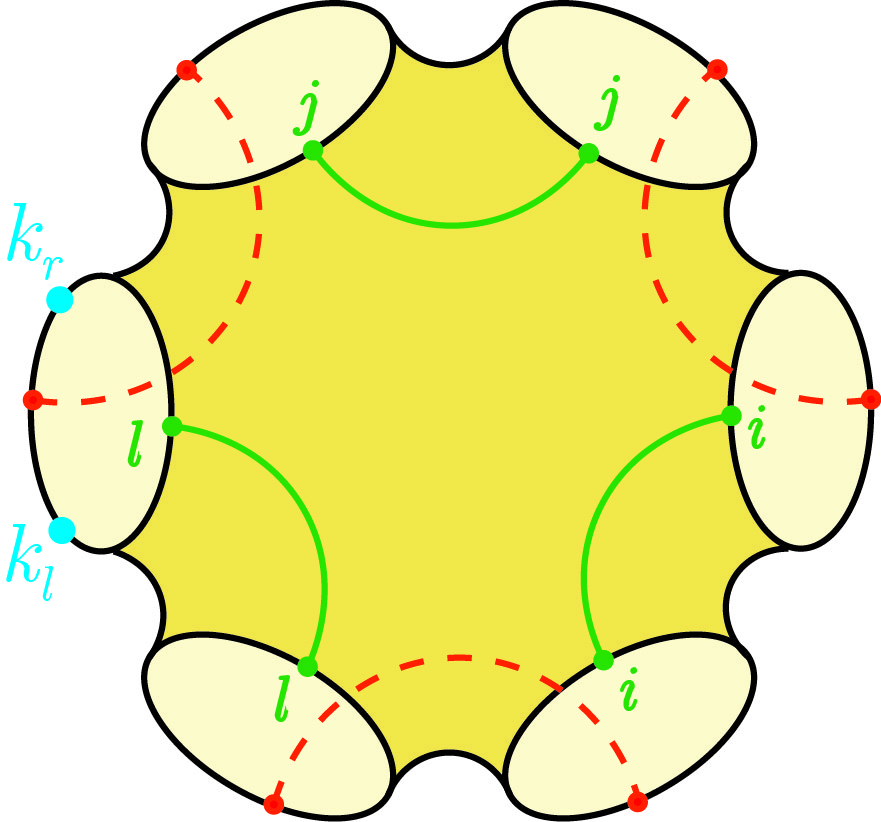}
    \caption{Genus-zero pinwheel geometry contributing to the $n=2$ term $\overline{\langle\psi_l|\psi_i\rangle\langle\psi_i|\psi_j\rangle\langle\psi_j|k_lk_r|\psi_l\rangle}$ on the right hand side of equation \eqref{eq:traceres}.}
    \label{fig:kpinwheelclosed}
\end{figure}
This geometry gives a contribution\footnote{Here we label the energies by $l$ and $r$ to indicate the patches to the left and right of the observer. These are the same patches that we labeled by 1 and 2 in the calculation of the Hilbert space dimension in Section \ref{sec:closeduniverse}, see the caption of Figure \ref{fig:closedstate}.}
\begin{equation}
    Z_{n+1}^{\textrm{closed}}(k_l,k_r)=e^{2S_0}\int ds_lds_r\rho_0(s_l)\rho_0(s_r)\tilde{x}^n(s_l,s_r)\tilde{x}_k(s_l,s_r)
\end{equation}
where $\tilde{x}$ is given in equation \eqref{eq:tildex} and $\tilde{x}_k$ is defined as
\begin{equation}
    \tilde{x}_k(s_l,s_r)=e^{-2S_0}e^{-\frac{(\beta+\beta_l)s_l^2}{2}}e^{-\frac{(\beta+\beta_r)s_r^2}{2}}\gamma_{\Delta_O}(s_l,s_r)\gamma_{\Delta_m}(s_l,s_r).
    \label{eq:tildexk}
\end{equation}
After setting $i=k$ and summing over $i$, the Schwinger-Dyson equation obtained from the diagrammatic expansion of equation \eqref{eq:traceres} reads
\begin{equation}
\begin{aligned}
    \overline{\sum_{ij}R_{ij}(\lambda)\langle\psi_j|k_lk_r|\psi_i\rangle}&=\sum_{n=0}^\infty Z_{n+1}^{\textrm{closed}}(k_l,k_r)R^{n+1}(\lambda)\\
    &=e^{2S_0}\int ds_lds_r\rho_0(s_l)\rho_0(s_r)\frac{R(\lambda)\tilde{x}_k(s_l,s_r)}{1-R(\lambda)\tilde{x}(s_l,s_r)}.
    \end{aligned}
\end{equation}
We can now finally compute $\overline{\tr_{\Hr(K)}(k_lk_r)}$ using equation \eqref{eq:traceresolventclosed}:
\begin{equation}
    \overline{\tr_{\Hr(K)}(k_lk_r)}=\frac{e^{2S_0}}{2\pi i}\int ds_lds_r\rho_0(s_l)\rho_0(s_r)\tilde{x}_k(s_l,s_r)\oint_{C}\frac{d\lambda}{\lambda}\frac{R(\lambda)}{1-R(\lambda)\tilde{x}(s_l,s_r)}.
\end{equation}
Because $R(\lambda)\sim K/\lambda^2$ for $\lambda\to\infty$, the integral over the counterclockwise contour $C_\infty$ at infinity vanishes, and we can focus solely on the integral over the clockwise contour $C_0$ around $\lambda=0$. As we have discussed above, if we consider $K<d^2$ basis states, $R(\lambda)$ has no poles. We can thus write $R(0)=-R_0(K)$ \cite{Boruch:2024kvv} and obtain
\begin{equation}
    \overline{\tr_{\Hr(K)}(k_lk_r)}=e^{2S_0}\int ds_lds_r\rho_0(s_l)\rho_0(s_r)\frac{R_0(K)\tilde{x}_k(s_l,s_r)}{1+R_0(K)\tilde{x}(s_l,s_r)} \quad \quad \quad K<d^2.
\end{equation}
In this case, the trace clearly does not factorise due to the non-trivial dependence of $\gamma_{\Delta_O}(s_l,s_r)$ and $\gamma_{\Delta_m}(s_l,s_r)$ on $s_l$ and $s_r$, see equation \eqref{eq:gamma}. The reason for this lack of factorisation is that $K<d^2$ states do not span the full quantum gravity relational Hilbert space $\Hr$, but rather an arbitrary subspace of $\Hr$, which need not factorise in general. 

On the other hand, for $K>d^2$, $\Hr(K)=\Hr$, namely, the $K$ states span the entire quantum gravity
relational Hilbert space. Note that for $K>d^2$, $R(\lambda)$ has a pole at $\lambda=0$ and therefore drops out of the integral over $C_0$. We then obtain
\begin{equation}
    \begin{aligned}
        \overline{\tr_{\Hr(K)}(k_lk_r)}&=\overline{\tr_{\Hr}(k_lk_r)}=e^{2S_0}\int ds_lds_r\rho_0(s_l)\rho_0(s_r)\frac{\tilde{x}_k(s_l,s_r)}{\tilde{x}(s_l,s_r)}\\[10pt]
        &=\left(e^{S_0}\int ds_l\rho_0(s_l)e^{-\frac{\beta_ls_l^2}{2}}\right)\left(e^{S_0}\int ds_r\rho_0(s_r)e^{-\frac{\beta_rs_r^2}{2}}\right)\\[10pt]
        &=\overline{\tr_{\Hrl}(k_l)}\overline{\tr_{\Hrr}(k_r)}\quad \quad \quad K>d^2
    \end{aligned}
\end{equation}
We then conclude that the full quantum gravity relational Hilbert space $\Hr$ factorises into a Hilbert space to the right of the observer and a Hilbert space to the left of the observer.

\subsection{Two-sided black hole}
\label{app:BHfactorisation}

Similarly, it is simple to argue that the Hilbert space $\Hr$ of two-sided black holes factorises into a tensor product of four Hilbert spaces
\begin{equation}
    \Hr=\HrR\otimes \Hrr\otimes \HrL\otimes \Hrl,
\end{equation}
where $\HrR$, $\HrL$ are associated with the right and left asymptotic boundaries, and $\Hrr$, $\Hrl$, with the right and left sides of the observer. Here we will only sketch the main steps of the argument, because it is very similar to that used to show the factorisation of $\Hr$ in the closed universe case.

We can consider the insertion of four operators $k_l=e^{-\beta_lH_l}$, $k_r=e^{-\beta_rH_r}$, $k_L=e^{-\beta_LH_L}$, $k_R=e^{-\beta_RH_R}$ in the preparation of a ket $\ket{\psi_{ii'}}$, $i,i'=1,...,K$ as shown in Figure \ref{fig:BHfactor}.
\begin{figure}[h]
    \centering
    \includegraphics[width=0.8\linewidth]{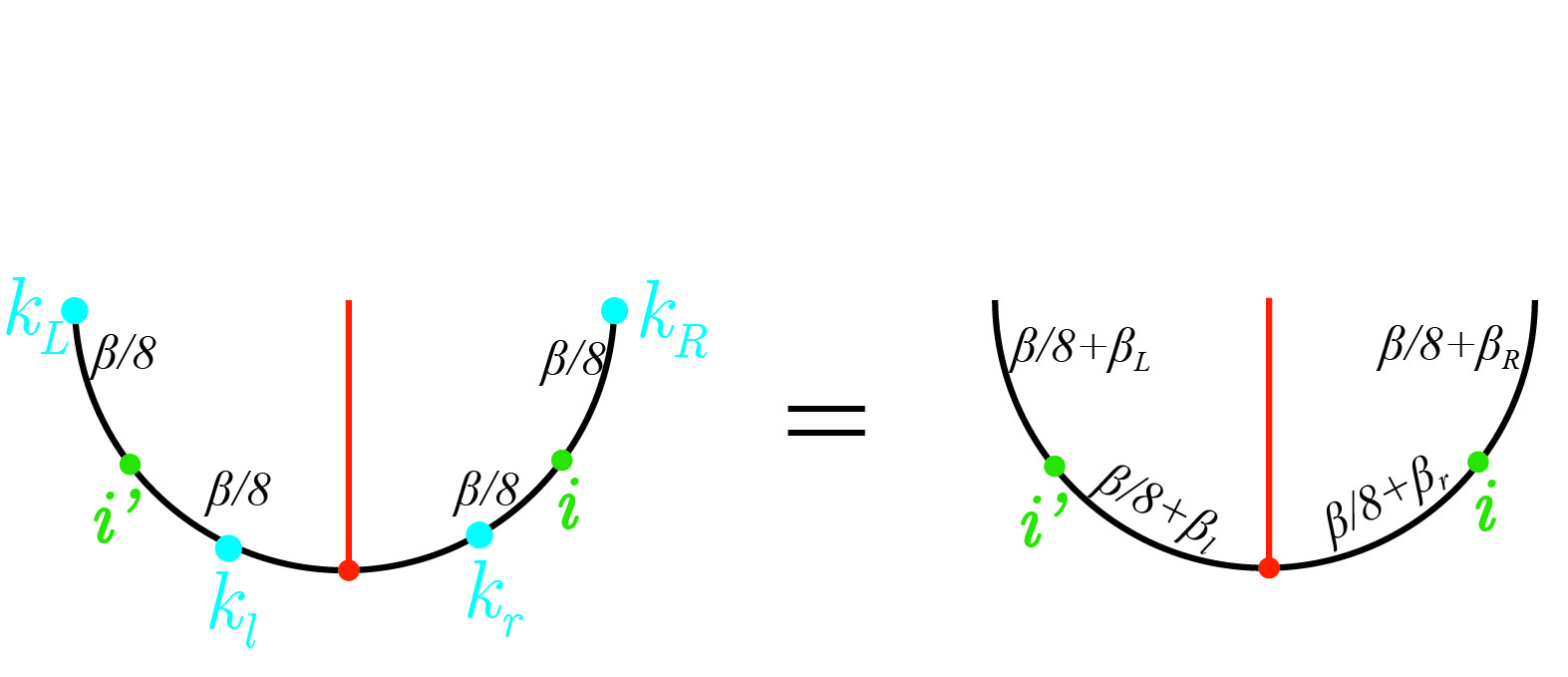}
    \caption{Boundary conditions for a two-sided black hole state prepared by the insertion of the observer, two matter operators to the left and right of the observer, and operators $k_l=e^{-\beta_lH_l}$, $k_r=e^{-\beta_rH_r}$, $k_L=e^{-\beta_LH_L}$, $k_R=e^{-\beta_RH_R}$. The insertion of $k_l$, $k_r$, $k_L$, $k_R$ in a state with every boundary segment of length $\beta/8$ (left picture) is equivalent to a state with boundary segments of length $\beta/8+\beta_l$, $\beta/8+\beta_r$, $\beta/8+\beta_L$, and $\beta/8+\beta_R$ (right picture).}
    \label{fig:BHfactor}
\end{figure}
We then have
\begin{equation}
    \sum_{jj'}R_{ii',jj'}(\lambda)\langle\psi_{jj'}|k_Lk_lk_Rk_r|\psi_{kk'}\rangle=\frac{1}{\lambda}\sum_{n=0}^{\infty}\frac{\sum_{jj'}\left(M^n\right)_{ii',jj'}\langle\psi_{jj'}|k_Lk_lk_Rk_r|\psi_{kk'}\rangle}{\lambda^n}
    \label{eq:reskBH}
\end{equation}
and the leading connected contribution to the $n$-th term in the sum comes from the genus-$n$ geometry with $2n+1$ boundaries depicted in Figure \ref{fig:kpinwheelBH}.
\begin{figure}[h]
    \centering
    \includegraphics[width=0.5\linewidth]{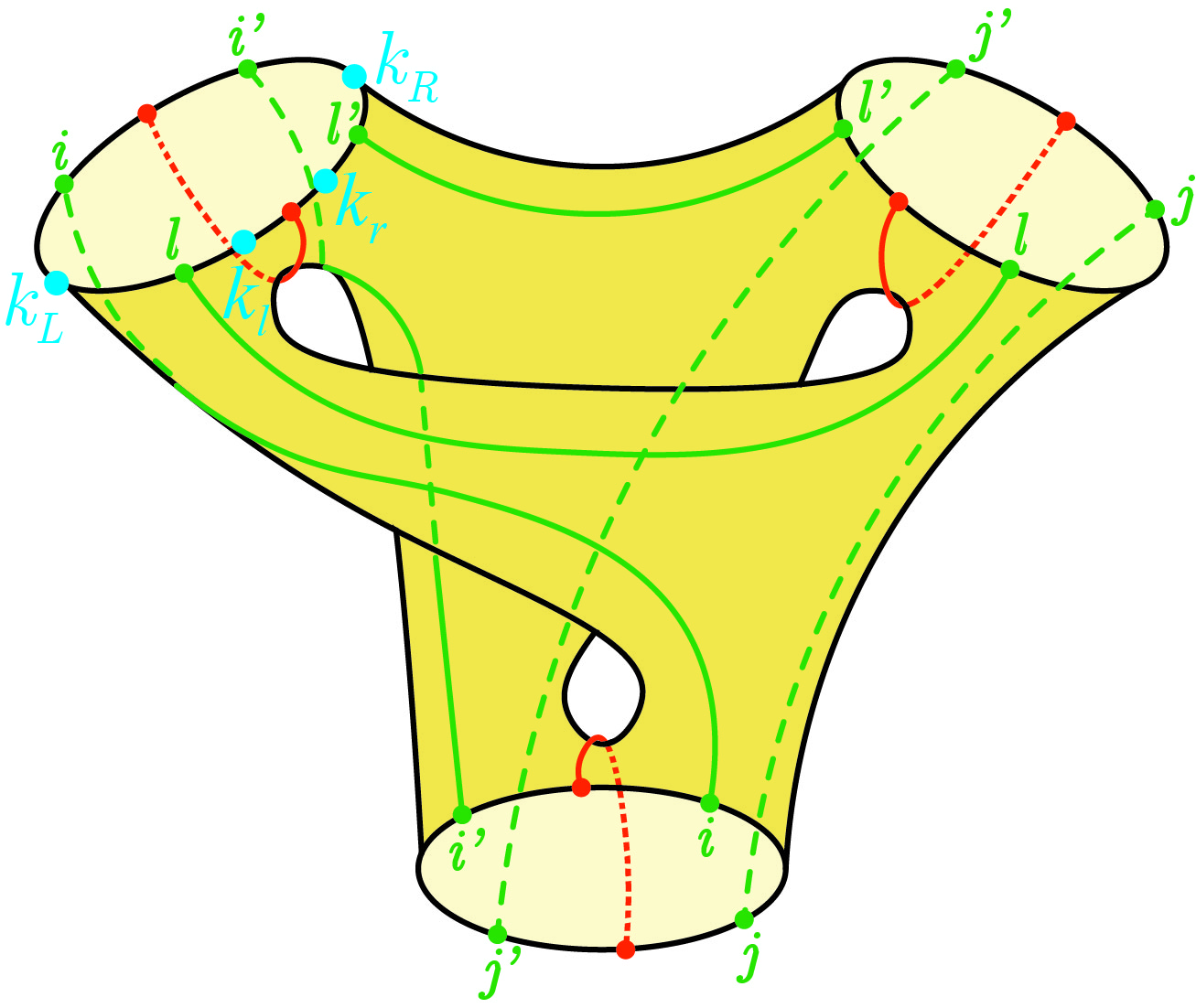}
    \caption{Genus $g=2$ geometry contributing to the $n=2$ term $\overline{\langle\psi_{ll'}|\psi_{jj'}\rangle\langle\psi_{jj'}|\psi_{ii'}\rangle\langle\psi_{ii'}|k_Lk_lk_Rk_r|\psi_{ll'}\rangle}$ on the right hand side of equation \eqref{eq:reskBH}.}
    \label{fig:kpinwheelBH}
\end{figure}
Its contribution is given by\footnote{$s_L$, $s_l$, $s_r$, $s_R$ here correspond respectively to $s_1$, $s_2$, $s_4$, $s_3$ in Section \ref{sec:twosided}, see the caption of Figure \ref{fig:2sidedstate}.} 
\begin{equation}
    Z_{n+1}^{BH}(k_L,k_l,k_R,k_r)=e^{4S_0}\int ds_Lds_lds_Rds_r\rho_0(s_L)\rho_0(s_l)\rho_0(s_R)\rho_0(s_r)\tilde{y}^n(s_L,s_l,s_R,s_r)\tilde{y}_k(s_L,s_l,s_R,s_r)
    \label{eq:zkBH}
\end{equation}
where $\tilde{y}(s_L,s_l,s_R,s_r)$ is given in equation \eqref{eq:tildey} and
\begin{equation}
    \tilde{y}_k(s_L,s_l,s_R,s_r)=e^{-3S_0}e^{-\frac{(\beta+4\beta_L)s_L^2}{8}}e^{-\frac{(\beta+4\beta_l)s_l^2}{8}}e^{-\frac{(\beta+4\beta_R)s_R^2}{8}}e^{-\frac{(\beta+4\beta_r)s_r^2}{8}}\gamma_{\Delta_m}(s_L,s_l)\gamma_{\Delta_m}(s_r,s_R)\gamma_{\Delta_O}(s_l,s_r).
\end{equation}
Using a diagrammatic expansion of equation \eqref{eq:reskBH} along with equation \eqref{eq:zkBH}, we obtain the Schwinger-Dyson equation
\begin{equation}
    \overline{\sum_{ii'jj'}R_{ii',jj'}(\lambda)\langle\psi_{jj'}|k_Lk_lk_Rk_r|\psi_{ii'}\rangle}=e^{4S_0}\int ds_Lds_lds_Rds_r\rho_0(s_L)\rho_0(s_l)\rho_0(s_R)\rho_0(s_r)\frac{R(\lambda)\tilde{y}_k(s_L,s_l,s_R,s_r)}{1-R(\lambda)\tilde{y}(s_L,s_l,s_R,s_r)}.
    \label{eq:integrandtraceBH}
\end{equation}
Similar to equation \eqref{eq:traceresolventclosed} in the closed universe case, we can then relate the left-hand side of equation $\eqref{eq:integrandtraceBH}$ to $\overline{\tr_{\Hr(K)}(k_Lk_lk_Rk_r)}$. Performing the contour integral, we find that the Hilbert space does not factorise if $K^2<d^4$. Again, this is due to the $K^2<d^4$ states spanning a subspace of $\Hr$, which need not factorise in general. On the other hand, if $K^2>d^4$, the $K^2$ states span the full quantum gravity relational Hilbert space, i.e. $\Hr(K)=\Hr$. In this case, we obtain
\begin{equation}
\begin{aligned}
    &\overline{\tr_{\Hr}(k_Lk_lk_Rk_r)}=e^{4S_0}\int ds_Lds_lds_Rds_r\rho_0(s_L)\rho_0(s_l)\rho_0(s_R)\rho_0(s_r)\frac{\tilde{y}_k(s_L,s_l,s_R,s_r)}{\tilde{y}(s_L,s_l,s_R,s_r)}\\[10pt]
    &=\left(e^{S_0}\int ds_L\rho_0(s_L)e^{-\frac{\beta_Ls_L^2}{2}}\right)\left(e^{S_0}\int ds_l\rho_0(s_l)e^{-\frac{\beta_ls_l^2}{2}}\right)\left(e^{S_0}\int ds_R\rho_0(s_R)e^{-\frac{\beta_Rs_R^2}{2}}\right)\left(e^{S_0}\int ds_r\rho_0(s_r)e^{-\frac{\beta_rs_r^2}{2}}\right)\\[10pt]
    &=\overline{\tr_{\HrL}(k_L)}\overline{\tr_{\Hrl}(k_l)}\overline{\tr_{\HrR}(k_R)}\overline{\tr_{\Hrr}(k_r)}.
    \end{aligned}
\end{equation}
This result signals the factorisation of $\Hr$ into four Hilbert spaces, with two of them ($\HrL$ and $\HrR$) associated with the two asymptotic boundaries of the two-sided black hole, and the other two ($\Hrl$ and $\Hrr$) associated with the two sides of the observer.

\section{Chern-Simons theory as a model of relational dynamics}
\label{app:chernsimons}

We find a practical and fully soluble example of a theory with constraints in Chern-Simons. Following closely the work of \cite{Witten:1988hf,Elitzur:1989nr}, the quantization procedure we outline has several analogues to our prescription for quantizing gravity in the presence of an observer. In particular, this example illustrates how the introduction of an observer, or ``source charge" as this modification is more commonly known in the Chern-Simons literature, changes the constraint equation and the associated Hilbert space of physically relevant states. 

We formulate the theory on a manifold $M=\R\times\Sigma$ with $\Sigma$ a 2-manifold. The action is
\begin{equation}
    S=\frac{k}{4\pi}\int_M \Tr[AdA+\frac{2}{3}A^3]
\end{equation}
where $A$ is a Lie-Algebra valued one-form $A=A^a_\mu T_a dx^\mu$ which serves as a connection over $M$ and $\Tr{\cdot}$ is its associated Killing form. For now we leave both the group $G$ and the manifold $\Sigma$ unspecified. Let us call the real-line in $M$ time and decompose the connection into a temporal part, labeled by the variable $t$, and spatial parts, labeled by the index $i$, so that 
\begin{equation}
    A=A_tdt+A_idx^i.
\end{equation}
We define $\Tilde{A}=A_idx^i$ and $d[\cdot]=dt\wedge\partial_t[\cdot]+\Tilde{d}[\cdot]$. With these conventions the action becomes
\begin{equation}\label{eq:cs}
    S=-\frac{k}{4\pi}\int_{M}\Tr[dt\wedge\Tilde{A}\wedge \partial_t\Tilde{A}]+\frac{k}{2\pi}\int_M\Tr[A_t\wedge (\Tilde{d}\Tilde{A}+\Tilde{A}\wedge\Tilde{A})]
\end{equation}
where we have unsupressed the wedge product for clarity. The derivative of the time component of the connection, $A_t$, does not enter into the action. Integrating it out imposes a constraint 
\begin{equation}\label{eq:cs_constraint}
    0=\Tilde{F}=\Tilde{d}\Tilde{A}+\Tilde{A}\wedge\Tilde{A}.
\end{equation}
The equations of motion constrain the connection to be flat on $\Sigma$. We refer to this equation as Gauss's law in analogy with electromagnetism.

After integrating out $A_t$ we are left with
\begin{equation}
    S_{CS}(A)=\frac{k}{4\pi}\int_{M}\Tr{dt\wedge\partial_t\Tilde{A}\wedge\Tilde{A} }.
\end{equation}
As we outlined in Section \ref{sec:discussion}, we have two options for quantizing this theory. The first is to elevate the Poisson bracket
\begin{equation}
    \left\{\Tilde{A}_i^a(x),\Tilde{A}_j^b(y)\right\}\propto\epsilon_{ij}\delta^{ab}\delta^{(2)}(x-y)
\end{equation}
to a commutator. This phase space includes connections which violate Gauss's law. In the two-sided black hole this would be like forgetting to entangle the two horizons. There, as here, neglecting the constraint causes us to overestimate the number of physical states. In principle, we can promote the constraint to an operator statement and solve the resulting differential equation, but in practice, this is quite difficult. We referred to this approach as ``quantize then constrain" in Section \ref{sec:relational}. Instead, it is easier to ``constrain then quantify" through geometric quantization, a process described in \cite{guillemin1982geometric,Witten:1988hf,Alekseev:1994nzg}. This technique is well-suited to the phase space of flat connections on a closed manifold which can easily be equipped with a symplectic structure, so we proceed by limiting the allowed field configurations to this subset. The resulting Hilbert space depends on $\Sigma$. 

For concreteness, we consider the case of Chern-Simons on a disk $(\Sigma=D)$ \cite{Elitzur:1989nr}. We solve the constraint by parametrizing the gauge field as 
\begin{equation}
    \Tilde{A} = -\Tilde{d}U U^{-1}
\end{equation}
for some function $U:\Sigma\to G$. Substituting this into equation \eqref{eq:cs} yields a 2D CFT known as the chiral Wess-Zumino-Witten (cWZW) model
\begin{equation}
    S_D(U)=\frac{k}{4\pi}\int_{\partial D\times\R}d\varphi dt\  \Tr [U^{-1}\partial_\varphi UU^{-1}\partial_t U] +\frac{k}{12\pi}\int_{D\times\R}\Tr [(U^{-1}dU)^3].
\end{equation}
Currents in cWZW come in representations of the Kac-Moody algebra of $G$. One can show that $A_\varphi$ is one such current. The model is redundant under transformations of the form 
\begin{equation}
    U\mapsto V(\phi)UW(t)
\end{equation}
with $V$ and $W$ maps from the circle and real-line into $G$ respectively. The identification makes the phase space a quotient of the loop group $LG$, the set of maps from $S^1$ into $G$. Since the action is first-order in time-derivatives the symplectic form is
\begin{equation}
    \omega_D=\frac{k}{4\pi}\oint \Tr [(U^{-1}\delta U)\partial_\varphi (U^{-1}\delta U)]
\end{equation}
where $\delta$ is the exterior derivative on phase space. It follows that the Hilbert space is the trivial representation of the Kac-Moody algebra.  

So far, we have only considered pure Chern-Simons. This is analogous to considering pure gravity. Of course, there is no way to perform experiments in such a theory since observations require a detector. In gravity, we call this detector the observer. In this setting, it's more fitting to refer to it as a source. But, in both cases, we have to modify the constraints to impose the presence of our detectors. Just as in $AdS_2$ gravitational detectors come in representation of the isometry group $SL(2,\R)$ (c.f. Section \ref{sec:review}), Chern-Simons detectors come in representations of the gauge symmetry. Call this representation $R$ and let $\{T^a\}_{a=1}^{\dim G}$ form a basis of this representation. Given $N$ sources, each at position $x_n$, we require
\begin{equation}\label{eq:cs_source_constraint}
    \frac{k}{8\pi}\epsilon^{ij}F^a_{ij}=\sum_{n=1}^N \delta^2(x-x_n)T^a_{(n)}.
\end{equation}
The left-hand side is the Chern-Simons equivalent of $H_0=H_L-H_R$. The constraint vanishes in the pure theory and is a function of the matter representation after coupling. However, unlike equation \eqref{eq:matter_constraint}, here we are forcing the theory to put the source at a fixed position. We can, in principle, impose this constraint before quantizing. This is reminiscent of our prescription introduced in Section \ref{sec:setup} of fixing the worldline of the observer in the gravitational path integral. Just as the observer's worldline is protected from topology change, the source is protected from the influence of nearby charges. We will see that, as with gravity, these types of constraints have significant consequences for the resultant Hilbert space.

Unfortunately, imposing equation \eqref{eq:cs_source_constraint} is even harder than imposing equation \eqref{eq:cs_constraint}. In the sourceless case, we relied on the fact that the space of flat connections on a closed manifold can easily be endowed with a symplectic structure. But the connection is no longer flat. What is worse, the classical phase space of connections must be non-commutative since generally, the $\{T^a\}$'s satisfies some non-abelian algebra. We can make progress by taking an approach intermediate between the two strategies outlined above \cite{Alekseev:1994nzg}. Suppose that equation \eqref{eq:cs_source_constraint} is a statement in some quantum theory. How would we go about finding a classical theory that realizes this equation? Away from the $\{x_n\}$ we would just have vanilla Chern-Simons, but the sources require special treatment. Let $T\subset G$ be the maximal torus in $G$.\footnote{
A torus $T'$ is a compact, connected, abelian subgroup of a group $G$. They are the Lie group equivalents of Cartan subalgebras. All torii look like $T=(S^1)^k=\R^k/\Z^k$ for some integer $k$. The maximal torus $T$ is such that $T\subset T'$ implies $T=T'$. The group $G/T$ is often called the flag manifold.
} The Borel-Weil-Bott theorem tells us that there is a unique symplectic structure $\omega_R$ we can put on $G/T$ such that $R$ is the Hilbert space associated with the phase space $G/T$. We realize the constraint by placing a copy of $G/T$ at each $\{x_n\}$. Quantizing $G/T$ using $\omega_R$ gives a distinct Hilbert space living on select worldlines in the spacetime manifold. This is exactly analogous to what we referred to earlier as $\HO$.  

All that remains is to choose a functional on this phase space, which gives $T^a_{(n)}$ after quantization. If $\lambda_n$ is the highest weight in $R$ then the desired functional is the Wilson line \cite{Elitzur:1989nr,Alekseev:1994nzg,Murayama:1989we}
\begin{equation}
    S_{x_n}(A,g_n) = \int d^3x\ \delta^2(x-x_n)\Tr[\lambda_n g_n^{-1}(t)(\partial_t+A_t)g_n(t)].
\end{equation}
It is clear that $g_n \in G/T$ since the integral is invariant under $g\mapsto gh$ for $h\in T$. The analogy to gravity is clear: this is an integral over the worldline of our detector like the one we used to introduce the observer in equation \eqref{eq:action-of-observer}. The full theory is given by the action 
\begin{equation}\label{eq:d13}
    S=S_{CS}+\sum_{n=1}^N S_{x_n}(A,g_n)
\end{equation}
and an associated constraint equation
\begin{equation}
    \frac{k}{2\pi}\Tilde{F}-\sum_{n=1}^N \delta^{2}(x-x_n)g_n(t)\lambda_n g_n^{-1}(t)=0.
\end{equation}
Observables are calculated via the path integral \cite{Murayama:1989we}
\be\begin{aligned}
    W_R(\Gamma) &= \int\mathcal{D}A\ e^{iS_{CS}}\Tr_R[ \mathcal{P}\{e^{\int_\Gamma A}\}], \\
    &= \int\mathcal{D}A\mathcal{D}g\ e^{iS},
\end{aligned}\ee
where $\Gamma=x_1(t)$. We calculate an expectation value for a fixed connection $A$ by only integrating over $g$. This is analogous to how we propose calculating observables in $\Hr$ by integrating over the gravitational degrees of freedom for some fixed observer. Note that we had no way of computing this observable before introducing the source since the trace over $R$ is outside of the exponentials. 

Let us study the example of the disk for the new action in equation \eqref{eq:d13} in the presence of a source \cite{Elitzur:1989nr}. Much is the same as before. The action is modified by the addition of a $\lambda$-dependent term
\begin{equation}
    S_D \mapsto S_D+\frac{1}{2\pi}\int_{\partial D\times\R}\Tr [\lambda U^{-1}\partial_t U].
\end{equation}
We can continue parametrizing the gauge field using $U$, only now we need
\begin{equation}
    U\mapsto U\exp{\frac{1}{k}g(t)\lambda g^{-1}(t)\varphi}
\end{equation}
to satisfy the equations of motion. The global symmetry on $U$ also remains intact, although now we also need $W(t)$ to commute with $\lambda$. The symplectic form is also slightly modified
\begin{equation}
    \omega_D \mapsto \omega_D+\frac{1}{2\pi}\oint\Tr[\lambda(U^{-1}\delta U)^2].
\end{equation}
The theory is no longer cWZW, but quantization still gives a representation of the Kac-Moody algebra, now with the highest weight $\lambda$. Adding multiple sources gives a tensor product of the associated representations. 

The disk is a slightly misleading example. It is not always true that the addition of a source changes the Hilbert space. For example, the Hilbert space of Chern-Simons on a sphere does not change if we add one or two marked points, assuming we choose representations compatible with charge conservation. But the core point remains. The introduction of a detector changes what defines physical states. That detector should be considered a protected degree of freedom that is not subject to the usual dynamics of the theory. In gravity, as in the Chern-Simons examples outlined above, this can lead to a marked difference in the size and structure of the Hilbert space. 

\bibliographystyle{jhep}
\bibliography{references}

\providecommand{\href}[2]{#2}\begingroup\raggedright\begin{thebibliography}{100}

\bibitem{Susskind:1994vu}
L.~Susskind, {\it {The World as a hologram}},  {\em J. Math. Phys.} {\bf 36} (1995) 6377--6396, [\href{http://arxiv.org/abs/hep-th/9409089}{{\tt hep-th/9409089}}].

\bibitem{tHooft:1993dmi}
G.~'t~Hooft, {\it {Dimensional reduction in quantum gravity}},  {\em Conf. Proc. C} {\bf 930308} (1993) 284--296, [\href{http://arxiv.org/abs/gr-qc/9310026}{{\tt gr-qc/9310026}}].

\bibitem{Maldacena:1997re}
J.~M. Maldacena, {\it {The Large N limit of superconformal field theories and supergravity}},  {\em Adv. Theor. Math. Phys.} {\bf 2} (1998) 231--252, [\href{http://arxiv.org/abs/hep-th/9711200}{{\tt hep-th/9711200}}].

\bibitem{Witten:1998qj}
E.~Witten, {\it {Anti-de Sitter space and holography}},  {\em Adv. Theor. Math. Phys.} {\bf 2} (1998) 253--291, [\href{http://arxiv.org/abs/hep-th/9802150}{{\tt hep-th/9802150}}].

\bibitem{Aharony:1999ti}
O.~Aharony, S.~S. Gubser, J.~M. Maldacena, H.~Ooguri, and Y.~Oz, {\it {Large N field theories, string theory and gravity}},  {\em Phys. Rept.} {\bf 323} (2000) 183--386, [\href{http://arxiv.org/abs/hep-th/9905111}{{\tt hep-th/9905111}}].

\bibitem{gibbons}
G.~W. Gibbons and S.~W. Hawking, {\it Action integrals and partition functions in quantum gravity},  {\em Phys. Rev. D} {\bf 15} (May, 1977) 2752--2756.

\bibitem{Hartle:1983ai}
J.~B. Hartle and S.~W. Hawking, {\it {Wave Function of the Universe}},  {\em Phys. Rev. D} {\bf 28} (1983) 2960--2975.

\bibitem{Maldacena:2004rf}
J.~M. Maldacena and L.~Maoz, {\it {Wormholes in AdS}},  {\em JHEP} {\bf 02} (2004) 053, [\href{http://arxiv.org/abs/hep-th/0401024}{{\tt hep-th/0401024}}].

\bibitem{Lewkowycz:2013nqa}
A.~Lewkowycz and J.~Maldacena, {\it {Generalized gravitational entropy}},  {\em JHEP} {\bf 08} (2013) 090, [\href{http://arxiv.org/abs/1304.4926}{{\tt arXiv:1304.4926}}].

\bibitem{Faulkner:2013ana}
T.~Faulkner, A.~Lewkowycz, and J.~Maldacena, {\it {Quantum corrections to holographic entanglement entropy}},  {\em JHEP} {\bf 11} (2013) 074, [\href{http://arxiv.org/abs/1307.2892}{{\tt arXiv:1307.2892}}].

\bibitem{Marolf:2020xie}
D.~Marolf and H.~Maxfield, {\it {Transcending the ensemble: baby universes, spacetime wormholes, and the order and disorder of black hole information}},  {\em JHEP} {\bf 08} (2020) 044, [\href{http://arxiv.org/abs/2002.08950}{{\tt arXiv:2002.08950}}].

\bibitem{Saad:2018bqo}
P.~Saad, S.~H. Shenker, and D.~Stanford, {\it {A semiclassical ramp in SYK and in gravity}},  \href{http://arxiv.org/abs/1806.06840}{{\tt arXiv:1806.06840}}.

\bibitem{Saad:2019lba}
P.~Saad, S.~H. Shenker, and D.~Stanford, {\it {JT gravity as a matrix integral}},  \href{http://arxiv.org/abs/1903.11115}{{\tt arXiv:1903.11115}}.

\bibitem{Saad:2021rcu}
P.~Saad, S.~H. Shenker, D.~Stanford, and S.~Yao, {\it {Wormholes without averaging}},  {\em JHEP} {\bf 09} (2024) 133, [\href{http://arxiv.org/abs/2103.16754}{{\tt arXiv:2103.16754}}].

\bibitem{Saad:2021uzi}
P.~Saad, S.~H. Shenker, and S.~Yao, {\it {Comments on wormholes and factorization}},  {\em JHEP} {\bf 10} (2024) 076, [\href{http://arxiv.org/abs/2107.13130}{{\tt arXiv:2107.13130}}].

\bibitem{Penington:2019kki}
G.~Penington, S.~H. Shenker, D.~Stanford, and Z.~Yang, {\it {Replica wormholes and the black hole interior}},  {\em JHEP} {\bf 03} (2022) 205, [\href{http://arxiv.org/abs/1911.11977}{{\tt arXiv:1911.11977}}].

\bibitem{Almheiri:2019qdq}
A.~Almheiri, T.~Hartman, J.~Maldacena, E.~Shaghoulian, and A.~Tajdini, {\it {Replica Wormholes and the Entropy of Hawking Radiation}},  {\em JHEP} {\bf 05} (2020) 013, [\href{http://arxiv.org/abs/1911.12333}{{\tt arXiv:1911.12333}}].

\bibitem{Blommaert:2021fob}
A.~Blommaert, L.~V. Iliesiu, and J.~Kruthoff, {\it {Gravity factorized}},  {\em JHEP} {\bf 09} (2022) 080, [\href{http://arxiv.org/abs/2111.07863}{{\tt arXiv:2111.07863}}].

\bibitem{Blommaert:2022ucs}
A.~Blommaert, L.~V. Iliesiu, and J.~Kruthoff, {\it {Alpha states demystified \textemdash{} towards microscopic models of AdS$_{2}$ holography}},  {\em JHEP} {\bf 08} (2022) 071, [\href{http://arxiv.org/abs/2203.07384}{{\tt arXiv:2203.07384}}].

\bibitem{Boruch:2023trc}
J.~Boruch, L.~V. Iliesiu, and C.~Yan, {\it {Constructing all BPS black hole microstates from the gravitational path integral}},  {\em JHEP} {\bf 09} (2024) 058, [\href{http://arxiv.org/abs/2307.13051}{{\tt arXiv:2307.13051}}].

\bibitem{Iliesiu:2021are}
L.~V. Iliesiu, M.~Kologlu, and G.~J. Turiaci, {\it {Supersymmetric indices factorize}},  {\em JHEP} {\bf 05} (2023) 032, [\href{http://arxiv.org/abs/2107.09062}{{\tt arXiv:2107.09062}}].

\bibitem{Iliesiu:2022kny}
L.~V. Iliesiu, S.~Murthy, and G.~J. Turiaci, {\it {Black hole microstate counting from the gravitational path integral}},  \href{http://arxiv.org/abs/2209.13602}{{\tt arXiv:2209.13602}}.

\bibitem{Marolf:2022ybi}
D.~Marolf, {\it {Gravitational thermodynamics without the conformal factor problem: partition functions and Euclidean saddles from Lorentzian path integrals}},  {\em JHEP} {\bf 07} (2022) 108, [\href{http://arxiv.org/abs/2203.07421}{{\tt arXiv:2203.07421}}].

\bibitem{Colafranceschi:2023moh}
E.~Colafranceschi, X.~Dong, D.~Marolf, and Z.~Wang, {\it {Algebras and Hilbert spaces from gravitational path integrals. Understanding Ryu-Takayanagi/HRT as entropy without AdS/CFT}},  {\em JHEP} {\bf 10} (2024) 063, [\href{http://arxiv.org/abs/2310.02189}{{\tt arXiv:2310.02189}}].

\bibitem{Marolf:2024jze}
D.~Marolf, {\it {On the nature of ensembles from gravitational path integrals}},  \href{http://arxiv.org/abs/2407.04625}{{\tt arXiv:2407.04625}}.

\bibitem{Balasubramanian:2022gmo}
V.~Balasubramanian, A.~Lawrence, J.~M. Magan, and M.~Sasieta, {\it {Microscopic origin of the entropy of black holes in general relativity}},  \href{http://arxiv.org/abs/2212.02447}{{\tt arXiv:2212.02447}}.

\bibitem{Balasubramanian:2022lnw}
V.~Balasubramanian, A.~Lawrence, J.~M. Magan, and M.~Sasieta, {\it {Microscopic origin of the entropy of astrophysical black holes}},  \href{http://arxiv.org/abs/2212.08623}{{\tt arXiv:2212.08623}}.

\bibitem{Sasieta:2022ksu}
M.~Sasieta, {\it {Wormholes from heavy operator statistics in AdS/CFT}},  {\em JHEP} {\bf 03} (2023) 158, [\href{http://arxiv.org/abs/2211.11794}{{\tt arXiv:2211.11794}}].

\bibitem{Ryu2006a}
S.~Ryu and T.~Takayanagi, {\it {Aspects of Holographic Entanglement Entropy}},  {\em JHEP} {\bf 08} (2006) 045, [\href{http://arxiv.org/abs/hep-th/0605073}{{\tt hep-th/0605073}}].

\bibitem{Ryu2006b}
S.~Ryu and T.~Takayanagi, {\it Holographic derivation of entanglement entropy from the anti--de sitter space/conformal field theory correspondence},  {\em Physical review letters} {\bf 96} (2006), no.~18 181602.

\bibitem{Hubeny:2007xt}
V.~E. Hubeny, M.~Rangamani, and T.~Takayanagi, {\it {A Covariant holographic entanglement entropy proposal}},  {\em JHEP} {\bf 07} (2007) 062, [\href{http://arxiv.org/abs/0705.0016}{{\tt arXiv:0705.0016}}].

\bibitem{Engelhardt:2014gca}
N.~Engelhardt and A.~C. Wall, {\it {Quantum Extremal Surfaces: Holographic Entanglement Entropy beyond the Classical Regime}},  {\em JHEP} {\bf 01} (2015) 073, [\href{http://arxiv.org/abs/1408.3203}{{\tt arXiv:1408.3203}}].

\bibitem{Almheiri:2012rt}
A.~Almheiri, D.~Marolf, J.~Polchinski, and J.~Sully, {\it {Black Holes: Complementarity or Firewalls?}},  {\em JHEP} {\bf 02} (2013) 062, [\href{http://arxiv.org/abs/1207.3123}{{\tt arXiv:1207.3123}}].

\bibitem{Susskind:2012rm}
L.~Susskind, {\it {Singularities, Firewalls, and Complementarity}},  \href{http://arxiv.org/abs/1208.3445}{{\tt arXiv:1208.3445}}.

\bibitem{Almheiri:2013hfa}
A.~Almheiri, D.~Marolf, J.~Polchinski, D.~Stanford, and J.~Sully, {\it {An Apologia for Firewalls}},  {\em JHEP} {\bf 09} (2013) 018, [\href{http://arxiv.org/abs/1304.6483}{{\tt arXiv:1304.6483}}].

\bibitem{Harlow:2013tf}
D.~Harlow and P.~Hayden, {\it {Quantum Computation vs. Firewalls}},  {\em JHEP} {\bf 06} (2013) 085, [\href{http://arxiv.org/abs/1301.4504}{{\tt arXiv:1301.4504}}].

\bibitem{Stanford:2022fdt}
D.~Stanford and Z.~Yang, {\it {Firewalls from wormholes}},  \href{http://arxiv.org/abs/2208.01625}{{\tt arXiv:2208.01625}}.

\bibitem{IliLev24}
L.~V. Iliesiu, A.~Levine, H.~W. Lin, H.~Maxfield, and M.~Mezei, {\it On the non-perturbative bulk hilbert space of jt gravity},  \href{http://arxiv.org/abs/2403.08696}{{\tt arXiv:2403.08696}}.

\bibitem{Blommaert:2024ftn}
A.~Blommaert, C.-H. Chen, and Y.~Nomura, {\it {Firewalls at exponentially late times}},  {\em JHEP} {\bf 10} (2024) 131, [\href{http://arxiv.org/abs/2403.07049}{{\tt arXiv:2403.07049}}].

\bibitem{LevSha22}
A.~Levine and E.~Shaghoulian, {\it Encoding beyond cosmological horizons in de sitter jt gravity},  \href{http://arxiv.org/abs/2204.08503}{{\tt arXiv:2204.08503}}.

\bibitem{McInnes:2004nx}
B.~McInnes, {\it {Answering a basic objection to bang / crunch holography}},  {\em JHEP} {\bf 10} (2004) 018, [\href{http://arxiv.org/abs/hep-th/0407189}{{\tt hep-th/0407189}}].

\bibitem{Cooper:2018cmb}
S.~Cooper, M.~Rozali, B.~Swingle, M.~Van~Raamsdonk, C.~Waddell, and D.~Wakeham, {\it {Black hole microstate cosmology}},  {\em JHEP} {\bf 07} (2019) 065, [\href{http://arxiv.org/abs/1810.10601}{{\tt arXiv:1810.10601}}].

\bibitem{Antonini:2019qkt}
S.~Antonini and B.~Swingle, {\it {Cosmology at the end of the world}},  {\em Nature Phys.} {\bf 16} (2020), no.~8 881--886, [\href{http://arxiv.org/abs/1907.06667}{{\tt arXiv:1907.06667}}].

\bibitem{VanRaamsdonk:2020tlr}
M.~Van~Raamsdonk, {\it {Comments on wormholes, ensembles, and cosmology}},  {\em JHEP} {\bf 12} (2021) 156, [\href{http://arxiv.org/abs/2008.02259}{{\tt arXiv:2008.02259}}].

\bibitem{VanRaamsdonk:2021qgv}
M.~Van~Raamsdonk, {\it {Cosmology from confinement?}},  {\em JHEP} {\bf 03} (2022) 039, [\href{http://arxiv.org/abs/2102.05057}{{\tt arXiv:2102.05057}}].

\bibitem{Antonini:2022blk}
S.~Antonini, P.~Simidzija, B.~Swingle, and M.~Van~Raamsdonk, {\it {Cosmology from the vacuum}},  \href{http://arxiv.org/abs/2203.11220}{{\tt arXiv:2203.11220}}.

\bibitem{Antonini:2022ptt}
S.~Antonini, P.~Simidzija, B.~Swingle, and M.~Van~Raamsdonk, {\it {Accelerating Cosmology from a Holographic Wormhole}},  {\em Phys. Rev. Lett.} {\bf 130} (2023), no.~22 221601, [\href{http://arxiv.org/abs/2206.14821}{{\tt arXiv:2206.14821}}].

\bibitem{Antonini:2022fna}
S.~Antonini, P.~Simidzija, B.~Swingle, M.~Van~Raamsdonk, and C.~Waddell, {\it {Accelerating cosmology from \ensuremath{\Lambda} \ensuremath{<} 0 gravitational effective field theory}},  {\em JHEP} {\bf 05} (2023) 203, [\href{http://arxiv.org/abs/2212.00050}{{\tt arXiv:2212.00050}}].

\bibitem{Sahu:2023fbx}
A.~Sahu, P.~Simidzija, and M.~Van~Raamsdonk, {\it {Bubbles of cosmology in AdS/CFT}},  \href{http://arxiv.org/abs/2306.13143}{{\tt arXiv:2306.13143}}.

\bibitem{Chakravarty:2024bna}
J.~Chakravarty, A.~Maloney, K.~Namjou, and S.~F. Ross, {\it {A new observable for holographic cosmology}},  {\em JHEP} {\bf 10} (2024) 184, [\href{http://arxiv.org/abs/2407.04781}{{\tt arXiv:2407.04781}}].

\bibitem{Antonini:2024bbm}
S.~Antonini and L.~G.~C. Bariuan, {\it {Magnetic braneworlds: cosmology and wormholes}},  {\em JHEP} {\bf 09} (2024) 070, [\href{http://arxiv.org/abs/2405.18465}{{\tt arXiv:2405.18465}}].

\bibitem{Betzios:2024oli}
P.~Betzios and O.~Papadoulaki, {\it {Inflationary Cosmology from Anti-de Sitter Wormholes}},  {\em Phys. Rev. Lett.} {\bf 133} (2024), no.~2 021501, [\href{http://arxiv.org/abs/2403.17046}{{\tt arXiv:2403.17046}}].

\bibitem{VanRaamsdonk:2024sdp}
M.~Van~Raamsdonk and C.~Waddell, {\it {Holographic motivations and observational evidence for decreasing dark energy}},  \href{http://arxiv.org/abs/2406.02688}{{\tt arXiv:2406.02688}}.

\bibitem{Antonini:2024mci}
S.~Antonini and P.~Rath, {\it {Do holographic CFT states have unique semiclassical bulk duals?}},  \href{http://arxiv.org/abs/2408.02720}{{\tt arXiv:2408.02720}}.

\bibitem{Sahu:2024ccg}
A.~Sahu and M.~Van~Raamsdonk, {\it {Holographic black hole cosmologies}},  \href{http://arxiv.org/abs/2411.14673}{{\tt arXiv:2411.14673}}.

\bibitem{Chandrasekaran:2022cip}
V.~Chandrasekaran, R.~Longo, G.~Penington, and E.~Witten, {\it {An algebra of observables for de Sitter space}},  {\em JHEP} {\bf 02} (2023) 082, [\href{http://arxiv.org/abs/2206.10780}{{\tt arXiv:2206.10780}}].

\bibitem{Strominger:2001pn}
A.~Strominger, {\it {The dS / CFT correspondence}},  {\em JHEP} {\bf 10} (2001) 034, [\href{http://arxiv.org/abs/hep-th/0106113}{{\tt hep-th/0106113}}].

\bibitem{Coleman:2021nor}
E.~Coleman, E.~A. Mazenc, V.~Shyam, E.~Silverstein, R.~M. Soni, G.~Torroba, and S.~Yang, {\it {De Sitter microstates from T$ \overline{T} $ + \ensuremath{\Lambda}$_{2}$ and the Hawking-Page transition}},  {\em JHEP} {\bf 07} (2022) 140, [\href{http://arxiv.org/abs/2110.14670}{{\tt arXiv:2110.14670}}].

\bibitem{dsds}
X.~Dong, E.~Silverstein, and G.~Torroba, {\it {De Sitter Holography and Entanglement Entropy}},  {\em JHEP} {\bf 07} (2018) 050, [\href{http://arxiv.org/abs/1804.08623}{{\tt arXiv:1804.08623}}].

\bibitem{Araujo-Regado:2022gvw}
G.~Araujo-Regado, R.~Khan, and A.~C. Wall, {\it {Cauchy slice holography: a new AdS/CFT dictionary}},  {\em JHEP} {\bf 03} (2023) 026, [\href{http://arxiv.org/abs/2204.00591}{{\tt arXiv:2204.00591}}].

\bibitem{McFadden:2009fg}
P.~McFadden and K.~Skenderis, {\it {Holography for Cosmology}},  {\em Phys. Rev. D} {\bf 81} (2010) 021301, [\href{http://arxiv.org/abs/0907.5542}{{\tt arXiv:0907.5542}}].

\bibitem{Anninos:2011af}
D.~Anninos, S.~A. Hartnoll, and D.~M. Hofman, {\it {Static Patch Solipsism: Conformal Symmetry of the de Sitter Worldline}},  {\em Class. Quant. Grav.} {\bf 29} (2012) 075002, [\href{http://arxiv.org/abs/1109.4942}{{\tt arXiv:1109.4942}}].

\bibitem{Anninos:2011zn}
D.~Anninos, T.~Anous, I.~Bredberg, and G.~S. Ng, {\it {Incompressible Fluids of the de Sitter Horizon and Beyond}},  {\em JHEP} {\bf 05} (2012) 107, [\href{http://arxiv.org/abs/1110.3792}{{\tt arXiv:1110.3792}}].

\bibitem{Leutheusser:2022bgi}
S.~Leutheusser and H.~Liu, {\it {Subregion-subalgebra duality: emergence of space and time in holography}},  \href{http://arxiv.org/abs/2212.13266}{{\tt arXiv:2212.13266}}.

\bibitem{Witten:2021unn}
E.~Witten, {\it {Gravity and the crossed product}},  {\em JHEP} {\bf 10} (2022) 008, [\href{http://arxiv.org/abs/2112.12828}{{\tt arXiv:2112.12828}}].

\bibitem{Chandrasekaran:2022eqq}
V.~Chandrasekaran, G.~Penington, and E.~Witten, {\it {Large N algebras and generalized entropy}},  {\em JHEP} {\bf 04} (2023) 009, [\href{http://arxiv.org/abs/2209.10454}{{\tt arXiv:2209.10454}}].

\bibitem{Witten:2023qsv}
E.~Witten, {\it {Algebras, regions, and observers.}},  {\em Proc. Symp. Pure Math.} {\bf 107} (2024) 247--276, [\href{http://arxiv.org/abs/2303.02837}{{\tt arXiv:2303.02837}}].

\bibitem{Witten:2023xze}
E.~Witten, {\it {A background-independent algebra in quantum gravity}},  {\em JHEP} {\bf 03} (2024) 077, [\href{http://arxiv.org/abs/2308.03663}{{\tt arXiv:2308.03663}}].

\bibitem{Gesteau:2023hbq}
E.~Gesteau, {\it {Large $N$ von Neumann algebras and the renormalization of Newton's constant}},  \href{http://arxiv.org/abs/2302.01938}{{\tt arXiv:2302.01938}}.

\bibitem{Kolchmeyer:2024fly}
D.~K. Kolchmeyer and H.~Liu, {\it {Chaos and the Emergence of the Cosmological Horizon}},  \href{http://arxiv.org/abs/2411.08090}{{\tt arXiv:2411.08090}}.

\bibitem{Chen:2024rpx}
C.-H. Chen and G.~Penington, {\it {A clock is just a way to tell the time: gravitational algebras in cosmological spacetimes}},  \href{http://arxiv.org/abs/2406.02116}{{\tt arXiv:2406.02116}}.

\bibitem{Kudler-Flam:2024psh}
J.~Kudler-Flam, S.~Leutheusser, and G.~Satishchandran, {\it {Algebraic Observational Cosmology}},  \href{http://arxiv.org/abs/2406.01669}{{\tt arXiv:2406.01669}}.

\bibitem{DeVuyst:2024pop}
J.~De~Vuyst, S.~Eccles, P.~A. Hoehn, and J.~Kirklin, {\it {Gravitational entropy is observer-dependent}},  \href{http://arxiv.org/abs/2405.00114}{{\tt arXiv:2405.00114}}.

\bibitem{Boruch:2024kvv}
J.~Boruch, L.~V. Iliesiu, G.~Lin, and C.~Yan, {\it {How the Hilbert space of two-sided black holes factorises}},  \href{http://arxiv.org/abs/2406.04396}{{\tt arXiv:2406.04396}}.

\bibitem{Antonini:2023hdh}
S.~Antonini, M.~Sasieta, and B.~Swingle, {\it {Cosmology from random entanglement}},  {\em JHEP} {\bf 11} (2023) 188, [\href{http://arxiv.org/abs/2307.14416}{{\tt arXiv:2307.14416}}].

\bibitem{Usatyuk:2024mzs}
M.~Usatyuk, Z.-Y. Wang, and Y.~Zhao, {\it {Closed universes in two dimensional gravity}},  {\em SciPost Phys.} {\bf 17} (2024), no.~2 051, [\href{http://arxiv.org/abs/2402.00098}{{\tt arXiv:2402.00098}}].

\bibitem{Usatyuk:2024isz}
M.~Usatyuk and Y.~Zhao, {\it {Closed universes, factorization, and ensemble averaging}},  \href{http://arxiv.org/abs/2403.13047}{{\tt arXiv:2403.13047}}.

\bibitem{McNamara:2020uza}
J.~McNamara and C.~Vafa, {\it {Baby Universes, Holography, and the Swampland}},  \href{http://arxiv.org/abs/2004.06738}{{\tt arXiv:2004.06738}}.

\bibitem{DiValentino:2019qzk}
E.~Di~Valentino, A.~Melchiorri, and J.~Silk, {\it {Planck evidence for a closed Universe and a possible crisis for cosmology}},  {\em Nature Astron.} {\bf 4} (2019), no.~2 196--203, [\href{http://arxiv.org/abs/1911.02087}{{\tt arXiv:1911.02087}}].

\bibitem{Handley:2019tkm}
W.~Handley, {\it {Curvature tension: evidence for a closed universe}},  {\em Phys. Rev. D} {\bf 103} (2021), no.~4 L041301, [\href{http://arxiv.org/abs/1908.09139}{{\tt arXiv:1908.09139}}].

\bibitem{Planck:2018vyg}
{\bf Planck} Collaboration, N.~Aghanim et~al., {\it {Planck 2018 results. VI. Cosmological parameters}},  {\em Astron. Astrophys.} {\bf 641} (2020) A6, [\href{http://arxiv.org/abs/1807.06209}{{\tt arXiv:1807.06209}}]. [Erratum: Astron.Astrophys. 652, C4 (2021)].

\bibitem{Jackiw:1984je}
R.~Jackiw, {\it {Lower Dimensional Gravity}},  {\em Nucl. Phys. B} {\bf 252} (1985) 343--356.

\bibitem{Teitelboim:1983ux}
C.~Teitelboim, {\it {Gravitation and Hamiltonian Structure in Two Space-Time Dimensions}},  {\em Phys. Lett. B} {\bf 126} (1983) 41--45.

\bibitem{Hsin:2020mfa}
P.-S. Hsin, L.~V. Iliesiu, and Z.~Yang, {\it {A violation of global symmetries from replica wormholes and the fate of black hole remnants}},  {\em Class. Quant. Grav.} {\bf 38} (2021), no.~19 194004, [\href{http://arxiv.org/abs/2011.09444}{{\tt arXiv:2011.09444}}].

\bibitem{Akers:2021fut}
C.~Akers and G.~Penington, {\it {Quantum minimal surfaces from quantum error correction}},  {\em SciPost Phys.} {\bf 12} (2022), no.~5 157, [\href{http://arxiv.org/abs/2109.14618}{{\tt arXiv:2109.14618}}].

\bibitem{Akers:2022qdl}
C.~Akers, N.~Engelhardt, D.~Harlow, G.~Penington, and S.~Vardhan, {\it {The black hole interior from non-isometric codes and complexity}},  \href{http://arxiv.org/abs/2207.06536}{{\tt arXiv:2207.06536}}.

\bibitem{Antonini:2024yif}
S.~Antonini, V.~Balasubramanian, N.~Bao, C.~Cao, and W.~Chemissany, {\it {Non-isometry, State-Dependence and Holography}},  \href{http://arxiv.org/abs/2411.07296}{{\tt arXiv:2411.07296}}.

\bibitem{Penington:2019npb}
G.~Penington, {\it {Entanglement Wedge Reconstruction and the Information Paradox}},  {\em JHEP} {\bf 09} (2020) 002, [\href{http://arxiv.org/abs/1905.08255}{{\tt arXiv:1905.08255}}].

\bibitem{Hartman:2020khs}
T.~Hartman, Y.~Jiang, and E.~Shaghoulian, {\it {Islands in cosmology}},  {\em JHEP} {\bf 11} (2020) 111, [\href{http://arxiv.org/abs/2008.01022}{{\tt arXiv:2008.01022}}].

\bibitem{Papadodimas:2012aq}
K.~Papadodimas and S.~Raju, {\it {An Infalling Observer in AdS/CFT}},  {\em JHEP} {\bf 10} (2013) 212, [\href{http://arxiv.org/abs/1211.6767}{{\tt arXiv:1211.6767}}].

\bibitem{Papadodimas:2013jku}
K.~Papadodimas and S.~Raju, {\it {State-Dependent Bulk-Boundary Maps and Black Hole Complementarity}},  {\em Phys. Rev. D} {\bf 89} (2014), no.~8 086010, [\href{http://arxiv.org/abs/1310.6335}{{\tt arXiv:1310.6335}}].

\bibitem{MarPol15}
D.~Marolf and J.~Polchinski, {\it Violations of the born rule in cool state-dependent horizons},  \href{http://arxiv.org/abs/1506.01337}{{\tt arXiv:1506.01337}}.

\bibitem{StaSus14}
D.~Stanford and L.~Susskind, {\it Complexity and shock wave geometries},  {\em Phys. Rev. D} {\bf 90} (2014) 126007, [\href{http://arxiv.org/abs/1406.2678}{{\tt arXiv:1406.2678}}].

\bibitem{Harlow:2020fpj}
D.~Harlow and D.~Jafferis, {\it The factorization problem in jackiw-teitelboim gravity},  {\em Journal of High Energy Physics} {\bf 2020} (2020), no.~2 1--32.

\bibitem{Mertens:2022irh}
T.~G. Mertens and G.~J. Turiaci, {\it {Solvable models of quantum black holes: a review on Jackiw\textendash{}Teitelboim gravity}},  {\em Living Rev. Rel.} {\bf 26} (2023), no.~1 4, [\href{http://arxiv.org/abs/2210.10846}{{\tt arXiv:2210.10846}}].

\bibitem{Penington:2023dql}
G.~Penington and E.~Witten, {\it {Algebras and States in JT Gravity}},  \href{http://arxiv.org/abs/2301.07257}{{\tt arXiv:2301.07257}}.

\bibitem{Kolchmeyer:2023gwa}
D.~K. Kolchmeyer, {\it {von Neumann algebras in JT gravity}},  {\em JHEP} {\bf 06} (2023) 067, [\href{http://arxiv.org/abs/2303.04701}{{\tt arXiv:2303.04701}}].

\bibitem{Iliesiu:2024cnh}
L.~V. Iliesiu, A.~Levine, H.~W. Lin, H.~Maxfield, and M.~Mezei, {\it {On the non-perturbative bulk Hilbert space of JT gravity}},  {\em JHEP} {\bf 10} (2024) 220, [\href{http://arxiv.org/abs/2403.08696}{{\tt arXiv:2403.08696}}].

\bibitem{Yang:2018gdb}
Z.~Yang, {\it {The Quantum Gravity Dynamics of Near Extremal Black Holes}},  {\em JHEP} {\bf 05} (2019) 205, [\href{http://arxiv.org/abs/1809.08647}{{\tt arXiv:1809.08647}}].

\bibitem{Kitaev:2018wpr}
A.~Kitaev and S.~J. Suh, {\it {Statistical mechanics of a two-dimensional black hole}},  {\em JHEP} {\bf 05} (2019) 198, [\href{http://arxiv.org/abs/1808.07032}{{\tt arXiv:1808.07032}}].

\bibitem{Isham:1992ms}
C.~J. Isham, {\it {Canonical quantum gravity and the problem of time}},  {\em NATO Sci. Ser. C} {\bf 409} (1993) 157--287, [\href{http://arxiv.org/abs/gr-qc/9210011}{{\tt gr-qc/9210011}}].

\bibitem{Held:2024rmg}
J.~Held and H.~Maxfield, {\it {The Hilbert space of de Sitter JT: a case study for canonical methods in quantum gravity}},  \href{http://arxiv.org/abs/2410.14824}{{\tt arXiv:2410.14824}}.

\bibitem{JafKol22}
D.~L. Jafferis, D.~K. Kolchmeyer, B.~Mukhametzhanov, and J.~Sonner, {\it Jt gravity with matter, generalized eth, and random matrices},  \href{http://arxiv.org/abs/2209.02131}{{\tt arXiv:2209.02131}}.

\bibitem{Stanford:2019vob}
D.~Stanford and E.~Witten, {\it {JT gravity and the ensembles of random matrix theory}},  {\em Adv. Theor. Math. Phys.} {\bf 24} (2020), no.~6 1475--1680, [\href{http://arxiv.org/abs/1907.03363}{{\tt arXiv:1907.03363}}].

\bibitem{Iliesiu:2021ari}
L.~V. Iliesiu, M.~Mezei, and G.~S\'arosi, {\it {The volume of the black hole interior at late times}},  {\em JHEP} {\bf 07} (2022) 073, [\href{http://arxiv.org/abs/2107.06286}{{\tt arXiv:2107.06286}}].

\bibitem{Chen:2020tes}
Y.~Chen, V.~Gorbenko, and J.~Maldacena, {\it {Bra-ket wormholes in gravitationally prepared states}},  {\em JHEP} {\bf 02} (2021) 009, [\href{http://arxiv.org/abs/2007.16091}{{\tt arXiv:2007.16091}}].

\bibitem{Akers:2023fqr}
C.~Akers, A.~Levine, G.~Penington, and E.~Wildenhain, {\it {One-shot holography}},  \href{http://arxiv.org/abs/2307.13032}{{\tt arXiv:2307.13032}}.

\bibitem{DHoker:2002nbb}
E.~D'Hoker and D.~Z. Freedman, {\it {Supersymmetric gauge theories and the AdS / CFT correspondence}},  in {\em {Theoretical Advanced Study Institute in Elementary Particle Physics (TASI 2001): Strings, Branes and EXTRA Dimensions}}, pp.~3--158, 1, 2002.
\newblock \href{http://arxiv.org/abs/hep-th/0201253}{{\tt hep-th/0201253}}.

\bibitem{Sus15}
L.~Susskind, {\it The typical-state paradox: Diagnosing horizons with complexity},  \href{http://arxiv.org/abs/1507.02287}{{\tt arXiv:1507.02287}}.

\bibitem{Sus20}
L.~Susskind, {\it Black holes at exp-time},  \href{http://arxiv.org/abs/2006.01280}{{\tt arXiv:2006.01280}}.

\bibitem{Jafferis:2022wez}
D.~L. Jafferis, D.~K. Kolchmeyer, B.~Mukhametzhanov, and J.~Sonner, {\it {Jackiw-Teitelboim gravity with matter, generalized eigenstate thermalization hypothesis, and random matrices}},  {\em Phys. Rev. D} {\bf 108} (2023), no.~6 066015, [\href{http://arxiv.org/abs/2209.02131}{{\tt arXiv:2209.02131}}].

\bibitem{Stanford:2020wkf}
D.~Stanford, {\it {More quantum noise from wormholes}},  \href{http://arxiv.org/abs/2008.08570}{{\tt arXiv:2008.08570}}.

\bibitem{deBoer:2023vsm}
J.~de~Boer, D.~Liska, B.~Post, and M.~Sasieta, {\it {A principle of maximum ignorance for semiclassical gravity}},  {\em JHEP} {\bf 2024} (2024) 003, [\href{http://arxiv.org/abs/2311.08132}{{\tt arXiv:2311.08132}}].

\bibitem{Climent:2024trz}
A.~Climent, R.~Emparan, J.~M. Magan, M.~Sasieta, and A.~Vilar~L\'opez, {\it {Universal construction of black hole microstates}},  {\em Phys. Rev. D} {\bf 109} (2024), no.~8 086024, [\href{http://arxiv.org/abs/2401.08775}{{\tt arXiv:2401.08775}}].

\bibitem{wheeler}
J.~A. Wheeler, {\it Superspace and the nature of quantum geometrodynamics.},  {\em pp 615-724 of Topics in Nonlinear Physics. Zabusky, Norman J. (ed.). New York, Springer-Verlag New York, Inc., 1968.} (10, 1969).

\bibitem{dewitt}
B.~S. DeWitt, {\it Quantum theory of gravity. i. the canonical theory},  {\em Phys. Rev.} {\bf 160} (Aug, 1967) 1113--1148.

\bibitem{pagewooters}
D.~N. Page and W.~K. Wootters, {\it Evolution without evolution: Dynamics described by stationary observables},  {\em Phys. Rev. D} {\bf 27} (Jun, 1983) 2885--2892.

\bibitem{Hoehn:2019fsy}
P.~A. Hoehn, A.~R.~H. Smith, and M.~P.~E. Lock, {\it {Trinity of relational quantum dynamics}},  {\em Phys. Rev. D} {\bf 104} (2021), no.~6 066001, [\href{http://arxiv.org/abs/1912.00033}{{\tt arXiv:1912.00033}}].

\bibitem{Sha21}
E.~Shaghoulian, {\it The central dogma and cosmological horizons},  \href{http://arxiv.org/abs/2110.13210}{{\tt arXiv:2110.13210}}.

\bibitem{Sha23}
E.~Shaghoulian, {\it Quantum gravity and the measurement problem in quantum mechanics},  \href{http://arxiv.org/abs/2305.10635}{{\tt arXiv:2305.10635}}.

\bibitem{Witten:1988hf}
E.~Witten, {\it {Quantum Field Theory and the Jones Polynomial}},  {\em Commun. Math. Phys.} {\bf 121} (1989) 351--399.

\bibitem{Frohlich:1989gr}
J.~Frohlich and C.~King, {\it {The {Chern-Simons} Theory and Knot Polynomials}},  {\em Commun. Math. Phys.} {\bf 126} (1989) 167.

\bibitem{guillemin1982geometric}
V.~Guillemin and S.~Sternberg, {\it Geometric quantization and multiplicities of group representations},  {\em Inventiones mathematicae} {\bf 67} (1982), no.~3 515--538.

\bibitem{Murayama:1989we}
H.~Murayama, {\it {Explicit Quantization of the {Chern-Simons} Action}},  {\em Z. Phys. C} {\bf 48} (1990) 79--88.

\bibitem{Alekseev:1994nzg}
A.~Y. Alekseev and A.~Z. Malkin, {\it {Symplectic geometry of the Chern-Simons theory}},  {\em Lect. Notes Phys.} {\bf 436} (1994) 59--97.

\bibitem{Elitzur:1989nr}
S.~Elitzur, G.~W. Moore, A.~Schwimmer, and N.~Seiberg, {\it {Remarks on the Canonical Quantization of the Chern-Simons-Witten Theory}},  {\em Nucl. Phys. B} {\bf 326} (1989) 108--134.

\bibitem{Maldacena:2001kr}
J.~M. Maldacena, {\it {Eternal black holes in anti-de Sitter}},  {\em JHEP} {\bf 04} (2003) 021, [\href{http://arxiv.org/abs/hep-th/0106112}{{\tt hep-th/0106112}}].

\bibitem{Zamolodchikov:2004ce}
A.~B. Zamolodchikov, {\it {Expectation value of composite field T anti-T in two-dimensional quantum field theory}},  \href{http://arxiv.org/abs/hep-th/0401146}{{\tt hep-th/0401146}}.

\bibitem{Smirnov:2016lqw}
F.~A. Smirnov and A.~B. Zamolodchikov, {\it {On space of integrable quantum field theories}},  {\em Nucl. Phys. B} {\bf 915} (2017) 363--383, [\href{http://arxiv.org/abs/1608.05499}{{\tt arXiv:1608.05499}}].

\bibitem{Cavaglia:2016oda}
A.~Cavagli\`a, S.~Negro, I.~M. Sz\'ecs\'enyi, and R.~Tateo, {\it {$T \bar{T}$-deformed 2D Quantum Field Theories}},  {\em JHEP} {\bf 10} (2016) 112, [\href{http://arxiv.org/abs/1608.05534}{{\tt arXiv:1608.05534}}].

\bibitem{Gorbenko:2018oov}
V.~Gorbenko, E.~Silverstein, and G.~Torroba, {\it {dS/dS and $ T\overline{T} $}},  {\em JHEP} {\bf 03} (2019) 085, [\href{http://arxiv.org/abs/1811.07965}{{\tt arXiv:1811.07965}}].

\bibitem{Ivo:2024ill}
V.~Ivo, Y.-Z. Li, and J.~Maldacena, {\it {The no boundary density matrix}},  \href{http://arxiv.org/abs/2409.14218}{{\tt arXiv:2409.14218}}.

\bibitem{Goel:2020yxl}
A.~Goel, L.~V. Iliesiu, J.~Kruthoff, and Z.~Yang, {\it {Classifying boundary conditions in JT gravity: from energy-branes to $\alpha$-branes}},  {\em JHEP} {\bf 04} (2021) 069, [\href{http://arxiv.org/abs/2010.12592}{{\tt arXiv:2010.12592}}].

\end{thebibliography}\endgroup
\end{document}